\documentclass[10pt]{book}
\usepackage[english]{babel}
\usepackage[a4paper,top=3 cm,bottom=3 cm,inner=2.5 cm,outer=2.5 cm]{geometry} 
\usepackage[toc,page]{appendix}
\usepackage{graphicx}
\usepackage{titlepic}
\usepackage{bm}
\usepackage{mathtools}
\usepackage{amsmath}
\usepackage{amsfonts}
\usepackage{subcaption}
\usepackage{mathbbol}
\usepackage{enumitem}
\usepackage{cancel}
\DeclareMathOperator{\Tr}{Tr}

\setlist[itemize]{label=\textbullet}
\usepackage{float} 
\usepackage[table]{xcolor} 
\usepackage{booktabs} 
\usepackage[utf8]{inputenc} 
\usepackage{soul,color}
\usepackage{xcolor}
\usepackage{hyperref}
\usepackage {amsmath, amssymb, amsthm} 
\usepackage{eurosym}
\usepackage{tikz}
\newtheorem{lem}{Lemma} [section]
\newtheorem{prop}{Proposition} [section]
\newtheorem{defn}{Definition} [section]
\newtheorem{thm}{Theorem}[section]

\newtheorem{rem}{Remark}[section]

\theoremstyle{definition}
\newtheorem{ex}{Example}[section]
\usepackage{wasysym}
\hypersetup{pdfborder={0 0 0}} 
\usepackage{url} 
\usepackage{titlesec}

\titleformat{\paragraph}
{\normalfont\normalsize\bfseries}{\theparagraph}{1em}{}
\titlespacing*{\paragraph}
{0pt}{3.25ex plus 1ex minus .2ex}{1.5ex plus .2ex}
\usepackage{siunitx}
\usepackage{mathrsfs}

\titleformat{\subparagraph}
{\normalfont\normalsize\bfseries}{\theparagraph}{1em}{}
\titlespacing*{\paragraph}
{0pt}{3.25ex plus 1ex minus .2ex}{1.5ex plus .2ex}
\usepackage{siunitx}
\usepackage{mathrsfs}

\usepackage{leftidx}
\usepackage{braket}
\usepackage{physics}

\newcommand{\sch}{Schr{\"o}dinger }

\usepackage[square, sort&compress, comma, numbers]{natbib}
\begin{document}
 \begin{titlepage}
	\begin{center}

 		{{\Large{\textsc{UNIVERSITÄT LEIPZIG}}}}\\ 
 		\vspace{4mm}
 		{{\Large{\textsc{FAKULTÄT FÜR PHYSIK UND GEOWISSENSCHAFTEN}}}}
 		\rule[0.1cm]{15cm}{0.1mm}
 		\rule[0.5cm]{15cm}{0.1mm}\\
 		{\Large{\bf M.Sc. Mathematical Physics}}
 	\end{center}
 	\vspace{11mm}
 	\begin{center}
 	{\bf \Large{RELATIVE ENTROPY FOR}}\\
 		\vspace{3mm}
 		{\bf \Large{FERMIONIC QUANTUM FIELD THEORY}}\\

 		\vspace{4mm}
 	
 		\vspace{4mm}

 	\end{center}
 	\vspace{10mm}
 	\par
 \centerline{\bf \LARGE{Stefano Galanda} }
 	\vspace{40mm}
 	 	\par
 	\noindent
 	\begin{minipage}[t]{0.47\textwidth}
 		{\large{ \textit{Supervisor:}\\
 				\bf  Dr. Albert Much}}
 	\end{minipage}
 	\hfill
 	\begin{minipage}[t]{0.47\textwidth}\raggedleft
 		{\large{ \textit{Second Assessor:}\\
 				\bf Prof. Dr. Rainer Verch}}
 	\end{minipage}
 	\vspace{25mm}
 	\begin{center}
 		{\large{October 2022}}
 		\end{center}
\end{titlepage}

 \newpage
\thispagestyle{empty}
 \topskip0pt
\vspace*{\fill}
 \begin{center}
 		{\bf\large{Abstract}}
 \end{center}
 We study the relative entropy, in the sense of Araki, for the representation of a self-dual CAR algebra $\mathfrak{A}_{SDC}(\mathcal{H},\Gamma)$. We notice, for a specific choice of $f \in \mathcal{H}$, that the associated element in $\mathfrak{A}_{SDC}(\mathcal{H},\Gamma)$ is unitary. As a consequence, we explicitly compute the relative entropy between a quasifree state over $\mathfrak{A}_{SDC}(\mathcal{H},\Gamma)$ and an excitation of it with respect to the abovely mentioned unitary element. The generality of the approach, allows us to consider $\mathcal{H}$ as the Hilbert space of solutions of the classical Dirac equation over globally hyperbolic spacetimes, making our result, a computation of relative entropy for a Fermionic Quantum Field Theory.\\
 Our result, extends those of Longo \cite{Longo:2019mhx}, Casini et al. \cite{PhysRevD.99.125020} for the relative entropy between a quasifree state and a coherent excitation of it for a free Scalar Quantum Field Theory, to the case of fermions.\\
 As a first application, we computed such a relative entropy for a Majorana field on an ultrastatic spacetime.
 \vspace*{\fill}
 \newpage
 \thispagestyle{empty}
 \topskip0pt
\vspace*{\fill}
 \begin{center}
 		{\bf\large{Acknowledgements}}
 \end{center}
I would like to thank Dr. Albert Much for supervising this thesis, especially for his great quality of creating an amicable and professional working enviroment and for the large amount of time he dedicated to me. While completing my thesis, I was given the opportunity to deepen my understanding in a variety of fascinating topics which I greatly appreaciate and which will profoundly influence my future. I also thank Prof. Dr. Rainer Verch for his help and the discussions we had. Indeed, his guidance during the last months presents a considerable part behind the realization of this work.\\
In addition, I have benefited from discussions with Dr. Markus Fröb (on the definition of fermionic field algebras and Tomita-Takesaki modular theory), whom I want to thank for this reason.\\\\
With this work, I am closing an important chapter of my professional as well as personal life. It is for this reason, that I must thank the people that made these years special and supported me during any difficult time. Above all, and from the bottom of my heart, my family: my parents Stefania and Paolo, my elder brother Francesco and my grandparents Giuliana and Benito. Your love and the awareness of having you on my side were the biggest source of motivation.\\
Distinct words are required for Leonardo. Not only did we share this experience abroad together, but we especially share the same passion for mathematics and physics. The discussions with you are a central part of this thesis and are, in general, among the things I enjoy most. You are one of my best friends and these years together will remain unforgettable.\\
Furthermore, I want to thank Filippo for his support, especially in the most difficult period, my lifelong friends Madiara, Enes, Alessandro, Matthias and also the people that helped me just for part of these two years.\\\\
Finally, among the people I met here, I am especially grateful to Tim for his motivational support and contagious enthusiasm and to Paula for her virtue of being able to say the right words in the right moment.\\\\
I am indebted to you all.
 \vspace*{\fill}
\thispagestyle{empty}
\newpage
 \thispagestyle{empty}
 \tableofcontents
\vspace{10mm}
\thispagestyle{empty}

\newpage
\thispagestyle{empty}
\chapter*{Introduction and conventions}
\section{Introduction}
The first half of the past century was marked by the search for a fundamental theory describing nature, incorporating both the classical understanding of empty space as well as electromagnetism. During the same period, the particle-wave "duality" was one of the most debated concepts. Attempting to clarify the latter, de Broglie was the first to assume, in his work of 1926 \cite{deBroglie}, "the existence of a certain periodic phenomenon of a yet to be determined character, which is to be attributed to each and every isolated energy particle". Looking back, we know he was referring to the existence of what we call a \textit{Quantum field}. The first to introduce the notion of such a field was Dirac in \cite{1927Dirac}, who attempted to formulate a relativistically invariant theory describing a charged particle interacting with an electromagnetic field. Up to this date, enormous contributions to the formulation and understanding of the theory have led to what we nowadays call \textit{Quantum Field Theory}. However, referring to it as a \textit{theory}, however, may be slightly misleading. Instead, we should rather call it a framework, in which the "physical theories" are established. \textit{Quantum Electrodynamics (QED)}, for instance, presents one of the most precisely tested theories ever formulated and is nothing but a a specific type of Quantum field Theory obtained by incorporating electromagnetism in such a framework.\\
However, during its developement, certain mathematical problems arose (see for example \cite{Haag:1955ev}) and agreement between the predictions and experimental outcomes was prioritized at the expense of sacrificing mathematical rigor. For this reason, the subject has increasingly gained interest among both mathematicians as well as mathematical physicists. An attempt to a more mathematically rigorous formulation of QFT, in the spirit of Heisenberg matrix mechanics, was given in the pioneering work of Haag and Kastler (\cite{Haag:1963dh}). In their work, they define QFT in an axiomatic way, where the focus is on the study of the properties of observables, that's why is called \textit{Algebraic Quantum Field Theory}. One of the great innovations in their work, was the freedom regarding the choice of the underlying spacetime, allowing gravitational effects to be included in the theory. In this framework, a QFT is formulated on a spacetime where gravity is treated classically (in the sense that the evolution of the quantum field, does not have any backreaction on the spacetime geometry itself), and is hence called \textit{semiclassical gravity}.\\
The study of semiclassical gravity, using the tools of AQFT led, during the last $50$ years, to a deeper understanding of gravity itself, in particular in the context of Black Holes. In fact, in 1974, Hawking (\cite{Hawking:1974rv}) showed that despite classically, crossing the event horizon is a point of no return, Black Holes possess a grey-body spectrum of emission of radiation in the form of quantum fields. Together with the subsequent assignment of a temperature, this observation made concrete, the previously existing analogy between the laws of thermodynamics and Black Hole dynamics. \\
Once the area of a Black Hole horizon had been identified with its entropy, the interest in entropy measures in this context became greater. However, aiming for mathematical rigor, the entropy measures in QFT need to be reviewed as, for instance, the von Neumann formula is no longer well defined in this context. The reason for this mainly stems from the fact that when the number of degrees of freedom of the theory become uncountably many, certain mathematical properties of the observable algebra pertain. This led to a focus on the study of von Neumann algebras and their classification, at least from the point of view of a mathematical physicists. The first to solve the problem were Araki \cite{Araki1976} and Uhlmann \cite{Uhlmann}, who introduced a reformulated notion of relative entropy, generalizing the one by von Neumann.\\
However, such a result remains very abstract. In order to make it more concretely applicable to contexts of physical relevance, Longo \cite{Longo:2019mhx} and Casini et al. \cite{PhysRevD.99.125020} recently carried out the analysis for a particular type of unitary excitation of the vacuum of a free scalar QFT, for which the relative entropy is expressed in a form much more easily interpreted. Such a result was used, for instance, in the context of semiclassical gravity by Kurpicz, Pinamonti and Verch \cite{Kurpicz_2021} as well as by D'Angelo \cite{DAngelo:2021yat} to investigate thermodynamical properties of Black Holes that are not stationary (i.e. that incorporate a dynamically changing mass). By doing so, attempts were made of extending the thermodynamical analogy to the case of Black Holes away from equilibrium (stationary case).\\\\
This seeks to present a generalization of the mentioned result, regarding the computation of the relative entropy for the free scalar case, to the case of a fermionic quantum field theory. The text is organized in three main chapters. The first one aims at introducing the mathematical framework of AQFT and at defining Fermi quantum fields on a curved spacetime background. The second chapter is devoted to the discussion of entropy measures that, as we are going to see, need to be reviewed in the context of Quantum Field Theory due to ultraviolet divergences, arising from the type of observables algebras. In chapter two still, we we present the results for the bosonic case, allowing for the computation of the relative entropy for specific types of excitations. Lastly, in chapter three, we present our work regarding the computation of the relative entropy in the fermionic case.

\section{Conventions}
I here list the convention we adopt:
\begin{itemize}
    \item The four vectors are denoted in the abstract index notation. Namely $x^a$ denotes a vector field over the spacetime manifold and $x_a$ the corresponding covector with indices lowered and raised using the Lorentzian metric tensor $g_{ab}$ over $M$. In a specific coordinate chart over $M$ the four vectors are denoted with greek indices, e.g. $x^{\mu}$, while its spatial components either with latin indices or with bold font: $x^{\mu} = (x^0, x^i) = (x^0, \mathbf{x})$
    \item We adopt the Einstein summation convention for repeated indices
    \item The adjoint of an operator over an Hilbert space, is denoted as $^*$. While, if that operator is in particular a matrix over $\mathbb{C}^n$ with the standard inner product, we will use the more familiar notation $^{\dagger}$.
    \item We work with natural units, namely:
    \begin{equation*}
         \hbar = k_B = c = G = 1
    \end{equation*}
    \item The set of $k$-continuous functions over the spacetime $M$, are denoted by $\mathcal{C}^{k}(M)$ while the smooth functions as $\mathcal{C}^{\infty}(M)$. Analogously, a smooth $k$-continuous (resp. smooth) function over $M$ with compact support are denoted $\mathcal{C}^{k}_0(M)$ (resp. $\mathcal{C}^{\infty}_0(M)$). Finally, we denote by $\mathcal{S}(M)$ the Schwartz space of functions.
    \item We will denote by $[\cdot,\cdot]$ the commutator and by $[\cdot,\cdot]_+$ the anticommutator. 
    \item The fundamental group of a manifold $M$ is denoted by $\pi_1(M)$.
    \item A spacetime is a Lorentzian Manifold $M$ with metric $g$ of signature $(+---)$
    \item We denote the \textit{causal future/past} of a closed subset $C \subset M$ by $J^{\pm}(C)$. For a closed achronal set $S \subset M$ we denote the \textit{future/past causal developement} by $D^{\pm}(S)$. Finally the set of all \textit{causally complete regions} on $M$ is denoted as $\mathcal{K}$.
\end{itemize}
A brief review of the fundamentals of Lorentzian geometry is reported in Appendix \ref{app: Lorentz}, where we explain better the concepts mentioned in the last point.

\newpage
\chapter{Algebraic Quantum Field Theory}
Quantum Field Theory (QFT) is the theory that combines Quantum Mechanics (QM) and Special Relativity, namely it aims to incorporate the relativity principle into the theory describing the microscopic scale. This implies that the fundamental constituents of nature (the Quantum fields) as measured by different inertial observers, must transform according to the transformation connecting the two inertial observers. That is, there must exist a unitary representation of the Lorentz group, encoding the transformation properties of the quantum fields. One of the biggest achievements of the theory is its capability to give meaning to the concept of Spin: a consequence of the symmetry group of the spacetime that we encounter, once we ask, what the elementary properties are, that classify the fields and that further both observers agree on.\\
Algebraic Quantum Field Theory (AQFT) is an approach to QFT that differs from the standard one but that incorporates it as a particular case, such that we may see it as a generalization. The term \textit{Algebraic} refers to the way in which the Quantum theory is treated, in this sense we are formulating an Algebraic approach to Quantum Mechanics that allows us to incorporate special relativity with minimal effort.\\
The mathematical description of the standard approach to QM, starts with a Hilbert space $\mathcal{H}$ and the set of bounded linear operators defined over it $\mathfrak{B}(\mathcal{H})$. The subset $\mathfrak{B}_{sa}(\mathcal{H}) \subset \mathfrak{B}(\mathcal{H})$ of self adjoint operators is the set of all operators representing quantities that can be measured (as their point spectrum is real) on the system and for that reason is called the set of \textit{observables}. Finally a state is described by a so called \textit{density matrix}, i.e. an element of:
\begin{equation*}
    \mathfrak{D}(\mathcal{H}) := \{ \rho \in \mathfrak{B}_{sa}(\mathcal{H}) : \rho \geq 0 \,\, \mathrm{and} \,\, \Tr(\rho) = 1 \}
\end{equation*}
Any density matrix gives rise to a \textit{state functional} over the set of bounded operators defined as: $\omega_{\rho}: \mathfrak{B}(\mathcal{H}) \to \mathbb{C}$ via the relation:
\begin{equation*}
    \omega_{\rho} (A) := \Tr(A\rho) \hspace{15pt} \forall A \in \mathfrak{B}(\mathcal{H})
\end{equation*}
\begin{rem}
In only considering bounded operators as observables, in this first brief discussion about the motivation behind AQFT, we are not giving up generality. In fact, from the spectral theorem, it follows that we can decompose any self adjoint operator $A$ as:
\begin{equation*}
    A = \int_{\mathbb{R}} \lambda dE_{A}(\lambda)
\end{equation*}
where $E_A(\Delta)$ are the spectral measures corresponding to the measurable set $\Delta \in \mathbb{R}$ and in particular $E_A(\Delta) \in \mathfrak{B}(\mathcal{H})$ whenever $\Delta \in \mathbb{R}$ is bounded. Now, whenever we want to measure some quantity associated with a possibly unbounded operator $A$, we employ a measurement device which always has limited capacities, i.e. there are a highest and lowest values that is capable of registering. This means, that what we measure is not the unbounded operator $A$, but rather:
\begin{equation*}
    \Tilde{A} = \int_{a_{min}}^{a_{max}} \lambda dE_A(\lambda)
\end{equation*}
which is bounded.\\
Therefore, unbounded operators, that we will consider and study in what follows, are in fact just "idealized observables" and for the purpose of outlining the standard approach to QM, considering only bounded observables is perfectly fine.
\end{rem}
Given the ideas of the standard approach, two questions may arise:
\begin{itemize}
    \item[1)] How do we know that we are considering and/or able to measure all possible elements in $\mathfrak{B}_{sa}(\mathcal{H})$?
    \item[2)] How does our results depend on the choice of the initial Hilbert space and can we somehow relate the same physically measurable quantity in terms of bounded operators on two different such Hilbert spaces?
\end{itemize}
The answer to the first question is that in general we do not have access to the all set of observables. Therefore we need to investigate the consequences of knowing only a restrict subset of them. A concrete example in which we have access to just a subset of the observables, arises when introducing causality, for instance in QFT. As a consequence, we will be able to measure just what lies in our causal future, any other observable of the general theory defined over the entire spacetime is practically unaccessible.\\
For what concerns the second question, as long as we deal with finitely generated algebras, the Stone- von Neumann theorem ensures that there is a unitary transformation mapping the representation in one Hilbert space into that of another Hilbert space whenever the representations are irreducible. However, the same theorem tells us that when the algebras have infinitely many generators, there are unitarily inequivalent representations and thus the choice of the initial Hilbert space might lead to unwanted consequences, like preventing us from describing any physically relevant state of the system as a vector in that single Hilbert space. It is worth noting that, as we are interested in QFT, the algebras in this context, as we will see, have infinitely many generators: the solutions of the field equations.\\
AQFT solves both these two issues as it takes as primary object the algebra of "accessible" observables and as states the positive, normal functionals defined over it. In this approach, the notion of Hilbert space is derived via the so called \textit{GNS construction} (\textit{Gelfand-Naimark-Segal}) once a state is chosen. In this way, one is able to treat all states equally: just as functionals.\\\\
The chapter starts off with a section that aims at inroducing all the necessary mathematical theorems and definitions required for the formulation of AQFT. We will continue on in the second section with an axiomatic definition of AQFT. Finally, we will discuss an example of an AQFT presenting the Dirac field.

\section{Preliminary definitions and results}
As the algebraic approach starts off by specifying as primary data the Algebras of observables (i.e. those sets of opertors that are relevant for the study of the physical system), the first thing we will introduce is the notion of algebras of operators. Then, in order to specify a configuration of a system, we will introduce the notion of functionals over the Algebra and correspondingly the notion of a state. Finally, combining these two we will derive the Hilbert space and thus the standard formulation of QFT.
\subsection{Operator algebras}
The first core concept is that of an algebra, which is nothing but a vector space with an additional structure:
\begin{defn}
An \textit{algebra} $\mathfrak{A}$ is a vector space $V$ on a field $\mathbb{K}$ (if $\mathbb{K} = \mathbb{R}, \mathbb{C}$ we speak of a \textit{real} resp. \textit{complex} algebra) on which a product operation is defined $\circ: V \times V \to V$ which is associative:
\begin{equation*}
    A \circ (B \circ C) = (A \circ B) \circ C \hspace{15pt} \forall A,B,C \in \mathfrak{A}
\end{equation*}
and fulfills distributivity over the vector space structures. 
\end{defn}
\begin{ex}
The usual example is that of $\mathbb{R}^3$ with respect to the vector product $\times: \mathbb{R}^3 \times \mathbb{R}^3 \to \mathbb{R}^3$, defined for any $\mathbf{v}, \mathbf{w} \in \mathbb{R}^3$ as:
\begin{equation*}
    \mathbf{v} \times \mathbf{w} = \det \begin{pmatrix}
    \hat{x} & \hat{y} & \hat{z}\\
    v_x & v_y & v_z \\
    w_x & w_y & w_z
    \end{pmatrix}
\end{equation*}
Where we have denoted with $x,y,z$ the cartesian directions and $\hat{x}, \hat{y}, \hat{z}$ the corresponding versors.
\end{ex}
\begin{ex}
The set of bounded linear operators $\mathfrak{B}(\mathcal{H})$ over an Hilbert space $\mathcal{H}$ is an algebra with respect to composition. The simplest example holds when $\dim \mathcal{H} = n < \infty$, then the bounded operators $\mathfrak{B}(\mathcal{H}) = \mathrm{Mat}(n,\mathbb{C})$ and the composition simply becomes the matrix multiplication.
\end{ex}
Algebras that contain the neutral element of the multiplication operation, are called \textit{unital algebras}. Namely, by calling $\mathfrak{A}$ the unital algebra, there exist an element $\mathbb{1}\in \mathfrak{A}$ such that:
\begin{equation*}
    \mathbb{1} \cdot A = A \cdot \mathbb{1} = A
\end{equation*}
for any $A \in \mathfrak{A}$. As far as we are concerned, all algebras that we are going to consider in the following will be unital.\\
It is possible to introduce further structures on an algebra, one of which is the $*$-operation:
\begin{defn}
Let $\mathfrak{A}$ be a complex algebra. An antilinear map:
\begin{align*}
    *: \mathfrak{A} &\to \mathfrak{A}\\
    A &\mapsto A^*
\end{align*}
is called an \textit{involution} if $A = (A^*)^*$ and $(AB)^*=B^*A^*$ for all $A,B \in \mathfrak{A}$. An algebra with an involution is called a $*$\textit{-algebra}.
\end{defn}
We can also introduce a topology over an algebra and, taking advantage of the vector space structure, this can be done by introducing a norm:
\begin{defn}
Let $\mathfrak{A}$ be a complex algebra. We can define a norm:
\begin{align*}
    \|\cdot\| : \mathfrak{A} &\to \mathbb{R}^+\\
    A &\mapsto \|A\|
\end{align*}
With the usual properties:
\begin{enumerate}
    \item $\|A\| = 0$ if and only if $A = 0$
    \item $\|k A \| = |k| \, \|A\|$
    \item $\|A + B\| \leq \|A\| + \|B\|$
    \item $\|AB\| \leq \|A\| \, \|B\|$
\end{enumerate}
for any $A,B \in \mathfrak{A}$ and all $k \in \mathbb{C}$. An algebra with a norm is called a normed algebra.
\end{defn}
\begin{defn}
A normed algebra $(\mathfrak{A}, \|\cdot\|)$ that is complete in the topology induced by $\|\cdot\|$, is called a Banach algebra
\end{defn}
The next presents the type of algebra, we will be interested in the most:
\begin{defn}
Let $(\mathfrak{A}, \|\cdot\|)$ be a $*$-Banach algebra. $(\mathfrak{A}, \|\cdot\|)$ is called a $C^*$-algebra if the norm satisfies:
\begin{equation*}
    \|A A^*\| = \|A\|^2
\end{equation*}
\end{defn}
\begin{ex}
Consider the algebra of bounded operators over a Hilbert space $\mathfrak{B}(\mathcal{H})$. We can make it a $C^*$-algebra by considering the operator norm and as involution the adjoint with respect to the inner product on $\mathcal{H}$. The corresponding induced topology will be called \textit{operator topology} or \textit{uniform topology}. The operator norm is defined as follows:
\begin{equation*}
    \|A\|_{op} := \sup_{v \in \mathcal{H}} \frac{\|A v\|}{\|v\|}
\end{equation*}
for all $A \in \mathfrak{B}(\mathcal{H})$, and the norms on the right hand side are those induced by the scalar product on $\mathcal{H}$. We can always write a vector $v \in \mathcal{H}$ as:
\begin{equation*}
    v = \|v\| \frac{v}{\|v\|}
\end{equation*}
And clearly $\hat{v} := v/\|v\|$ is a normalized vector. For this reason we can rewrite:
\begin{equation*}
    \|A\|_{op} = \sup_{v \in \mathcal{H}\, , \, \|v\|=1} \|A v\|
\end{equation*}
The operator norm fulfills all the properties we listed before such that $(\mathfrak{B}(\mathcal{H}), \|\cdot \|_{op})$ is a normed algebra:
\begin{enumerate}
    \item $\|A\|_{op} = 0 \Leftrightarrow Av=0 \,\,\,\, \forall v \in \mathcal{H}$ but this implies that $A = 0$
    \item It easily follows from the properties of the norm over the Hilbert space:
    \begin{align*}
        \|k A\|_{op} &= \sup_{v \in \mathcal{H}, \|v\| = 1} \|k A v\|\\
        &= \sup_{v \in \mathcal{H}, \|v\| = 1} |k| \, \|A v\|\\
        &= |k| \, \|A\|_{op}
    \end{align*}
    \item Again, one has from the Hilbert space norm:
    \begin{align*}
        \|A + B\|_{op} &= \sup_{v \in \mathcal{H}, \|v\| = 1} \|(A+B) v\|\\
        &\leq \sup_{v \in \mathcal{H}, \|v\| = 1} (\|A v\| + \|B v\|)\\
        &= \|A\|_{op} + \|B\|_{op}
    \end{align*}
    \item To prove the last property, notice that from the definition one has:
    \begin{equation*}
        \|Av\| \leq \|A\|_{op} \|v\|
    \end{equation*}
    Then, if we look at:
    \begin{align*}
        \|AB\|_{op} &= \sup_{v \in \mathcal{H}, \|v\| = 1} \|ABv\|\\
        &\leq \sup_{v \in \mathcal{H}, \|v\| = 1} \|A\|_{op}\|Bv\|\\
        &= \|A\|_{op} \sup_{v \in \mathcal{H}, \|v\| = 1} \|Bv\| = \|A\|_{op} \|B\|_{op}
    \end{align*}
\end{enumerate}
$(\mathfrak{B}(\mathcal{H}), \|\cdot\|_{op})$ is also a Banach algebra as can be derived from the general result:
\begin{prop}
Let $\mathcal{X}, \mathcal{Y}$ be normed vector spaces. If $\mathcal{Y}$ is a Banach space, then $(\mathfrak{B}(\mathcal{X},\mathcal{Y}), \|\cdot\|_{op})$ is a Banach algebra
\end{prop}
\begin{proof}
Consider a Cauchy sequence $\{ A_n \} \subset \mathfrak{B}(\mathcal{X},\mathcal{Y})$ for the operator norm. By definition of an operator norm we have:
\begin{equation*}
    \|A_n v - A_m v\| \leq \|A_n - A_m\|_{op} \, \|v\|
\end{equation*}
As $\{ A_n \}$ is a Cauchy sequence, $\{A_n v\}$ must be as well. But we know that $\mathcal{Y}$ is complete so there is a vector in $\mathcal{Y}$ such that:
\begin{equation*}
    A v := \lim_{n \to \infty} A_n v
\end{equation*}
As all the $A_n$ are linear operators, so must be the operator $A$.\\
Now, as $\{ A_n \}$ is Cauchy, for any choice of $\varepsilon$ if we take $m,n$ sufficiently large we have:
\begin{equation*}
    \|A_m - A_n\|_{op} < \varepsilon
\end{equation*}
Then we can estimate:
\begin{equation*}
    \|A_n v - A_m v\| \leq \|A_n - A_m\|_{op} \|v\| \leq \varepsilon \|v\|
\end{equation*}
As we are taking the norm over $\mathcal{Y}$:
\begin{equation*}
    \|A v - A_m v\| = \lim_{n \to \infty} \|A_n v - A_m v\| \leq \varepsilon \| v\|
\end{equation*}
where in the inequality, we implicitly assume that $m$ is sufficiently large. Now since $\|Av\| \leq \|Av - A_nv\| + \|A_n v\|$, we have:
\begin{equation*}
    \|A v\| \leq (\varepsilon + \|A_n\|_{op}) \|v\|
\end{equation*}
The above proves that $A$ is bounded. Finally we need to show that as $n \to \infty$, $\|A - A_n\|_{op} \to 0$. To see this, just notice that from the above:
\begin{equation*}
    \|A - A_n\|_{op} = \sup_{v \in \mathcal{H}, \|v\| = 1} \|A v - A_n v\| \leq \varepsilon
\end{equation*}
The result follows from the freedom of choice of $\varepsilon$ for sufficiently large $n$.
\end{proof} 
Since in our case, we are considering bounded operators over a Hilbert space, the fact that $\mathfrak{B}(\mathcal{H})$ is a Banach algebra, is an immediate consequence of this proposition.\\\\
Finally, let us show that the $C^*$-algebra condition holds. First of all let us show that $\|A^*\|_{op} = \|A\|_{op}$. We have for $v,x \in \mathcal{H}$:
\begin{equation*}
    \braket{A^* v}{x} = \braket{v}{A x} \leq \|v\| \, \|A\|_{op} \|x\|
\end{equation*}
But, if we choose $x = A^*v$ we get:
\begin{equation*}
    \|A^*v\|^2 \leq \|v\| \, \|A\|_{op} \|A^*v\|
\end{equation*}
This implies: $\|A^*\|_{op} \leq \|A\|_{op}$. Now, as the adjoint is involutive:
\begin{equation*}
    \braket{v}{A^* x} = \overline{\braket{A^* x}{v}} = \overline{\braket{x}{A v}} = \braket{Av}{x}
\end{equation*}
We also have: $\|A\|_{op} =\|(A^*)^*\|_{op} \leq \|A^*\|_{op}$ which implies:
\begin{equation*}
    \|A^*\|_{op} = \|A\|_{op}
\end{equation*}
Further, we have from the multiplicative property of the operator norm, proven above: $\|A^* A\|_{op} \leq \|A^*\|_{op} \|A\|_{op} = \|A\|_{op}^2$. For the converse inequality, observe that:
\begin{align*}
    \|A\|_{op}^2 &= \bigg( \sup_{v \in \mathcal{H}, \|v\| = 1} \|Av\| \bigg)^2\\
    &= \sup_{v \in \mathcal{H}, \|v\| = 1} \|A v\|^2\\
    &= \sup_{v \in \mathcal{H}, \|v\| = 1} \braket{Av}{Av}\\
    &= \sup_{v \in \mathcal{H}, \|v\| = 1} \braket{v}{A^*Av}\\
    &\leq \sup_{v \in \mathcal{H}, \|v\| = 1} \|A^*A v\| = \|A^*A\|_{op}
\end{align*}
Therefore, $\|A^*A\|_{op} = \|A\|_{op}^2$ and we have proven that $(\mathfrak{B}(\mathcal{H}), \|\cdot\|_{op})$ is a $C^*$-algebra.
\end{ex}
Another important concept is that of a $C^*$-subalgebra:
\begin{defn}
Let $\mathfrak{A}$ be a $C^*$-algebra, and let $\mathfrak{B} \subset \mathfrak{A}$ be a $*$-subalgebra in the sense that it is closed under the product and involution of its elements and let it be closed in the topology induced by the norm. Then, we say that $\mathfrak{B}$ is a $C^*$-subalgebra of $\mathfrak{A}$, and clearly, it is a $C^*$-algebra on its own.
\end{defn}
We now introduce specific type of $C^*$-algebras that are going to be widely used and mentioned in what follows.
\vspace{5mm}
\subsubsection{CCR $*$-algebra and Weyl $C^*$-algebra}\label{sec: Weyl}
We start by abstractly defining what a $CCR$ algebra is, the motivation for it will be explained in an incoming example:

\begin{defn}
Let $\mathbf{P}$ be a symplectic space with symplectic form $\sigma(\cdot, \cdot)$. A CCR (Canonical Commutation Relation) algebra $\mathfrak{A}_{CCR}(\mathbf{P}, \sigma)$, is the quotient of the $*$-algebra, generated by $A^*(f)$, $A(f)$ and the identity $\mathbb{1}$ (where $f \in \mathbf{P}$), by the following relations:
\begin{align*}
    (A^*(f))^* &= A(f)\\
    [A(f), A(g)] &= i \sigma(f,g) \mathbb{1}\\
    A(c_1 f_1 + c_2 f_2) &= c_1 A(f_1) + c_2 A(f_2)\\
    A^*(h) &= A(\overline{h})
\end{align*}
for $c_1, c_2 \in \mathbb{C}$ and $\overline{h}$ denotes the complex conjugate of $h \in \mathbf{P}$.
\end{defn}
An example of such an abstract algebra, is the algebra of free scalar field operators, i.e. the operator valued distributions, as the fundamental solutions of the Klein-Gordon equation, smeared with test functions.\\
Unfortunately, on the $CCR$ $*$-algebra one cannot define a $C^*$-norm. This is due to the following Proposition:
\begin{prop}\label{prop: CCRno}
There are no self-adjoint operators $q$ and $p$ such that, on a common invariant subspace, $[q,p] = i c \mathbb{1}$ (for $c \in \mathbb{C}$) and at the same time $q$ and $p$ are bounded.
\end{prop}
\begin{proof}
Suppose that $[q,p] = i c \mathbb{1}$ on a common invariant space $D$ where both $q$ and $p$ are bounded. Restrict to the closure $\overline{D}$, extending $q$, $p$ as self adjoint operators over it, and consider it as the Hilbert space. From the commutation relations we get:
\begin{equation*}
    [p, q^n] = -i c (n q^{n-1})
\end{equation*}
Using that, as a self-adjoint, bounded operator, $q^k = (q^k)^*$ for any $k \in \mathbb{N}$ and assuming $n$ odd:
\begin{align*}
    n \| q \|_{op}^{n-1} &= n \| q^{n-1} \|_{op}\\
    &\leq 2 \| p\|_{op} \| q^n \|_{op}\\
    &\leq 2 \| p \|_{op} \|q \|_{op} \| q^{n-1} \|_{op} \\
    &= 2 \| p \|_{op} \|q \|_{op} \| q \|_{op}^{n-1}
\end{align*}
Noticing now that, substituting $q' = q + c_1 \mathbb{1}$ for $c_1 \in \mathbb{C}$, will not alter the assumptions of the proposition, we always have $\| q \|_{op} \neq 0$. Therefore, we have obtained for any $n = 1,3, \dots$:
\begin{equation*}
    n \leq 2 \| p \|_{op} \|q \|_{op} < +\infty
\end{equation*}
that contraddicts the boundedness.
\end{proof}

So, the elements of a finitely generated $CCR$ algebra do not admit representations as buonded operators, forbidding the definition of a norm. As a consequence, we cannot even define a norm over the abstract $*$-algebra. Finally, being a general $CCR$ algebra the union of the finitely generated ones, this prevents the definition of a norm in general.\\
The $CCR$ not being a $C^*$-algebra leads inevitably, as we will state when introducing the \textit{GNS theorem}, to a representation of its elements as unbounded operators over a Hilbert space. As a consequence, domain issues in defining products of the represented elements in the algebra may arise. Therefore, in order to avoid them, the idea is to "turn" the $CCR$ into a $C^*$-algebra, for which the same \textit{GNS theorem} ensures that the elements are represented as bounded operators. For this reason, let us define the so called \textit{Weyl $C^*$-Algebra} that are a specific type of $C^*$-algebras that we can construct once over the same symplectic space:

\begin{defn}
Let $(\mathbf{P},\sigma)$ be a non-trivial real symplectic space, with $\sigma: \mathbf{P} \times \mathbf{P} \to \mathbb{R}$ non degenerate. The Weyl $C^*$-algebra $C\mathcal{W}(\mathbf{P},\sigma)$ associated to $(\mathbf{P},\sigma)$ is a $C^*$-algebra generated by elements $W(f)$ with $f \in \mathbf{P}$ that satisfy the following relations:
\begin{equation*}
    W(f)W(f') = e^{-\frac{i}{2}\sigma(f,f')}W(f + f') \hspace{20pt} W(f)^* = W(-f)
\end{equation*}
\end{defn}
From the given definition however, is just clear that $C\mathcal{W}(\mathbf{P},\sigma)$ is a $*$-algebra. We still need to argue that any Weyl $*$-algebra admits a norm with the $C^*$-property. Let me here only quote the procedure and refer to Remark $11.49$ number $\mathbf{(5)}$ in \cite{Moretti:2013cma} for a rigorous explanation. For any two Weyl $*$-algebras on the same underlying symplectic space, one can show that there always exist a unique $*$-isomorphism between them. Then, one shows that given $(\mathbf{P},\sigma)$  one can always construct a Weyl $*$-algebra of bounded operators $\mathfrak{B}(\mathcal{H})$ over a Hilbert space $\mathcal{H} = L^2(\mathbf{P},\mu)$, where $\mu$ is the counting measure defined on $(\mathbf{P},\Sigma)$ with $\Sigma$ the power set of $\mathbf{P}$ that we take as our sigma algebra. In particular, for any $A \in \Sigma$ the counting measure is:
\begin{equation*}
    \mu(A) := \left\{ \begin{aligned}
    &\sharp A \hspace{20pt} \mathrm{if} \,\, A \,\, \mathrm{is} \,\, \mathrm{finite}\\
    &+\infty \hspace{20pt} \mathrm{if} \,\, A \,\, \mathrm{is} \,\, \mathrm{infinite}
    \end{aligned}\right.
\end{equation*}
Where $\sharp \cdot$ denotes the cardinality of the set. However, since we know that any $(\mathfrak{B}(\mathcal{H}), \| \cdot \|_{op})$ is a $C^*$-algebra, we can define a $C^*$-norm on any other Weyl algebra on $(\mathbf{P},\sigma)$ starting from $\| \cdot \|_{op}$ using the existing $*$-isomorphism between any Weyl algebra on the same underlying symplectic space.\\\\
Finally, to give some motivation and make these abstract definitions more concrete we present an example, that we are familiar with, from standard Quantum Mechanics.
\begin{ex}[\textbf{CCR in QM}]
In \sch mechanics, the main operators are position and momentum. The \textit{Canonical Commutation Relation algebra (CCR)} for finitely many degrees of freedom is defined as a $*$-algebra generated by self-adjoint $q_1,\dots,q_n,p_1,\dots, p_n$, with the relations:
\begin{align*}
    [q_j,p_k] = i\delta_{jk}\mathbb{1} &\hspace{15pt} [q_j,q_k] = [p_j,p_k] = 0\\
    q_k = q_k^* \,\,\,\, &\mathrm{and} \,\,\,\, p_k = p_k^*
\end{align*}
Such an abstract $*$-algebra does not admit a representation in terms of bounded operators over a Hilbert space, as we have seen in Prop. \ref{prop: CCRno}, preventing the definition of a norm. The canonical way to bypass this issue is to transform such self-adjoint operators into unitaries, namely defining a Weyl algebra for it.\\
In this case, we have an underlying symplectic space stemmin from the classical theory that describes the positions and momenta of our collection of classical pointlike particles: $(\mathbb{R}^{2n},\sigma)$. Therefore, $\mathbf{P} = \mathbb{R}^{2n} = \{ (Q,P)| Q,P \in \mathbb{R}^n \}$ is the classical phase space and (denoting with $\cdot$ the standard inner product):
\begin{equation*}
    \sigma((Q_1,P_1),(Q_2,P_2)) := Q_1 \cdot P_2 - Q_2 \cdot P_1
\end{equation*}
This expression is compatible with the Poisson bracket as follows: given $f= (Q,P) \in \mathbf{P}$, the map $\sigma_f: h \mapsto \sigma(f,h)$ is a linear observable and, taking advantage of the non-degeneracy of the symplectic form, we can write any linear map $\mathbf{P}\to \mathbb{R}$ in this way. For instance, we can express the position and momentum operators $q_i(Q,P) := Q_i$ and $p_i(Q,P) := P_i$ as:
\begin{equation}\label{eq: sympos}
    q_i = \sigma_{(0,-e_i)} \hspace{20pt} p_i = \sigma_{(e_i,0)}
\end{equation}
where $e_i$ is the $i$-th canonical vector in $\mathbb{R}^n$. The compatibility with the Poisson bracket then follows:
\begin{align*}
  \{ \sigma_{(Q_1,P_1)}, \sigma_{(Q_2,P_2)} \}_{P.B.} &= \frac{\partial \sigma_{(Q_1,P_1)}}{\partial Q} \cdot \frac{\partial \sigma_{(Q_2,P_2)}}{\partial P} - \frac{\partial \sigma_{(Q_1,P_1)}}{\partial P} \cdot \frac{\partial \sigma_{(Q_2,P_2)}}{\partial Q}\\
  &= (-P_1) \cdot Q_2 - (-P_2) \cdot Q_1\\
  &= \sigma((Q_1,P_1), (Q_2,P_2))
\end{align*}
We can rewrite the CCR $*$-algebra relations using Eq. \eqref{eq: sympos} noticing that:
\begin{equation*}
    \sigma((0,-e_j),(0,-e_k)) = \sigma((e_j,0),(e_k,0)) = 0 \hspace{15pt} \sigma((0,-e_j), (e_k,0)) = \delta_{jk}
\end{equation*}
Then, by quantizing the classical theory:
\begin{align*}
    [\sigma_{(0,-e_j)},\sigma_{(e_k,0)}] = i\delta_{jk}\mathbb{1} \hspace{15pt} &[\sigma_{(0,-e_j)},\sigma_{(0,-e_k)}] = [\sigma_{(e_j,0)},\sigma_{(e_k,0)}] = 0\\
    \sigma_{(0,-e_k)} = \sigma_{(0,-e_k)}^* \,\,\,\,\,\, &\mathrm{and} \,\,\,\,\,\, \sigma_{(e_k,0)} = \sigma_{(e_k,0)}^*
\end{align*}
where we have used the fact that $e_i \in \mathbb{R}^n$. In this formulation, for any $f \in \mathbf{P}$ (e.g. $f = (0,-e_i)$), we define the associated Weyl algebra as generated by elements $e^{i\sigma_f}$, checked to fulfill:
\begin{equation*}
    e^{i \sigma_f} e^{i \sigma_h} = e^{-\frac{i}{2}\sigma(f,h)}e^{i\sigma_{f+h}} \hspace{20pt} (e^{i\sigma_f})^* = e^{i\sigma_{-f}}
\end{equation*}\\
\end{ex}

\subsubsection{Self-dual CAR algebra}\label{sec: SDCAR}
Similarly with the way we defined an abstract $CCR$ algebra, we can define an abstract $*$-algebra that has \textit{Canonical Anticommutation Relations (CAR)} implemented in it as done in \cite{Araki1968}. In this case, the underlying structure is that of a Hilbert space:

\begin{defn}\label{def: self}
Let $\mathcal{H}$ be a complex Hilbert space with inner product $(\cdot, \cdot)_{\mathcal{H}}$, and let $\Gamma$ be an antiunitary operator satisfying:
\begin{align*}
    \Gamma^2 = \mathbb{1} &\hspace{20pt} \Gamma i = -i \Gamma\\
    (\Gamma f_1, \Gamma f_2)_{\mathcal{H}} &= (f_2, f_1)_{\mathcal{H}}
\end{align*}
A selfdual CAR algebra $\mathfrak{A}_{SDC}(\mathcal{H}, \Gamma)$ is the quotient of the $*$-algebra, generated by $B^*(f)$, $B(f)$ and the identity $\mathbb{1}$ (where $f \in \mathcal{H}$), by the following relations:
\begin{align*}
    (B^*(f))^* &= B(f)\\
    [B^*(f), B(g)]_+ &= (g,f)_{\mathcal{H}} \mathbb{1}\\
    B(c_1 f_1 + c_2 f_2) &= c_1 B(f_1) + c_2 B(f_2)\\
    B^*(h) &= B(\Gamma h)
\end{align*}
for $c_1, c_2 \in \mathbb{C}$.
\end{defn}
Contrary to the $CCR$ case, one can define a $C^*$-norm over a self-dual $CAR$ algebra, making it a $C^*$-algebra after taking the completion of the previous $*$-algebra with respect to this norm. Let me just quote the main argument, and refer to Prop. $12.50$ in \cite{Derezinski:2013dra} for a rigorous proof. Note first of all that $\mathfrak{A}_{SDC}(\mathcal{H}, \Gamma)$ is an infinite dimensional Clifford algebra, due to the anticommutation relations. Moreover, each finite dimensional Clifford algebra admits an injective reppresentation, as an algebra of bounded operators over some Hilbert space, see \cite{Lawson} Section $4$. Then, as bounded operators, they admit a $C^*$-norm. Follows, that each finite dimensional Clifford algebra admits a $C^*$-norm as well. Finally, as $\mathfrak{A}_{SDC}(\mathcal{H}, \Gamma)$ is the completion of the union of these finite dimensional $C^*$-subalgebras, we must have that $\mathfrak{A}_{SDC}(\mathcal{H}, \Gamma)$ is also a $C^*$-algebra.\\
As a consequence of this, one obtains that the corresponding operators in $\mathfrak{A}_{SDC}(\mathcal{H}, \Gamma)$, are reppresented, by the \textit{GNS theorem}, as bounded operators on a Hilbert space.\\

\vspace{3mm}
In what follows, $\mathfrak{A}_{(\mathbf{K},s)}$ will denote an algebra built on $(\mathbf{K},s)$ with $s$ a bilinear form on $\mathbf{K}$, i.e. the elements in $\mathfrak{A}_{(\mathbf{K},s)}$ are maps $(\mathbf{K},s) \ni f \to T(f) \in \mathfrak{A}_{(\mathbf{K},s)}$. In particular, if $(\mathbf{K},s)$ is a symplectic space, $\mathfrak{A}_{(\mathbf{K},s)}$ denotes either a \textit{Weyl algebra} or the $CCR$ algebra. While, if $(\mathbf{K},s)$ is a Hilbert space, $\mathfrak{A}_{(\mathbf{K},s)}$ denotes a self-dual CAR algebra.
\vspace{5mm}
\subsubsection{Von Neumann algebras}
The last important type of $C^*$-algebra that we will study are the so called \textit{von Neumann algebras} (named after John von Neumann) that are the types of algebras we will be interested in, once we represent the abstract algebras on a concrete Hilbert space. These are unital algebras that we will encounter when studying the Tomita-Takesaki modular theory in the next chapter.\\
The first thing we need to introduce, is the notion of commutant algebra.
\begin{defn}
Let $\mathfrak{M} \subset \mathfrak{B}(\mathcal{H})$ be a subset of the bounded operators over the Hilbert space $\mathcal{H}$. The commutant algebra of $\mathfrak{M}$ is defined:
\begin{equation*}
    \mathfrak{M}' := \{ T \in \mathfrak{B}(\mathcal{H}): TA - AT = 0 \,\,\,\, \mathrm{for} \,\,\,\, \mathrm{any}  \,\,\,\,A \in \mathfrak{M}  \}
\end{equation*}
\end{defn}
\begin{rem}
If $\mathfrak{M}$ is closed under hermitian conjugation, then $\mathfrak{M}'$ is a $*$-algebra with unit. One can also convince himself that $\mathfrak{M}_1 \subset \mathfrak{M}_2$ then $\mathfrak{M}'_2 \subset \mathfrak{M}'_1$ and that $\mathfrak{M} \subset (\mathfrak{M}')'$. This implies that one cannot reach more than the first order of commutant as $\mathfrak{M}' = ((\mathfrak{M}')')' := (\mathfrak{M}')''$. This is true as, from the second statement applied on $\mathfrak{M}'$, we get:
\begin{equation*}
    \mathfrak{M}' \subset (\mathfrak{M}')''
\end{equation*}
but at the same time since $\mathfrak{M} \subset (\mathfrak{M}')'$ we must have, from the first statement above, that for the commutant algebras:
\begin{equation*}
    ((\mathfrak{M}')')' \subset \mathfrak{M}'
\end{equation*}
\end{rem}
The next theorem proves that the second commutant of $\mathfrak{M} \subset \mathfrak{B}(\mathcal{H})$ corresponds to its closure in the weak topology and, as we will see, this is required for defining von Neumann algebras:
\begin{thm}[\textbf{Double commutant theorem}]\label{thm: DouComm}
If $\mathcal{H}$ is a complex Hilbert space and $\mathfrak{A}$ is a unital $*$-subalgebra of $\mathfrak{B}(\mathcal{H})$ the following statements are equivalent:
\begin{itemize}
    \item[i)] $\mathfrak{A}'' = \mathfrak{A}$
    \item[ii)] $\mathfrak{A}$ is closed in the weak topology\footnote{A sequence of operators $T_i \in \mathfrak{B}(\mathcal{H})$ converges in the weak operator topology to $T$ if for any $x \in \mathcal{H}$ and $y \in \mathcal{H}^*$ (the dual space of $\mathcal{H}$ with respect to the inner product): $\braket{y}{(T_i - T) x} \to 0$}
    \item[iii)] $\mathfrak{A}$ is closed in the strong topology\footnote{A sequence of operators $T_i \in \mathfrak{B}(\mathcal{H})$ converges in the strong operator topology to $T$ if for any $x \in \mathcal{H}$: $\| T_i x - Tx\| \to 0$}
\end{itemize}
Then:
\begin{equation*}
    \mathfrak{A}'' = \overline{\mathfrak{A}}^s = \overline{\mathfrak{A}}^w
\end{equation*}
Where $\overline{(\cdot)}^s$ denotes the closure w.r.t. the strong topology and $\overline{(\cdot)}^w$ the closure in the weak topology.
\end{thm}
\begin{proof}
See Theorem $3.88$ in \cite{Moretti:2013cma}
\end{proof}
\begin{rem}
Notice that if  $\mathfrak{A}$, as introduced above, is closed in the strong topology it must be closed also in the operator (uniform) topology. Indeed, if $\| T_i - T\|_{op} \to 0$ in the operator topology, it must also converge in the strong one as:
\begin{equation*}
    \| T_i x - T x\| \leq \|T_i - T\|_{op} \|x\| 
\end{equation*}
The converse is not true since, by taking as operator $\mathbb{1}_N \in \mathfrak{B}(l^2(\mathbb{N}))$ defined via:
\begin{equation*}
    \mathbb{1}_N = \sum_{n=0}^{N} \ketbra{n}{n}
\end{equation*}
its action on any $x \in l^2(\mathbb{N})$ is:
\begin{equation*}
    \mathbb{1}_N = \sum_{n=0}^{N} \ket{n}\braket{n}{x} = \sum_{n=0}^{N} \ket{n}x_n
\end{equation*}
As $\ket{n} \in l^2(\mathbb{N})$ is nothing but a sequence of zeros except a single $1$ at the $n$-th place. But then in the strong operator topology, assuming w.l.o.g that $M > N$, we have that for all $x \in l^2(\mathbb{N})$:
\begin{equation*}
    \|\mathbb{1}_N x - \mathbb{1}_M x \| = \| \sum_{n=N+1}^{M} \ket{n}x_n \| \to 0
\end{equation*}
As, by assumption $x \in l^2(\mathbb{N})$, which means that only a finite number of $x_n$ must be non-zero, so by taking $N$ sufficiently large we can ascertain that no other $x_n$ for $n \geq N$ are nonvanishing. That means that $\mathbb{1}_N \to \mathbb{1}$ in the strong topology. On the other hand:
\begin{equation*}
    \| \mathbb{1}_N - \mathbb{1}_{N+1}\|_{op} = \|\ket{N+1}\braket{N+1}{N+1}\| = 1
\end{equation*}
proves that this sequence is not Cauchy, i.e. in particular $\mathbb{1}_N$ cannot converge in the operator topology.\\
So, we have proven:
\begin{equation*}
    \mathrm{uniform} \Rightarrow \mathrm{strong}
\end{equation*}
But then, if we take the closure of a set $C$ in the uniform topology, this closure must itself be contained in the closure of the strong topology of that set, as the convergence in the uniform topology of each sequence implies that all the accumulation points must be contained also in the closure w.r.t. the strong topology:
\begin{equation*}
    \overline{C}^{op} \subseteq \overline{C}^s
\end{equation*}
showing that if a set is closed in the strong topology sense, it must be closed in the uniform topology as well.\\\\
The main implication of this result is that, if $\mathfrak{A}''$ is $*$-closed, then it is a $C^*$-subalgebra of $\mathfrak{B}(\mathcal{H})$, as it is also closed in the operator topology.
\end{rem}
\begin{rem}
$\mathfrak{M}'$ is always closed in the uniform topology. This can be seen by noticing that $\mathfrak{M}' = (\mathfrak{M}')''$ and as such it is closed in the strong topology from the above theorem.
\end{rem}
Now we have introduced all the necessary formal notions to define a von Neumann algebra:
\begin{defn}
Let $\mathcal{H}$ be a complex Hilbert space. A von Neumann algebra over $\mathcal{H}$ is a unital $*$-subalgebra of $\mathfrak{B}(\mathcal{H})$ that fulfills one of the three equivalent statements $i);ii);iii)$ of the double commutant theorem.\\
Given a von Neumann algebra $\mathfrak{M}$, its center is the subset $\mathfrak{M} \cap \mathfrak{M}'$.
\end{defn}
We see then, from the above remarks, that any von Neumann algebra is in particular a $C^*$-subalgebra of $(\mathfrak{B}(\mathcal{H}), \| \cdot \|_{op})$.\\
Consequently, in general, we can obtain a von Neumann algebra whenever we are given $\mathfrak{M} \subset \mathfrak{B}(\mathcal{H})$ that is closed under Hermitian conjugation, by simply taking its double commutant $\mathfrak{M}''$. It turns out that this is the smallest von Neumann algebra that one can obtain starting from $\mathfrak{M}$ as, whenever we take $\mathfrak{M} \subset \mathfrak{A}$ with $\mathfrak{A}$ a von Neumann algebra containing $\mathfrak{M}$, we have $\mathfrak{A}' \subset \mathfrak{M}'$. But then: $\mathfrak{M}'' \subset \mathfrak{A}'' = \mathfrak{A}$ showing that $\mathfrak{M}''$ is the smallest von Neumann algebra containing $\mathfrak{M}$.

\subsection{State functionals}
As mentioned at the beginning, now that we have introduced the concept of what an algebra of operators is, we need to define the notion of a state if we want to specify the configuration of the system. These are linear functionals over the algebra $\mathfrak{A}$, namely $\omega$ is a linear functional over $\mathfrak{A}$ if $\omega: \mathfrak{A} \to \mathbb{C}$ and:
\begin{equation*}
    \omega(\alpha A + \beta B) = \alpha \omega(A) + \beta \omega(B)
\end{equation*}
For $\alpha, \beta \in \mathbb{C}$ and any $A,B \in \mathfrak{A}$. The space of all linear functionals corresponds to the dual space $\mathfrak{A}^*$ of $\mathfrak{A}$.\\
Let us start by introducing some additional structure over this space of functionals. As usual, we want a notion of convergence and of normalization for the functionals, in order to be able to tell how the configuration of the system can change continuously, and to recast the probabilistic interpretation of Quantum Mechanics into our new formalism. We thus start by defining a norm over the space of functionals $\omega \in \mathfrak{A}^*$ as:
\begin{equation*}
    \| \omega \| := \sup\{ |\omega(A)|\, : \, \| A\|=1 \}
\end{equation*}
In fact, with the above norm, one can show that $(\mathfrak{A}^*, \| \cdot \|)$ is a Banach space.\\
The functionals we will be interested in, for physical reasons, are listed in the following definition:
\begin{defn}
Let $\mathfrak{A}$ be a $C^*$-algebra. A linear functional $\omega \in \mathfrak{A}^*$ is defined to be:
\begin{itemize}
    \item positive if for all $A \in \mathfrak{A}$:
    \begin{equation*}
        \omega(A^*A) \geq 0
    \end{equation*}
    \item normalized if $\|\omega\| = 1$
    \item a \textbf{state} if it is positive and normalized
\end{itemize}
We denote the space of states by $\mathfrak{G}(\mathfrak{A})$.\\
A state is called faithful if it is strictly positive, i.e. if $\omega(A^*A) = 0$ implies that $A=0$.
\end{defn}
\begin{rem}
For a state, one can in fact drop the assumption of it being continuous, as being a positive and linear functional over a $C^*$-algebra also implies continuity. To see it, assume by absurd that $\omega$ is a positive linear but discontinuous functional. Then, let us pick a sequence of operators $A_n$ in the $C^*$-algebra, such that:
\begin{equation*}
    \| A_n \| \leq \frac{1}{2^n}
\end{equation*}
But then by the discontinuity, we can always find a state $\omega$ such that:
\begin{equation*}
    \omega(A^*_n A_n) \geq 1
\end{equation*}
As we are on a $C^*$-algebra, we know that $\|A_n^* A_n\| = \|A_n\|^2$, so:
\begin{align*}
    \| \sum_{n=0}^{\infty} A_n^* A_n\| &\leq \sum_{n=0}^{\infty} \| A_n \|^2 \\
    &\leq \sum_{n=0}^{\infty} \frac{1}{4^{n}}\\
    &= \frac{4}{3}
\end{align*}
where we have used the geometric series. This implies that:
\begin{equation*}
    \omega\bigg( \sum_{n=0}^{\infty} A_n^* A_n \bigg) < \infty
\end{equation*}
since:
\begin{align*}
    \omega\bigg( \sum_{n=0}^{\infty} A_n^* A_n \bigg) &\leq \bigg| \omega\bigg( \sum_{n=0}^{\infty} A_n^* A_n \bigg) \bigg| \\
    &= \bigg| \omega\bigg( \frac{\sum_{n=0}^{\infty} A_n^* A_n}{\| \sum_{n=0}^{\infty} A_n^* A_n\|} \| \sum_{n=0}^{\infty} A_n^* A_n \|\bigg) \bigg|\\
    &= \bigg| \| \sum_{n=0}^{\infty} A_n^* A_n \| \omega\bigg( \frac{\sum_{n=0}^{\infty} A_n^* A_n}{\| \sum_{n=0}^{\infty} A_n^* A_n\|} \bigg) \bigg|\\
    &\leq \| \omega \| \| \sum_{n=0}^{\infty} A_n^* A_n \|\\
    &= \| \sum_{n=0}^{\infty} A_n^* A_n \| < \infty
\end{align*}
But then, by linearity and the discontinuity of the state $\omega$, we will also have:
\begin{equation*}
    \omega\bigg( \sum_{n=0}^N A_n^* A_n \bigg) \geq N
\end{equation*}
for all finite $N \in \mathbb{N}$. So, by positivity of the state, it follows:
\begin{equation*}
    \infty > \omega\bigg( \sum_{n=0}^{\infty} A_n^* A_n \bigg) \geq \omega\bigg( \sum_{n=0}^N A_n^* A_n \bigg) \geq N
\end{equation*}
which is a contraddiction, i.e. $\omega$ must be continuous
\end{rem}
\begin{rem}
Any positive functional is automatically Hermitian, i.e. $\omega(A) \in \mathbb{R}$ if $A$ is self-adjoint. To see it, consider $B = A + \mathbb{1}$, then from the positivity of the functional and linearity we have:
\begin{align*}
    0 &\leq \omega((A^* + \mathbb{1})(A + \mathbb{1}))\\
    &= \omega(A^*A) + \omega(A^*) + \omega(\mathbb{1}) + \omega(A)
\end{align*}
But then, by taking the imaginary part of this, we must get a vanishing expression as the whole expression must be positive:
\begin{equation*}
    0 = \Im(\omega(A^*A) + \omega(A^*) + \omega(\mathbb{1}) + \omega(A)) = \Im(\omega(A^*) + \omega(A))
\end{equation*}
This implies $\Im(\omega(A^*)) = - \Im(\omega(A))$. On the other hand, if we consider $B = A + i\mathbb{1}$ and we take into account positivity:
\begin{equation*}
    0 \leq \omega(A^*A) + i\omega(A^*) + \omega(\mathbb{1}) -i\omega(A)
\end{equation*}
Then, by again taking the imaginary part of this we obtain: $\Re(\omega(A)) = \Re(\omega(A^*))$. From which it follows that:
\begin{equation*}
    \omega(A^*) = \overline{\omega(A)}
\end{equation*}
\end{rem}

\begin{rem}
From the assumption of positivity, one can prove that $\| \omega \| = \omega(\mathbb{1})$. Then the normalization property is equivalent to $\omega(\mathbb{1}) = 1$.
\end{rem}
From the point of view of a physicist, $\omega(A)$ is the expectation value of the operator $A$, if $A$ is symmetric, with respect to the system described by $\omega$. Furthermore, to give the intuition behind the definition of state, we notice:
\begin{itemize}
    \item The positivity assumption implies that we can consider $\omega(B^*A)$ as defining a degenerate inner product on the elements in $\mathfrak{A}$ and this will imply, as we will see, that we can get a Hilbert space once we specify a state on the algebra.
    \item The normalization assumption is needed in order to give a probabilistic interpretation to the expectation values on the states. To see it, consider any $A \in \mathfrak{A}$ with $\|A\| = 1$, while we also have $\| A^*A \| = 1$ from the $C^*$-algebra condition. But since for any normalized element $A$ of the algebra we also have: $|\omega(A)| \leq \|\omega\|$ and in this case $\| \omega \| = 1$, we have:
    \begin{equation*}
        0 \leq \omega(A^*A) \leq 1
    \end{equation*}
\end{itemize}
\vspace{10mm}
\subsubsection{Pure and mixed states}
Following the goal of building a bridge to the standard approach to QM, we want to be able to describe with state functionals also systems corresponding to ensembles. In the standard approach to quantum theory, we can talk of pure and mixed states by investigating the corresponding form of the density matrix.
Before defining \textit{pure} and \textit{mixed} state functionals, let us state and prove the following:
\begin{lem}
Let $\mathfrak{A}$ be a $C^*$-algebra. The set of states is convex, i.e. whenever we take $\omega_1,\omega_2 \in \mathfrak{G(A)}$ as states and $0 \leq \lambda \leq 1$ we have that:
\begin{equation*}
    \omega := \lambda \omega_1 + (1-\lambda) \omega_2 
\end{equation*}
is still a state.
\end{lem}
\begin{proof}
We need to prove positivity and normalization as linearity is manifestly true.\\
For what concerns normalization let us compute:
\begin{align*}
    \omega(\mathbb{1}) &= \lambda \omega_1(\mathbb{1}) + (1-\lambda) \omega_2(\mathbb{1})\\
    &= \lambda + (1-\lambda) = 1
\end{align*}
For what concerns positivity, consider any $A \in \mathfrak{A}$, then:
\begin{align*}
    \omega(A^*A) &= \lambda \omega_1(A^*A) + (1-\lambda) \omega_2(A^*A)\\
    &\geq 0
\end{align*}
as $0 \leq\lambda \leq 1$ and both $\omega_2(A^*A), \omega_1(A^*A) \geq 0$, so the above is just a sum of positive terms
\end{proof}
Given this geometric aspect of $\mathfrak{G(A)}$, we define pure states:
\begin{defn}
A pure state, is an extremal element of $\mathfrak{G(A)}$. That means that it cannot be written as a non trivial convex linear combination of any two other states. Moreover, any state that is not pure is called a mixed state.
\end{defn}
\vspace{5mm}
\subsubsection{Quasi-free states}
The last important class of states that I want to discuss is that of \textit{quasi-free states}. The definition is motivated by physics, as these are those states for which any correlation between operators is either zero, if the number of operators is odd or, if the number is even, can be decomposed in a product of two point correlation functions. This is precisely what happens for Gaussians distributions, where each moment of the distribution can be expressed via the variance (the moment of order one) recursively; for this reason such states are also called \textit{Gaussian states}. The physical motivation comes, for instance, from the quantum harmonic oscillator as the wavefunction resulting by solving the \sch equation for the ground state is a Gaussian: the corresponding state functional is quasi-free. Moreover, as one can prove, the \textit{ground} and the states describing thermal equilibrium (called \textit{KMS}) of a free QFT are quasi-free. This, in a very heuristic way, can be justified for the free sclar field, remembering that is nothing but a infinite set of decoupled harmonic oscillators.\\
We will say more about \textit{ground} and \textit{KMS states} in a later section. For the moment, we start defining the \textit{quasi-free} property.
\begin{defn}
A state $\omega$ over an algebra $\mathfrak{A}_{(\mathbf{K}, s)}$, where $(\mathbf{K}, s)$ is either a Hilbert space with inner product $s(\cdot, \cdot)$ or a symplectic space with $s(\cdot, \cdot)$ denoting a symplectic form, is quasifree if for any $A(f_n) \in \mathfrak{A}_{(\mathbf{K}, s)}$ ($f_n \in (\mathbf{K}, s)$) we have:
\begin{equation*}
    \omega(A(f_1) \cdots A(f_n)) = \left\{\begin{aligned}
    &0 \hspace{20pt} \mathrm{for} \,\, n \,\, \mathrm{odd}\\
    &\sum_{\pi \in \mathcal{P}_n} \omega(A(f_{\pi_1}) A(f_{\pi_2})) \cdots \omega(A(f_{\pi_{n-1}}) A(f_{\pi_{n}})) \hspace{20pt} n \,\, \mathrm{even} \,\, \mathrm{and} \,\, (\mathbf{K},s) \,\, \mathrm{symplectic}\\
    &\sum_{\pi \in \mathcal{P}_n} (-1)^{|\pi|}\omega(A(f_{\pi_1}) A(f_{\pi_2})) \cdots \omega(A(f_{\pi_{n-1}}) A(f_{\pi_{n}})) \hspace{20pt} n \,\, \mathrm{even} \,\, \mathrm{and} \,\, (\mathbf{K},s) \,\, \mathrm{Hilbert}
    \end{aligned} \right.
\end{equation*}
Where $\mathcal{P}_n$ denotes the group of ordered permutations of $n$ elements, for which it holds:
\begin{align*}
    \pi_{m-1} < \pi_{m} \,\, &\mathrm{with} \,\, 1 \leq m \leq n\\
    \pi_{m-1} < \pi_{m+1} \,\, &\mathrm{with} \,\, 1 \leq m < n
\end{align*}
and $|\pi|$ denotes the order of the permutation with respect to the initial order $\{ 1,2, \dots, n\}$.
\end{defn}
\begin{ex}
In the case in which we are considering a CCR algebra over a sympllectic space $(\mathbf{P}, \sigma)$, the condition of being quasi-free is formulated equivalently at the Weyl algebra level by the condition:
\begin{equation*}
    \omega\big( e^{i\sigma_f} \big) = e^{-\mu(f,f)/2} \hspace{15pt} \forall f \in \mathbf{P}
\end{equation*}
Where:
\begin{equation*}
    \mu(f,h) := \frac{1}{2}(\omega(\sigma_h \sigma_f) + \omega(\sigma_f \sigma_h))
\end{equation*}
and $\sigma_f$ are the generators of the CCR algebra for $f \in \mathbf{P}$. To show the equivalence, consider a CCR algebra defined via the relations:
\begin{align*}
    [\sigma_f,\sigma_h] &= i \sigma(f,h) \mathbb{1}\\
    \sigma_{\lambda f + h} &= \lambda \sigma_f + \sigma_h\\
    \sigma_f^* &= \sigma_{\overline{f}}
\end{align*}
for any $f,h \in \mathbf{P}^{\mathbb{C}} := \{ f + ih| f,h \in \mathbf{P} \}$ where we are assuming $\mathbf{P}$ real. From the definition of $\mu(\cdot,\cdot)$ we have:
\begin{equation*}
    \omega(\sigma_f \sigma_h) := w(f,h) = \mu(f,h) + \frac{i}{2}\sigma(f,h)
\end{equation*}
Then for $f_1,f_2,f_3,f_4 \in \mathbf{P}^{\mathbb{C}}$, if the state is assumed to be quasi-free (introducing the shorthand notation $\sigma_1 := \sigma_{f_1}$):
\begin{align*}
    \omega(\sigma_{1},\dots, \sigma_{4}) &= \frac{1}{2}\omega\big(\{ \sigma_1, \sigma_2 \} \sigma_3 \sigma_4\big) +\frac{1}{2}\omega\big([\sigma_1,\sigma_2]\sigma_3\sigma_4\big)\\
    &= \frac{1}{4}\omega\big(\{ \sigma_1, \sigma_2 \} \{ \sigma_3, \sigma_4\}\big) + \frac{1}{4}\omega\big(\{ \sigma_1, \sigma_2 \} [\sigma_3, \sigma_4]\big) + \frac{1}{4}\omega\big([\sigma_1, \sigma_2] \{ \sigma_3, \sigma_4\}\big) + \frac{1}{4}\omega\big([ \sigma_1, \sigma_2 ] [ \sigma_3, \sigma_4]\big)\\
    &= \frac{1}{4}\frac{d}{d\lambda_4} \dots \frac{d}{d\lambda_1} \omega\big( e^{i\sigma_{\lambda_1 f_1 + \lambda_2 f_2}} e^{i\sigma_{\lambda_3 f_3 + \lambda_4 f_4}} \big)\bigg|_{\lambda_1 = \dots = \lambda_4 = 0} - \frac{i}{2}\sigma(f_3,f_4)\frac{d}{d\lambda_2}\frac{d}{d\lambda_1}\omega\big( e^{i\sigma_{\lambda_1f_1 + \lambda_2f_2}} \big)\bigg|_{\lambda_1 = \lambda_2 = 0}\\
    &\,\,\,\,\,\,- \frac{i}{2}\sigma(f_1,f_2)\frac{d}{d\lambda_4}\frac{d}{d\lambda_3}\omega\big( e^{i\sigma_{\lambda_3f_3 + \lambda_4f_4}} \big)\bigg|_{\lambda_3 = \lambda_4 = 0} - \frac{1}{4} \sigma(f_1,f_2) \sigma(f_3,f_4)
\end{align*}
Using now the quasi-free condition:
\begin{equation*}
    \omega\big( e^{i\sigma_{\lambda_1 f_1 + \lambda_2 f_2}} \big) = e^{-\mu(\lambda_1 f_1 + \lambda_2 f_2, \lambda_1 f_1 + \lambda_2 f_2)/2}
\end{equation*}
Computing all the derivatives one gets:
\begin{equation*}
    \omega(\sigma_{1},\dots, \sigma_{4}) = w(f_1,f_2)w(f_3,f_4) + w(f_1,f_3)w(f_2,f_4) + w(f_1,f_4)w(f_2,f_3)
\end{equation*}
That is an example showing the equivalence.
\end{ex}
\vspace{10mm}

As mentioned at the beginning, together with a notion of normalization, we also want a notion of convergence over the space of linear functionals. Therefore, to end the section, let us discuss how to define a topology on $\mathfrak{A}^*$ called the \textit{weak $*$-topology}. We will need it in the next chapter to argue that the physics is encoded in the algebra and not in the Hilbert space on which the algebra is represented, giving in this way the main arguments in favour of the Algebraic approach.
\begin{defn}
For any set of operators $A_1,\dots,A_n \in \mathfrak{A}$, define a seminorm\footnote{A seminorm has all the properties of a norm, except that $\sigma_{A_1, \dots, A_n}(\phi) = 0$ does not imply $\phi=0$} $\sigma_{A_1, \dots, A_n}$ on $\mathfrak{A}^*$ by:
\begin{equation*}
    \sigma_{A_1,\dots,A_n}(\phi) = \sup\{ |\phi(A_k)| : k = 1,\dots, n \}
\end{equation*}
\end{defn}
\begin{rem}
This is a seminorm, as the states are not assumed to be faithful
\end{rem}
Then, the weak $*$-topology on $\mathfrak{A}^*$ is defined as the topology generated by the open neighborhoods:
\begin{equation*}
    \mathcal{U}(\phi;A_1,\dots,A_n;\epsilon) = \{\phi' \in \mathfrak{A}^*: \sigma_{A_1,\dots,A_n}(\phi-\phi')<\epsilon \}
\end{equation*}
for all $\phi  \in \mathfrak{A}^*$, $A_1,\dots,A_n \in \mathfrak{A}$ and $\epsilon >0$. In particular a sequence $\phi_k \in \mathfrak{A}^*$ ($k\in \mathbb{N}$) converges to $\phi \in \mathfrak{A}^*$ in the weak $*$-topology iff:
\begin{equation*}
    \phi_k(A) \to \phi(A) \hspace{15pt} \forall A \in \mathfrak{A}
\end{equation*}
Such a topology has a strikingly concrete interpretation. Suppose that our system is initially in the configuration $\omega \in \mathfrak{G}(\mathfrak{A})$. As the accuracy of experiments is limited and as the number of experiments we can perform is finite, we can practically determine the system to be in the state $\omega$ just up to a small neighborhood. This neighborhood is exactly $\mathcal{U}(\omega;A_1,\dots,A_n;\epsilon)$ if we interpet $A_i$ for $i=1,\dots,n$ as the finite measures that we perform and $\epsilon$ the limited accuracy in the knowledge of the measurement to agree with the value of $\omega(A_i)$ for all the $i=1,\dots,n$. This means that $\omega$ is physically equivalent to all other states belonging to the $*$-weak neighborhood as, by performing measurements on the system, we are not able to distinguish states belonging to $\mathcal{U}(\omega;A_1,\dots,A_n;\epsilon)$.

\subsection{GNS theorem and Fock space representation}
The algebraic approach gives a way to abstractly investigate the theory, but for practical computations we need to represent these algebras on some Hilbert space, recasting the standard approach. This section aims at presenting the fundamental theorems leading to such a representation.\\\\
The following is the theorem ensuring the existence of a Hilbert space representation that, by the name of the auhors, we will refer to as \textit{GNS theorem}.
\begin{thm}[\textbf{Gelfand-Naimark-Segal}]\label{thm: GNS}
Let $\mathfrak{A}$ be a $C^*$-algebra with unit and let $\omega \in \mathfrak{G}(\mathfrak{A})$. Then:
\begin{itemize}
    \item[\textbf{a)}] There exist a triple $(\mathcal{H}_{\omega}, \pi_{\omega}, \Omega_{\omega})$, where: $\mathcal{H}_{\omega}$ is a Hilbert space, $\pi_{\omega}: \mathfrak{A} \to \mathfrak{B}(\mathcal{H}_{\omega})$ a $\mathfrak{A}$-representation on $\mathcal{H}_{\omega}$ and $\Omega_{\omega} \in \mathcal{H}_{\omega}$, such that:
    \begin{enumerate}
        \item $\Omega_{\omega}$ is cyclic for $\pi_{\omega}$, that means that $\mathcal{D}_{\omega} := \pi_{\omega}(\mathfrak{A}) \Omega_{\omega}$ is a dense subset of $\mathcal{H}_{\omega}$
        \item For every $A \in \mathfrak{A}$ we have: $\omega(A) = \braket{\Omega_{\omega}}{\pi_{\omega}(A) \Omega_{\omega}}$
    \end{enumerate}
    \item[\textbf{b)}] If $((\mathcal{H}'_{\omega}, \pi'_{\omega}, \Omega'_{\omega}))$ satisfies $1.$ and $2.$, then there must exist a unitary operator $U: \mathcal{H}_{\omega} \to \mathcal{H}'_{\omega}$, such that $\Omega'_{\omega} = U \Omega_{\omega}$ and for any $A \in \mathfrak{A}$:
    \begin{equation*}
        \pi'_{\omega}(A) = U \pi_{\omega}(A) U^{-1}
    \end{equation*}
\end{itemize}
\end{thm}
\begin{proof}
See Theorem $14.4$ in \cite{Moretti:2013cma}
\end{proof}
\begin{rem}
There's a version of the GNS theorem also for unital $*$-algebras $\mathfrak{C}$, in which case the elements of the algebra are represented as unbounded operators over the Hilbert space with same domain of definition $D_1 := \pi_{\omega}(\mathfrak{C}) \Omega_{\omega}$. The explicit statement and proof of it can be found in \cite{Moretti:2013cma} Theorem $14.24$. In particular, this version of the theorem, provides a representation of elements of a $CCR$ algebra (e.g. the one of Bosonic fields) as unbounded operators over a Hilbert space $\mathcal{H}_{\omega}$.
\end{rem}
\begin{rem}
Starting from an abstract algebra and performing a \textit{GNS construction}, one can always get a von Neumann algebra on a Hilbert space. Let $\mathfrak{A}$ be an abstract $C^*$-algebra and consider $\omega \in \mathfrak{G}(\mathfrak{A})$. Performing the GNS, we get a subset $\pi_{\omega}(\mathfrak{A}) \subset \mathfrak{B}(\mathcal{H}_{\omega})$ that is a $*$-subalgebra of bounded operators. However, we can always take its double commutant $\pi_{\omega}(\mathfrak{A})''$ making it closed in the uniform topology and thus a von Neumann algebra. In the following, without mentionig explicitly these steps every time, we will say that any abstract $C^*$-algebra is represented as a von Neumann algebra via its GNS construction.
\end{rem}
\begin{rem}
For our purposes, all representations in what follows are considered to be \textit{faithful} i.e. the map $\pi_{\omega}$ is injective. This is not a loss of generality, as also in the case of non faithful representation $\pi_{\omega}$ we can always define $\Tilde{\mathfrak{A}} := \mathfrak{A} \backslash \ker(\pi_{\omega})$ and start with this other algebra instead.
\end{rem}
\begin{rem}
One can show that a state $\omega$ is pure iff the corresponding GNS repesentation is irreducible, i.e., $\mathcal{H}_{\omega}$ does not have subspaces that are invariant under $\pi_{\omega}(\mathfrak{A})$ except for $\mathcal{H}_{\omega}$ itself and $\{ 0\}$.
\end{rem}
To abbreviate the notation, when referring to the GNS construction of a state $\omega$ on an algebra $\mathfrak{A}$ we refer to the uniquely determined, up to unitary equivalence, triple $(\mathcal{H}_{\omega}, \pi_{\omega}, \Omega_{\omega})$.

\subsubsection{Equivalence of representations}
Once a representation is specified, we ask how different ones are related. For that purpose, we start identifying a specific subset of $\mathfrak{G(A)}$
\begin{defn}
Let $\mathfrak{A}$ be a $C^*$-algebra and $\pi: \mathfrak{A} \to \mathfrak{B}(\mathcal{H}_{\pi})$ a representation of it over a Hilbert space $\mathcal{H}_{\pi}$. Then, a state $\omega \in \mathfrak{G(A)}$ is said to be normal with respect to the representation $\pi$ or $\pi$-normal if there's a density matrix $\rho_{\omega} \in \mathfrak{B}(\mathcal{H}_{\pi})$ such that $\omega(\cdot) = \Tr(\rho_{\omega} \pi(\cdot))$. The set of all $\pi$-normal states, denoted $\mathfrak{G}^{(\pi)} \mathfrak{(A)}$ is a convex subset of $\mathfrak{G(A)}$ and is called the folium of $\pi$.\\
Finally, any normal state of the form:
\begin{equation*}
    \omega(\cdot) = \Tr(\ketbra{\Psi}{\Psi} \pi(\cdot))
\end{equation*}
with unit vector $\ket{\Psi} \in \mathcal{H}_{\pi}$, is called a vector state of the representation $\pi$.
\end{defn}
As proven in Section III$.2.2$ of \cite{Haag:1992hx}, any state in $\mathfrak{G(A)}$ can be approximated with arbitrary accuracy in the weak $*$- topology by a sequence of normal states.\\\\
With this notion we can now define different levels of equivalence between representations:
\begin{defn}
Let $\mathfrak{A}$ be a $C^*$-algebra, and let $\pi_1: \mathfrak{A} \to \mathfrak{B}(\mathcal{H}_1)$ and $\pi_2: \mathfrak{A} \to \mathfrak{B}(\mathcal{H}_2)$ be two different representations. We say that:
\begin{itemize}
    \item $\pi_1$ is uinitarily equivalent to $\pi_2$ ($\pi_1 \simeq \pi_2$) if it exist a unitary $U: \mathcal{H}_1 \to \mathcal{H}_2$ such that:
    \begin{equation*}
        U \pi_1(A) U^* = \pi_2(A) \hspace{15pt} \forall A \in \mathfrak{A}
    \end{equation*}
    \item $\pi_1$ is quasi-equivalent to $\pi_2$ if $\mathfrak{G}^{(\pi_1)} \mathfrak{(A)} = \mathfrak{G}^{(\pi_2)} \mathfrak{(A)}$
    \item $\pi_1$ is physically-equivalent to $\pi_2$ if for every state $\omega_1 \in \mathfrak{G}^{(\pi_1)} \mathfrak{(A)}$ and every weak $*$-neighborhood $\mathcal{U}(\omega_1; A_1, \dots, A_n; \epsilon)$ of $\omega_1$ there exist a state $\omega_2 \in \mathfrak{G}^{(\pi_2)}(\mathfrak{A})$ such that $\omega_2 \in \mathcal{U}(\omega_1; A_1, \dots, A_n; \epsilon)$. Equivalently, the folium of $\pi_1$ is physically indistinguishable from that of $\pi_2$.
\end{itemize}
\end{defn}
The natural question that arises now is: which equivalence do we have for GNS representations of different states over the same algebra? The most important result, in finite dimensional QM, is the \textit{Stone-von Neumann theorem}:
\begin{thm}[\textbf{Stone-von Neumann}]
Let $\mathcal{H}$ be a complex Hilbert space and let $(\mathbf{P},\sigma)$ be a symplectic vector space of real dimension $2n$. Suppose that we have on $\mathcal{H}$ a representation of a Weyl $C^*$-algebra as a subset of bounded operators $\pi(C\mathcal{W}(\mathbf{P},\sigma)) \subset \mathfrak{B}(\mathcal{H})$ such that:
\begin{itemize}
    \item $\mathcal{H}$ is irreducible under $\pi(C\mathcal{W}(\mathbf{P},\sigma))$
    \item For every $\mathbf{x} \in \mathbf{P}$ we have that the representation is strongly continuous:
    \begin{equation*}
        s-\lim_{s \to 0} \pi(W(s\mathbf{x})) = \pi(W(0))
    \end{equation*}
\end{itemize}
Then, in a given standard symplectic basis of $\mathbf{P}$ for which $\mathbf{x} \in \mathbf{P}$ is determined by $(\mathbf{t}^{\mathbf{x}}, \mathbf{u}^{\mathbf{x}}) \in \mathbb{R}^n \times \mathbb{R}^n$, there exist a Hilbert space isomorphism $S: \mathcal{H} \to L^2(\mathbb{R}^n,dx)$ (where $dx$ denotes the Lebesgue measure) such that for any $\mathbf{x}$:
\begin{equation*}
    S \pi(W(\mathbf{x})) S^{-1} := \exp{i \sum_{k=1}^n t_k^{(\mathbf{x})}q_k + u_k^{(\mathbf{x})}p_k}
\end{equation*}
Where $q_k, p_k: \mathcal{S}(\mathbb{R}^n) \to L^2(\mathbb{R}^n,dx)$ are:
\begin{align*}
    (q_k \psi)(\mathbf{x}) &= x_k \psi(\mathbf{x})\\
    (p_k \psi)(\mathbf{x}) &= -i \frac{\partial \psi}{\partial x_k}(\mathbf{x})
\end{align*}
As a consequence, $\mathcal{H}$ must be separable as $L^2(\mathbb{R}^n,dx)$ is.
\end{thm}
\begin{proof}
See Theorem $11.43$ in \cite{Moretti:2013cma}
\end{proof}
\begin{rem}
This theorem shows that any irreducible representation on different Hilbert spaces $\mathcal{H}$, i.e. with respect to different pure states in the algebraic formalism, of the Weyl algebra on a finite dimensional symplectic space is isomorphic to a Weyl algebra on $L^2(\mathbb{R}^n,dx)$. This means, that the different representations themselves are isomorphic. Therefore, the representation is unique up to unitary equivalence that, in the category of equivalences introduced above, means that all the irreducible representations are unitarily equivalent.
\end{rem}
However, the Stone-von Neumann theorem holds just for finite dimensional symplectic spaces and, when we deal with QFT, the symplectic space is that of  classical solutions of field equations: that is infinite dimensional. So, we expect the existence of unitarily inequivalent representations. This might give as a consequence that some state, representing a physically accessible configuration for the system, might not be represented by a vector in the chosen Hilbert space associated to another state via a GNS.\\
Of course this is an issue of the standard approach, where one starts with just a fixed Hilbert space i.e. one works in a specified GNS representation. Although, one of the strengths of AQFT is that we do not choose a representation in the first place and thus we can treat all states on equal footing and thus study this issue in $\mathfrak{A}^*$. In fact, one can show that, despite the absence of unitary equivalence of different representations, we always have \textit{physical equivalence}, i.e. the physics is independent from the choice of the representation. The following theorem holds:
\begin{thm}
Any two faithful representations of the algebra of observables are physically equivalent
\end{thm}
\begin{proof}
I refer to Theorem $2.2.13$ in \cite{Haag:1963dh} for further discussion and the, there cited, original paper \cite{Fell} for a proof.
\end{proof}
The conclusion is that, as one should expect, the representation is just a convenient choice for practical computations, but it does not have any real physical implication. We will make use of this in a concrete example in Chapter $4$. For this reason, if we want to study some process related to a state $\omega \in \mathfrak{G}(\mathfrak{A})$ is more convenient to choose that representation where such a state is normal (as they are dense in $\mathfrak{G}(\mathfrak{A})$), i.e. is represented by a density matrix. What we have presented, shows that the physically relevant informations are encoded in the algebra $\mathfrak{A}$ rather than in the Hilbert space representation.

\subsubsection{Fock representation}
Finally, we want to study under which conditions the GNS construction leads to a Fock space representation:
\begin{thm}\label{thm: quasi}
We distinguish:
\begin{enumerate}
    \item For $\omega$ quasifree, defined on the CCR algebra $\mathfrak{A}_{(\mathbf{P},\sigma)}$,
    there's a unique (up to unitary equivalence) Hilbert space $\mathfrak{h}$ and a real linear map $K: \mathbf{P} \to \mathfrak{h}$ such that:
    \begin{align*}
        \overline{K \mathbf{P} + i K \mathbf{P}} &= \mathfrak{h}\\
        \mu(f,h) &= \Re(\braket{Kf}{Kh}_{\mathfrak{h}})\\
        \sigma(f,h) &= 2\Im(\braket{Kf}{Kh}_{\mathfrak{h}})
    \end{align*}
    The GNS construction applied to $\omega$ yields the symmetric Fock space $\mathcal{H}$ built upon $\mathfrak{h}$, with the GNS vector given by the Fock space vacuum $\Omega$ and:
    \begin{equation*}
        \pi_{\omega}(A(f)) = a_B^{*}(Kf) + a_B(K\overline{f})
    \end{equation*}
    With $a_B^*(\psi), a_B(\psi)$ the creation annihilation operators for $\psi \in \mathfrak{h}$, fulfilling:
    \begin{align*}
        a_B(\psi) \Omega &= 0\\
        [a_B(\psi), a_B^*(\phi)] &= \braket{\psi}{\phi}_{\mathfrak{h}}
    \end{align*}
    \item Let $\mathfrak{A}_{(\mathbf{K},s)}$ be a self-dual CAR algebra. If $\mathbf{K}$ has even finite dimension or is infinite dimensional, there is a \textit{basis projection}\footnote{Such a basis projection $P$ is motivated by the physical model we consider. Usually is the projection onto the positive energy eigenmodes of the one particle Hamiltonian of the system.} $P$ on $\mathbf{K}$ such that $\Gamma P \Gamma = 1 - P$. We have $P: \mathbf{K} \to \mathfrak{h} := P\mathbf{K}$ and we equip $\mathfrak{h}$ with the restriction of the inner product $s(\cdot, \cdot)$. Then, we can define a quasifree state $\omega_P$ from $P$ which GNS construction yields the antisymmetric Fock space $\mathcal{H}$ built upon $\mathfrak{h}$, with the GNS vector given by the Fock space vacuum $\Omega$. Moreover:
    \begin{equation*}
        \pi_{P}(B(f)) = a_F^*(Pf) + a_F(P\Gamma f)
    \end{equation*}
    Where $f \in \mathbf{K}$ and $B(f) \in \mathfrak{A}_{(\mathbf{K},s)}$ and the creation and annihilation operators are such that:
    \begin{align*}
        a_F(\psi) \ket{\Omega} &= 0\\
        [ a_F(\psi), a_F^*(\phi) ]_+ &= \braket{\psi}{\phi}_{\mathfrak{h}}
    \end{align*}
\end{enumerate}
\end{thm}
\begin{proof}
For details, see Lemma $3.3$ in \cite{Araki1968} together with Lemma $4.3.$ in \cite{Araki:1971id} for the CAR case and see \cite{KAY199149} Prop. $3.2$ and the discussion thereafter for the bosonic case. Moreover, for details regarding the construction of the state $\omega_P$ see Lemma \ref{Lemma 3.2}.
\end{proof}
\begin{rem}
The Fock space $\mathcal{H}$ mentioned in the above theorem is built as usual:
\begin{equation*}
    \mathcal{H} = \mathbb{C} \oplus \bigoplus_{n > 0} E_n \mathfrak{h}^{\otimes n}
\end{equation*}
Where $E_n$ is either the projector onto the totally symmetric subspace (CCR case) or onto the totally antisymmetric subspace (CAR case). The vacuum vector corresponds to $\Omega = (1,0,0, \dots)$. Finally, the creation and annihilation operators mentioned above are defined for $\chi \in \mathfrak{h}$ and for any $\ket{\Psi_n} = E_n( \ket{\psi_1} \otimes \dots \otimes \ket{\psi_n}) \in E_n \mathcal{H}^{\otimes n}$ as follows:
\begin{align*}
    a_B^*(\chi) \ket{\Psi_n} &= \sqrt{(n+1)} E_{n+1} \ket{\chi \otimes \Psi_n}\\
    a_B(\chi) \ket{\Psi_n} &= \sqrt{\frac{1}{n}} \sum_{j=1}^n \braket{\chi}{\psi_j}_{\mathfrak{h}} E_{n-1} \ket{\psi_1} \otimes \dots \hat{\ket{\psi_j}} \otimes \dots \otimes \ket{\psi_n}
\end{align*}
for the Bosonic case and as:
\begin{align*}
    a_F^*(\chi) \ket{\Psi_n} &= \sqrt{(n+1)} E_{n+1} \ket{\chi \otimes \Psi_n}\\
    a_F(\chi) \ket{\Psi_n} &= \sqrt{\frac{1}{n}} \sum_{j=1}^n \braket{\chi}{\psi_j}_{\mathfrak{h}} (-1)^{j-1}E_{n-1} \ket{\psi_1} \otimes \dots \hat{\ket{\psi_j}} \otimes \dots \otimes \ket{\psi_n}
\end{align*}
for the Fermionic case. Where the hat means that the vector is removed. In particular, we see consistently that, e.g. in the $CAR$ case:
\begin{align*}
    [ a_F^*(\chi), a_F(\chi') ]_+ \ket{\Psi_n} &= a_F^*(\chi) \bigg( \frac{1}{\sqrt{n}} \sum_{j=1}^n \braket{\chi'}{\psi_j} E_{n-1} \ket{\psi_1} \otimes \dots  \hat{\ket{\psi_j}} \dots \otimes \ket{\psi_n} \bigg) + a_F(\chi') \bigg(\sqrt{n+1} E_{n+1} \ket{\chi \otimes \Psi_n}\bigg)\\
    &= \bigg( \frac{1}{\sqrt{n}} \sum_{j=1}^n \braket{\chi'}{\psi_j} \sqrt{n} E_{n} \ket{\chi} \ket{\psi_1} \otimes \dots  \hat{\ket{\psi_j}} \dots \otimes \ket{\psi_n} \bigg)\\
    &\,\,\,\,\,\,+ \sqrt{n+1} \frac{1}{\sqrt{n+1}}\bigg( \braket{\chi'}{\chi} \ket{\Psi_n} - \sum_{j=1}^n \braket{\chi'}{\psi_j} E_{n} \ket{\chi} \otimes \ket{\psi_1} \otimes \dots \hat{\ket{\psi_j}} \dots \otimes \ket{\psi_n} \bigg)\\
    &= \braket{\chi'}{\chi} \ket{\Psi_n}
\end{align*}
Where the minus in the braket at the second step, follows from the fact that $E_n$ is the antisymmetrization projection and thus by putting the annihilation operator through the $\ket{\chi}$ we need to account for a minus sign. From this follows that:
\begin{equation*}
    [ a_F^*(\chi), a_F(\chi') ]_+ = \braket{\chi'}{\chi} \mathbb{1}
\end{equation*}
In the same way one can show that in the $CCR$ case:
\begin{equation*}
    [a_B(\chi'), a_B^*(\chi)] = \braket{\chi'}{\chi} \mathbb{1}
\end{equation*}
In the following we will drop the $B$ and $F$ subscripts, when the context is clear.
\end{rem}

\subsection{KMS and ground states}\label{sec: KMS}
In this section, we describe a particular class of states of high physical relevance: the \textit{ground} and \textit{thermal equilibrium} or \textit{KMS (Kubo-Martin-Schwinger)} states.\\
Whenever we have a $C^*$-algebra $\mathfrak{A}_{(\mathbf{K},s)}$ and a transformation $U$ over $(\mathbf{K},s)$, preserving $s(\cdot,\cdot)$, we can lift it to a $*$-automorphism over the abstract algebra $\mathfrak{A}_{(\mathbf{K},s)}$\footnote{In order $U$ to give a $*$-automorphism in the case of fermionic fields, one further needs to require $[\Gamma, U] = 0$, with $\Gamma$ the involution over the Hilbert space $(\mathbf{K},s)$.}. Namely, for any $A(f) \in \mathfrak{A}_{(\mathbf{K}, s)}$ we have:
\begin{equation*}
    \alpha_U A(f) := A(U^{-1} f) \hspace{20pt} \forall f \in \mathbf{K}
\end{equation*}
One particular example is that of a time-evolution. Namely, when the structure preserving transformation is a one parameter group $V_t$ (either of symplectomorphisms or unitaries). Then, if $t \to V_t$ is strongly continuos, by Stone's theorem:
\begin{thm}[\textbf{Stones's Theorem}]
Let $(V_t)_{t \in \mathbb{R}}$ be a strongly continuous one-parameter unitary group. Then, there exist a unique, possibly unbounded, operator $A: \mathcal{D}_A \to \mathcal{H}$, that is self-adjoint and such that:
\begin{equation*}
    V_t = e^{itA} \hspace{20pt} \forall t \in \mathbb{R}
\end{equation*} 
Where:
\begin{equation*}
    \mathcal{D}_A = \bigg\{ \psi \in \mathcal{H} | \lim_{\varepsilon \to 0} \frac{-i}{\varepsilon}(U_{\epsilon} \psi - \psi) \,\, \mathrm{exists} \bigg\}
\end{equation*}
\end{thm}
\begin{proof}
See Theorem $9.33$ in \cite{Moretti:2013cma}.
\end{proof}
We have existence of a \textit{"Hamiltonian"} $\mathbf{h}$, generating the dynamics $V_t$. In this specific case, we shall denote the induced $*$-automorphism over the algebra as:
\begin{equation*}
    \alpha_t(A(f)) = A(V_{-t} f)
\end{equation*}
Finally, one may further raise this to an automorphism over the states as:
\begin{equation*}
    \alpha_t^*\omega (A(f)) := \omega(\alpha_t(A(f)))
\end{equation*}
For any $\omega$ state over $\mathfrak{A}_{(\mathbf{K},s)}$ and $A(f) \in \mathfrak{A}_{(\mathbf{K},s)}$.\\\\
Now that we know how states evolve, we define those that are stationary with respect to the time evolution:
\begin{equation*}
    \alpha_t^* \omega = \omega
\end{equation*}
Two such examples are the \textit{ground} and the \textit{KMS (Kubo-Martin-Schwinger) states}. We start introducing the first, proving that is stationary.
\begin{defn}
A state $\omega$ on $\mathfrak{A}_{(\mathbf{K},s)}$ is called a \textit{ground state}, with respect to the time evolution $\alpha_t$, if:
\begin{equation*}
    -i\partial_t \omega(A^* \alpha_t (A))|_{t=0} \geq 0 \hspace{20pt} \forall A \in \mathfrak{A}_{(\mathbf{K},s)}
\end{equation*}
\end{defn}
\begin{prop}
A ground state $\omega$ over $\mathfrak{A}_{(\mathbf{K},s)}$ is stationary
\end{prop}
\begin{proof}
Note that to prove the stationarity, is enough to show that $\partial_t \alpha_t^* \omega|_{t=0} = 0$. Choose an $A$ self-adjoint\footnote{An element $A \in \mathfrak{A}_{(\mathbf{K},s)}$ is self-adjoint if $A^* = A$}:
\begin{align*}
    \partial_t \omega (\alpha_t(A^2))|_{t=0} &= \partial_t \omega(\alpha_t(A) A)|_{t=0} + \partial_t \omega(A \alpha_t(A))|_{b=0}\\
    &= 2 \Re \partial_t \omega(A \alpha_t(A))|_{t=0}\\
    &= 0
\end{align*}
Proving the invariance when we evaluate squares of self-adjoint operators. But now, any element $A \in \mathfrak{A}_{(\mathbf{K},s)}$ can be written as a combination of two sel-adjoint elements $A_{1/2} \in \mathfrak{A}_{(\mathbf{K},s)}$ as $A = A_1 + i A_2$. Moreover, we can write any of the self-adjoints:
\begin{equation*}
    A_{1/2} = \frac{1}{2}(\mathbb{1} + A_{1/2})(\mathbb{1} + A_{1/2}) - \frac{1}{2}\mathbb{1} - \frac{1}{2}A_{1/2}A_{1/2}
\end{equation*}
where all the terms involved are squares of self adjoint elements. So, by linearity of the state functional and the abovely proven result for squares of self-adjoint operators, follows the claim.
\end{proof}
\begin{rem}
The physical motivation behind the definition of ground state is the following. By the just proven stationarity of the ground state, we have that the automorphism $\alpha_t$ is implemented over the GNS by a unitary $U_t = e^{-itH}$, which leaves $\ket{\Omega}$ invariant:
\begin{align*}
    -i\partial_t \omega(A^* \alpha_t (A))|_{t=0} &= -i\partial_t \braket{\Omega}{ \pi_{\omega}(A)^*  e^{itH}\pi_{\omega}(A)e^{-itH}\Omega}\\
    &= \braket{\pi_{\omega}(A)\Omega}{H \pi_{\omega}(A)\Omega}
\end{align*}
Hence, $\omega$ ground state is equivalent to demand $H \geq 0$ with $H \ket{\Omega} = 0$. Namely, choosing a ground state, corresponds to pick that representation in which the "Hamiltonian" generating the time evolution, is bounded below by $0$.
\end{rem}
Let us move on to the definition of KMS state:
\begin{defn}
A state $\omega$ over $\mathfrak{A}_{(\mathbf{K},s)}$ is called a Kubo-Martin-Schwinger (KMS) state, with inverse temperature $\beta$ (with respect to the time evolution $\alpha_t$) if for all $A,B \in \mathfrak{A}_{(\mathbf{K},s)}$:
\begin{itemize}
    \item The funciton \begin{equation*}
    F_{A,B}(z) := \omega(A \alpha_z(B))
\end{equation*}
of complex variable $z$, is bounded for $0 \leq \Im(z) \leq \beta$ and analytic in the complex strip $0< \Im(z) < \beta$
\item The function $F_{A,B}(z)$ and the funciton $G_{A,B}(z) := \omega(\alpha_z(A) B)$ are related:
\begin{equation*}
    F_{A,B}(t) = G_{A,B}(t + i \beta) \hspace{20pt} \forall t \in \mathbb{R}
\end{equation*}
\end{itemize}
\end{defn}
\begin{prop}
A KMS state $\omega$ over $\mathfrak{A}_{(\mathbf{K},s)}$ is stationary
\end{prop}
\begin{proof}
Take in the KMS condition $A = \mathbb{1}$, then for $z \in \mathbb{C}$ it becomes:
\begin{equation*}
    F_B(z) = \omega(\alpha_z (B))
\end{equation*}
This function is analytic in $0 < \Im z < \beta$ and bounded in $0 \leq \Im z \leq \beta$, with boundary values which are periodic, $F_B(t + i\beta) = F_B(t)$. But then, as an analytic function on a strip with periodic boundary conditions, it must be constant. We have proven that for all $B \in \mathfrak{A}_{(\mathbf{K},s)}$ and $t \in \mathbb{R}$, $\omega(\alpha_t(B)) = \omega(B)$ that is equivalent to: $\alpha_t^* \omega(B) = \omega(B)$
\end{proof}
\begin{rem}
The appaently very abstract definition of KMS state, is actually just the generalization of the condition that one gets in QM on finite dimensional Hilbert spaces in defining \textit{thermal equilibrium or Gibbs states}. To see it explicitly, consider on a system of finitely many degrees of freeedom the Gibbs state of inverse temperature $\beta$:
\begin{equation*}
    \omega_{\beta}(A) := \frac{\Tr(A e^{-\beta H})}{Z_{\beta}} \hspace{20pt} \forall A \in \mathfrak{A}_{(\mathbf{K},s)}
\end{equation*}
where $Z_{\beta}$ is the partition function associated to the thermal state. The time evolution is given by the Heisenberg equation:
\begin{equation*}
    \alpha_t A = e^{itH} A e^{-itH}
\end{equation*}
with Hamiltonian $H$ bounded from below. Therefore:
\begin{equation*}
    \omega_{\beta}(A \alpha_t (B)) = \frac{\Tr(A e^{itH} B e^{-(\beta + it)H})}{Z_{\beta}}
\end{equation*}
As $H$ is semibounded, the map $z \to e^{izH}$ into bounded operators is bounded for $\Im(z) \geq 0$ and analytic $\Im(z) > 0$. As a consequence, the map $z \to e^{-(\beta + iz)H}$ is bounded for $\Im(z) \leq \beta$ and analytic for $\Im(z) < \beta$. It follows, that map $z \to \omega_{\beta}(A \alpha_z(B))$ is bounded for $0 \leq \Im(z) \leq \beta$ and has domain of analyticity $0< \Im(z) < \beta$. To prove also the second condition of a $KMS$ state, notice:
\begin{align*}
    \omega_{\beta}(A \alpha_{t + i\beta}(B)) &= \frac{\Tr(A e^{-(\beta -it)H}B e^{-itH})}{Z_{\beta}}\\
    &= \frac{\Tr(e^{itH}B e^{-itH} A e^{-\beta H})}{Z_{\beta}}\\
    &= \omega_{\beta}(\alpha_t(B) A)
\end{align*}
Therefore, the Gibbs state at inverse temperature $\beta$ is KMS at inverse temperature $\beta$.\\
In this sense, the KMS condition generalizes the notion of thermal equilibrium in the context of field theories where, the lacking of a notion of trace on local operator algebras (we discuss this more in Chapter $2$ when we introduce von Neumann factors), prevents us from defining a density matrix.\\ Therefore, the motivation behind the KMS condition, is to keep the same analyticity conditions of the finite dimensional case that are independent from the fact that $e^{-\beta H}$ must be of trace class. These same conditions, become analyticity conditions on the state functional with respect to a given time evolution.
\end{rem}
\begin{rem}
The existence of a time evlution is crucial, from the physical point of view, in order to define the notions of ground and KMS states. In fact, is the general definition of time evolution that allows to assign temperatures in an observer independent way. In the same way, is the existence of a global time evolution that allows to define a global notion of energy and as such fix the scale of the lowest energy level to zero.
\end{rem}
We close this section, noticing that a ground state is also $KMS$. Defining: \begin{equation*}
    F_{A,B}(t) = \braket{\pi_{\omega}(A^*) \Omega}{e^{itH} \pi_{\omega}(B) \Omega}
\end{equation*}
if $\omega$ is a ground state, we have $H \geq 0$ and therefore the function $z \to F_{A,B}(z)$ is bounded for $\Im(z) \geq 0$ and analytic for $\Im(z) > 0$. One can even show that the converse holds, i.e. a KMS state with associated function $F_{A,B}(z)$ bounded for $\Im(z) \geq 0$ and analytic for $\Im(z) > 0$, defines a ground state (see Prop. $5.3.19$ in \cite{Bratteli:1996xq}). In this way, a ground state can be understood as the $\beta \to \infty$ limit of a KMS state.

\section{Axioms of AQFT}\label{sec: Axiom}
Now that we have introduced all the necessary mathematical background and explanied most of the reasons in favour of the algebraic approach to QM, we define what is a QFT in the algebraic sense. This was done in an axiomatic way, first by Wightman and Garding in \cite{Wightman} and later to nets of local algebras by Haag and Kastler in \cite{Haag:1963dh}. Here we adopt the Haag-Kastler axioms, and present them in a more modern fashion following Section $2.3$ of \cite{Hollands:2017dov}.\\
In algebraic quantum field theory, the algebraic relations between the quantum fields are encoded in nets of $C^*$-algebras associated with spacetime regions, with partial ordering given by the usual inclusion of sets on the spacetime. The formalism is then perfect to define Quantum Fields even on curved backgrounds (in a semiclassical way), generalizing already the standard approach. The structure and type of such nets depends of course on the type of field theory and on the spacetime background.\\
To start and establish the connection with the standard approach, consider Minkowski spacetime $(\mathbb{R}^4,\eta)$. Consider on it the set of causally complete regions $\mathcal{K}$ and the subset of \textit{causal diamonds} $\mathcal{O}$. The Poincaré group acts freely on Minkowski, i.e. for any $g \in \mathcal{P}^{\uparrow}_+ = \mathcal{L}_+^{\uparrow} \ltimes \mathbb{R}^4$, that we write as $g = (\Lambda,a)$, we have:
\begin{equation*}
    g \cdot x = \Lambda x + a \hspace{15pt} \forall x \in \mathbb{R}^4
\end{equation*}
Now, since Poincaré transformations are isometries of Minkowski, they map causal diamonds in other causal diamonds as they are causally complete regions.

\begin{defn}[\textbf{Algebraic Quantum Field Theory}]
An algebraic quantum field theory (AQFT) is defined by a $C^*$-algebra $\mathfrak{A}$, called the algebra of observables\footnote{The name is a bit misleading, as elements of $A \in \mathfrak{A}$ are not all self adjoint, i.e. not all of them are proper observables in the sense of being actually measurable.}, and an assignement to every causal diamond $\mathcal{O}$ of a $C^*$-subalgebra of $\mathfrak{A}$:
\begin{equation*}
    \mathcal{K} \ni \mathcal{O} \mapsto \mathfrak{A}(\mathcal{O}) \subset \mathfrak{A}
\end{equation*}
which are called the local algebras of observables. This net of algebras $\mathfrak{A}(\mathcal{O})$ must satisfy the following set of axioms:
\begin{itemize}
    \item[\textbf{A1}](Isotony) If $\mathcal{O}_1 \subset \mathcal{O}_2$ we have $\mathfrak{A}(\mathcal{O}_1) \subset \mathfrak{A}(\mathcal{O}_2)$. Moreover:
    \begin{equation*}
        \mathfrak{A} = \overline{\bigcup_{\mathcal{O}} \mathfrak{A}(\mathcal{O})}
    \end{equation*}
    With the completion taken with respect to the $C^*$-norm
    \item[\textbf{A2}](Causality) If $\mathcal{O}_1$ is spacelike separated from $\mathcal{O}_2$ we must have: $[\mathfrak{A}(\mathcal{O}_1), \mathfrak{A}(\mathcal{O}_2)] = 0$. That means that algebras at spacelike separation must commute. This can also be rewritten as:
    \begin{equation*}
        \mathfrak{A}(O') \subset \mathfrak{A}(O)'
    \end{equation*}
    \item[\textbf{A3}](Relativistic Covariance) For each transformation $g \in \Tilde{\mathcal{P}}^{\uparrow}_+$ (the universal covering group of $\mathcal{P}^{\uparrow}_+$), there's an automorphism $\alpha_g$ on $\mathfrak{A}$ such that:
    \begin{equation*}
        \alpha_g \mathfrak{A}(\mathcal{O}) = \mathfrak{A}(\Lambda \mathcal{O} + a) 
    \end{equation*}
    for all causal diamonds $\mathcal{O}$ and such that: $\alpha_g \alpha_{g'} = \alpha_{gg'}$, $\alpha_{(\mathbb{1},0)} = \mathbb{1}$
    \item[\textbf{A4}](Vacuum) There's a unique state $\omega_0$ over $\mathfrak{A}$ that is invariant under $\alpha_g$. On its GNS representation $(\pi_0, \mathcal{H}_0, \ket{0})$, $\alpha_g$ is implemented by a projective representation $U$, with positive energy, of the universal covering $\Tilde{\mathcal{P}}^{\uparrow}_+$. This means that $\pi_0(\alpha_g(A)) = U(g)\pi_0(A)U(g)^*$ for all $A \in \mathfrak{A}$ and $g \in \Tilde{\mathcal{P}}^{\uparrow}_+$. Positive energy means that the representation is strongly continuous and if we pick $x \in \mathbb{R}^4$ and the corresponding translation operator $(\mathbb{1},x) \in \Tilde{\mathcal{P}}^{\uparrow}_+$ we have:
    \begin{equation*}
        U(x) = \exp(-iP^{\mu}x_{\mu})
    \end{equation*}
    And the vector generator $P = (P^{\mu})$ has spectrum $p = p^{\mu}$ in the forward lightcone $p \in \overline{V}^+ = \{ p| p^2 \geq 0, p^0 > 0 \}$
\end{itemize}
\end{defn}
\begin{rem}
The axioms are motivated by our expectation. The first reflects the intuition that, if $\mathfrak{A}(\mathcal{O}_2)$ represents everything that can be observed in the region $\mathcal{O}_2$, whenever we consider a smaller region $\mathcal{O}_1 \subset \mathcal{O}_2$ all that can be observed here is for sure contained in what was observable in the wider region. The second axiom, as the name suggests, is just the reflection of the classical idea of causality translated in the language of Quantum Mechanics: as the operators commute, we can in an abstract sense measure something in $\mathfrak{\mathcal{O}}$ without affecting the measure in $\mathfrak{A}(O')$ as these are \textit{compatible} observables. The third axiom reflects the attachement of the net of algebras to the spacetime regions and as such they should transform according to the way in which the spacetime regions transform. The fourth is motivated by the fact that each inertial observer should measure the same physics in his system of reference.
\end{rem}

If one considers a more general spacetime background, one needs to modify these axioms. Clearly axioms $\mathbf{A1}, \mathbf{A2}$ can be generalized straightforwardly to the case of curved backgrounds as they do not involve any specific feature of $(\mathbb{R}^4,\eta)$. At the same time, if the spacetime has a specific symmetry group $G$, one can generalize also $\mathbf{A3}$ requiring the same to hold after replacing $\Tilde{\mathcal{P}}^{\uparrow}_+$ with $G$. What il less easy to generalize is $\mathbf{A4}$ as, in general spacetimes, there's no unique vacuum. This is due to the general absence of a timelike Killing vector field. If we would have rather had such a symmetry (Minkowski spacetime or more generally on static spacetimes), we can impose a globally defined notion of positive and negative energy modes giving a globally defined notion of vacuum once we consider any two different observers moving along the flow lines of the timelike Killing field.\\\\
However, in the axiomatic definition, there's still a caveat: from standard QFT and statistical physics we know that fermions have anticommutation relations. This translates in fermionic field opeators that do not commute at spacelike separation but rather anticommute. This violates the \textit{Causality axiom} of an AQFT. There are two way out for this problem:
\begin{enumerate}
    \item Modify the axioms by requiring a weaker version of locality called \textit{graded locality}
    \item Consider a subset of the fermionic algebra of fields corresponding to field binomial such that causality holds on that subset
\end{enumerate}
A good reason in favour of the second choice, is that we are just able to measure quantities composed by fermionic binomials, as one can see also from the form of the Dirac Lagrangian. However, as it is easier to work with the full algebra, we will here adopt the first approach and consider the fermionic algebras to be an example of an AQFT with respect to \textit{graded locality}.\\
So, let me introduce this concept:
\begin{defn}
Let $(\mathfrak{I}(\mathcal{O}))_{\mathcal{O}}$ be the net of local $C^*$-algebras acting on a Hilbert space $\mathcal{H}$, i.e. $\mathfrak{I}(\mathcal{O}) \subset \mathfrak{B}(\mathcal{H})$ and let $\ket{\Omega}$ be the cyclic vector for $\mathcal{H}$. A $\mathbb{Z}_2$-grading on $\mathfrak{I}(\mathcal{O})$ is defined by $\Gamma \in \mathfrak{B}(\mathcal{H})$ such that $\Gamma = \Gamma^{-1} = \Gamma^*$ and:
\begin{align*}
    \Gamma \ket{\Omega} &= \ket{\Omega}\\
    \Gamma \mathfrak{I}(\mathcal{O}) \Gamma &= \mathfrak{I}(\mathcal{O}) \hspace{15pt} \forall \mathcal{O} \in \mathcal{K}\\
    [\Gamma, U(\Tilde{g})] &= 0 \hspace{15pt} \forall \Tilde{g} \in \Tilde{P}^{\uparrow}_+
\end{align*}
An operator $A \in \mathfrak{I}$ such that $\Gamma A \Gamma = \pm A$ is called homogeneous \textbf{Bose} (\textbf{even}) or \textbf{Fermi} (\textbf{odd}) depending on the alternative $\pm$. For this reason, we introduce a degree $I_A$ for the Homogeneous operator $A$ that is $0$ if it is Bose and $1$ if it is Fermi
\end{defn}
We will call a net of local algebra with a $\mathbb{Z}_2$-grading a \textit{$\mathbb{Z}_2$-graded net of local algebras}. Moreover, any $A \in \mathfrak{I}$ can be decomposed as a sum $A = A_+ + A_-$ of a Bose and a Fermi operator:
\begin{equation*}
    A_{\pm} := \frac{A \pm \Gamma A \Gamma}{2}
\end{equation*}
Where of course $A_{\pm} \in \mathfrak{I}(\mathcal{O})$.\\
Now that we have introduced the notion of grading, we define the \textit{graded commutator}
\begin{defn}
Let $(\mathfrak{I}(\mathcal{O}))_{\mathcal{O}}$ be a $\mathbb{Z}_2$-graded net of local algebras. We define the graded commutator for $A,B \in \mathfrak{I}$ homogenous as:
\begin{equation*}
    [A,B]_{\Gamma} := AB - (-1)^{I_A \cdot I_B} BA
\end{equation*}
and we extend this to general $A,B \in \mathfrak{I}$ by linearity.
\end{defn}
From the definition follows immediately that, for both homogeneous operators of the Fermi type, the graded commutator reduces to an anticommutator, while is a standard commutator for Bose homogeneous operators.\\
We can now replace the axiom $\mathbf{A2}$ with a generalized version:
\begin{itemize}
    \item[\textbf{A2'}] Let $(\mathfrak{I}(\mathcal{O}))_{\mathcal{O}}$ be a $\mathbb{Z}_2$-graded net of local algebras. We say that the local algebra satisfies graded locality if, for any pair of spacelike separated regions $\mathcal{O}_1, \mathcal{O}_2$, we have that $[\mathfrak{I}(\mathcal{O}_1), \mathfrak{I}(\mathcal{O}_2)]_{\Gamma} = 0$
\end{itemize}
In this manner Fermionic Quantum Field Theories fit in the formalism of AQFT.\\
The notion of $\mathbb{Z}_2$-grading is often replaced in litterature by twisted locality \cite{Bisognano:1975ih},\cite{Bisognano:1976za} and \cite{DAntoni:2001ido}:
\begin{defn} \label{def: Z2grad}
Let $(\mathfrak{I}(\mathcal{O}))_{\mathcal{O}}$ be a $\mathbb{Z}_2$-graded net of local algebras and let:
\begin{equation*}
    Z := \frac{1-i\Gamma}{1-i}
\end{equation*}
We define the twisted commutant of a local algebra as $\mathfrak{I}(\mathcal{O})^{t'} := Z \mathfrak{I}(\mathcal{O})'Z^*$. We say that the net satisfies twisted locality if, for any two spacelike separated regions $\mathcal{O}_1,\mathcal{O}_2$, we have $\mathfrak{I}(\mathcal{O}_1) \subset \mathfrak{I}(\mathcal{O}_2)^{t'}$
\end{defn}
One can prove that the two are equivalent:
\begin{lem}
Let $\mathfrak{I}(\mathcal{O})$ be a $\mathbb{Z}_2$-graded net of local algebras. Then, it satisfies graded locality if and only if it satisfies twisted locality
\end{lem}
\begin{proof}
$(\Rightarrow)$ By assumption $[\mathfrak{I}(\mathcal{O}_1), \mathfrak{I}(\mathcal{O}_2)]_{\Gamma} = 0$. We want to show that $\mathfrak{I}(\mathcal{O}_1) \subset  Z \mathfrak{I}(\mathcal{O}_2) Z^* =: \mathfrak{I}(\mathcal{O}_2)^{t'}$. Start noticing:
\begin{align*}
    Z Z^* &= \frac{1}{2}(1 - i \Gamma)(1 + i \Gamma)\\
    &= \frac{1}{2}(1 + i \Gamma - i \Gamma + \Gamma^2)\\
    &= 1
\end{align*}
Let us first look at the case in which the $\mathfrak{I}(\mathcal{O}_i)$ is homogeneous for all $i$. For Bose operators, we know:
\begin{align*}
    [\mathfrak{I}(\mathcal{O}_1), \mathfrak{I}(\mathcal{O}_2)]_{\Gamma} &= \mathfrak{I}(\mathcal{O}_1) \mathfrak{I}(\mathcal{O}_2) - \mathfrak{I}(\mathcal{O}_2) \mathfrak{I}(\mathcal{O}_1) = 0\\
    Z \mathfrak{I}(\mathcal{O}_2)' Z^* &= \mathfrak{I}(\mathcal{O}_2)'
\end{align*}
But the first of these equations implies that:
\begin{equation*}
    \mathfrak{I}(\mathcal{O}_1) \subset \mathfrak{I}(\mathcal{O}_2)' = \mathfrak{I}(\mathcal{O}_2)^{t'}
\end{equation*}
For what concerns the case of Fermi homogeneous operators we know:
\begin{equation*}
    [\mathfrak{I}(\mathcal{O}_1), \mathfrak{I}(\mathcal{O}_2)]_{\Gamma} = \mathfrak{I}(\mathcal{O}_1) \mathfrak{I}(\mathcal{O}_2) + \mathfrak{I}(\mathcal{O}_2) \mathfrak{I}(\mathcal{O}_1) = 0
\end{equation*}
But then, if we look at:
\begin{align*}
    \mathfrak{I}(\mathcal{O}_1) Z \mathfrak{I}(\mathcal{O}_2) Z^* - Z \mathfrak{I}(\mathcal{O}_2) Z^* \mathfrak{I}(\mathcal{O}_1) &= \frac{1}{2} \mathfrak{I}(\mathcal{O}_1) (\mathfrak{I}(\mathcal{O}_2) + i \mathfrak{I}(\mathcal{O}_2)\Gamma - i \Gamma \mathfrak{I}(\mathcal{O}_2) + \Gamma \mathfrak{I}(\mathcal{O}_2) \Gamma)\\
    &\,\,\,\,\,\,- \frac{1}{2}(\mathfrak{I}(\mathcal{O}_2) + i\mathfrak{I}(\mathcal{O}_2)\Gamma -i \Gamma \mathfrak{I}(\mathcal{O}_2) + \Gamma \mathfrak{I}(\mathcal{O}_2) \Gamma)\mathfrak{I}(\mathcal{O}_1)\\
    &= \frac{1}{2} \mathfrak{I}(\mathcal{O}_1) (+ i \mathfrak{I}(\mathcal{O}_2)\Gamma - i \Gamma \mathfrak{I}(\mathcal{O}_2)) - \frac{1}{2}(+ i\mathfrak{I}(\mathcal{O}_2)\Gamma -i \Gamma \mathfrak{I}(\mathcal{O}_2))\mathfrak{I}(\mathcal{O}_1)\\
    &= \frac{1}{2} (i\mathfrak{I}(\mathcal{O}_1) \mathfrak{I}(\mathcal{O}_2)\Gamma + i \Gamma \mathfrak{I}(\mathcal{O}_1) \mathfrak{I}(\mathcal{O}_2)) - \frac{1}{2}(- i\mathfrak{I}(\mathcal{O}_2) \mathfrak{I}(\mathcal{O}_1)\Gamma -i \Gamma \mathfrak{I}(\mathcal{O}_2) \mathfrak{I}(\mathcal{O}_1))\\
    &= \frac{i}{2} \Gamma [\mathfrak{I}(\mathcal{O}_1), \mathfrak{I}(\mathcal{O}_2)]_{\Gamma}+ \frac{i}{2} [\mathfrak{I}(\mathcal{O}_1), \mathfrak{I}(\mathcal{O}_2)]_{\Gamma} \Gamma = 0
\end{align*}
Where we have first used the fact that for homogenous Fermi operators $\Gamma \mathfrak{I}(\mathcal{O}_i) \Gamma = - \mathfrak{I}(\mathcal{O}_i)$ and at the third step that $\Gamma^2 = 1$.\\
The general case follows from the fact that any $A \in \mathfrak{I}(\mathcal{O}_i)$ can be written as $A = A_+ + A_-$.\\\\
$(\Leftarrow)$ From the above computations, we have seen for homogeneous Bose: $\mathfrak{I}(\mathcal{O}_i)' = \mathfrak{I}(\mathcal{O}_i)^{t'}$, so if $\mathfrak{I}(\mathcal{O}_1) \subset \mathfrak{I}(\mathcal{O}_2)^{t'}$ we have:
\begin{equation*}
    [\mathfrak{I}(\mathcal{O}_1), \mathfrak{I}(\mathcal{O}_2)]_{\Gamma} = 0
\end{equation*}
While, for the Fermi case we have seen above that:
\begin{equation*}
    \mathfrak{I}(\mathcal{O}_1) Z \mathfrak{I}(\mathcal{O}_2) Z^* - Z \mathfrak{I}(\mathcal{O}_2) Z^* \mathfrak{I}(\mathcal{O}_1) = \frac{i}{2} \Gamma [\mathfrak{I}(\mathcal{O}_1), \mathfrak{I}(\mathcal{O}_2)]_{\Gamma}+ \frac{i}{2} [\mathfrak{I}(\mathcal{O}_1), \mathfrak{I}(\mathcal{O}_2)]_{\Gamma} \Gamma
\end{equation*}
But if $\mathfrak{I}(\mathcal{O}_1) \subset \mathfrak{I}(\mathcal{O}_2)^{t'}$, the left hand side must vanish, from which it follows the claim.
\end{proof}
We mentioned this because we are going to use the notion of graded locality pointing although at litterature where the notion of twisted locality is used.\\\\

From what we said in the previous sections, formulating a Quantum field theory just requires the construction of an algebra, corresponding to the operators representing the quantum fields, fulfilling the above axioms. Once that is done we can formulate the standard approach to a free QFT by picking a specific Fock representation using a quasi-free state. In the following sections, I will present an example of QFT that fulfills the above listed axioms: the Dirac-Majorana field. From what we have just mentioned, we need to show how to construct the algebras of observables starting from the classical field equations. This process is often called \textit{quantisation}.\\
Moreover, as mentioned above, AQFT is the perfect framework to generalize to aribtrary spacetimes backgrounds, therefore we will work on curved spacetimes. This will require some technicalities in the classical definition of Dirac fields, that we are going to discuss at the beginning of the next chapter.\\
The choice of presenting just the Dirac field, despite it is more complicated than the free scalar field, is due to the results of this thesis that are presented in Chapter $3$ and that regard fermionic fields. Morevoer, the free scalar field is treated in most of the literature so I refer the interested reader to \cite{Wald:1995yp}, \cite{Haag:1992hx}, \cite{Moretti:2013cma} and the literature cited there, for the original papers.

\newpage
\section{Dirac and/or Majorana fields}\label{sec: DiMaj}
The free scalar fields on arbitrary curved spacetime $(M,g)$, are just sections of a trivial bundle $M\times \mathbb{C}$. As such, they are always well defined on each manifold, provided that we have existence and uniqueness of solutions of fields equations, i.e. if the spacetime is globally hyperbolic. In contrast the formulation of the classical Dirac equation on a spacetime manifold $(M,g)$ requires the spacetime manifold to have additional properties and the Dirac fields are going to be sections of a non-trivial bundle.\\
We will start the chapter by discussing which additional properties the Manifold must have in order to be able to formulate the classical Dirac equation.\\
I refer to Chapter $13$ of \cite{Wald:1984rg} for the physical motivation and applications behind the introduction of Spin manifolds, to the Lecture notes \cite{BaerSpin} for a modern mathematically rigorous approach to Spin geometry and to \cite{Dimock1982DiracQF}, \cite{VerchFewster} and \cite{Sanders} for literature related to the quantization of Dirac fields.

\subsection{Spin bundle}
The concpet of Spin on Minkowski spacetime, relies on the existing projective representation of its symmetry group (the Poincaré group or to be more precise the proper orthocronus part of its universal covering) over the Hilbert space of the QFT. Following Wigner's approach \cite{Wigner:1939cj}, the "elementary" fields are classified according to the irreducible representations of the symmetry group. These are labelled by the eigenvalues of the Casimir operators $m^2$ and $S^2$, respectively the mass $m$ a positive real number and a discrete\footnote{In the massless case is known the existence of irreducible representations with non-trivial translations that admit continuous spin representations. However, such fields have, so far, never been of physical relevance. For further details on it and additional literature I refer to \cite{Weinberg}} integer or semiinteger that labels the so called spin or, in the massless case, helicity.\\
The strict connection with the underlying spacetime geometry is manifest and one should expect that notions like spin and spinors (vectors that transform according to the two fold covering of the proper orthocronus Lorentz group) have to be reviewed on arbitrary spacetime geometries where the isometry group can also be trivial.\\
The idea is to rely on Einstein's Equivalence Principle (EEP): each spacetime is, at least locally, Minkowski. Namely, if we consider an observer $O$ specified at each spacetime point $x \in M$ by a tetrad of orthonormal vectors $\{e_{a}\}_{a = 0,1,2,3} \in T_x M$, where $e_0$ is the normalized vector tangent to the worldline of $O$, another observer $\Tilde{O}$ at the same spacetime point $x \in M$ has its tetrad of vectors that must be related to that of $O$ by a Lorentz transformation $\phi_g$. Moreover, each measurement done by $O$ on a system $\mathcal{S}$ at $x$ must agree with a measurement done by $\Tilde{O}$ at the same spacetime point $x$ on $\Tilde{\mathcal{S}}$ that is: $\Tilde{\mathcal{S}} = \Tilde{\phi}_g \mathcal{S}$ where the transformation $\Tilde{\phi}_g$ is the induced action on physical observables existing by EEP.\\\\
Formally, the above means that at each $x \in M$, for each time oriented tetrads, we can find Lorentz transformations mapping one tetrad into another. This transformation must induce, by EEP, a transormation over the elements of a theory satisfying the EEP. This, holding for tetrads at the same spacetime point $x$, is a fiberwise statement. Our aim, is to extend the relation between the tetrads in the fiber globally, in a continuous way, as a relation between sections of the fiber bundle.\\
So, start defining the principal fiber bundle as the bundle of oriented, time oriented and orthonormal tetrads on which we have a free action of the proper orthocronus Lorentz group $\mathcal{L}_+^{\uparrow}$. A section of this bundle is an assignement of a frame at each spacetime point. In this way the fibers are diffeomorphic to the proper orthocronus Lorentz group:
\begin{itemize}
    \item The time orientation of the orthonormal basis gives that each observer has the same time orientation, thus they must be related by an orthochronous Lorentz transformation
    \item The orientation assumption tells us that the observers must be related to each other by a transformation in the same connected component of the Lorentz group and since the identity transformation maps "two" such observers, we are considering the proper part of the Lorentz group
\end{itemize}
Therefore, starting from the standard tetrad coming from the orientation and time orientation assumptions, we have a one to one correspondence between elements in $\mathcal{L}_+^{\uparrow}$ and tetrads $\{ e^a\}_{a=0,1,2,3}$.\\
We call such a fiber bundle a \textit{Frame bundle} and denote it as $FM$ and we will denote the principal frame bundle, to emphasize the existence of a free right action of the proper orthocronus Lorentz group, by $(FM, \mathcal{L}^{\uparrow}_+, M, \phi)$ with $\phi: FM \times \mathcal{L}^{\uparrow}_+ \to FM$ the free right action defined as:
\begin{equation*}
    \phi: (v,A) \mapsto v \cdot A
\end{equation*}
Where $\cdot$ denotes the matrix multiplication between the element in the Lorentz group associated to the section $v(x)$ at each spacetime point $x \in M$ and the element $A \in \mathcal{L}^{\uparrow}_+.$\\\\
However, as we know from the flat spacetime case, in order to define spinors with non trivial behavior under a $2\pi$ rotation, we need the universal covering of $\mathcal{L}_+^{\uparrow}$. Therefore we "duplicate" each fiber to produce a principal $Spin_{1,3}^0$ bundle over $M$, where $Spin_{1,3}^0$ denotes the universal covering\footnote{The explicit definition of the Spin group is reported in Appendix \ref{app}} of $\mathcal{L}_+^{\uparrow}$.
\begin{defn}
A spin structure on $(M,g)$ is a pair $(SM,p)$, where $SM$ is a principal $Spin_{1,3}^0$-bundle over $M$, the spin frame bundle, which carries a right action $R_S$ with respect to $S \in Spin_{1,3}^0$. While, $p: SM \to FM$ is a base-point preserving bundle homomorphism such that:
\begin{equation*}
    p \circ R_S = R_{\Lambda(S)} \circ p
\end{equation*}
Where $S \mapsto \Lambda(S)$ is the canonical universal covering map of the Lorentz group:
\begin{align*}
    \Lambda : Pin_{1,3} &\to \mathcal{L}\\
    S &\mapsto \Lambda^a_{\,\,b}(S)
\end{align*}
see Appendix \ref{app}.
\end{defn}
The just outlined construction of spin structures, may be prevented by topological obstructions of the manifold $M$. To understand why, consider first the simple case of $M$ simply connected and thus orientable and time orientable\footnote{Suppose M is not orientable by absurd. Consider then its oriented double covering $p:X\to M$ and pick $x_0 \in X$. Let $\lambda$ be a path connecting $x_0$ and some other point in the different fiber of the double covering, call this point $x_1$ (here we are using that the double cover of a non-orientable $M$ is connected - since it is a manifold, it is also path-connected). Now, $p \circ \lambda$ is a loop in $M$, which lifts to $\lambda$. Since $M$ is simply connected, $p \circ \lambda$ is homotopically trivial. But this is an absurd, since $\lambda$ is not a loop and this will imply $x_0 = x_1$}. Construct on it the frame bundle $(FM,\mathcal{L}^{\uparrow}_+, M, \phi)$ of oriented and time oriented frames. Consider a curve $\gamma \in FM$, from the simply connectedness of $M$, we know that $\mathfrak{p} \circ \gamma$ (the projection of the path $\gamma \in FM$ to the base manifold $M$) is contractible to the trivial closed curve at a point $x$ lying on $\mathfrak{p} \circ \gamma$. This means that we can continuously deform $\gamma$ in $FM$ to a loop restricted to $\mathfrak{p}^{-1}(x) = (\mathcal{L}^{\uparrow}_+)_x$ (the subscript denotes the fiber in $FM$ at the point $x$). But we know that $\pi_1 (\mathcal{L}^{\uparrow}_+) = \mathbb{Z}_2$, therefore this loop is not contractible within $\mathcal{L}^{\uparrow}_+$. The fact that such a loop is not contractible is, from standard QFT, crucial for the definition of spinors on Minkowski spacetime, as it assigns a sign to a spinor depending on their behaviour under a $2\pi$ rotation. 
In this case, $(\mathcal{L}^{\uparrow}_+)_x$ is the fiber of $FM$ and thus if we want a global notion of spin we need this loop to be incontractible also throughout $FM$, otherwise we may start with Spin half field on a fiber move along a closed curve in $FM$ and end up having a spin zero field. Therefore, we need to make sure that such loops, not just fiberwise, but also by deforming them throughout $FM$, are not all contractible: \textit{if $FM$ is simply connected, the notion of spinors on M cannot be defined}.\footnote{ Actually, it exists a more general result: $FM$ is not simply connected if and only if its second Stiefel-Whitney class of $M$ vanishes. I refer to \cite{Wald:1984rg} for further literature on it, where the Stiefel-Whitney class is defined.}.\\
In the case in which $M$ is not simply connected, the manifold may first of all fail to be orientable or time orientable. In that case, even $FM$ as the frame bundle of oriented and time oriented basis cannot be defined. Therefore, we need to asume $M$ to be at least orientable and time orientable if we want to define a notion of Spin.\\
Then, $M$ is assumed orientable and time orientable but not simply connected. In constructing $SM$, we cannot simply take the universal covering of $FM$ as now also $M$ is not simply connected. If we were to do that, the universal covering will also have a modified base manifold $M$ while we want to "duplicate" just the fibers. We can still demand, just the fibers to be "unwrapped", if the fundamental group is canonically splitted:
\begin{equation}\label{eq: spinor}
    \pi_1 (FM) = \pi_1(\mathcal{L}^{\uparrow}_+) \times \pi_1(M)
\end{equation}
In this case, we can split the task of finding the universal covering of the fibers and of the base manifold, as each loop in the bundle can be splitted in a loop in $M$ and another in the fiber. Moreover, if $M$ is simply connected we get back the previously discussed case.\\
Therefore, Eq. \eqref{eq: spinor} is a necessary condition to have spinors\footnote{Actually the above splitting might not be unique, leading to different notions of spinors introducing an ambiguity that results in the so called \textit{exotic spin structures} and the number of them equals that of generators of the first cohomology group $H^1(M; \mathbb{Z}_2)$. A possible physical interpretation of them was discussed in \cite{Isham:1978ec}}.\\\\

Now that we have discussed the conditions for existence of Spinors, it is natural to ask whether the manifolds we are interested in admit such structures. In particular, as we want to formulate field theories, to have existence and uniqueness of solutions of the Cauchy problem for differential equations we need to assume the spacetime to be globally hyperbolic \footnote{If we are dealing with spacetimes with boundaries, like the one of the Casimir effect, $I\times \mathbb{R}$ for $I \subset \mathbb{R}$ some compact real interval, then we can still have existence and uniqueness of solutions for boundary value problems by assigning Dirichlet boundary condition, despite such spacetimes are not globally hyperbolic. However, in what follows, we will study Cauchy problems for the field equation, therefore the spacetimes will always be assumed to be globally hyperbolic.}. For this reason, it's natural to ask whether a globally hyperbolic, orientable spacetime admits spinors. Let me start quoting a general theorem proved by Geroch \cite{doi:10.1063/1.1664507}:
\begin{thm}
Let $M$ be an open manifold with a Lorentzian metric $g$, then $M$ admits a spinor structure if and only if there exists on $M$ a global system of orthonormal tetrads
\end{thm}
The condition of admitting a global orthonormal tetrad, is equivalent for the spacetime to be parallelizable\footnote{A Manifold $M$ is parallelizable, if there exist smooth vector fields $V_1,\dots,V_n$ on the manifold, such that at every point $p \in M$ the tangent vectors $V_1(p),\dots,V_n(p)$ provide a basis of the tangent space at $p$. Equivalently a manifold is parallelizable if the tangent bundle is trivial, i.e. if the frame bundle has a global section on $M$}. A globally hyperbolic spacetime is parallelizable, by the following two theorems (\cite{Wald:1995yp}, \cite{Stiefel1935/36}):

\begin{thm}\label{thm: globhy}
A globally hyperbolic spacetime is diffeomorphic to $\mathbb{R} \times \Sigma$ and isometric to $\mathbb{R} \times \Sigma$, equipped with metric:
\begin{equation*}
    g = -\beta(\tau, \mathbf{x}) d\tau^2 + \gamma_{ij}(\tau, \mathbf{x})dx^i dx^j
\end{equation*}
with $\tau \in \mathbb{R}$ and $\mathbf{x} = \{ x^i \}$ coordinates on $\Sigma$, such that $\beta > 0$, $\Sigma_{\tau} = \{ \tau \} \times \Sigma$ is a Cauchy surface and $(\Sigma_{\tau}, \gamma_{ij}(\tau, \dot))$ is a Riemannian Manifold
\end{thm}

\begin{thm}[\textbf{Stiefel}]
An orientable $3$ dimensional manifold $\Sigma$ is parallelizable
\end{thm}

Therefore, an orientable globally hyperbolic spacetime\footnote{Because time orientability is implied by global hyperbolicity from the previous theorem}, by the first theorem, is diffeomorphic to $M \simeq \mathbb{R} \times \Sigma$ and since, by assumption, it was orientable also $\Sigma$ must be orientable. Therefore, by the \textit{Stiefel theorem}, $\Sigma$ is parallelizable. So, since the cartesian product of parallelizabe manifolds is still parallelizable, we conclude from Geroch's result that each globally hyperbolic manifold admits a spinor structure.

\subsection{Dirac spinors and cospinors}
During the next section we will use, especially for proofs, concepts and results related to the Dirac Clifford algebra that are reported in the Appendix \ref{app}. However, leaving aside proofs, the section can be followed without all the details reported in Appendix and we will always cite definitions and statements necessary to follow the discussion. Moreover, from now on, $(M,g)$ is assumed to be globally hyperbolic admitting thus a spin structure.\\
We turn to the construction of spinors. Let's start by choosing a complex irreducible representation $\pi$ of the Dirac algebra $D$ and matrices $A,C \in GL(4,\mathbb{C})$ reppresenting the charge and hermitian conjugation (see Appenix \ref{app} and Definition \ref{def: hacc}).\\
Consider a globally hyperbolic spin spacetime $SM$, we define the assocaited vector bundle:
\begin{equation*}
    DM := SM \times_{Spin_{1,3}^0} \mathbb{C}^4
\end{equation*}
Where $Spin_{1,3}^0$ acts on $SM$ from the right and on $\mathbb{C}^4$ from the left via the representation $\pi$. In other words, $DM$ is obtained from the product bundle $SM \times \mathbb{C}^4$ by identifying:
\begin{equation*}
    [E,z] = [R_S E, \pi(S^{-1})z]
\end{equation*}
Where we think of $z \in \mathbb{C}^4$ as a column vector. A spinor is thus an equivalence class of tetrads or of a $2\pi$ rotation of it (due to the double covering definition of $SM$) and an ordered quadrupole of comlpex numbers transforming according to a representation of the $Spin^0_{1,3}$. The equivalence relation identifies what observers, differing by a Lorentz transformation, see. A spinor, is thus Lorentz invariant by definition. In this way, if we let $M_{2\pi}$ denote the $2\pi$ rotation in $Spin_{1,3}^0$ with respect to a frame element $e_{3}$ of the section $E \in SM$, we see that $[E,z]$ and $[R_{2\pi} E,z] = [E, \pi(M_{2\pi})z]$ describe different spinors, allowing to distinguish them depending on their sign under a $2\pi$ rotation. If we have worked with $FM$ instead of $SM$ this distinction would not have been possible.\\
We can analogously introduce the dual bundle $D^*M$, or use the usual isomorphism from $\mathbb{C}^4 \to (\mathbb{C}^4)^*$ induced by the standard inner product over $\mathbb{C}^4$, to define elements of $D^*M$:
\begin{equation*}
    [E,w^*] = [R_S E, w^* \pi(S)]
\end{equation*}
In this way, we can pointwise on the base manifold (fiberwise) interpret $D^*M$ as linear functionals over $DM$ via the following notation:
\begin{equation*}
    \langle [E,w^*],[E,z] \rangle := w^*(z) = \langle w,z \rangle
\end{equation*}
Where we need both elements having the same base "point" $E$, in order this inner product to be well defined. The right most expression above, denotes the standard inner product of $\mathbb{C}^4$.

\begin{defn}
The vector bundle $DM$ is called the Dirac spinor bundle, its elements are called Dirac spinors and a section of $DM$ is called a Dirac spinor field. The space of all smooth spinor fields is denoted as $C^{\infty}(DM)$, and the compactly supported ones as $C^{\infty}_0(DM)$.\\
The vector bundle $D^*M$ is called the Dirac cospinor bundle with relatively analogous definitions and notations.
\end{defn}
For notational convenience we indicate the canonical pairing of a spinor field $u$ with a cospinor $v$ by:
\begin{equation*}
    vu(x) := \langle v(x), u(x) \rangle
\end{equation*}
that defines a sesquilinear map (by the sesquilinearity of the standard inner product on $\mathbb{C}^4$): $C^{\infty}(D^*M) \times C^{\infty}(DM) \to C^{\infty}(M)$ as the image is a function of the spacetime point $x$.\\
Let us now introduce the charge conjugation and Dirac adjoint maps, for spinor and cospinors (later we will generalize this for spinor and cospinor fields) as we need these notions to formulate the Dirac equation. As we'll see, we need the chosen matrices $A,C \in GL(4, \mathbb{C})$ mentioned at the beginning and whose properties are discussed in Appendix \ref{app}:
\begin{defn}
We define the maps:
\begin{align*}
        (\cdot)^{\sharp} &: DM \to D^*M\\
        (\cdot)^{\sharp} &: D^*M \to DM
    \end{align*}
\begin{align*}
        (\cdot)^c &: DM \to DM\\
        (\cdot)^c &: D^*M \to D^*M
\end{align*}
as follows:
\begin{align*}
    [E,z]^{\sharp} := [E,z^*A] \hspace{15pt} [E,z^*]^{\sharp} := [E,A^{-1}z] \\
    [E,z]^c := [E,C^{-1}\overline{z}] \hspace{15pt} [E,z^*]^c := [E,\overline{z}^* C]
\end{align*}
\end{defn}
\begin{rem}
These maps are base-point preserving and act just on the spinorial component.
\end{rem}
These are well defined maps, compatibly with the definitions of $DM$ and $D^*M$, i.e. compatibly with the identifications of frames and vectors under the free action of $Spin_{1,3}^0$. To see it one needs Lemma \ref{lem: tec1} and:
\begin{align*}
    [R_S E, \pi(S^{-1})z]^{\sharp} &= [R_S E, z^* \pi(S^{-1})^* A]\\
    &= [R_S E, z^* A \pi(S)]\\
    &= [E, z^* A] = [E,z]^{\sharp}
\end{align*}
Moreover, the abovely defined maps, are vector bundle anti-isomorphisms:
\begin{lem}
For $q = [E,z] \in DM$ and $p = [E,w^*] \in D^*M$ we have:
\begin{align*}
    q^{\sharp \sharp} = q = q^{cc} &\hspace{20pt} p^{\sharp \sharp} = p = p^{c c}\\
    q^{\sharp c} = - q^{c \sharp} &\hspace{20pt} p^{\sharp c} = - p^{c \sharp}\\
    \langle q^{\sharp}, p^{\sharp}\rangle = &\overline{\langle p,q \rangle} = \langle p^c, q^c \rangle
\end{align*}
\end{lem}
\begin{proof}
Recalling that, from Definition \ref{def: hacc}, we have $A = A^*$ and $\overline{C}C = \mathbb{1}$, we can compute for $q$:
\begin{align*}
    q^{\sharp \sharp} &= [E,z]^{\sharp \sharp} = [E, z^* A]^{\sharp} = [E, A^{-1} A^* z] = [E,z] = q\\
    q^{cc} &= [E,z]^{cc} = [E, C^{-1} \overline{z}]^{c} = [E, C^{-1} \overline{C^{-1} \overline{z}}] = [E,z] = q
\end{align*}
While for $p$ we have:
\begin{align*}
    p^{\sharp \sharp} &= [E,w^*]^{\sharp \sharp} = [E, A^{-1} w]^{\sharp} = [E, w^* (A^{-1})^* A] = [E,w^*] = p\\
    p^{cc} &= [E,w^*]^{cc} = [E, \overline{w^*} C]^{c} = [E, \overline{\big(\overline{w^*} C \big)} C] = [E,w^*] = p
\end{align*}
For what concerns the second set of equalities, compute using Theorem \ref{thm: ex}:
\begin{align*}
    q^{c \sharp} &= [E,z]^{c \sharp} = [E, C^{-1} \overline{z}]^{\sharp} = [E, \overline{z}^* (C^{-1})^* A]\\
    &= - [E, \overline{z^* A} C] = - [E, z^* A]^c = - [E,z]^{\sharp c} = - q^{\sharp c}\\
    p^{c \sharp} &= [E,w^*]^{c \sharp} = [E, \overline{w}^* C]^{\sharp} = [E, A^{-1} C^* \overline{w}]\\
    &= - [E, C^{-1} \overline{A^{-1} w}] = - [E, A^{-1} w]^c = - [E,w^*]^{\sharp c} = - p^{\sharp c}
\end{align*}
Finally for what concerns the last result:
\begin{align*}
    \langle q^{\sharp}, p^{\sharp}\rangle &= \langle [E,z]^{\sharp}, [E,w^*]^{\sharp} \rangle = \langle [E,z^*A] , [E,A^{-1}w] \rangle\\
    &= z^*w  =\overline{w^*z} = \overline{\langle p,q\rangle}\\
    \langle p^{c}, q^{c}\rangle &= \langle [E, \overline{w}^* C], [E, C^{-1} \overline{z}] \rangle = \overline{w}^* \overline{z}\\
    &= \overline{w^* z} = \overline{\langle p,q \rangle}
\end{align*}
\end{proof}
\begin{defn}
The maps $(\cdot)^{\sharp}$ are called Dirac adjoint, while $(\cdot)^c$ Dirac charge conjugation. 
\end{defn}
For spinor and cospinor fields these maps are defined pointwise and extended gobally after. This means that for $u \in C^{\infty}(DM)$ we have $u^{\sharp}(x) = u(x)^{\sharp}$, manifesting that the charge conjugation and adjoint preserve the support.\\
Now that we have introduced the spinor and cospinor bundles, defining pointwise on the manifold vector spaces, we can construct tensors and tensor fields starting from them. Therefore, on $M$, we have now the following bundle structures: $DM$, $D^*M$, $TM$, $T^*M$. Therefore, we can construct tensors combining both the spin and the standard tangent bundles, forming in this way a mixed spinor-tensor algebra. In presenting it, I will adopt the "physicist" notation working in local coordinates but, in order to do that, we need to introduce a proper index notation to distinguish spinor and vector indices.\\
Let us start from a local section $E$ of $SM$ and a real basis $b_A$ of $\mathbb{C}^4$. Then, we get local coordinates $E_A$ on $DM$ and, as a consequence, also on $TM$:
\begin{equation*}
    E_A = [E,b_A] \,\,\,\, \Rightarrow \,\,\,\, e = p \circ E
\end{equation*}
Where we have used the canonical projection map $p: SM \to FM$. Now $(E_A)$ represents a local frame for $DM$ while $\{ e_a \}_{a=0,\dots,3}$ is a frame for $TM$. One can further assume these to be normalized:
\begin{equation*}
    e^b(e_a) = \delta^b_{\,\,a} \hspace{15pt} E^B E_A = \delta^B_{\,\,A}
\end{equation*}
Where we have introduced the dual basis $e^b$, $E^B$ canonically obtained from the inner product structure that we have. The difference, with standard tensors over a manifold, stems from the way they transform under change of coordinates. In order to explain it, let $E'$ be a different section of $SM$ on the same spacetime region $\mathcal{O} \subset M$. Then, calling $S: \mathcal{O} \to Spin_{1,3}^0$ the map that associates to spacetime points transformations in the universal covering of the fibers of $SM$, we know that: $E' = R_{S^{-1}}E$. As a consequence:
\begin{equation*}
    E'_A = [E',b_A] = [R_{S^{-1}} E, b_A] = [E, b_B \pi(S^{-1})^B_{\,\, A}] = E_B \pi(S^{-1})^B_{\,\,A}
\end{equation*}
So we have derived, the following transformation rules for frames of $DM$:
\begin{equation*}
    E'_A = E_B \pi(S^{-1})^B_{\,\,A} \hspace{15pt} (E')^A = \pi(S)^A_{\,\,B} E^B
\end{equation*}
and the standard ones for frames of $TM$ ($\Lambda$ is the double covering map of the proper orthocronus Lorentz group):
\begin{equation*}
    e'_a = e_b \Lambda(S^{-1})^b_{\,\,a} \hspace{20pt} (e')^a = \Lambda(S)^a_{\,\,b} e^b
\end{equation*}
From which, for a general spinor-tensor $T = T^{Aa}_{\,\,\,\,\,\,Bb} E_A \otimes e_a \otimes E^B \otimes e^b$, the transformation rule of its components becomes:
\begin{equation*}
    (T')^{Aa}_{\,\,\,\,\,\,Bb} = \pi(S)^A_{\,\,F} \pi(S^{-1})^D_{\,\,B} \Lambda(S^{-1})^{d}_{\,\,b} \Lambda(S)^a_f T^{Ff}_{\,\,\,\,\,\,Dd}
\end{equation*}
To better understand the local coordinate notation, we express charge conjugation and hermitian adjointeness using local frames. In particular, from the transformation rules of frames:
\begin{equation*}
    E_A^{\sharp} = \delta_{AB}A^B_{\,\,C}E^C \hspace{15pt} E_A^c = E_B(C^{-1})^B_{\,\,A}
\end{equation*}
we find, the following general transformation rules:
\begin{align*}
    u^{\sharp} &= (u^A E_A)^{\sharp} =\overline{u^A} \delta_{AB} A^B_{\,\,C} E^C = (u^{\sharp})_C E^C\\
    u^c &= (u^A E_A)^c = \overline{u^A} E_B(C^{-1})^B_{\,\,A} =(u^c)^B E_B\\
    v^{\sharp} &= (v_A E^A)^{\sharp} = \overline{v_A} E_C\delta^{AB} (A^{-1})^C_{\,\,B} = (v^{\sharp})^C E_C\\
    v^c &= (v_A E^A)^c = \overline{v_A} E_C C^C_{\,\, B} \delta^{BA} = (v^c)^C E_C
\end{align*}
From which we extract:
\begin{align*}
    (u^{\sharp})_C &= \overline{u}^A \delta_{AB} A^B_{\,\,C} \hspace{20pt} (u^c)^A = (C^{-1})^A_{\,\,B} \overline{u}^B\\
    (v^{\sharp})^A &= (A^{-1})^A_{\,\,B} \delta^{BC} \overline{v_C} \hspace{20pt} v^c_A = \overline{v_B} C^B_{\,\,A}
\end{align*}\\\\
Finally, let us introduce for vector fields $v$ and covector fields $k$ the Feynman slash notation:
\begin{equation*}
    \cancel{v} := v^a \gamma^A_{\,\,\,\,aB} \hspace{20pt} \cancel{k} := k_a \gamma^{Aa}_{\,\,\,\,\,\,B}
\end{equation*}\\\\

Before moving on to the next section, where we present and study the Dirac equation, we still owe a proof of independence, of the described spin bundles, from the choice of representation of the Dirac algebra $D$ and matrices $A,C \in GL(4,\mathbb{C})$ that we did at the beginning:
\begin{prop}
Consider the Dirac spinor bundle $DM_0$ defined analogously to $DM$, but with respect to different $\pi_0$ and different matrices $A_0 = \gamma_0$, $C_0 = \gamma_2$ that induce the different notions of charge conjugation and hermitian adjoint denoted respectively as: $(\cdot)^-$ and  $(\cdot)^+$. Let also $\cancel{\nabla}_0$ denote the Dirac operator defined through $\pi_0$. Then, there exist a base-point preserving, vector bundle isomorphism $\lambda: DM \to DM_0$ and induced isomorphism $\lambda^*: D^*M \to D^*M_0$ such that:
\begin{equation*}
    \lambda \circ (\cdot)^\sharp = (\cdot)^+ \circ \lambda^* \hspace{20pt} \lambda \circ \cancel{\nabla} = \cancel{\nabla}_0 \circ \lambda \hspace{20pt} \lambda \circ (\cdot)^c = (\cdot)^- \circ \lambda
\end{equation*}
This isomorphism is unique, up to an overall sign.
\end{prop}
\begin{proof}
On each fiber, the bundle isomorphism must be given by Theorem \ref{thm: fund} by:
\begin{equation*}
    \lambda: [E,z] \mapsto [E, Lz]_0
\end{equation*}
For some $L \in GL(4,\mathbb{C})$. As a consequence, still fiberwise:
\begin{equation*}
    \lambda^*: [E,z^*] \mapsto [E,z^*L^{-1}]_0
\end{equation*}
In order to get the right intertwining properties, we need $L$ to satisfy:
\begin{equation*}
    L A^{-1} = A_0 (L^*)^{-1} \hspace{20pt} LC^{-1} = C_0^{-1} \overline{L} \hspace{20pt} L \pi = \pi_0 L
\end{equation*}
Then, we apply Theorem \ref{thm: ex} to ensure that such an $L$ exist unique up to a sign. By continuity of the attachement of the fibers in the definition of a bundle, $L$ is locally constant on $M$ and defined at each point of $M$ by its connectedness. This allows us to extend the map $\lambda$ globally, making it well defined and unique up to a global sign.
\end{proof}

To summarize, before moving on to define further structures over $DM$, we have seen how spinors and cospinors are defined, introduced charge conjugation and Dirac adjoint and showed how these construction are all independent from the choice of the representation of the Dirac algebra.\\
In the next section, on $DM$ and $D^*M$ we will define a dynamics via the \textit{Dirac equation}.

\subsection{The classical Dirac equation}
We now have all the necessary formalism to formulate the Dirac equation on globally hyperbolic spacetimes.\\\\
The first thing we need, is to define a notion of derivation on spinors. We start expressing the Levi-Civita connection in the frame bundle $FM$. This, is done using the Levi-Civita connection $\nabla$ that we have on $TM$. In fact, for $v \in TM$:
\begin{equation*}
    \nabla v = (\nabla_{\mu} v^{\rho}) dx^{\mu} \otimes \frac{\partial}{\partial x^{\rho}} = (\partial_{\mu} v^{\rho} + \Gamma^{\rho}_{\,\, \mu \nu}v^{\nu}) dx^{\mu} \otimes \frac{\partial}{\partial x^{\rho}}
\end{equation*}
The advantage of having a frame, is the expansion of vector fields in terms of the elements of the frame. Namely, if $v\in TM$ we decompose $v = v^a e_a$, where the decomposition is pointwise defined. But, we also know, as elements in the tangent space at each point, that $x^a := e^a_{\mu}x^{\mu}$, where the greek index denotes the components with respect to a coordinate chart over the base manifold. In particular, we recognize the elements $e_{\mu}^a$ as those vectors:
\begin{equation*}
    g_{\mu \nu} e^{\mu}_a e^{\nu}_b = \eta_{ab}
\end{equation*}
from which we can get:
\begin{equation*}
    e_{\mu}^a = g_{\mu \nu} \eta^{ab} e^{\nu}_b
\end{equation*}
as both sides have the same action on basis vectors. From this, it follows:
\begin{equation*}
    \eta_{ab}e^a_{\mu} e^b_{\nu} = g_{\mu \nu} \Rightarrow g^{\mu \nu} e^a_{\mu} e^b_{\nu} = \eta^{ab} \Rightarrow \eta^{ab} e^{\mu}_a e^{\nu}_b = g^{\mu \nu}
\end{equation*}
So, we see that we can raise and lower indices using $\eta^{ab}$ and $\eta_{ab}$. Thus, we express the covariant derivative of the vector $v$ using the frame:
\begin{equation}\label{eq: covind}
    \nabla v = (\nabla_b v^a) e^b \otimes e_a = (\partial_b v^a + \Gamma^a_{\,\, bc}v^c) e^b \otimes e_a
\end{equation}
Where, in local coordinates over the manifold, the "frame derivation" is a directional derivative $\partial_b = e^{\mu}_b \partial_{\mu}$. Therefore, the only thing still undetermined in this expression, that we need to somehow relate to $\Gamma^{\rho}_{\,\,\mu \nu}$, are the Christoffel symbols $\Gamma^a_{\,\,bc}$. For this purpose, compare the Levi-Civita derivatives expressions in coordinate and in the frame:
\begin{align*}
    (\nabla_{\mu} v^{\rho}) dx^{\mu} \otimes \frac{\partial}{\partial x^{\rho}} &= (\nabla_b v^a) e^b \otimes e_a\\
    &= (\nabla_b v^a) e^b_{\mu} e_a^{\rho} dx^{\mu} \otimes \frac{\partial}{\partial x^{\rho}}\\
\end{align*}
From which it follows:
\begin{equation*}
    \partial_b v^a + \Gamma^a_{\,\, bc}v^c = (\partial_{\mu} v^{\rho} + \Gamma^{\rho}_{\,\, \mu \nu}v^{\nu}) e_b^{\mu} e^a_{\rho} = \partial_b (v^{\rho}e^a_{\rho}) - v^{\rho}\partial_b e^a_{\rho} + \Gamma^{\rho}_{\,\, \mu \nu} v^{\nu} e^{\mu}_b e^a_{\rho}
\end{equation*}
Where, at the second step we have used the Leibniz rule. Hence, using $\partial_b(e^a_{\rho} e^{\rho}_c) = \partial_b \delta^a_c = 0$, 
\begin{equation}\label{eq: spte}
    \Gamma^a_{\,\, bc} = - e^{\rho}_c \partial_b e^a_{\rho} + e^a_{\rho} e^{\mu}_b e^{\nu}_c \Gamma^{\rho}_{\,\, \mu \nu} = e^a_{\rho} \partial_b e^{\rho}_c + e^a_{\rho} e^{\mu}_b e^{\nu}_c \Gamma^{\rho}_{\,\, \mu \nu}
\end{equation}
That is the type of relation we were searching for, in order to fully understand the covariant derivative in the frame bundle.\\\\
Before defining the connection over $DM$, starting from that on $FM$, notice that we can also express it in terms of the so called connection forms. These are introduced noticing that, to fully determine the covariant differentiation, we just need to know how $\nabla$ acts on basis vectors:
\begin{equation*}
    \nabla e_a = (\partial^b e_a + \Gamma^b_{\,\,a c} e^c) \otimes e_b = \Gamma^b_{\,\,a c} e^c \otimes e_b
\end{equation*}
So, we can define the connection form $\omega^b_{\,\, a} = \Gamma^b_{\,\,a c} e^c$ (a one-form, see Eq. \eqref{eq: covind}) and express the terms involving the Christoffel symbols in the covariant derivatives, using instead the connection forms.\\
With this in mind, let us now introduce the connection over $DM$. 
To start, define the one-form $\mathbf{\Omega} \in \Omega^1(FM, \mathfrak{lie}(\mathcal{L}_+^{\uparrow}))$ as that one form over the frame bundle $FM$, taking value in the Lie algebra $\mathfrak{lie}(\mathcal{L}^{\uparrow}_+)$, 
such that its pullback with respect to any local section $e$ of $FM$ is: $e^*\mathbf{\Omega}^a_{\,\,c} = \omega^a_{\,\,c}$. We define then (here $d \Lambda: \mathfrak{lie}(Spin_{1,3}^0) \to \mathfrak{lie}(\mathcal{L}_+^{\uparrow})$ and is explicitly defined Prop. \ref{prop: lieal}):
\begin{equation*}
    \mathbf{\Sigma} = (d\Lambda)^{-1} p^*(\mathbf{\Omega}) = \frac{1}{4}p^*(\mathbf{\Omega}^a_{\,\,c})\gamma_a \gamma^c
\end{equation*}
Where $\mathbf{\Sigma}$ is then a $\mathfrak{lie}(Spin_{1,3}^0)$ valued one-form on $SM$ that defines a connection over $DM$, that we call the spin connection. In a local section $E$ of $SM$ the spin connection one forms are denoted as $\sigma^A_{\,\,\,\,b \, C}$ and are then given by the pullback of $\mathbf{\Sigma}$ with respect to $E$. Since we have $E^*p^* =(p \circ E)^*=e^*$, we get:
\begin{equation*}
    \sigma_b = \frac{1}{4}\Gamma^a_{\,\, bc} \gamma_a \gamma^c
\end{equation*}
\begin{figure}[H]
 		\centering
 		\includegraphics[width=1.30	\columnwidth]{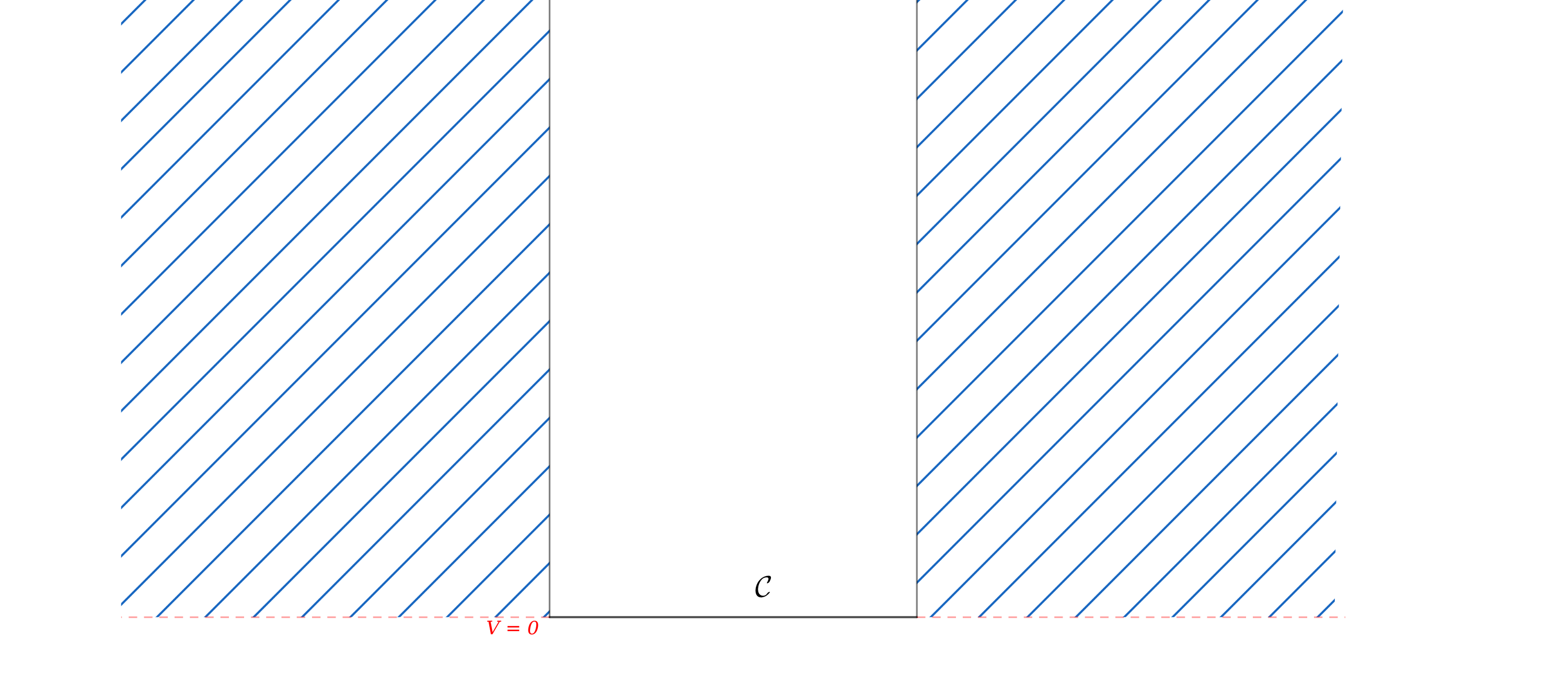}
\end{figure}
Therefore, the covariant derivative of a spinor field $u \in DM$ is:
\begin{equation*}
    \nabla u = (\nabla_b u^A)e^b \otimes E_A = (\partial_b u^A + \sigma_{\,\,\,\,b \, C}^A u^C)e^b \otimes E_C
\end{equation*}
while, for cospinor fields $v$, using the relation $\partial_a(vu) = (\nabla_a v)u + v (\nabla_a u)$:
\begin{equation*}
    \nabla v = (\nabla_b v_C)e^b \otimes E^C = (\partial_b v_C - v_A \sigma_{\,\,\,\,b \, C}^A)e^b \otimes E^C
\end{equation*}
Using the short-hand notation, dropping spinor indices, we rewrite the above as:
\begin{equation*}
    \nabla_b u = \partial_b u + \sigma_b u \hspace{20pt} \nabla_b v = \partial_b v - v \sigma_b 
\end{equation*}\\\\
Now that we have a notion of covariant derivation for spinors, we can study variation of tensors. The first important result, regards the variation of the Dirac gamma matrices that, as we now show, are constant with respect to the covariant differentiation:
\begin{lem}
The section $\gamma$ is covariantly constant
\end{lem}
\begin{proof}
Pick a specific local frame, to get $\gamma_{\,\,\,\,a \, B}^A$, and compute:
\begin{equation*}
    \nabla_b \gamma_{\,\,\,\,a \, B}^A = \sigma_{\,\,\,\,b\,C}^A \gamma_{\,\,\,\,a \,B}^C - \sigma_{\,\,\,\,b\,B}^C \gamma_{\,\,\,\,a \, C}^A - \Gamma^c_{ba} \gamma^A_{\,\,\,\,c \, B}
\end{equation*}
Dropping spinor indices:
\begin{align*}
    \nabla_b \gamma_a &= \sigma_b \gamma_a - \gamma_a \sigma_b - \Gamma^c_{ba} \gamma_c = \frac{1}{4} \Gamma^d_{bc}(\gamma_d \gamma^c \gamma_a - \gamma_a \gamma_d \gamma^c) - \Gamma^c_{ba} \gamma_c\\
    &= \frac{1}{4} \Gamma^d_{bc}([\gamma_d \gamma^c, \gamma_a] - 4 \delta^c_a \gamma_d)\\
    &= \frac{1}{4} \Gamma^d_{bc}(\gamma_d \{\gamma^c, \gamma_a\} - \{\gamma_d,\gamma_a\}\gamma^c - 4 \delta^c_a \gamma_d)\\
    &= \frac{1}{4} \Gamma^d_{bc}(2\gamma_d \delta^c_a - 2\eta_{ad}\gamma^c - 4 \delta^c_a \gamma_d)\\
    &= -\frac{1}{2} \Gamma^d_{bc}(\eta_{ad}\gamma^c + \delta^c_a \gamma_d)\\
    &= -\frac{1}{2}\gamma^c(\eta_{ad} \Gamma^d_{bc} + \eta_{cd} \Gamma^d_{ba})
\end{align*}
Where at the first step we have used the above definition of $\sigma$. But, after a standard but tedious computation using Eq, \eqref{eq: spte} with the fact that $\partial_{\mu} g_{\alpha \beta} = g_{\alpha \tau}\Gamma^{\tau}_{\mu \beta} + g_{\beta \tau}\Gamma^{\tau}_{\mu \alpha}$, one gets:
\begin{equation*}
    \eta_{ad} \Gamma^d_{bc} + \eta_{cd} \Gamma^d_{ba} = 0
\end{equation*}
From which it follows that:
\begin{equation*}
    \nabla_b \gamma_a = 0
\end{equation*}
\end{proof}
We will need this lemma later on, as covariant derivatives of $\gamma$-matrices will arise.\\
Now, we define the differential operator used in formulating the Dirac equation, called \textit{Dirac derivative}:
\begin{defn}
Let $\cancel{\nabla} : C^{\infty}(DM) \to C^{\infty}(DM)$ be a first order partial differential operator defined by:
\begin{equation*}
    \cancel{\nabla} := \gamma^a \nabla_a
\end{equation*}
we call it the Dirac operator and $\gamma^a$ is seen as a map from $DM$ to itself acting from the left.\\
Analogously we define $\cancel{\nabla} : C^{\infty}(D^*M) \to C^{\infty}(D^*M)$, still called the Dirac operator, via:
\begin{equation*}
    \cancel{\nabla} := \gamma^a \nabla_a
\end{equation*}
now $\gamma^a$ is seen as a map from $D^*M$ to itself acting from the right.
\end{defn}
\begin{rem}
The Dirac operator in a local frame becomes:
\begin{align*}
    \cancel{\nabla} u &= E_A(\cancel{\nabla} u)^A = E_A \gamma^{aA}_{\,\,\,\,\,\,B} \nabla_a u^B = E_A \gamma^{aA}_{\,\,\,\,\,\,B}(\partial_au^B + \sigma_{\,\,\,\,a \, C}^B u^C)\\
    \cancel{\nabla} v &=(\cancel{\nabla} v)_A E^A = (\nabla_a v_B)\gamma^{aB}_{\,\,\,\,\,\,A} E^A = (\partial_a v_B - v_C \sigma_{\,\,\,\,a \,B}^C) \gamma^{aB}_{\,\,\,\,A} E^A
\end{align*}
and by dropping spinor indices:
\begin{align*}
     \cancel{\nabla} u &= \gamma^a \nabla_a u = \gamma^a(\partial_a u + \sigma_a u)\\
    \cancel{\nabla} v &= (\nabla_a v) \gamma^a = (\partial_a v - v\sigma_a)\gamma^a
\end{align*}
\end{rem}
Now we have all that is needed to define unambiguously the Dirac equation:
\begin{defn}
The Dirac equation for $u \in C^{\infty}(DM)$, resp. $v \in C^{\infty}(D^*M)$, is:
\begin{align*}
    (-i\cancel{\nabla} + m)u &=0\\
    (i\cancel{\nabla} + m)v &= 0
\end{align*}
\end{defn}
In particular, if we know a solution, also its adjoint is:
\begin{lem} \label{lem: uscnorm}
For all spinor field $u$ and cospinor field $v$:
\begin{align*}
    (\cancel{\nabla} u)^{\sharp} = \cancel{\nabla} u^{\sharp} &\hspace{20pt} (\cancel{\nabla} v)^{\sharp} = \cancel{\nabla}v^{\sharp}\\
    (\cancel{\nabla} u)^c = -\cancel{\nabla}u^c &\hspace{20pt} (\cancel{\nabla} v)^c = -\cancel{\nabla}v^c
\end{align*}
Moreover, for any timelike future pointing vector $n$ on $M$ we have that $u^{\sharp}\cancel{n}u \geq 0$
\end{lem}
\begin{proof}
Using the definitions of the charge conjugation ($\overline{\pi(S)}C = C \pi(S)$), adjoint matrices ($\pi(S)^*A = A \pi(S^{-1})$) and the constancy of the entries of the matrices:
\begin{align*}
    (\cancel{\nabla} v)^c &= ((\partial_a v - v \sigma_a)\gamma^a)^c = (\partial_a \overline{v} - \overline{v \sigma_a})\overline{\gamma}^a C\\
    &= -(\partial_a(\overline{v}C) - \overline{v}C\sigma_a)\gamma^a = -\cancel{\nabla}(\overline{v}C) = -\cancel{\nabla}v^c\\
    (\cancel{\nabla} u)^{\sharp} &= (\gamma^a(\partial_a u + \sigma_a u))^{\sharp} = ((\partial_a u^* + u^*\sigma_a^* )(\gamma^a)^*)A\\
    &= (\partial_a(u^*A) - u^*A\sigma_a)\gamma^a = \cancel{\nabla}(u^*A) = \cancel{\nabla}u^{\sharp}
\end{align*}
Where the minus sign in the second line of the second equation, comes from the property of adjoint matrix $A$ cited above, therefore the two $\gamma$ matrices in the definition of $\sigma_a$ are switched and so we need to switch them back. For the other two identities:
\begin{align*}
    (\cancel{\nabla} v)^{\sharp} &= (\cancel{\nabla} v^{\sharp \sharp})^{\sharp} = (\cancel{\nabla} v^{\sharp})^{\sharp \sharp} = \cancel{\nabla} v^{\sharp}\\
    (\cancel{\nabla} u)^c &= (\cancel{\nabla} u^{\sharp})^{\sharp c} = - (\cancel{\nabla} u^{\sharp})^{c \sharp} = (\cancel{\nabla} u^{\sharp c})^{\sharp} = - (\cancel{\nabla} u^{c \sharp})^{\sharp} = - (\cancel{\nabla} u^c)
\end{align*}
where we have used the above results for spinors with $(\nabla v^{\sharp \sharp}) = (\nabla v^{\sharp})^{\sharp}$ as $v^{\sharp}$ is a spinor now.\\
Finally for $u(x) = [E,z]$ we can compute:
\begin{equation*}
    u^{\sharp}(x) \cancel{n} u(x) = \langle z, A \cancel{n} z\rangle \geq 0
\end{equation*}
That follows from the definition of $A$ in \ref{def: hacc}
\end{proof}

\subsection{Quantisation}\label{sec: quantisation}
Finally, by studying the space of solutions of the Dirac equation we will be able to quantize the theory. As we will see, the standard procedure for the scalar field needs to be modified as the space of solutions is endowed with an hermitian product instead of a symplectic form, i.e. we have an Hilbert instead of a symplectic space.\\
The most convenient approach is to start by introducing the \textit{double spinor bundle} as the following direct sum of vector bundles: $DM \oplus D^*M$. Of course a double spinor field is going to be a section of this vector bundle, and we shall denote the space of sections of it as:
\begin{equation*}
    \mathcal{D}(M) := C^{\infty}(DM \oplus D^*M)
\end{equation*}
We extend in particular the definitions of adjointeness and charge conjugation as pointwise maps: $(u \oplus v)^c := u^c\oplus v^c$ and $(u \oplus v)^{\sharp} = v^{\sharp} \oplus u^{\sharp}$. Of course a spinor field is just $u \oplus 0$.\\
Let us also define the first order differential operators:
\begin{align*}
    D &:= (-i\cancel{\nabla} + m) \oplus (i\cancel{\nabla} + m)\\
    \Tilde{D} &:= (i\cancel{\nabla} + m) \oplus (-i\cancel{\nabla} + m)
\end{align*}
So with this notation $f \in \mathcal{D}(M)$ is said to fulfill the Dirac equation if $Df=0$. Moreover, using the sesquilinear non-degenerate form:
\begin{equation*}
    \langle u_1 \oplus v_1, u_2 \oplus v_2 \rangle := \int_M (u_1^{\sharp} u_2 - v_2 v_1^{\sharp}) d\mathrm{vol_g}
\end{equation*}
We can then turn any $u_1 \oplus v_1 \in \mathcal{D}(M)$ into a distribution acting over the space of compactly supported sections.\\
Let us start proving some properties for the first order differential operator $D$:
\begin{lem}\label{lem: Dirac}
Consider $f \in \mathcal{D}_0(M)$ and $g \in \mathcal{D}(M)$, then:
\begin{enumerate}
    \item $Dg^{\sharp} = (Dg)^{\sharp}$, $Dg^c = (Dg)^c$ and finally $g^{c \sharp} = -g^{\sharp c}$
    \item $\langle f^{\sharp}, g^{\sharp}\rangle = \langle f^c, g^c \rangle = - \overline{\langle f, g \rangle} = - \langle g,f \rangle$
    \item $f^{\sharp}, f^c \in \mathcal{D}_0(M)$ and we have that $\langle f, Dg \rangle = \langle Df, g \rangle$
\end{enumerate}
\end{lem}
\begin{proof}
Let us call for simplicity $g = g_1 \oplus g_2$, then:
\begin{equation*}
    D g^{\sharp} = \big( ((-i\cancel{\nabla} + m) \oplus (i\cancel{\nabla} + m)) \big) (g_2^{\sharp} \oplus g_1^{\sharp}) =\big( (-i\cancel{\nabla} + m) g_2^{\sharp} \big) \oplus \big( (i\cancel{\nabla} + m) g_1^{\sharp}\big) = (D g)^{\sharp}
\end{equation*}
In the same way one proves the other identity for the charge conjugation, while:
\begin{equation*}
    g^{c \sharp} = g_2^{c \sharp} \oplus g_1^{c \sharp} = -\big( g_2^{\sharp c} \oplus g_1^{\sharp c} \big) = - \big( g_2^{\sharp} \oplus g_1^{\sharp} \big)^c = - \big( g_1 \oplus g_2 \big)^{\sharp c}
\end{equation*}
For what concerns the second statement, compute:
\begin{align*}
    \langle f^{\sharp}, g^{\sharp} \rangle &= \int_M f_2 g_2^{\sharp} - g_1^{\sharp} f_1\\
    &= \int_M f_2^{\sharp \sharp} g_2^{\sharp} - g_1^{\sharp} f_1^{\sharp \sharp}\\
    &= \int_M \overline{g_2 f_2^{\sharp}} - \overline{f_1 g_1^{\sharp}}\\
    &= - \overline{\langle f, g \rangle}
\end{align*}
The other identities are proven in the same way.\\
Finally for all $u \oplus v \in \mathcal{D}(M)$ we have (from the invariance of $\gamma$ matrices):
\begin{equation*}
    \nabla_a(v \gamma^a u) = (\cancel{\nabla} v) u + v (\cancel{\nabla} u)
\end{equation*}
Now, if either $u$ or $v$ are compactly supported (or at least vanish at infinity) we can integrate by parts to get:
\begin{equation*}
    \int_M (\cancel{\nabla} v)u = - \int_M v (\cancel{\nabla} u)
\end{equation*}
But then:
\begin{align*}
    \langle f, Dh \rangle &= \langle f_1 \oplus f_2, D(h_1 \oplus h_2) \rangle \\
    &= \int_M f_1^{\sharp} \big( (-i \cancel{\nabla} + m) h_1 \big) -  \big( (i\cancel{\nabla} + m ) h_2 \big) f_2^{\sharp}\\
    &= \int_M \big( (i\cancel{\nabla} + m)f_1^{\sharp} \big) h_1 - h_2 \big( (-i\cancel{\nabla} + m)f_2^{\sharp} \big)\\
    &= \langle Df, h \rangle
\end{align*}
\end{proof}
After these first properties, in order to study the existence and uniqueness of solutions of the Dirac equation, we define the second order differential operator $\Tilde{D} D = D \Tilde{D}$. For each factor of the direct sum we have:
\begin{equation*}
    (-i\cancel{\nabla} + m) (i \cancel{\nabla} + m) = (\cancel{\nabla} \cancel{\nabla} + m^2)
\end{equation*}
and since:
\begin{align*}
    \cancel{\nabla} \cancel{\nabla} u^c&= \gamma^a \nabla_a \gamma^b \nabla_b u^c = \gamma^a \gamma^b \nabla_a \nabla_b u^c\\
    &= (-\gamma^b \gamma^a + 2\eta^{ab}) \nabla_a \nabla_b u^c\\
    &= (-\gamma^b \gamma^a \nabla_a \nabla_b + 2 \Box)u^c\\
    &= -\gamma^b \gamma^a (\nabla_b \nabla_a u^c + R_{abd}^c u^d) + 2 \Box u^c\\
\end{align*}
So by relabeling of indices and putting the first factor on the left:
\begin{align*}
    \cancel{\nabla} \cancel{\nabla} u^c &= -\frac{1}{2} \gamma^b \gamma^a R_{abd}^c u^d + \Box u^c\\
    &= \frac{1}{2}\gamma^a \gamma^b R^c_{abd} u^d + \Box u^c
\end{align*}
But then, as the principal part of $D \Tilde{D}$ is given by $\Box$ we have that $D \Tilde{D}$ is an hyperbolic differential operator. Therefore, we can apply the following theorem to argue about existence of retarded and advanced propagators:
\begin{thm}
Let $M$ be a globally hyperbolic Lorentzian manifold. Let $P$ be a second order differential operator acting on sections in a vector bundle $E$ over $M$ defined in local coordinates as:
\begin{equation*}
    P = - \sum_{\mu, \nu = 1}^n g^{\mu \nu}(x) \frac{\partial^2}{\partial x^{\mu} \partial x^{\nu}} + \sum_{\mu = 1}^n A^{\mu}(x) \frac{\partial}{\partial x^{\mu}} + B_1(x)
\end{equation*}
Where $A^{\mu}(x), B_1(x)$ are matrix-valued functions over $M$.\\
Then, for each $x \in M$ there are unique fundamental solutions of $P$ denoted $E^{\pm}(x)$ with past respectively future compact support. They must also satisfy:
\begin{itemize}
    \item $supp(E^{\pm}(x)) \subset J^{\pm}(x)$
    \item for each test function $f \in \mathcal{D}_0(M,E^*)$ the map $x \mapsto E^{\pm}(x)[f]$ are smooth sections in $E^*$, satisfying:
    \begin{equation*}
        P^* (E^{\pm}(\cdot)[f]) = f
    \end{equation*}
    Where $P^*$ is the same differential operator just acting on elements in $E^*$.
\end{itemize}
\end{thm}
\begin{proof}
See Theorem $3.3.1$ in \cite{Baer}.
\end{proof}

From the physical point of view, the motivations behind what we have just introduced, are the same as that on Minkowski spacetime when formulating the Dirac equation as the differential equation describing the motion of a relativistic particle. Namely, we have shown that a solution of the Dirac equation is in particular a solution of a "Klein-Gordon like" equation giving in this way the right dispersion relations for the solutions. In the same manner, from the mathematical side, using the existence of advanced and retarded propagators for the operator $D \Tilde{D}$ gives also advanced and retarded propagators for the Dirac operator:
\begin{thm}\label{thm: prop}
The maps $S^{\pm} : \mathcal{D}_0(M) \to \mathcal{D}(M)$ defined by $S^{\pm} := \Tilde{D}E^{\pm}$ are the unique advanced $S^-$ and retarded $S^+$ propagators for $D$, such that:
\begin{itemize}
    \item $supp(S^{\pm} f) \subset J^{\pm}(supp f)$ for all $f \in \mathcal{D}_0(M)$
    \item Moreover, $S^{\pm} f^{\sharp} = (S^{\pm} f)^{\sharp}$ and $S^{\pm} f^c = (S^{\pm}f)^c$
    \item For all $f,h \in \mathcal{D}_0(M)$ we have:
    \begin{equation*}
        \langle f, S^{\pm} h \rangle = \langle S^{\mp} f, h \rangle
    \end{equation*}
\end{itemize}
\end{thm}
\begin{proof}
For any $f \in \mathcal{D}_0(M)$ we see, from the above theorem, that $D S^{\pm} f = f$ and also the claimed support properties follow from the support properties of $E^{\pm}$ mentioned in the previous theorem. This proves that $S^{\pm}$ are right fundamental solutions.\\
Now, consider $f, h \in \mathcal{D}_0(M)$ and, as a consequence of the global hyperbolicity, that $supp(S^{\mp}f) \cap supp(S^{\pm} h)$ is compact. Then, we can perform the following integration by parts throwing away boundary terms, to show the last statement:
\begin{align*}
    \langle f, S^{\pm} h \rangle &= \langle D S^{\mp} f, S^{\pm} h \rangle \\
    &= \langle S^{\mp} f, D S^{\pm} h \rangle = \langle S^{\mp} f,h \rangle
\end{align*}
Indeed, by taking the adjoint of $D S^{\pm} = \mathbb{1}$, we even get $S^{\mp} D = \mathbb{1}$ showing that $S^{\pm}$ is also a left fundamental solution.\\
Finally, compute:
\begin{align*}
    \langle S^{\pm} f^{\sharp}, h \rangle &= \langle f^{\sharp}, S^{\mp} h \rangle = \langle (D S^{\pm} f)^{\sharp}, S^{\mp} h \rangle\\
    &= \langle D ( S^{\pm} f)^{\sharp}, S^{\mp} h \rangle = \langle (S^{\pm} f)^{\sharp}, D S^{\mp} h \rangle =  \langle (S^{\pm} f)^{\sharp}, h \rangle
\end{align*}
showing that $S^{\pm} f^{\sharp} = (S^{\pm} f)^{\sharp}$.
\end{proof}
With the advanced and retarded propagators, we can define, as usual, that operator that takes a test function into a solution of $D$: the \textit{causal propagator}, defined as $S := S^- - S^+$.
\begin{rem}
Of course, one can define and prove, as we did on $\mathcal{D}(M)$, the separate existence and properties of the propagators over spinors $DM$ and cospinors $D^*M$. Let us denote them as $S^{\pm}_{sp}, S_{sp}, S^{\pm}_{cosp}, S_{cosp}$. Moreover, by the uniqueness, we have for $u \in DM$:
\begin{equation*}
    S = S_{sp} \oplus S_{cosp} \Rightarrow (S_{sp} u)^{\sharp} = S_{cosp} u^{\sharp}
\end{equation*}
In addition, it holds:
\begin{equation*}
    \int_M v(S_{sp} u) = - \int_M(S_{cosp} v) u
\end{equation*}
This can be seen by taking $V,U \in \mathcal{D}_0(M)$, defined as $V = v^{\sharp} \oplus 0$ and $U = u \oplus 0$ where $u \in DM$ and $v \in D^*M$, and computing:
\begin{align*}
    \langle V, S U \rangle &= \langle (v^{\sharp} \oplus 0), ((S_{sp} u) \oplus 0) \rangle\\
    &= \int_M v (S_{sp} u)
\end{align*}
But at the same time $ \langle V, S U \rangle = -  \langle S V, U \rangle$ and:
\begin{align*}
    -\langle S V, U \rangle &= -\langle (S_{sp} v^{\sharp} \oplus 0), (u \oplus 0) \rangle\\
    &= - \int_M (S_{sp} v^{\sharp})^{\sharp} u\\
    &= - \int_M (S_{cosp} v) u
\end{align*}
\end{rem}
Before coming to quantization, we need one more step regarding the space of solutions of Dirac's equation. Namely, we show that, using the propagator, we can define an inner product over this space. In this way, by completing this same space with respect to the just introduced inner product, we make the space of solutions a Hilbert space. As we'll see, introducing such a structure is needed in order to define anticommutation relations between the field operators. The necessity can already be infered, recalling the definition of self-dual CAR algebras:
\begin{lem} \label{lem: HilbD}
For the propagator $S: \mathcal{D}_0(M) \to \mathcal{D}(M)$ we have $\ker S = D(\mathcal{D}_0(M))$. Moreover, the bilinear map:
\begin{equation*}
    (f,h) := i \langle f, S h \rangle 
\end{equation*}
defines an inner product over $\mathcal{D}_0(M) \backslash \ker S$. The inner product is hermitian and fulfills the following sequence of equalities:
\begin{equation*}
    (f^{\sharp}, h^{\sharp}) = (f^c,h^c) = \overline{(f,h)} = (h,f)
\end{equation*}
\end{lem}
\begin{proof}
Of course, if $f = D h$ for some $h \in \mathcal{D}_0(M)$, then $S f = 0$ from what we have said above regarding the advanced and retarded propagators. It follows that $D(\mathcal{D}_0(M)) \subset \ker S$. Conversely, if $S f = 0$ with $f \in \mathcal{D}_0(M)$, define $h := S^- f = S^+ f$ that has support in $J^+(supp\,\, f) \cap J^-(supp\,\, f)$. But then, from the definition of the retarded and advanced propagators, we also have $f = Dh$ giving $\ker S \subset  D(\mathcal{D}_0(M))$.\\
For what concerns the last sequence of equalities, from Lemma $\ref{lem: Dirac}$ and from Theorem \ref{thm: prop} we must have:
\begin{align*}
    (f,h) &= -i\langle S f, h \rangle\\
    &= -i \overline{\langle h, Sf \rangle}\\
    &= \overline{i \langle h, S f \rangle}\\
    &= \overline{(h,f)}
\end{align*}
and this must be a well defined sesquilinear map over $\mathcal{D}_0(M) \backslash \ker S$. Still from the previous Theorem $\ref{thm: prop}$, we have that adjoint and charge conjugation translate to operators on the factors of the inner product. In particular, denoting for simplicity both charge conjugation and adjointeness by $(\cdot)^*$, using Lemma \ref{lem: Dirac}:
\begin{equation*}
    (f^*, h^*) = i \langle f^* , Sh^* \rangle = i \langle f^*, (S h)^* \rangle = -i \overline{\langle f, S h \rangle} = \overline{(f,h)} = (h,f)
\end{equation*}
What remains to be proven, is the positivity of the sesquilinear form  $(f,f) \geq 0$ with equality if and only if $Sf = 0$. Compute:
\begin{align*}
    (u \oplus v, u \oplus v) &= i \int_M u^{\sharp} (S_{sp} u) - (S_{cosp} v) v^{\sharp}\\
    &= i \int_M u^{\sharp} (S_{sp} u) + v (S_{sp} v^{\sharp})
\end{align*}
Let us choose a Cauchy surface $\Sigma \subset M$ and denote by $D_{sp} := -i\cancel{\nabla} + m$:
\begin{align*}
    i \int_M u^{\sharp} (S_{sp} u) &= i \int_{J^+(\Sigma)} (D_{sp} S^-_{sp} u)^{\sharp} (S_{sp} u) + i \int_{J^-(\Sigma)} (D_{sp} S^+_{sp} u)^{\sharp} (S_{sp} u)\\
    &= i \int_{J^+(\Sigma)} ( S^-_{sp} u)^{\sharp} (D_{sp} S_{sp} u) + i\nabla_a \big( (S^-_{sp}u)^{\sharp} \gamma^a S_{sp} u \big)\\
    &\,\, i \int_{J^-(\Sigma)} ( S^+_{sp} u)^{\sharp} (D_{sp} S_{sp} u) + i\nabla_a \big( (S^+_{sp}u)^{\sharp} \gamma^a S_{sp} u \big)\\
    &= - \int_{J^+(\Sigma)} \nabla_a \big( (S^-_{sp}u)^{\sharp} \gamma^a S_{sp} u \big) - \int_{J^-(\Sigma)} \nabla_a \big( (S^+_{sp}u)^{\sharp} \gamma^a S_{sp} u \big)\\
    &= \int_{\Sigma} n_a(S^-_{sp} u - S^{+}_{sp} u)^{\sharp} \gamma^a S_{sp} u = \int_{\Sigma} (S_{sp} u)^{\sharp} \cancel{n} S_{sp} u
\end{align*}
Where we've used $D_{sp} S_{sp} = 0$ at the third step and the Stokes theorem at the fourth, introducing the future pointing normal to $\Sigma$. However, we have seen in Lemma \ref{lem: uscnorm} that the last integrand is a positive function vanishing only if the argunet is zero.
\end{proof}
We define the completion, with respect to this inner product, to be $\mathcal{H}(M) := \overline{\mathcal{D}_0(M) \backslash \ker S}$: the Hilbert space of solutions of the Dirac equation.\\\\
After this very long but necessary introduction of all these structures, we are finally able to quantise the classical theory leading the a free Dirac QFT on the curved globally hyperbolic background. 

\subsubsection{Dirac Quantum fields}
In particular, we define the \textit{Dirac algebra} $\mathfrak{A}_{DIRAC}$, as the unique (up to $*$-isomorphism) $C^*$-algebra generated by elements $B(f)$, with $f \in \mathcal{H}(M)$, such that\footnote{Notice that as $\mathcal{H}(M) = \overline{\mathcal{D}_0(M) \backslash \ker S}$ the dynamics determined by the Dirac equation is implicitly included into the algebra}:
\begin{enumerate}
    \item $f \mapsto B(f)$ is $\mathbb{C}$-linear 
    \item $B(f^{\sharp}) = B(f)^*$
    \item $[ B(f)^*, B(g)]_+ = (f,g) \mathbb{1}$
\end{enumerate}
What we immediately notice, is that the Dirac algebra is a self-dual CAR algebra, as defined in Def. \ref{def: self}, constructed over the Hilbert space of solutions of the Dirac equation $\mathcal{H}(M)$. In particular, the involution over such a Hilbert space is given by the Dirac adjoint, i.e. $\Gamma f := f^{\sharp}$, that we have proven in Lemma $\ref{lem: HilbD}$ to fulfill all the properties, reported in Def. \ref{def: self}, to be an antiunitary involution. Therefore:
\begin{equation*}
    \mathfrak{A}_{DIRAC} = \mathfrak{A}_{SDC}(\mathcal{H}(M), \,^{\sharp}\,)
\end{equation*}
But, as already discussed, a self-dual CAR algebra is a $C^*$-algebra, proving in this way that also the \textit{Dirac algebra} is. For the uniqueness, suppose that we have another such Dirac algebra built on $\mathcal{H}(M)$, with the same algebraic relations mentioned above $\mathfrak{A}'_{DIRAC}$. Call $B_1(f)$ the elements of $\mathfrak{A}_{DIRAC}$ and by $B_2(f)$ the elements of $\mathfrak{A}'_{DIRAC}$. Define a map $\alpha: \mathfrak{A}_{DIRAC} \to \mathfrak{A}'_{DIRAC}$:
\begin{equation*}
    \alpha(B_1(f)) = B_2(f)
\end{equation*}
This map preserves the algebraic relations between the spaces and is a $*$-isomorphism. Therefore, the algebra is unique up to $*$-isomorphisms.
\begin{rem}
The proof of the Dirac algebra being a $C^*$-algebra allows us, as explained in the Section \ref{sec: Axiom}, to define a QFT. In particular, the localization of the field operators is given by the support of the corresponding test functions $f \in \mathcal{H}(M)$. At the end of the section I'll argue about the fulfillement of the axioms.
\end{rem}

Let us introduce over our Dirac algebra a charge conjugation map that, as we will see once we define it, distinguishes the Dirac from the Majorana field and defines a grading over $\mathfrak{A}_{DIRAC}$:
\begin{prop}
The map $f \mapsto f^{c \sharp} := T f$ raises to a $*$-automorphism to the algebra $\mathfrak{A}_{DIRAC}$ determined by:
\begin{equation*}
    \alpha_C(B(f)) = B(f^{c \sharp})
\end{equation*}
that squares to $\alpha_C^2(B(f)) = - B(f)$
\end{prop}
\begin{proof}
First, from what we have seen previously, we have that the map is an isometry:
\begin{equation*}
   (T f, Tg) = (f^{c \sharp}, g^{c \sharp}) = (g^c, f^c) = (f, g)
\end{equation*}
From here, one can define a spceific type of Bogolubov transformations that imply the existence of the $*$-automorphism $\alpha_C$, for details see section $5.2.2.1$ of \cite{Bratteli:1996xq}.\\
Finally, for what concerns the square:
\begin{align*}
    \alpha_C^2(B(f)) &= B(f^{\sharp c \sharp c})\\
    &= B(-f^{\sharp \sharp c c})\\
    &= -B(f)
\end{align*}
\end{proof}
The way we have introduced the local algebra of Dirac fields, might be unfamiliar with the standard notation that is used in most textbooks of QFT. For this reason, to establish a connection with the standard approaches, we rewrite our operators introducing the following notation:
\begin{defn}
Define the maps $\psi: C^{\infty}_0(D^*M) \to \mathfrak{A}_{DIRAC}$ and $\psi^{\sharp}: C^{\infty}_0(DM) \to \mathfrak{A}_{DIRAC}$ by:
\begin{equation*}
    \psi(h) := B(0 \oplus h), \hspace{20pt} \psi^{\sharp}(f) := B(f \oplus 0)
\end{equation*}
We also define:
\begin{align*}
    \psi^c(h^c) &:= \alpha_C(B(0 \oplus h^c)) = \psi(h)^*\\
    \psi^{\sharp c}(f^c) &:= \alpha_C(B(f^c \oplus 0)) = \psi^{\sharp}(f)^*
\end{align*}
\end{defn}
Let us now prove that these are operators valued distributions and satisfy the standard anticommutation relations:
\begin{prop}
The maps $B$, $\psi$ and $\psi^{\sharp}$ are operator valued distributions in the $C^*$-algebra $\mathfrak{A}$ and:
\begin{enumerate}
    \item $\psi^{\sharp}(f) = \psi(f^{\sharp})^*$
    \item $[ \psi^{\sharp}(f), \psi(h) ]_+ = (f^{\sharp} \oplus 0, h \oplus 0)\mathbb{1} = i \int_M f(S_{sp}h)\mathbb{1}$ and all other anticommutator vanish
    \item We have:
    \begin{equation*}
        (-i\cancel{\nabla} + m)\psi = 0 \hspace{20pt} (i\cancel{\nabla} + m) \psi^{\sharp} = 0
    \end{equation*}
    Where the action of the differential operator on the operator valued distribution is defined as usual for any distribution $U$ on smooth sections $f \in D^*M$ as:
    \begin{equation*}
        (\nabla_a U)(f) = -U(\nabla_a f) \hspace{20pt} (\gamma_aU)(f) = U(f \gamma_a)
    \end{equation*}
\end{enumerate}
\end{prop}
\begin{proof}
For the first statement:
\begin{align*}
    \psi^{\sharp}(f) &= B(f \oplus 0)\\
    &= B(f^{\sharp \sharp} \oplus 0)\\
    &= B(0 \oplus f^{\sharp})^*
\end{align*}
For what concerns the second statement let us compute:
\begin{align*}
    [ \psi^{\sharp}(f), \psi(h)]_+ &= [ B(f \oplus 0), B(h^{\sharp} \oplus 0)^*]_+\\
    &= (h^{\sharp} \oplus 0, f\oplus 0)\\
    &= i \int_M h(S_{sp} f) \mathbb{1}
\end{align*}
Where we have used the result of the first point.\\
For the third statement compute:
\begin{align*}
    ((-i\cancel{\nabla} +m) \psi)(h) &= \psi((i\cancel{\nabla} + m)h)\\
    &= B(D(0 \oplus h)) = 0
\end{align*}
Where we have used the fact that $D(0 \oplus h) \in \ker S$ and the fact that the arguments of $B$ are in $\mathcal{H}(M)$. With the same procedure, one proves the same to hold also for $\psi^{\sharp}$.\\
It remains to show that $B$, $\psi$ and $\psi^{\sharp}$ are operator valued distributions. Consider the $C^*$-subalgebra generated by $\mathbb{1}$, $\psi(h)$ and $\psi(h)^*$. This is a Clifford algebra and we can also show it to be isomorphic to $M(2,\mathbb{C})$. Define for it the following map:
\begin{equation*}
    \psi(h) \mapsto \begin{pmatrix}
    0 & \sqrt{c}\\
    0 & 0
    \end{pmatrix} \hspace{20pt} \psi(h)^* \mapsto \begin{pmatrix}
    0 & 0\\
    \sqrt{c} & 0
    \end{pmatrix}
\end{equation*}
Where $c = (0\oplus h, 0 \oplus h) = i \int_M h(S_{sp} h^{\sharp}) d\mathrm{vol}_g > 0$. This is an isomorphism as the injectivity follows from the fact that $c = 0 $ iff $h=0$ and the surjectivity follows from the nullity rank theorem. Then, if we look at the norm:
\begin{align*}
    \| \psi(h) \| &= \| \begin{pmatrix}
    0 & \sqrt{c}\\
    0 & 0
    \end{pmatrix} \|_{op}\\
    &= \sqrt{c}
\end{align*}
So:
\begin{equation*}
    \| \psi(h) \|^2 = i \int_M h(S_{sp}h^{\sharp}) d \mathrm{vol}_g = c
\end{equation*}
This shows that this is a $C^*$-norm as:
\begin{equation*}
    \psi(h) \psi(h)^* \mapsto \begin{pmatrix}
    0 & \sqrt{c}\\
    0 & 0
    \end{pmatrix} \begin{pmatrix}
    0 & 0\\
    \sqrt{c} & 0
    \end{pmatrix} = \begin{pmatrix}
    c & 0\\
    0 & 0
    \end{pmatrix}
\end{equation*}
so:
\begin{equation*}
    \| \psi(h) \psi(h)^* \| = \| \begin{pmatrix}
    c & 0\\
    0 & 0
    \end{pmatrix} \|_{op} = c
\end{equation*}
Showing that $\| \psi(h) \psi(h)^* \| = \| \psi(h) \|^2$ and that the algebra genrated by  $\mathbb{1}$, $\psi(h)$ and $\psi(h)^*$ with this norm is a $C^*$-algebra.\\
But then, in the topology of test spinors, we have that $h \mapsto h \oplus h^{\sharp} \mapsto i \int_M h(S_{sp} h^{\sharp}) d\mathrm{vol}_g$ are all continuous: it follows, from the above expression of $\| \psi(h) \|^2$, that the map $h \mapsto \psi(h)$ is norm continuous. This proves that $\psi$ must be an operator valued distribution and the corresponding operators are in the abovely mentioned $C^*$-algebra, as we have proven it to be a continuous linear functional over spinors. The proof for $\psi^{\sharp}$ is anaogous and as a consequence follows the result for $B$.
\end{proof}

We still need to argue why the $C^*$-algebra $\mathfrak{A}_{DIRAC}$ defined above defines an AQFT, i.e. it fulfills the listed axioms in Section \ref{sec: Axiom}. Notice that, as a $C^*$-algebra, whenever we choose a state $\omega$ over $\mathfrak{A}_{DIRAC}$ we will always get a representation $\pi_{\omega}$ of it as bounded operators over a Hilbert space $\mathcal{H}_{\omega}$ by the \textit{GNS Theorem} \ref{thm: GNS}.\\
As mentioned earlier, the localization of algebras is given by the supports of the test-spinors, i.e. for $\mathcal{O}_1 \subset M$ we define $\mathfrak{A}_{DIRAC}(\mathcal{O}_1)$ as the $C^*$-subalgebras generated by the elements $B(f)$ for all those $f \in \mathcal{H}(M)$ such that $supp \, f \subset \mathcal{O}_1$. But then, if we take $\mathcal{O}_2 \subset \mathcal{O}_1$ as $supp \, f \subset \mathcal{O}_2 \Rightarrow supp \, f \subset \mathcal{O}_1$ we must also have $\mathfrak{A}_{DIRAC}(\mathcal{O}_2) \subset \mathfrak{A}_{DIRAC}(O_2)$, proving $\mathbf{A1}$. For what concerns causality, let us first notice that, once $\mathfrak{A}_{DIRAC}$ is represented, the canonical anticommutation relations become:
\begin{equation*}
    [ \pi_{\omega}(B(f)^*), \pi_{\omega}(B(g)) ]_+ = [\pi_{\omega}(B(f)^*), \pi_{\omega}(B(g))]_{\Gamma} = (f,g)\mathbb{1}
\end{equation*}
Where we have introduced the graded commutator as $\mathfrak{A}_{DIRAC}$ is a $\mathbb{Z}_2$-graded algebra with grading given by the $*$-automorphism represented by $\gamma(B(f)) = B(\pi(-\mathbb{1})f) = -B(f)$ (see Section II.$2$ in \cite{DAntoni:2001ido}). But then, from the definition of the inner product, if the support of $f$ is spacelike separated from that of $g$ the RHS of the above expression vanishes proving $\mathbf{A2'}$. The third axiom is fulfilled by raising the eventual existing symmetries of the spacetime as automorphisms over the algebra using similar argument as those used for $\alpha_C$. The validity of last axiom, in the general context of curved spacetimes, depends on the existence of a timelike symmetry. However, in the case of Minkowski, we can find a unique quasifree ground state over $\mathfrak{A}_{DIRAC}$, whose \textit{GNS construction} gives the standard vacuum vector that is left invariant by Poincaré transformations.\\\\

\subsubsection{Majorana Quantum fields}
Finally let us discuss how to get a Majorana quantum field theory. In particular, we need to figure out how to impose the "reality" Majorana condition.\\
First of all, consider just the space of test $C^{\infty}_0(DM) \backslash \ker S_{sp}$ and equip it with an inner product $(u_1,u_2)' := (u_1 \oplus 0, u_2 \oplus 0)$. then, we complete it to a Hilbert also in this case: $\mathcal{H}'(M) = \overline{C^{\infty}_0(DM) \backslash \ker S_{sp}}$. The map $(\cdot)^c$ defines a conjugation map over this space. Now we quantise the classical space of solutions, to obtain a \textit{Majorana algebra} $\mathfrak{A}_{MAJ}$, as a $C^*$-algebra generated by elements $B'(f)$ for $f \in \mathcal{H}'(M)$, satisfying:
\begin{enumerate}
    \item $f \mapsto B'(f)$ is $\mathbb{C}$-linear
    \item $B'(f^c) = B'(f)^*$\\
    Notice that in this case we use the $(\cdot)^c$ and not the hermitian conjugation as for Dirac spinors, since our operator valued distributions are defined just on spinors
    \item $[ B'(f_1)^*, B'(f_2) ]_+ = (f_1, f_2)'\mathbb{1}$
\end{enumerate}
For cospinors $h$ we define $B'(h) := B'(h^{\sharp})^*$. This, however, adds nothing new to the algebra as we were able to express it just in terms of spinors. Although, for them, we have a difference from the Dirac cospinor case as now $[ B'(h)^*, B'(f) ]_+ = (h^{\sharp c}, f)' \mathbb{1}$ while in the case of Dirac:
\begin{equation*}
    [ \psi(h)^*, \psi^{\sharp}(f) ]_+ = [ B(0 \oplus h)^*, B(f \oplus 0) ]_+ = 0
\end{equation*}
That, represents the anticommutation, in contrast, of Dirac field operators $\psi^{\sharp}(h^{\sharp})$ with $\psi^{\sharp}(f)$ (same holding for $\psi(h)$ and $\psi(f^{\sharp})$).\\
As for the Dirac algebra, we see:
\begin{equation*}
    \mathfrak{A}_{MAJ} = \mathfrak{A}_{SDC}(\mathcal{H}'(M), \, ^c \,)
\end{equation*}
Therefore, also the Majorana algebra is a $C^*$-algebra.\\
For what concerns the map $\alpha_C$, in this case it acts as an identity on Majorana spinors:
\begin{align*}
    \alpha_C (B'(f)) &= B'(f^{\sharp c})\\
    &= B'(f^{\sharp})^* = B'(f)
\end{align*}
Finally, in order to recast the standard field notation, we introduce the following notation for Majorana fields:
\begin{equation*}
    \psi(f) := B'(f)
\end{equation*}
So that $\psi(f)^* = \psi(f^c)$.\\\\

\subsubsection{Alternative, isomorphic, construction of the algebra of Dirac fields}\label{sec: alternative}
The Hilbert space of solutions can be constructed also in an alternative way, leading to an isometric Hilbert space and a  Dirac algebra isomorphic to the one previously constructed. Here, we briefly want to outline this alternative construction as will be used in Chapter 3, in constructing the algebra for a specific example.\\
We start considering a Cauchy surface $\Sigma$, exisitng by the global hyperbolicity assumption, and consider the space of test spinors and cospinors over $\Sigma$: $\mathcal{D}_0(\Sigma) := C^{\infty}_0(D\Sigma \oplus D^*\Sigma)$. On this space, we introduce the following sesquilinear, non-degenerate form:
\begin{equation*}
    \langle u_1 \oplus v_1, u_2 \oplus v_2 \rangle_{\Sigma} := \int_{\Sigma} (u_1^{\sharp} u_2 - v_2 v_1^{\sharp}) d\mathrm{vol_h}
\end{equation*}
Where we are denoting with $h$ the restriction of the metric $g$ on $\Sigma$, as a Riemannian metric $h$ (See Theorem \ref{thm: globhy}). As done before, we can reformulate and prove Lemma \ref{lem: Dirac} for this sesquilinear form and for analogous definitions of $D$, $\Tilde{D}$ and Dirac adjoint and charge conjugation.\\
Then, we define an inner product as follows:
\begin{lem}
Te bilinear map:
\begin{equation*}
(f,h)_{\Sigma} := \langle f, \cancel{n} h \rangle_{\Sigma}
\end{equation*}
with $n$ the unit timelike, forward normal to the Cauchy surface $\Sigma$, defines an hermitian inner product. Moeover, it fulfills the following sequence of equalities for $f,h \in \mathcal{D}_0(\Sigma)$:
\begin{equation*}
    (f^{\sharp}, h^{\sharp})_{\Sigma} = (f^{c}, h^{c})_{\Sigma} = \overline{(f,h)_{\Sigma}} = (h, f)_{\Sigma}
\end{equation*}
\end{lem}
\begin{proof}
Define a map $\beta(f) = \Tilde{f}$, with $f \in \mathcal{D}(M)$ and $\Tilde{f} \in \mathcal{D}(\Sigma)$, as the map that projects a spinor in its component over $\Sigma$. Furthermore, this defines, taking the composition with $S$, the map $\beta \circ S$ on $\mathcal{D}_0(M) \backslash \ker(S)$. Such a map, associates to a test spinor $f$ the corresponding initial datum, over $\Sigma$, of the solution $Sf$ of the Dirac equation. This association is unique (See Theorem $2.3$ in \cite{Dimock1982DiracQF}). Moreover, by existence and uniqueness of solutions, the association is unique also in the opposite direction. Namely, given $\Tilde{f} \in \mathcal{D}_0(\Sigma)$, there is a unique solution of the Dirac equation associated to it and in particular a unique $f \in \mathcal{D}_0(M)\backslash \ker(S)$ such that the solution is $Sf$. Let us call this bijective  map:
\begin{align*}
    K: \mathcal{D}_0(M)\backslash \ker(S) &\to \mathcal{D}_0(\Sigma)\\
    f &\mapsto \beta(S(f)) =: f_0
\end{align*}
For notational convenience, as $S = S_{sp} \oplus S_{cosp}$, we also decompose $K = K_{sp} \oplus K_{cosp}$.\\
Now, using this map, we show that $(\cdot, \cdot)_{\Sigma}$ defines an inner product over $\mathcal{D}_0(\Sigma)$, from the knowledge of $(\cdot, \cdot)$ being an inner product over $\mathcal{D}_0(M) \backslash \ker(S)$. Take $\mathcal{D}_0(M)\backslash \ker(S) \ni f = f^{(1)} \oplus f^{(2)}$, then:
\begin{align*}
    (f_0, g_0)_{\Sigma} &= (K f, K g)_{\Sigma}\\
    &= \int_{\Sigma} [(K_{sp}f^{(1)})^{\sharp} \cancel{n}K_{sp}g^{(1)} - K_{cosp}g^{(2)} \cancel{n}(K_{cosp}f^{(2)})^{\sharp}]\\
    &= i \int_M [(f^{(1)})^{\sharp} S_{sp}g^{(1)} + g^{(2)} (S_{cosp}f^{(2)})^{\sharp}] = (f,g)
\end{align*}
Where at the third step we have used the identities:
\begin{align*}
    i \int_M (f^{(1)})^{\sharp} (S_{sp} g^{(1)}) &= \int_{\Sigma} (S_{sp} f^{(1)})^{\sharp} \cancel{n} S_{sp} g^{(1)} = \int_{\Sigma} (K_{sp} f^{(1)})^{\sharp} \cancel{n} K_{sp} g^{(1)}\\
    i \int_M g^{(2)} (S_{cosp} f^{(2)})^{\sharp} &= -\int_{\Sigma} (S_{cosp} g^{(2)}) \cancel{n}(S_{cosp} f^{(2)})^{\sharp} = -\int_{\Sigma} (K_{cosp} g^{(2)})\cancel{n} (K_{cosp} f^{(2)})^{\sharp}
\end{align*}
that are derived with the procedure outlined in the proof of Lemma \ref{lem: HilbD}. Therefore, holding this for any $f_0, g_0 \in \mathcal{D}_0(\Sigma)$ and since the analogous statements of the Lemma hold for the inner product $(\cdot, \cdot)$ over $\mathcal{D}_0(M) \backslash \ker(S)$, follows the claim.

\end{proof}
As a consequence of this lemma, we define the Hilbert space $\mathcal{H}(\Sigma)$ by completing $\mathcal{D}_0(\Sigma)$ with respect to this inner product. Therefore, $K$ becomes an isometry between the Hilbert spaces: $K: \mathcal{H}(\Sigma) \to \mathcal{H}(M)$. As a further consequence, if we define the self-dual CAR algebra associated to $\mathcal{H}(\Sigma)$ with involution $\Gamma$ (respectively Dirac adjoint or charge conjugation for Dirac or Majorana field operators), we obtain an algebra $\mathfrak{A}_{SDC}(\mathcal{H}(\Sigma), \Gamma)$ that is isomorphic to $\mathfrak{A}_{SDC}(\mathcal{H}(M), \Gamma)$ (See Section $2.4.2$ in \cite{Hollands:2017dov}).\\
Finally, one can show that these construction do not depend on the initial choice of Cauchy surface $\Sigma$, see Theorem $4.2$ in \cite{Dimock1982DiracQF}.

\newpage
\chapter{The role of entropy and information in modern physics}
Information theory formalizes the way to store and process information. It is a well established theory on its own, namely without introducing or relating it to any discrete or continuous systems. In fact, its classical formulation lead to most of the big progresses in engineering, telecomunications and computer science. Even more amazingly, in the second half of the last century, it turned out that to overcome our ignorance about quantum processes, the generalization of classical information theory to the quantum scale provided a theoretical understanding of aspects of the quantum theory giving birth to the era of Quantum Technologies.\\
One of the key concepts in it is entropy: information about the system hidden in the collection of degrees of freedom too small and numerous to keep track of all of them. Which role does it play in modern physics? A way to explain it is, for example, in the context of Black Holes is via the gedankenexperiment of the hot cup of coffee. Namely, take a cup of hot coffee and throw it into a Black Hole. What happens to the entropy of the cup once it has fallen into the Black Hole, if nothing can escape from it? If we say that it is simply lost, we will be violating the second law of thermodynamics. Therefore, there must be some compensation in order for an external observer to measure an increase or at most a conservation of entropy from this process. This was in fact proven to be the case via the so called Generalized Second Law of Black Hole dynamics stating that in each process the sum of the external entropy and the Black Hole area must always increase, suggesting the assignment of entropy to a Black Hole via its area.\\
Hence, our statistical and information-theoretical interpretation of entropy together with (thanks to Hawking radiation and Black Hole thermodynamics) the above stated relation to a geometrical quantity like the area of the Black Hole, most likely constitutes the first true equation of an hypothetical theory of Quantum Gravity. As such, it lead to speculations on how to derive it within the zoo of the existing theories of quantum gravity. On a more humble and cautious perspective, understanding how processes leading to the variation of the Black Hole entropy, like evaporation due to Hawking radiation, have arisen as fundamental questions.\\\\
Therefore, the aim of this chapter, is to rigorously introduce such notions in the framework of QFT and, more precisely, at the level of local algebras. We will see, that relative entropy is the only information theoretical notion that is well defined in QFT and is interpreted as a measure of distinguishability among functionals on local algebras of observables. The ultimate goal is thus to introduce the Tomita-Takesaki modular theory, the Araki formula for the relative entropy and prove some properties for it. For related literature, I mainly refer to \cite{Vedral:2002zz} for a review on Information theory and on the original works of Araki \cite{Araki1976}, \cite{Araki1977}, the book of Takesaki \cite{Takesaki:1970aki}, the work of Uhlmann \cite{Uhlmann}, and the recent intuitive introduction \cite{Witten:2018zxz}, meant for physicists, on Tomita-Takesaki theory and Araki's relative entropy.

\section{Quantum information theory}
This section aims at reviewing the basics of quantum information theory, in the extent that is needed in order to discuss in the next sections the Araki's relative entropy. I will assume the reader to be familiar with classical information theory, and I refer to \cite{Vedral:2002zz} for a review of it. Let me just mention here that the Shannon entropy is interpreted in classical information theory as an averaged measure of uncertainty (represented by the logarithm of the inverse of the classical probabilities) and that the Shannon relative entropy is a measure of distinguishability between different probability distributions.\\
Since, as mentiond before, the final goal is to generalize these notions to QFT, we will also introduce the concept of von Neumann factors. In particular, we will see that finite dimensional quantum systems, are those for which the observables belong to a so called Type $I$ von Neumann algebra.

\subsection{The Von Neumann relative entropy}
Among the main differences between a classical and a quantum system are the notions of entanglement and superposition. Formally, for a finite dimensional system, this means that while a classical mixed system was described by a vector with "classical" probabilities for each possible subsystem, in the quantum case we need to use matrices in order to account for the entanglement between the various subsystems. We call such a matrix $\rho$ the \textit{density matrix} of the system. From the probabilistic role that such a matrix has in the description of the system, it must be of unit trace, positive semi-definite and (since the moduli squared of its entries must be probabilities) we want it to be hermitian.\\\\
A very important class of states is the one described by pure states, namely those for which the density matrix is idempotent: $\rho^2 = \rho$. The notion of purity is complemented by that of a quantum composite system that is composed of quantum subsystems. When these subsystems are entangled, it is possible to assign a definite quantum state to each of them. The usual toy model to explain it is that of two bits $\ket{0}, \ket{1}$ in a so called qubit or of two spins along a certain quantization axis $\ket{\uparrow}, \ket{\downarrow}$ being in the Bell or "EPR" state:
\begin{equation*}
    \ket{\Psi} = \frac{1}{\sqrt{2}}(\ket{\uparrow}\otimes \ket{\downarrow} + \ket{\downarrow}\otimes \ket{\uparrow}). 
\end{equation*}
This is different from the classical configuration: there is no way in which one can say that the two bits or spins have a definite state. At most, we can say that when one is found in a given configuration, the other must be in the opposite one.\\
To make this notion of entanglement more general, one introduces the so called Schmidt decomposition (see Section II.D in \cite{Vedral:2002zz}), that works for Hilbert spaces $\mathcal{H}$ that have are tensor product $\mathcal{H} = \mathcal{H}_1 \otimes \mathcal{H}_2$, where $\mathcal{H}_1, \mathcal{H}_2 \subset \mathcal{H}$ are Hilbert spaces associated to two different subsystems. Assuming that $\dim \mathcal{H}_1 = N$ and $\dim \mathcal{H}_2 = M$ we have that a general state $\ket{\psi} \in \mathcal{H}$ can be written as:
\begin{equation*}
    \ket{\psi} = \sum_{i = 1}^{L} c_i \ket{u'_i}\otimes \ket{v'_i} = \sum_{i=1}^L c_i \ket{u'_i , v'_i}, 
\end{equation*}
where:
\begin{equation*}
    L := \min\{N,M\}, 
\end{equation*}
and $\{ u'_i \}_{i=1, \dots N}$, $\{ v'_i \}_{i= \, \dots, M}$ are two orthonormal bases of $\mathcal{H}_1$ respectively $\mathcal{H}_2$. This clearly generalizes the Bell state. As usual, starting from a pure state\footnote{In the case of mixed states, we cannot obtain a Schmidt decomposition of the vector}, we can construct the associated density matrix as $\rho = \ket{\psi}\bra{\psi}$ and the corresponding reduced density matrices:
\begin{align*}
    \rho_{1} &= \Tr_{\mathcal{H}_2}{\rho}\\
    \rho_2 &= \Tr_{\mathcal{H}_1}{\rho}. 
\end{align*}
One of the first things we can infer about entanglement, is that a system with $N$ many degrees of freedom can be entangled with at most $N$ other orthogonal states of another one.\\
Therefore, if we want to quantify the quantity corresponding to the Shannon entropy in the case of a finite dimensional quantum state, we need an expression involving the density matrix associated to the system. Otherwise, any expression involving just the vector state cannot account for entanglement features. For this reason, we define the \textit{von Neumann entropy}:
\begin{defn}
Given a finite dimensional quantum system described by a density matrix $\rho$, we define the associated von Neumann entropy as:
\begin{equation*}
    S_N(\rho) := -\Tr(\rho \ln{\rho})
\end{equation*}
\end{defn}
However, the Shannon entropy can also be defined in the quantum context. In particular, consider the reduced density matrices for a pure state obtained starting from the Schmidt decomposition of $\ket{\psi}$:
\begin{align*}
    \rho_1 &= \sum_{i=1}^L |c_i|^2 \ket{v'_i} \bra{v'_i}\\
    \rho_2 &= \sum_{i=1}^L |c_i|^2 \ket{u'_i} \bra{u'_i}. 
\end{align*}
These show, how the $i$-th state has, in both cases, a probability of $|c_i|^2$ associated with it. Therefore, we can associate with it a quantum analogue of the Shannon entropy in the same fashion. If $A$ is an observable, pertaining to the system described by $\rho$, by the spectral decomposition theorem we can write:
\begin{equation*}
    A = \sum_i a_i P_i, 
\end{equation*}
where $P_i$ are projectors onto the states with the eigenvalue $a_i$. The probability of obtaining the eigenvalue $a_j$ is given by $p_j = Tr(\rho P_j)$. Then, let us associate $A$ to a subsystem $1$ and $B$ to $2$, demanding that the spectra are discrete and non degenerate with corresponding probabilities $p(a_i)$ and $p(b_j)$ for an outcome of the observable $A$ (resp. $B$) being $a_i$ (resp. $b_j$). Calling the joint entropy $p(a_i,b_j)$, we can define the Shannon entropies as:
\begin{align*}
    S(A) &:= - \sum_i p(a_i) \ln{p(a_i)}\\
    &= - \sum_{i,j} p(a_i, b_j) \ln{\sum_j p(a_i, b_j)}\\
    S(A,B) &= - \sum_{i,j} p(a_i, b_j) \ln{p(a_i, b_j)}, 
\end{align*}
where we have used the fact that $\sum_j p(a_i,b_j) = p(a_i)$. From this definition, we see that the Shannon entropy is equal to the von Neumann entropy when it describes the uncertainty for those observables that commute with the density matrix. This follows from the fact that, in that case, we can simultaneously diagonalize $\rho$ and $A$:
\begin{align*}
    S_{N} &= - \Tr(\rho \log \rho) = - \Tr\bigg(\sum_j p(a_j) P_j \log (\sum_i p(a_i) P_i)\bigg)\\
    &= - \sum_{k,j,i} \delta_{j,k} p(a_j) \delta_{j,i} \log p(a_i) \delta_{i,k}\\
    &= -\sum_i p(a_i) \log p(a_i). 
\end{align*}
However, for a general observable $A$, we have a huge difference between classical and quantum information:
\begin{equation*}
    S(A) \geq S_N(\rho).
\end{equation*}
Counterintuitively, there is more uncertainty in a single observable than in the state as a whole.\\
Following the same line of reasoning as in the classical case, just starting with von Neumann entropy to quantify quantum uncertainty, we can define:
\begin{defn}
The von Neumann relative entropy\footnote{Commonly named also after Umegaki who first introduced it} between two states $\sigma$ and $\rho$ is defined as:
\begin{equation*}
    S_N(\rho \| \sigma) = \Tr{\rho(\ln{\rho} - \ln{\sigma})}. 
\end{equation*}
\end{defn}
In favour of an interpretation of it as a measure of distinguishability of states, as in the classical case, there are strong arguments but no rigorous proofs. I refer to \cite{Vedral:2002zz} Section II.E and the articles cited there for further discussion in this regard.

\subsection{Introduction to von Neumann factors I}
Our ultimate goal in this chapter is to argue why relative entropy is the fundamental notion in QFT when it comes to quantify information-like measures for the quantum fields and distinguishability of the corresponding configurations. The notion of entropy in QFT is prevented, by the so called ultraviolet or vacuum divergences. Such divergences can not be avoided as they are part of the building blocks of the theory. The best way to see this, is in the algebraic approach, by investigating properties of the algebras. In order to deal with it in the next section, we first need to formulate the finite dimensional quantum case just presented, in terms of operator algebras.\\\\
First of all, we need to define what \textit{factors} are. Let $\mathfrak{A}$ be a von Neumann algebra and define for it:
\begin{equation*}
    \mathfrak{A} \vee \mathfrak{A}' := \{ AB \, | \, A \in \mathfrak{A} \,\, , \,\, B \in \mathfrak{A}' \}'', 
\end{equation*}
where $\mathfrak{A}'$ denotes the commutant of the algebra.
\begin{defn}
A von Neumann algebra $\mathfrak{A}$ on a Hilbert space $\mathcal{H}$ is called a \textit{factor} if $\mathfrak{A} \vee \mathfrak{A}' = \mathcal{B(H)}$.
\end{defn}
Intuitively speaking, factors are the pieces in terms of which we can factorize general bounded operators on the Hilbert space. Still roughly speaking, it is the von Neumann algebra that allows to write each bounded operator as a product of an element in it and an element in its commutant.\\
The definition of a factor has an equivalent formulation in terms of the center of the algebra. Namely, a factor is a von Neumann algebra such that $\mathfrak{A} \cap \mathfrak{A}' = \mathbb{C}\mathbb{1}$. To see it let me start with the following:
\begin{lem}
Consider $\mathcal{B(H)}$, then its commutant is:
\begin{equation*}
    \mathcal{B(H)}' := \{ \lambda \mathbb{1} | \lambda \in \mathbb{C} \}. 
\end{equation*}
\end{lem}
\begin{proof}
Let $A \in \mathcal{B(H)}'$. This means that $A$ must commute with all orthogonal projections that means that all subspaces of the Hilbert spaces are left invariant by the action of $A$. In particular, all nonzero vectors in the Hilbert space are eigenvectors for $A$. However, if there were different eigenvalues for $A$, then the sum of the two corresponding eigenvectors will be a nontrivial element of the Hilbert space but not an eigenvector anymore:
\begin{align*}
    A(v + w) &= Av + Aw\\
    &= \lambda v + \mu w\\
    &\neq \xi (v + w), 
\end{align*}
contradicting the fact that all nonzero vectors must be eigenvectors.
\end{proof}
This implies that, in the case of factors, since $\mathfrak{A} \vee \mathfrak{A}' = \mathcal{B(H)}$, the center of $\mathfrak{A}$ is precisely $\mathcal{B(H)} \cap \mathcal{B(H)}'$, which is trivial from the above lemma.\\
Starting from the definition we gave, it was possible (thanks to the pioneristic works of Murray and von Neumann \cite{vNMurray1936}, \cite{vNMurray1937}, \cite{vNMurray1943} and to the classification of Connes \cite{Connes1973}) to classify all types of factors. Moreover, von Neumann proved in \cite{vN1949} that a general von Neumann algebra, on a separable Hilbert space, is isomorphic to a direct integral of factors. Therefore, since we will always deal with separable Hilbert spaces, understanding and classifying factors has both a mathematical and a physical (that we will see in what follows) deep motivation.\\\\
In this section we dealt with finite dimensional quantum systems, these are the systems for which we can always find two or more separate subsystems, that allow for a factorization of the Hilbert space itself in two independent sub-Hilbert spaces. Namely, the Hilbert space $\mathcal{H}$ can be written as:
\begin{equation*}
    \mathcal{H} = \mathcal{H}_1 \otimes \mathcal{H}_2. 
\end{equation*}
Therefore, it is also clear that an observable on the first system is a selfadjoint operator in $\mathfrak{A} = \mathcal{B(H)} \otimes \mathbb{1}$ and one on the second system will be a selfadjoint operator in $\mathbb{1} \otimes \mathcal{B(H)}$ that is in particular $\mathfrak{A}'$. This means that we can factorize also the observables of the total system as:
\begin{equation*}
    \mathcal{B(H)} = \mathcal{B}(\mathcal{H}_1) \otimes \mathcal{B}(\mathcal{H}_2). 
\end{equation*}
An algebra admitting such a factorization is what we call a \textit{von Neumann algebra of \textbf{Type $I$}}.\\
A characteristic of Type $I$ algebras is that they admit minimal projections in $\mathfrak{A}$, this means that for each vector $\psi \in \mathcal{H}_1$ we can define a projection $E_{\psi} = \ketbra{\psi}{\psi}$ that is an operator $E_{\psi}: \mathcal{H} \to \mathcal{H}_1$ and that:
\begin{equation*}
    E := E_{\psi} \otimes \mathbb{1} 
\end{equation*}
is a minimal projection. This means, that there are no projection operators to smaller subsets except the trivial null projection.\\
These types of algebras have the feature that admit a function called trace: $\Tr: a \to \Tr(a)$ where $a \in \mathfrak{A}$. This is a function fulfilling the \textit{cyclicity property}:
\begin{equation*}
    \Tr(abc) = \Tr(cab) = \Tr(bca), 
\end{equation*}
and the \textit{positivity} for $a \neq 0$:
\begin{equation*}
    \Tr(a a^{\dagger}) > 0. 
\end{equation*}
Of course, the trace of an operator might be divergent if we are dealing with bounded operators on an infinite dimensional Hilbert space, e.g. the identity. For this reason, we distinguish between \textit{Type $I_d$}, those Type $I$ algebras that are represented on Hilbert spaces of dimension $d$, and \textit{Type $I_{\infty}$} algebras, those that are represented on infinite dimensional separable Hilbert spaces. Therefore, the trace is a well defined map over all Type $I_d$ algebras and only partially on those of Type $I_{\infty}$. Moreover, Type $I_d$ algebras are algebraically isomorphic to the matrix ring of square $d \times d$ matrices, while $I_{\infty}$ algebras are algebraically isomorphic to $\mathcal{B(H)}$.\\\\

In our discussion regarding quantum information theory, we argued that it is natural to consider states represented by density matrices and we assumed them to be normalized. We have also discussed that, given two separate subsystems, one can focus just on one of them by reducing density matrices tracing out the degrees of freedom of the other system. Moreover, any expectation value of an observable $A$ on a mixed state $\rho$ is given by: $\Tr(A \rho) = \langle A \rangle$.\\
Therefore, in light of what we have just discussed, we were always implicitly assuming that the algebra of observables was admitting a trace: namely it must be of Type $I$. We see, in particular, that the definition of entropy itself relies on this property of the algebra and we can anticipate that, for those factors that do not admit a trace map, the von Neumann entropy can not be used. The question, of course, becomes whether we should care about such types of von Neumann algebras. For this reason, in the next section, we argue that such algebras exist and we have to deal with them as they are the algebras of observables in Quantum field theory. In particular, we will present a \textit{gedankenexperiment}, presented by Fermi, that shows how a naive treatement of the observable algebra, as if it is always of Type $I$, may lead to serious causal issues.\\

\section{Relative entropy in QFT}
\subsection{Introduction to von Neumann factors II}
Let me start completing the classification of factors, that was started at the end of last section. First we need the following definition:
\begin{defn}
Let $\mathfrak{A}$ be a factor on a separable Hilbert space $\mathcal{H}$. Two projections $P_1, P_2 \in \mathfrak{A}$ are called equivalent, denoted $P_1 \sim P_2$, if there exist a $W \in \mathfrak{A}$ such that:
\begin{equation*}
   W^* W = P_1 \hspace{15pt} W W^* = P_2. \end{equation*}
The map $W$ isometrically maps $P_1 \mathcal{H}$ to $P_2 \mathcal{H}$ and annihilates all other vectors in the orthogonal complement of $\mathcal{H}_1 := P_1 \mathcal{H}$. $W$ is called a partial isometry.
\end{defn}
Now, we properly define what Type $III$ factors are: the von Neumann factors for which all projections $E$ are equivalent to the identity $\mathbb{1}$. For them, the Hilbert spaces must be infinite dimensional (as otherwise all the observables are representable as matrices and as such belong to the Type $I_d$ algebras) with the property that each subspace is isometric to the full Hilbert space. This shows, that even if we try to identify a subspace by tracing out some degree of freedom, we are in fact performing a trace over the entire Hilbert space, i.e. for any $A$ in such an algebra and for any projection $E$:
\begin{align*}
    \Tr_{E\mathcal{H}}A &= \Tr_{\mathcal{H}} W W^* A W W^*\\
    &= \Tr_{\mathcal{H}}A. 
\end{align*}
This shows, that there is no way in which we can properly talk about subsystems and in this way obtain a tensor product splitting of the Hilbert space. This is in contrast with the case of Type $I$ factors, where the existence of minimal projections ensures that we can always identify a subspace of the Hilbert space of the theory and possibily associate it to a subsystem. Moreover, this also shows that we cannot define a trace over a Type $III$ algebra. In fact, the only possibility for which the trace on each subspace corresponds to that over the entire space, is either if the trace is zero, but this implies that $A=0$, or that is always divergent, preventing the definition of trace as the identity holds for any $A \in \mathfrak{A}$.\\
As a side remark, let me mention that there are also so called \textit{Type $II$} algebras that somehow lay in between the $I$ and the $III$ types. These are those factors that have no minimal projections, but every non-zero projection $E$ has a subprojection $F < E$ (that means $F \mathcal{H} \subseteq E \mathcal{H}$) that is finite in the sense that each $F' < F$ must be such that $F \sim F'$. However, this case is not very interesting, as far as it is known, from the point of view of a mathematical physicist and we will thus not investigate it further here.

\subsection{The necessity of Type $III$ factors in QFT}
As already mentioned, Type $III$ algebras are the algebras of observables in any localized Quantum Field Theory. We want to discuss it, with an example following the work of \cite{Yngvason_2005} based on the original work \cite{Fermi}.\\
Let us start by considering two atoms $a$ and $b$ separated by a distance $R$. At time $t=0$, $b$ is in an excited configuration that can spontaneously decay to the ground state by emission of radiation. At the same time assume that atom $a$ is in its ground state and that the energy of the emitted radiation by $b$ can allow $a$ to pass from its ground to the excited state. Of course, by causal reasons, we expect that, from the moment in which atom $b$ decays, at least a time greater that $R/c$ must pass before atom $a$ gets excited. Therefore, by setting the decay instant at $t=0$, we expect that for $t < R/c$ atom $a$ stays in its ground state.\\
Let us first analyze this problem from a "Type $I$ point of view". In this spirit, we assume that the Hilbert space of the problem is $\mathcal{H} = \mathcal{H}_a \otimes \mathcal{H}_b \otimes \mathcal{H}_c$ where we have denoted with $\mathcal{H}_c$ the Hilbert space associated with the radiation. To each of the subspaces, we assign $\mathfrak{A}_a = \mathcal{B}(\mathcal{H}_a)$ as the associated von Neumann factor of Type $I$, and with obvious notation we have:
\begin{equation*}
    \mathcal{B}(\mathcal{H}) = \mathfrak{A}_a \otimes \mathfrak{A}_b \otimes \mathfrak{A}_c. 
\end{equation*}
The initial state of the system is:
\begin{equation*}
    \omega_0 = \omega_a \otimes \omega_b \otimes \omega_c, 
\end{equation*}
where initially $\omega_a$ is the ground state for the atom $a$, $\omega_b$ the excited configuration for atom $b$ and $\omega_c$ the vacuum state for the radiation field. The state at time $t > 0$ is:
\begin{equation*}
    \omega_t(\cdot) = \omega_0(e^{itH} \cdot e^{-itH}), 
\end{equation*}
where $H$ is, as usual, the Hamiltonian of the total system. In the assumption about the type of factors that we are dealing with, we can define a projection operator onto the excited configuration of atom $a$ in $\mathcal{H}_a$ as $E_a = \mathbb{1}_a - \ketbra{\psi_a}{\psi_a}$ where $\psi_a$ is the ground state of atom $a$. Therefore, we have the minimal projection onto the excited state given by:
\begin{equation*}
    E = E_a \otimes \mathbb{1}_b \otimes \mathbb{1}_c. 
\end{equation*}
The probability of finding atom $a$ at instant $t$ in the excited configuration, according to the probabilistic interpretation, then is:
\begin{equation*}
    P(t) = \omega_t(E). 
\end{equation*}
Therefore, according to what we said above, we should expect that $P(t) = 0$ for $t < R/c$.\\
In general, physical Hamiltonians have lower bounds for the energy spectrum. This implies, that the vector valued function: $t \to Ee^{-itH}\ket{\phi}$ for all $\ket{\phi} \in \mathcal{H}$, is analytic as a function of $t$ in the upper half plane. This can be seen by first noticing that $t \to Ee^{-itH}\ket{\phi}$ must vanish on a closed interval, since $E^2 = E$ as a projection and thus:
\begin{equation*}
    \omega_t(E) = \| Ee^{-itH} \ket{\phi} \|^2. 
\end{equation*}
But since $\omega_t(E) = 0$ for $t < R/c$, the same must hold for the vector valued function $t \to Ee^{-itH}\ket{\phi}$ that maps onto the zero vector. Write now the vector valued map via its spectral decomposition, displacing it in the complex plane:
\begin{align*}
   E e^{izH} \ket{\phi} &= \int_0^{\infty} E e^{iz\lambda} dE_{\lambda} \ket{\phi}\\
    &= \int_0^{\infty} E e^{-\Im(z)\lambda} e^{i\Re(z)\lambda} dE_{\lambda} \ket{\phi}. 
\end{align*}
Then, since we want to prove analyticity in the upper half plane, assume $\Im(z) > 0$. By integrating over any closed contour $\Gamma$:
\begin{equation*}
    \int_{\Gamma} dz E e^{izH} \ket{\phi} = 0
\end{equation*}
by Cauchy theorem. But then, by Morera theorem, this means that the function must be analytic in the upper half plane. Therefore, by Schwarz reflection principle \ref{app: Schw} and the vanishing of $t \to Ee^{-itH}\phi$ for $t < R/c$, we conclude that the vector valued function must be holomorphic on the closed interval $t \in [0,R/c - \epsilon]$ for some $\epsilon > 0$. This allows to conclude that, as a holomorphic function vanishing on a closed interval, it must be identically zero for all $t \in \mathbb{R}$. From the above this also gives that $\omega_t(E) = 0$ for all $t \in \mathbb{R}$.\\
Hence, if atom $a$ becomes excited at all, that must happen immediately at $t=0$. However, this cannot happen as it would be a violation of causality. The solution to this apparent paradox, stems from our inappropriate description of the algebra of observables to be of type $I$.\\
Before we reformulate it in terms of type $III$ factors, let us point out some differences that arise when dealing with type $III$ algebras. First, each projector $E \in \mathfrak{A}\mathcal{(O)}$ (now an element of a local algebra with $\mathcal{O} \subset M$, for $M$ some manifold) has an associated isometry $W \in \mathfrak{A}\mathcal{(O)}$ such that $E = W W^*$. This means, that we can make any state $\omega$ an eigenstate of $E$ by local operations without altering its properties in the causal complement. To motivate why this is true, first define:
\begin{equation*}
    \omega_{W}(A) := \omega(W^*AW). 
\end{equation*}
Then, we have:
\begin{equation*}
    \omega_W(E) = \omega(W^*WW^*W) = \omega(\mathbb{1}) = 1, 
\end{equation*}
but at the same time for $B \in \mathfrak{A}\mathcal{(O)}'$:
\begin{equation*}
    \omega_{W}(B) = \omega(W^*BW) = \omega(W^* W B) = \omega(B). 
\end{equation*}
Moreover, without going into detail here, according to Connes classification \cite{Connes1973} of von Neumann algebras, for Type $III_1$ algebras (that are those of QFT) each state $\varphi$ is locally approximated with arbitrary precision in the weak topology by some $\omega_{W}$. Finally, for local von Neumann algebras $\mathfrak{A}\mathcal{(O)}$, every state is a vector state: for $\omega$, a state over $\mathfrak{A}\mathcal{(O)}$, there exist a vector $\psi_{\omega} \in \mathcal{H}$ such that:
\begin{equation*}
    \omega(A) = \braket{\psi_{\omega}}{A \psi_{\omega}} \hspace{20pt} \forall A \in \mathfrak{A}\mathcal{(O)}. 
\end{equation*}
Keeping this in mind, let us go back to our decay process. The only thing that might change in the case of Type $III$ algebras is that the probability, as we have computed it relying on the standard Copenhagen interpretation, might need to be modified. We need to reformulate everything in terms of Local Quantum Field Theory. For that purpose (we are on Minkowski spacetime $M_D$) consider the atom $a$ represented by a small ball in a Cauchy surface $\Sigma_0$ at $t=0$: pick $\mathcal{O}_a = \mathcal{R}_a \times \{ 0 \}$ where $\mathcal{R}_a \subset \Sigma_0$ and $D(\Sigma_0) = M_D$. The observables, associated to the electromagnetic field in the region $\mathcal{O}_a$, are thus included in $\mathfrak{A}(\mathcal{O}_a)$ as well. Analogously, we use $\mathfrak{A}(\mathcal{O}_b)$ as a replacement for $\mathfrak{A}_b$. In this framework, if we want to tell whether the decay of atom $b$ has influenced the configuration of $a$, we need to compare two states on $\mathfrak{A}(\mathcal{O}_a)$ and see whether their "difference" (we are going to properly define this in a moment) changes in time. For that purpose, define, at $t=0$, $\omega_0$ and $\omega_0^{(0)}$ to be the configurations in which atom $b$ is in the excited and in the ground state, respectively, while atom $a$ in its ground state in both cases. Initially, the two states are only distinguished by local measurements within $\mathcal{O}_b$. With time, the states evolve in $\omega_t$ and $\omega_t^{(0)}$. In particular, we will say that the configuration of atom $a$ is unaltered at $t>0$, if for all $A \in \mathfrak{A}(\mathcal{O}_a)$:
\begin{equation*}
    \omega_t(A) = \omega_t^{(0)}(A). 
\end{equation*}
In this way, if we want to quantify the deviation of $\omega_t$ from $\omega_t^{(0)}$ (testing whether the decay of atom $b$ influenced the configuration of atom $a$), we define the quantity:
\begin{equation*}
    D(t) = \sup_{A \in \mathfrak{A}(\mathcal{O}_a); \|A\| < 1} \bigg| \omega_t(A) - \omega_t^{(0)}(A) \bigg|. 
\end{equation*}
At $t=0$, we have $D(0) = 0$, but also for each observable $A$ localized at a distance $< R$ from $\mathcal{R}_a$ we have that $\omega_0(A) - \omega_0^{(0)}(A) = 0$. This implies (as the local algebra of $\mathcal{O}$ corresponds to that of its causal development, looking at the time evolution in the Heisenberg picture) that $D(t)$ vanishes for $t < R/c$, giving no violation of causality.\\
Let us see how this case is different from the previous by looking at the projection operator. Start by restricting $\omega_0$ on $\mathfrak{A}(\mathcal{O}_a)$, to get a vector $\psi \in \mathcal{H}$ and define $E = \mathbb{1} - \ketbra{\psi}{\psi}$. In the type $I$ approach, this was the operator we used as a "test for the excitation of atom $a$", as it was an observable belonging to the algebra of the restricted subsystem. Moreover, as $E \in \mathfrak{A}_a$ in that case, we found $\omega_t(E) = P_t(E) = 0$ for all $t$ by analyticity and evidently $\omega_t^{(0)}(E) = 0$ for all $t$. This also shows that in this context, considering something like $D(t)$ to test the excitation of the atom $a$, is useless.\\
On the other hand, in the case of type $III$ algebras, looking simply at the non-vanishing of $\omega_t(E)$ is erroneous as in general $\omega_t(E) > 0$. This is a consequence of causality and of the Reeh-Schlieder theorem (see Appendix Theorem \ref{app: Reeh}). To see it, notice that $\omega_0^{(0)}$ represents the vacuum of the QFT, invariant under time evolution: $\omega_0^{(0)}(E) = \omega_t^{(0)}(E)$ for all $t \in \mathbb{R}$. By the Reeh-Schlieder theorem, the vacuum must be separating and cyclic, therefore since $E$ is non-trivial we must have $\omega_0^{(0)}(E) > 0$. At the same time, being $E \in \mathfrak{A}(\mathcal{O})$ for $\mathcal{O}_a \subset \mathcal{O}$ (a slighlty larger set, see \cite{Yngvason_2005} Eq. $(23)$), we must have by causality $\omega_t(E) - \omega_t^{(0)}(E) = 0$ for $t < R'/c$ ($R'$ is slightly smaller than $R$). This gives $\omega_t(E) > 0$, showing that in this context, a non-vanishing of it cannot be taken as a measure of the excitation of atom $a$. Moreover, as in $\omega_t^{(0)}$ the atom $b$ is not even excited, the specific value of $\omega_t(E) > 0$ for $t < R'/c$ cannot be caused by the presence of the atom $b$. Therefore, as $\omega_t(E)$ does not vanish on any closed interval, the analyticity argument used in the Type $I$ case cannot be applied here, since the norm $\| E e^{-itH} \ket{\phi} \|^2$ being constant does not imply that the vector valued function $t \to E e^{-itH} \ket{\phi}$ is constant (e.g. it might be a vector varying on a sphere). Therefore, as now $\omega_t(E)$ cannot be proven to be an analytic function (it involves both $e^{itH}$ and $e^{-itH}$), the constancy on an interval does not imply the constancy for all $t$. As a consequence, for $t > R'/c$ we can have $D(t) \neq 0$ measuring the excitation of the atom $a$ due to a decay of $b$.\\
To conclude, by considering T ype $III$ factors, the measure of excitation of the atom $a$ can pass from being zero to being finite after a finite amount of time, thus preserving causality.\\\\

Now that we have argued that local algebras of QFT must be of Type $III$, we want to ask whether in this context one also has a measure of distinguishability between different states of the system, as was the relative entropy for Type $I$ algebras of Quantum Mechanics. We seek a genereral definition of entropy that yields feasable results also on Type $III$ algebras, where we do not have a notion of trace. Such a definition will have to reduce to the known one, when we restrict ourselves to Type $I$ algebras. However, in order to do so, we first need to introduce the \textit{modular-} or \textit{Tomita-Takesaki modular theory}.

\subsection{Tomita-Takesaki modular theory}
Let $\mathcal{H}$ be the Hilbert space on which we have the von Neumann algebra $\mathfrak{A}$. We start with a couple of definitions:
\begin{defn}
A vector $\ket{\Psi} \in \mathcal{H}$ is called cyclic for the von Neumann algebra $\mathfrak{A}$ if the subset $\mathfrak{A} \ket{\Psi} \subset \mathcal{H}$ is a dense subset of $\mathcal{H}$
\end{defn}
\begin{defn}
A vector $\ket{\Psi}$ is called separating for the von Neumann algebra $\mathfrak{A}$ if the condition:
\begin{equation*}
    Q \ket{\Psi} = 0, 
\end{equation*}
for $Q \in \mathfrak{A}$ implies that $Q = 0$.
\end{defn}
In particular, recalling that one of the defining properties of von Neumann algebras is given by the fact that any such algebra and its commutant are in fact found to be each others commutants, we have:
\begin{prop}\label{prop: 112}
$\ket{\Psi} \in \mathcal{H}$ is cyclic for $\mathfrak{A}$ if and only if is separating for the commutant algebra $\mathfrak{A}'$.
\end{prop}
\begin{proof}
$(\Rightarrow)$ Assume that for $Q' \in \mathfrak{A}'$ we have:
\begin{equation*}
    Q' \ket{\Psi} = 0
\end{equation*}
and by assumption, $\ket{\Psi}$ is cyclic for $\mathfrak{A}$. Then, for any combination of elements in $\mathfrak{A}$ (that I will simply denote, with a slight abuse of notation, by $\mathfrak{A}$):
\begin{equation*}
    \mathfrak{A} Q' \ket{\Psi} = 0 \Leftrightarrow Q' \mathfrak{A} \ket{\Psi} = 0. 
\end{equation*}
Since, however, $\mathfrak{A} \ket{\Psi}$ is dense in $\mathcal{H}$ that implies $Q'$ is an operator that vanishes on a dense subset of $\mathcal{H}$: $Q' = 0$ and thus $\ket{\Psi}$ is separating for $\mathfrak{A}'$.\\\\
$(\Leftarrow)$ Suppose that a vector $\ket{\Psi}$ is not cyclic for $\mathfrak{A}$. Therefore, the vectors $Q \ket{\Psi}$ for $Q \in \mathfrak{A}$ generate a proper subspace $\mathcal{H}' \subset \mathcal{H}$. Let us call $P: \mathcal{H} \to \mathcal{H}$ the projection onto the orthogonal $\mathcal{H}'_{\bot}$. Then, we have that $P \in \mathfrak{A}'$ and is a bounded operator. But $P \ket{\Psi} = 0$, as $\mathbb{1} \in \mathfrak{A}$, which implies $\mathbb{1} \ket{\Psi} = \ket{\Psi} \in \mathcal{H}'$. However, this shows that $P$ is a non-vanishing operator acting on $\ket{\Psi}$ (that by assumption is separating for $\mathfrak{A}'$) that gives a vanishing result, yielding an absurdum. Therefore $\ket{\Psi}$ must be cyclic for $\mathfrak{A}$ and such a $P$ cannot exist
\end{proof}
Before continuing with the definition of the Tomita operator, let us prove a result of operator theory:
\begin{prop}\label{prop: Closable}
Let $X$ be a Banach space and let $T: D(T) \subset X \to Y$ be a linear operator. Then, $T$ is closable (i.e. it admits a closed extension) iff for any sequence $(x_n)_{n \in \mathbb{N}} \subset D(T)$ such that $\lim_{n \to \infty} x_n = 0$ and $\lim_{n \to \infty} Tx_n = y$, we have $y=0$. 
\end{prop}
\begin{proof}
$(\Rightarrow)$ Clearly if $T$ is closable, denoting with $\overline{T}$ its closure, we have by definition that for each sequence $(x_n)_{n \in \mathbb{N}} \subset D(\overline{T})$ that converges to $x$ and such that $\overline{T}x_n$ converges, we have:
\begin{equation*}
    \lim_{n \to \infty} \overline{T}x_n = y \Rightarrow y = \overline{T}x. 
\end{equation*}
Therefore if we take those sequences $(x_n)_{n \in \mathbb{N}} \in D(T) \subset D(\overline{T})$ that converge to $0$ and $y = \lim_{n \to \infty} T x_n$, by $D(T) \subset D(\overline{T})$ we also have:
\begin{equation*}
    y = \lim_{n \to \infty} \overline{T} x_n, 
\end{equation*}
but since $\overline{T}$ is closed we have that: $y = \overline{T} 0 = 0$.\\\\
$(\Leftarrow)$ Let us start by picking a sequence $(x_n)_{n \in \mathbb{N}} \in D(T)$ as in the statement of the proposition. Define a new sequence $x'_n = x_n + x$ such that:
\begin{equation*}
    \lim_{n \to \infty} x'_n = x
\end{equation*}
This sequence might, however, not be in $D(T)$. For this reason let us define the linear operator $\overline{T}$ that is an extension of $T$ (i.e. $D(T) \subset D(\overline{T})$ and $T \xi = \overline{T} \xi$ for all $\xi \in D(T)$). Then we have:
\begin{align*}
    \lim_{n \to \infty} \overline{T} x'_n &= \lim_{n \to \infty} (\overline{T} x_n + \overline{T} x)\\
    &= \lim_{n \to \infty} (T x_n + \overline{T} x)\\
    &= \overline{T} x,
\end{align*}
which proves that $\overline{T}$ is a closed extension of $T$.
\end{proof}
As any Hilbert space is in fact a Banach space, equally, this proposition holds in our case where we are dealing with Hilbert spaces.\\
Next up is a theorem about the existence of a unique polar decomposition for unbounded operators:
\begin{thm}\label{thm: polar}
Let $\mathcal{M}$ be a unital $C^*$-algebra. Then if $Q \in \mathcal{M}$ is an invertible operator, there exists a unique $U$, unitary or antiunitary, such that:
\begin{equation*}
    Q = U |Q|. 
\end{equation*}
\end{thm}
\begin{proof}
Let us start by noticing that the set of invertible operators is closed under the multiplication over the algebra and under the star operation. Therefore, let us define: $|Q|^2 := Q^*Q$, which must still be invertible. If we now define:
\begin{equation*}
    U = Q |Q|^{-1}, 
\end{equation*}
we have that:
\begin{equation*}
    Q = U |Q|. 
\end{equation*}
Moreover, the operator $U$ is unitary or antiunitary as stated:
\begin{align*}
    U^* U &= |Q|^{-1} Q^* Q |Q|^{-1}\\
    &= |Q|^{-1} |Q|^2 |Q|^{-1} = \mathbb{1}\\
    U U^* &= Q |Q|^{-1} |Q|^{-1} Q^*\\
    &= Q (Q^* Q)^{-1} Q^* \\
    &= Q Q^{-1} (Q^*)^{-1} Q^*= \mathbb{1}. 
\end{align*}
\end{proof}
This constitutes everything we need to introduce the Tomita-Takesaki modular theory.\\
Let us pick a vector $\ket{\Psi}$, which is cyclic and separating for the von Neumann algebra $\mathfrak{A}$ over the Hilbert space $\mathcal{H}$. Define the operator $S_{\Psi}: \mathcal{D} \to \mathcal{H}$, where $\mathcal{D} \subset \mathcal{H}$ is a dense subset of the Hilbert space, called the \textit{Tomita operator}, which for any $Q \in \mathfrak{A}$ acts as:
\begin{equation*}
    S_{\Psi} Q \ket{\Psi} = Q^* \ket{\Psi}. 
\end{equation*}
Firstly, we notice, that $S_{\Psi}$ is closable. In fact, if we take a sequence $Q_n \ket{\Psi} \to 0$, we can clearly take the adjoint and directly obtain $Q_n^*\ket{\Psi} \to 0$ as well. Therefore, from Proposition \ref{prop: Closable}, we have that $S_{\Psi}$ is closable and from now on we will consider its closure and, to simplify the notation, we will use the same symbol $S_{\Psi}$ for it.\\
Another thing to notice is that $S_{\Psi} \ket{\Psi} = \ket{\Psi}$ and that it squares to the identity: $S_{\Psi}^2 = \mathbb{1}$. In particular, the second identity shows that the Tomita operator is invertible. For this reason, Theorem \ref{thm: polar} yields its polar decomposition:
\begin{equation}
    S_{\Psi} = J_{\Psi} \Delta_{\Psi}^{1/2},
\end{equation}
where $J_{\Psi} = S_{\Psi} \Delta_{\Psi}^{-1/2}$ is a unitary operator that must necessarily be antilinear, as it is defined in terms of $S_{\Psi}$, which is antilinear. 
$J_{\Psi}$ is called the \textit{modular conjugation}. Moreover, $\Delta_{\Psi} = S_{\Psi}^* S_{\Psi}$, called the \textit{modular operator}, is a nonnegative self-adjoint operator (as it evidently is symmetric and $D(\Delta_{\Psi}) = D(S_{\Psi})$, as well as  $D(\Delta_{\Psi}^*) = D(S_{\Psi})$ holds). Moreover, we notice that:
\begin{equation*}
    \Delta_{\Psi} \ket{\Psi}= \ket{\Psi} \hspace{15pt} J_{\Psi} \ket{\Psi} = \ket{\Psi},
\end{equation*}
from which also follows that, given any function $f$, we have $f(\Delta_{\Psi}) \ket{\Psi} = f(1) \ket{\Psi}$.
Finally, let us notice that from $S_{\Psi}^2 = \mathbb{1}$, we have:
\begin{equation*}
    J_{\Psi} \Delta_{\Psi}^{1/2} J_{\Psi} \Delta_{\Psi}^{1/2} = \mathbb{1}
\end{equation*}
which implies:
\begin{equation}\label{eq: ant}
    J_{\Psi} \Delta_{\Psi}^{1/2} J_{\Psi} = \Delta_{\Psi}^{-1/2}. 
\end{equation}
Simultaneously, we also have:
\begin{equation*}
    J^2_{\psi} (J_{\psi}^{-1} \Delta_{\psi}^{1/2} J_{\psi} ) = \Delta_{\psi}^{-1/2}.
\end{equation*}
However, the uniqueness of the polar decomposition of $\Delta_{\psi}^{-1/2}$ enforces $J^{2}_{\psi} = \mathbb{1}$.\\
Moreover, composing Eq. \eqref{eq: ant} with itself, one finds:
\begin{align*}
    J_{\Psi} \Delta_{\Psi}^{1/2} J_{\Psi} J_{\Psi} \Delta_{\Psi}^{1/2} J_{\Psi} &= \Delta_{\Psi}^{-1}\\
    J_{\Psi} \Delta_{\Psi} J_{\Psi} &= \Delta_{\Psi}^{-1}. 
\end{align*}
Then, from the antilinearity of $J_{\Psi}$, we have for any function $f$: $J_{\Psi} f(\Delta_{\Psi}) J_{\Psi} = \overline{f}(\Delta_{\Psi}^{-1})$. Therefore, if we take as a function $f(x) = x^{is}$, we will have that:
\begin{equation*}
    J_{\Psi} \Delta_{\Psi}^{is} J_{\Psi} = \Delta_{\Psi}^{is}. 
\end{equation*}\\\\
Following this brief introduction to the Tomita operator, we are in a position to quote the main result of the preliminary discussion, given by
the \textit{theorem of Tomita-Takesaki}:
\begin{thm}[\textbf{Tomita-Takesaki}]\label{thm: TomTak}
The modular conjugation $J_{\Psi}$ and the modular group $\Delta_{\Psi}^{it}$ associated with the von Neumann algebra $\mathfrak{A}$ and the cyclic and separating vector $\ket{\Psi}$, are such that:
\begin{align*}
    J_{\Psi} \mathfrak{A} J_{\Psi} &= \mathfrak{A}'\\
    \Delta^{it}_{\Psi} \mathfrak{A} \Delta_{\Psi}^{-it} =  \mathfrak{A} &\hspace{15 pt} \Delta^{it}_{\Psi} \mathfrak{A}' \Delta_{\Psi}^{-it} =  \mathfrak{A}'
\end{align*}
for all $t \in \mathbb{R}$. 
\end{thm}
\begin{proof}
See \cite{Takesaki:1970aki} Theorem $10.1$ and Corollary $9.1$. 
\end{proof}
The above essentially states that the modular conjugation maps the algebra onto its commutant and that additionally the modular group is an automorphism of the von Neumann algebra.\\\\
Let us now turn to the relative Tomita operator. For this purpose, start by choosing two states $\ket{\Psi}$ and $\ket{\Phi}$, both assumed to be normalized $\braket{\Psi}{\Psi} = \braket{\Phi}{\Phi} = 1$. For the moment, we only assume the former, that is $\ket{\Psi}$, to be cyclic and separating for the algebra $\mathfrak{A}$. Define the \textit{relative Tomita operator} $S_{\Phi|\Psi}: D(S_{\Phi|\Psi}) \to \mathcal{H}$, where as before $D(S_{\Phi|\Psi})$ is a dense subset of $\mathcal{H}$, as:
\begin{equation*}
 S_{\Phi|\Psi} Q \ket{\Psi} = Q^* \ket{\Phi}
\end{equation*}
for $Q \in \mathfrak{A}$. As was the case for the Tomita operator, it is clear that this is a closable operator and from now on we will assume that the closure has been taken. In general, the operator $S_{\Phi|\Psi}$ is not invertible. A naive guess of $S_{\Psi|\Phi}$ being the inverse, is accompanied by the problem that $\ket{\Phi}$ is an arbitrary, and not necessarily a cyclic vector. So, $S_{\Psi|\Phi}$ is not defined on a dense subset of the Hilbert space. Therefore, under the additional assumption that also $\ket{\Phi}$ is cyclic and separating, we have that the relative Tomita operator is invertible and as such admits a unique polar decomposition:
\begin{equation*}
    S_{\Phi|\Psi} = J_{\Phi|\Psi} \Delta_{\Phi|\Psi}^{1/2}, 
\end{equation*}
where as before we have that $J_{\Phi|\Psi}$ is called the \textit{relative conjugate operator} and $\Delta_{\Phi|\Psi}$ the \textit{relative modular operator}. Like for the Tomita operator, we have that $S_{\Phi|\Psi} \ket{\Psi} = \ket{\Psi}$. As a side remark, we notice that if we set $\ket{\Phi} = \ket{\Psi}$, the relative Tomita operator reduces to the Tomita operator.\\

\subsection{Araki's relative entropy}
We have now gathered everything needed in order to define the most general notion of relative entropy, first introduced by Uhlmann in \cite{Uhlmann} and by Araki in \cite{Araki1976}. In what follows, after giving the definition of entropy, we will prove most of its properties based on the original work of Araki \cite{Araki1976}, \cite{Araki1977} and we will argue why it generalizes the one of von Neumann of Quantum Mechanics to type $III$ factors and thus to Quantum Field Theory. However, before introducing it, let me show why it suffices to only study relative entropy in this context, and not entropy, as it is divergent. The first thing to notice is that the vacuum state of QFT, as a cyclic and separating state, is a highly entangled state. To understand what we mean by that, let us revisit the example of the Bell pair:
\begin{ex}
Consider the Bell pair:
\begin{equation*}
    \ket{\Psi} := \frac{1}{\sqrt{2}}(\ket{0}\otimes \ket{0} + \ket{1} \otimes \ket{1}), 
\end{equation*}
which is an entangled state for the bipartite system $\mathfrak{A} = M(2\times 2, \mathbb{C}) \otimes M(2\times 2, \mathbb{C})$. In particular, whenever we act on $\ket{\Psi}$ with operators acting on one of the two subsystems only, we obtain any vector in the full Hilbert space $\mathbb{C}^2 \otimes \mathbb{C}^2$. To see it, let us consider:
\begin{equation*}
    A = \begin{pmatrix}
    a_{00} & a_{01} \\
    a_{10} & a_{11} \end{pmatrix} \in M(2\times 2, \mathbb{C}), 
\end{equation*}
acting on basis vectors as:
\begin{equation*}
    A \ket{0} = a_{00} \ket{0} + a_{10} \ket{1} \hspace{15pt} A\ket{1} = a_{11} \ket{1} + a_{01} \ket{0}. 
\end{equation*}
This, however, means that:
\begin{equation*}
    (A\otimes \mathbb{1}) \ket{\Psi} = \frac{1}{\sqrt{2}} (a_{00} \ket{0} \otimes \ket{0} + a_{10} \ket{1} \otimes \ket{0} + a_{01} \ket{0} \otimes \ket{1} + a_{11} \ket{1} \otimes \ket{1}),
\end{equation*}
which shows how, by conveniently choosing the operator $A$, we get any vector in $\mathbb{C}^2 \otimes \mathbb{C}^2$. Note that this result only holds for the above state. Had we chosen to start with a separable state like $\ket{0} \otimes \ket{0}$, the above consideration would not have been true.
\end{ex}
The similarity with QFT is evident from the statement of the Reeh-Schlieder theorem. Given a region $\mathcal{O}$ of the considered spacetime $M$ (assumed to be such that the Reeh-Schlieder property holds, see comments at the end of Section \ref{app: Reeh}), we are assuming that the full algebra of bounded operators over the Hilbert space $\mathcal{H}$ can be decomposed as: 
\begin{equation*}
    \mathcal{B}(\mathcal{H}) = \mathfrak{A}(\mathcal{O}) \vee \mathfrak{A}(\mathcal{O}'). 
\end{equation*}
Namely, the local von Neumann algebras are factors. The Reeh-Schlieder theorem asserts that acting with operators $A \in \mathfrak{A}(\mathcal{O})$ on the vacuum, we can approximate any vector in $\mathcal{H}$ with arbitrary precision. Moreover, since this holds true for any open region $\mathcal{O}$, we say that the vacuum is \textit{entangled at any distance}: Against intuition, the vacuum is a highly entangled state (which may be the reason why the Reeh-Schlieder theorem, when first formulated, was considered to be counter-intuitive).\\\\
Finally, we seek to define a notion of entropy for local algebras. In order to do so, we shall discuss the following properties:
\begin{defn}\label{def: Rentr}
We define an entropy in QFT as a map $\mathcal{O} \in \mathcal{K} \to S(\mathcal{O}) \in \mathbb{R}$ that satisfies the following properties:
\begin{itemize}
    \item (positivity) $S(\mathcal{O}) \geq 0$ for all $\mathcal{O} \in \mathcal{K}$
    \item (strong subadditivity) $S(\mathcal{O}_1 \vee \mathcal{O}_2) + S(\mathcal{O}_1 \cap \mathcal{O}_2) \leq S(\mathcal{O}_1) + S(\mathcal{O}_2)$\\
    for all commuting regions $\mathcal{O}_1$ and $\mathcal{O}_2$
    \item (Poincaré invariance) $S(g\mathcal{O}) = S(\mathcal{O})$ for all $g \in \mathcal{P}^{\uparrow}_{+}$
\end{itemize}
\end{defn}
\begin{rem}
The definition is presented for Minkowski spacetime, but can be trivially generalized to the case of curved backgrounds by replacing Poincaré invariance with the invariance under the symmetry group of the considered spacetime.
\end{rem}
In the above, we have denoted by $\mathcal{K}$ the set of all causally complete sets, i.e. those sets that can be obtained by taking the causal development of a portion of a Cauchy surface, and with $\mathcal{O}_1 \vee \mathcal{O}_2 = (\mathcal{O}_1 \cup \mathcal{O}_2)''$.\\
The last demanded property is motivated by the fact that the vacuum state is Poincaré invariant and the local algebras $\mathfrak{A}(\mathcal{O})$ and $\mathfrak{A}(g\mathcal{O})$ are unitarily equivalent by the Poincaré covariance of QFT. The first two requirements are motivated by the definition of entropy that we gave in the finite dimensional case. However, as I now want to argue, such a quantity cannot exist.\\
Define the \textit{entanglement surface} $\gamma_{\mathcal{O}}$ as the $2$-dimensional surface that is the boundary of the portion of Cauchy surface whose causal developement gives $\mathcal{O}$. Then, we have the following result: 
\begin{lem}
Let $\mathcal{K}_P \subset\mathcal{K}$ be the set of causally complete sets that have a polyhedral entanglement surface. Any positive, strongly subadditive, and Poincaré invariant function $\mathcal{O} \in \mathcal{K}_P \to S(\mathcal{O}) \in \mathbb{R}$ is of the form:
\begin{equation*}
    S(\mathcal{O}) = c_1 + c_2 \mathrm{vol}(\gamma_{\mathcal{O}}),
\end{equation*}
where $c_1,c_2 \geq 0$ are positive constants independent of the set $\mathcal{O}$, while $\mathrm{vol}(\gamma_{\mathcal{O}})$ is the geometrical volume of the spacelike surface $\gamma_{\mathcal{O}}$.
\end{lem}
\begin{proof}
See Theorem $3$ in \cite{Casini_2004}.
\end{proof}
The above Lemma shows that a nontrivial entropy measure in QFT, consistent with the above definition, cannot exist. In fact, if this were the case, we could define and study the mutual information between two strictly spacelike separated sets $\mathcal{O}_1$ and $\mathcal{O}_2$:
\begin{equation*}
    I(\mathcal{O}_1, \mathcal{O}_2) := S(\mathcal{O}_1) + S(\mathcal{O}_2) - S(\mathcal{O}_1 \cup \mathcal{O}_2) = c_1. 
\end{equation*}
From $\mathrm{vol}$ being a measure, we have: $\mathrm{vol}(\mathcal{O}_1 \cup \mathcal{O}_2) = \mathrm{vol}(\mathcal{O}_1) + \mathrm{vol}(\mathcal{O}_2)$. However, as is reasonable,
the mutual information between $\mathcal{O}_1$ and $\mathcal{O}_2$ should vanish whenever we take one of these regions to be at spacelike infinity as we then expect correlations to vanish. This would imply that the constant $c_1 = 0$ and therefore, for all strictly separated regions $\mathcal{O}_1,\mathcal{O}_2 \in \mathcal{K}_P$, we have that $I(\mathcal{O}_1, \mathcal{O}_2)=0$. Moreover, for any two open regions $\mathcal{O}_1, \mathcal{O}_2 \in \mathcal{K}$ we have the existence of sets $\Tilde{\mathcal{O}}_1, \Tilde{\mathcal{O}}_2 \in \mathcal{K}_P$ such that $\mathcal{O}_j \subset \Tilde{\mathcal{O}}_j$. Then, however, since for those with polyhedral boundary the mutual information vanishes, by monotonicity of mutual information we must have $I(\mathcal{O}_1, \mathcal{O}_2)=0$ for all strictly spacelike separated, open, causally complete sets $\mathcal{O}_1, \mathcal{O}_2$.\\
However, the vacuum state must always be entangled at any distace and this results proves exactly the opposite: there is no entanglement between any spacelike separated regions. This means that there is no way in QFT, in which we can define an entropy measure with the above conditions.\\\\
Moreover, the local algebras of QFT are of type $III$ where there is no notion of trace. However, despite the absence of a notion of entanglement entropy, there is no obstruction in defining a notion of relative entropy. The only thing that we need to be aware of is, as we argued in the previous section, that the algebras are type $III$ factors and thus we need to find a new definition that, by consistency, reduces to the von Neumann relative entropy if the algebras are taken to be of type $I$.

\begin{defn}
Let $\omega_{\phi}$ and $\omega_{\psi}$ be normal, faithful and positive linear functionals over a von Neumann algebra $\mathfrak{A}$. Then, by calling $\Phi$ and $\Psi$ the vector representatives on the Hilbert space $\mathcal{H}$ of the two functionals, define the Araki's relative entropy as:
\begin{equation*}
    S(\omega_{\phi}|\omega_{\psi}) := -\braket{\Psi}{\log(\Delta_{\Phi,\Psi})\Psi}.
\end{equation*}
To simplify the notation we will also denote $S(\omega_{\phi}|\omega_{\psi}) = S(\phi|\psi)$.
\end{defn}
\begin{rem}
Using our awareness of the interpretation of relative entropy in classical information theory and since, as we will show later, this expression reduces to the von Neumann relative entropy which in turn reduces to the classical relative entropy in absence of entanglement, we can interpret this as a measure of distinguishability between the two state functionals $\omega_{\phi}$ and $\omega_{\psi}$ over the algebra $\mathfrak{A}$. In fact, such a distinguishability measure, manifestly gives distinguishable physical effects. For example, let me mention that in the entropy area relation of Black Holes, the more the excited state is distinguished from the vacuum of the considered QFT, the more we can practically spot the difference by looking at how the area of a Black Hole has changed (See for example Section $5.5$ in \cite{Kurpicz_2021}). 
\end{rem}
\begin{rem} The above defined entropy, is independent of the choice of vector representatives $\Phi$ and $\Psi$ of the states $\omega_{\psi}, \omega_{\phi}$ defined over the algebra. To see it, notice that each different representative can be obtained from the previous, by a unitary transformation (by uniqueness of the GNS construction up to unitary equivalence), namely:
\begin{equation*}
    \ket{\Psi'} = U \ket{\Psi},
\end{equation*}
where $U$ is a unitary operator in $\mathfrak{A}$. We can now compute:
\begin{align*}
    \braket{U\Psi}{\Delta_{U\Phi,U\Psi} U\Psi} &= \braket{\Psi}{S^*_{U\Phi,U\Psi} S_{U\Phi,U\Psi} \Psi}\\
    &= \braket{U\Psi}{ S^*_{U\Phi,U\Psi} S_{U\Phi,U\Psi} U\Psi}\\
    &= \braket{U\Phi}{U \Phi}\\
    &= \braket{\Phi}{\Phi}\\
    &= \braket{\Psi}{\Delta_{\Phi,\Psi} \Psi}, 
\end{align*}
which proves the independence.
\end{rem}
In the next subsections we will start investigating the properties of Araki's definition of relative entropy.

\subsubsection{Positivity of relative entropy}
In this section we focus on strict positivity and reality, namely that for $\omega_{\phi}(\mathbb{1}) = \omega_{\psi}(\mathbb{1})$ we have:
\begin{equation*}
    S(\phi| \psi) \geq 0.
\end{equation*}
Start by picking a cyclic and separating vector $\ket{\Psi}$ and consider the set of vectors, called \textit{natural cone}, defined as follows:
\begin{equation*}
    \mathcal{P}_{\Psi} = \overline{\{ Aj_{\Psi}(A) \ket{\Psi} | A \in  \mathfrak{A}\}},
\end{equation*}
where the closure is taken with respect to the weak topology and $j_{\Psi}(A) := J_{\Psi}AJ_{\Psi}$\footnote{Note that $A j_{\Psi}(A)\ket{\Psi} = j_{\Psi}(A) A \ket{\Psi}$ from the Tomita-Takesaki theorem.}. Any other state $\omega_{\phi}$ over the algebra has a unique vector representative in the natural cone denoted by $\ket{\Phi}$ (see Theorem $4$ in \cite{Araki:1972zza}). Then, using the fact that $J_{\Phi,\Psi} = J_{\Phi} = J_{\Psi}$ for $\ket{\Phi} \in P_{\Psi}$ (Theorem $4$ in \cite{Araki:1972zza}):
\begin{equation*}
    \ket{\Phi} = S_{\Phi,\Psi} \ket{\Psi} = J_{\Phi,\Psi} \Delta_{\Phi, \Psi}^{1/2} \ket{\Psi},
\end{equation*}
which gives:
\begin{align*}
    J_{\Phi} \ket{\Phi} &= \Delta_{\Phi, \Psi}^{1/2} \ket{\Psi}\\
    \ket{\Phi} &= \Delta_{\Phi, \Psi}^{1/2} \ket{\Psi}.
\end{align*}
Denoting with $E_{\lambda}$ the spectral projection of $\Delta_{\Phi,\Psi}$, we have:
\begin{equation*}
    S(\phi| \psi) = - \int_0^{\infty} \log \lambda d(\Psi, E_{\lambda} \Psi).
\end{equation*}
From the previous equation and the normalization of the states we have:
\begin{equation*}
    \int_0^{\infty} \lambda d(\Psi, E_{\lambda} \Psi) = \phi(\mathbf{1}) < \infty
\end{equation*}
and at the same time $\phi(\mathbb{1}) > 0$. It follows that:
\begin{equation*}
    -\infty < -\int_0^{\infty} \lambda d(\Psi, E_{\lambda} \Psi) < - \int_0^{\infty} \log \lambda d(\Psi, E_{\lambda} \Psi),
\end{equation*}
which shows that the relative entropy is well defined, real and either finite or $+\infty$. Moreover, for any positive measurable function $\alpha(\lambda)$ of $\lambda \in (0,\infty)$ and any probability measure $\mu$ on $(0,\infty)$, we have, by the concavity of the logarithm:
\begin{equation*}
    \int_0^{\infty} \log \alpha(\lambda) d\mu(\lambda) \leq \log \int_0^{\infty} \alpha(\lambda) d\mu(\lambda).
\end{equation*}
Take as a function $\alpha(\lambda) = \lambda^{1/2}$ and $d\mu(\lambda) = d(\Psi,E_{\lambda} \Psi)/\|\Psi\|^2$. Then, using the fact that $\|\Psi\|^2 = \psi(\mathbb{1})$, we get in the above inequality:
\begin{equation*}
    S(\phi| \psi) \geq -2 \psi(\mathbb{1})\log\big( (\Phi,\Psi)/\psi(\mathbb{1}) \big),
\end{equation*}
where we have also used the fact that:
\begin{align*}
    \ket{\Phi} &= \Delta_{\Psi, \Phi}^{1/2} \ket{\Psi}\\
    &= \int_0^{\infty} \lambda^{1/2} dE_{\lambda} \ket{\Psi}.
\end{align*}
Now, using Schwarz inequality we get:
\begin{equation*}
    (\Phi, \Psi) \leq \|\Phi\| \|\Psi\| = (\phi(\mathbb{1})\psi(\mathbb{1}))^{1/2},
\end{equation*}
from which follows:
\begin{equation*}
    S(\phi| \psi) \geq -2 \psi(\mathbb{1})\log\bigg( \bigg(\frac{\phi(\mathbb{1})}{\psi(\mathbb{1})} \bigg)^{1/2} \bigg) = - \psi(\mathbb{1})\log\bigg( \frac{\phi(\mathbb{1})}{\psi(\mathbb{1})} \bigg).
\end{equation*}
But now, assuming that all functionals are normalized, i.e. $\phi(\mathbb{1}) = \psi(\mathbb{1}) = 1$, we get:
\begin{equation*}
     S(\phi| \psi) \geq 0.
\end{equation*}

\subsubsection{Lower semicontinuity of relative entropy}
The second property of Araki's definition we present, is the lower semicontinuity. For that purpose, consider $\phi_n,\psi_n,\phi,\psi$ faithful states on $\mathfrak{A}$, with unique vector representatives in the natural cone $\mathcal{P}_{\Psi}$ given by $\Phi_n,\Psi_n,\Phi,\Psi$ such that:
\begin{equation*}
    \lim_{n \to \infty} \|\Phi_n - \Phi\| = 0 \hspace{15pt} \lim_{n \to \infty} \|\Psi_n - \Psi\| = 0. 
\end{equation*}
Then, we have strong convergence for (see \cite{Araki:1972zza}):
\begin{equation*}
    \lim_{n \to \infty} (\mathbb{1}+ \Delta_{\Phi_n, \Psi_n}^{1/2})^{-1} = (\mathbb{1}+ \Delta_{\Phi, \Psi}^{1/2})^{-1},
\end{equation*}
and for any bounded and continuous function $f$:
\begin{equation}\label{eq: limf}
    \lim_{n \to \infty} f(\Delta_{\Phi_n, \Psi_n}) = f(\Delta_{\Phi, \Psi}).
\end{equation}
Let us now take for $N = 3,4 \dots$ the function:
\begin{equation*}
    f_N(\lambda) = \left\{ \begin{aligned}
    &\log N \hspace{15pt} \mathrm{if} \,\, \lambda \geq \log N\\
    &-\log N \hspace{15pt} \mathrm{if} \,\, \lambda \leq -\log N\\
    &\lambda \hspace{15pt} \mathrm{otherwise}
    \end{aligned} \right.
\end{equation*}
If we further denote by $E_{\lambda}^n$ the spectral projection of $\Delta_{\Phi_n, \Psi_n}$, we have that:
\begin{equation*}
    \int_0^{\infty} \lambda d(\Psi_n, E^n_{\lambda} \Psi_n) = \|\Phi_n\|^2 = \phi_n(\mathbb{1}),
\end{equation*}
from which we can compute:
\begin{align*}
    0 &\leq \int_N^{\infty} (\log \lambda - \log N) d(\Psi_n, E^n_{\lambda} \Psi_n)\\
    &= \int_N^{\infty} \{ \lambda^{-1}\log(\lambda/N)  \}\lambda d(\Psi_n, E^n_{\lambda} \Psi_n).
\end{align*}
If, however, we study the function $f(\lambda) = \lambda^{-1}\log(\lambda/N)$ defined for $\lambda \in (N,\infty)$:
\begin{equation*}
    f(\lambda)' = -\frac{1}{\lambda^2}\log \bigg( \frac{\lambda}{N} \bigg) + \frac{1}{\lambda^2} = \frac{1}{\lambda^2}\bigg( 1 - \log \bigg( \frac{\lambda}{N} \bigg) \bigg),
\end{equation*}
we see $f(\lambda)' \geq 0$ for $\lambda \leq e N$ and $f(\lambda)' < 0$ for $\lambda > e N$. Therefore:
\begin{equation*}
    f(\lambda) \leq f(eN) = (eN)^{-1} \hspace{15pt} \forall \lambda \in (N,\infty).
\end{equation*}
So, going back to the previous estimate: 
\begin{align*}
    0 &\leq \int_N^{\infty} \{ \lambda^{-1}\log(\lambda/N)  \}\lambda d(\Psi_n, E^n_{\lambda} \Psi_n)\\
    &\leq (eN)^{-1} \int_N^{\infty} \lambda d(\Psi_n, E^n_{\lambda} \Psi_n)\\
    &\leq (eN)^{-1} \int_0^{\infty} \lambda d(\Psi_n, E^n_{\lambda} \Psi_n) = (eN)^{-1} \phi_n(\mathbb{1}).
\end{align*}
At the same time we compute:
\begin{equation*}
    \int_0^{1/N}(\log \lambda + \log N) d(\Psi_n, E^n_{\lambda} \Psi_n) \leq 0.
\end{equation*}
Therefore, given our above defined $f_N(\lambda)$, the next lemma can be proven:
\begin{lem}
We have:
\begin{equation*}
    S(\phi_n\|\psi_n) \geq - (\Psi_n, f_N(\log \Delta_{\Phi_n,\Psi_n})\Psi_n) - (eN)^{-1} \phi_n(\mathbb{1}).
\end{equation*}
\end{lem}
\begin{proof}
Consider in this case $f_N(\log \Delta_{\Phi_n, \Psi_n})$. We will prove the result in the three regimes: $a)$ when $\log \Delta_{\Phi_n, \Psi_n} \geq \log N$, $b)$ for $\log \Delta_{\Phi_n, \Psi_n} \leq \log (1/N)$ and $c)$ when $\log(1/N) < \log \Delta_{\Phi_n \Psi_n} < \log(N)$.\\
The easiest case is $c)$. In this regime, we have $f(\log \Delta_{\Phi_n,\Psi_n}) = \log \Delta_{\Phi_n,\Psi_n}$, so:
\begin{align*}
    -(\Psi_n, f(\log \Delta_{\Phi_n, \Psi_n}) \Psi_n) -\phi_n(\mathbb{1}) (eN)^{-1} &= -(\Psi_n, \log \Delta_{\Phi_n, \Psi_n} \Psi_n) -\phi_n(\mathbb{1}) (eN)^{-1}\\
    &\leq S(\phi_n\|\psi_n),
\end{align*}
where we have used the fact $\phi_n(\mathbb{1})(eN)^{-1} \geq 0$ and the definition of relative entropy.\\
In case $a)$, we have:
\begin{align*}
    -(\Psi_n, f(\log \Delta_{\Phi_n, \Psi_n}) \Psi_n) -\phi_n(\mathbb{1}) (eN)^{-1} &= -(\Psi_n, \log N \Psi_n) -\phi_n(\mathbb{1}) (eN)^{-1}\\
    &\leq -(\Psi_n, \log N \Psi_n) - \int_N^{\infty}(\log \lambda - \log N) d(\Psi_n, E^n_{\lambda} \Psi_n)\\
    &\leq -\int_N^{\infty} \log \lambda d(\Psi_n, E_{\lambda}^n \Psi_n)\\
    &= S(\phi_n\|\psi_n),
\end{align*}
where, in the second step we have used the inequality:
\begin{equation*}
    -\phi_n(\mathbb{1})(eN)^{-1} \leq -\int_N^{\infty}(\log \lambda - \log N) d(\Psi_n, E_{\lambda}^n \Psi_n).
\end{equation*}
Let us finally focus on $b)$. We have:
\begin{align*}
    -(\Psi_n, f(\log \Delta_{\Phi_n, \Psi_n}) \Psi_n) -\phi_n(\mathbb{1}) (eN)^{-1} &= -(\Psi_n, f(\log (1/N)) \Psi_n) -\phi_n(\mathbb{1}) (eN)^{-1}\\
    &= (\Psi_n, f(\log N) \Psi_n) -\phi_n(\mathbb{1}) (eN)^{-1}.
\end{align*}
Now, since $-\phi_n(\mathbb{1}) (eN)^{-1} \leq 0$ and we have proven before that $-\int_0^{1/N}(\log \lambda + \log N) d(\Psi_n, E^n_{\lambda} \Psi_n) \geq 0$, we have:
\begin{equation*}
    -\phi_n(\mathbb{1}) (eN)^{-1} \leq -\int_0^{1/N}(\log \lambda + \log N) d(\Psi_n, E^n_{\lambda} \Psi_n),
\end{equation*}
from which we estimate the above:
\begin{align*}
    -(\Psi_n, f(\log \Delta_{\Phi_n, \Psi_n}) \Psi_n) -\phi_n(\mathbb{1}) (eN)^{-1} &\leq (\Psi_n, f(\log N) \Psi_n) -\int_0^{1/N}(\log \lambda + \log N) d(\Psi_n, E^n_{\lambda} \Psi_n) \\
    &\leq S(\phi_n\|\psi_n). 
\end{align*}
\end{proof}
By using Eq. $\eqref{eq: limf}$ for $f(x) = f_N(\log x)$, and by the last lemma:
\begin{equation*}
    \liminf_{n \to \infty} S(\phi_n\|\psi_n) \geq - (\Psi, f_N(\log \Delta_{\Phi,\Psi}) \Psi) - \phi(\mathbb{1}) (eN)^{-1}. 
\end{equation*}
Finally, since the right hand side of this expression tends to $S(\phi\|\psi)$ as $N \to \infty$ we have the lower semicontinuity of the relative entropy:
\begin{equation*}
    \liminf_{n \to \infty} S(\phi_n\|\psi_n) \geq S(\phi\|\psi).
\end{equation*}

\subsubsection{Equivalence with von Neumann relative entropy}
I have started the chapter mentioning that the relative entropy of Araki is a generalization of the von Neumann entropy to type $III$ algebras. Let us show, how the Araki entropy reduces to the von Neumann relative entropy if we deal with type $I$ algebras. As we discussed in the earlier section, in this case we can select a subsystem and assume that we are dealing with a bipartite quantum system: $\mathcal{H} = \mathcal{H}_1 \otimes \mathcal{H}_2$. We let $\mathfrak{A}$ be the von Neumann algebra acting on $\mathcal{H}_1$ and $\mathfrak{A}'$ the von Neumann algebra (the commutant algebra of $\mathfrak{A}$) on $\mathcal{H}_2$. As we have discussed, we can decompose any vector in the Hilbert space via its Schmidt decomposition:
\begin{equation*}
    \ket{\psi} = \sum_{k=1}^n c_k \ket{k,k'},
\end{equation*}
where we have denoted with $n = \min\{ \dim \mathcal{H}_1, \dim \mathcal{H}_2 \}$ Assume w.l.o.g. that $\dim \mathcal{H}_1 \leq \dim \mathcal{H}_2$ while $\{\ket{k}\}$ denotes an orthonormal basis of $\mathcal{H}_1$. Assume also that for $\ket{\psi}$, all the $c_k \neq 0$. To see whether $\ket{\psi}$ is a cyclic and separating vector for $\mathcal{H}_1$ let us act on it with a generic element in $\mathfrak{A}$:
\begin{equation*}
    (a \otimes 1) \ket{\psi} = \sum_{k=1}^n c_k \ket{a k,k'}.
\end{equation*} 
Therefore $\ket{\psi}$ is separating iff $a \ket{k} = 0$ for all $\ket{k}$ basis vectors. But since we took the $\ket{k}$ to be an orthonormal basis of $\mathcal{H}_1$, this implies that $a = 0$. So, $\ket{\psi}$ is separating for the algebra $\mathfrak{A}$ if and only if the basis $\ket{k}$ is an eigenbasis for $\mathcal{H}_1$. If we also have that $\ket{k'}$ is an eigenbasis for $\mathcal{H}_2$ (true only if $\dim \mathcal{H}_1 = \dim \mathcal{H}_2$), then for the same argument $\ket{\psi}$ is separating also for $\mathfrak{A}'$. In the case in which $\dim \mathcal{H}_1 \neq \dim \mathcal{H}_2$, by Schmidt decomposition, we can just restrict to the sub Hilbert space $\mathcal{H}_2' \subset \mathcal{H}_2$ that has as eigenbasis $\{\ket{k'}\}$. Therefore assume w.l.o.g. that  $\dim \mathcal{H}_1 = \dim \mathcal{H}_2$.\\
Thus, by Prop. \ref{prop: 112}$, \ket{\psi}$ is cyclic and separating for both algebras $\mathfrak{A}$ and $\mathfrak{A}'$ if and only if the sets $\ket{k}$ and $\ket{k'}$ are eigenbasis of the respective Hilbert spaces. 
The Tomita operator $S_{\Psi}: \mathcal{H} \to \mathcal{H}$ becomes in this context:
\begin{equation*}
    S_{\Psi}((a \otimes 1)\ket{\Psi}) = (a^{\dagger} \otimes 1) \ket{\Psi}.
\end{equation*}
In particular, we restrict the analysis to the case when $a$ is the matrix on $\mathcal{H}_1$ that acts as:
\begin{equation*}
    a \ket{i} = \ket{j} \hspace{15pt} a \ket{k} = 0 \,\,\,\,\,\, \mathrm{for} \,\, \mathrm{all} \,\, k \neq i
\end{equation*}
with adjoint:
\begin{equation*}
     a^{\dagger} \ket{j} = \ket{i} \hspace{15pt} a^{\dagger} \ket{k} = 0 \,\,\,\,\,\, \mathrm{for} \,\, \mathrm{all} \,\, k \neq i.
\end{equation*}
Therefore, the action of $a$ and of its adjoint on the cyclic and separating vector $\ket{\Psi}$ is:
\begin{equation*}
    (a \otimes 1) \ket{\Psi} = c_i \ket{j,i'} \hspace{15pt} (a^{\dagger} \otimes 1) \ket{\Psi} = c_j \ket{i,j'}.
\end{equation*}
From the definition of the Tomita operator, we must have:
\begin{equation*}
    S_{\Psi} (c_i \ket{j,i'}) = c_j \ket{i,j'}.
\end{equation*}
That, from the antilinearity of $S_{\Psi}$:
\begin{equation*}
    S_{\Psi} \ket{j,i'} = \frac{c_j}{\overline{c}_i} \ket{i,j'},
\end{equation*}
and taking the adjoint:
\begin{equation*}
    S_{\Psi}^* \ket{i,j'} = \frac{\overline{c}_j}{c_i} \ket{j,i'}.
\end{equation*}
Finally, since $\Delta_{\Psi} = S^{*}_{\Psi} S_{\Psi}$, we have:
\begin{equation*}
    \Delta_{\Psi} \ket{j,i'} = \frac{|c_j|^2}{|c_i|^2} \ket{j,i'},
\end{equation*}
which, by the polar decomposition $S_{\Psi} = J_{\Psi} \Delta^{1/2}_{\Psi}$, gives:
\begin{align*}
    \Delta_{\Psi}^{1/2} \ket{j,i'} &= \sqrt{\frac{|c_j|^2}{|c_i|^2}}\ket{j,i'}\\
    J_{\Psi} \ket{j,i'} &= \sqrt{\frac{c_j c_i}{\overline{c}_i \overline{c}_j}} \ket{i,j'}.
\end{align*}
On the other hand, to see the decomposition of the corresponding relative modular operator, let us first introduce another state:
\begin{equation*}
    \ket{\Phi} = \sum_{\alpha=1}^{n} d_{\alpha} \ket{\alpha, \alpha'}.
\end{equation*}
However, as it is only needed for the invertibility of the Tomita operator, we do not assume the $d_{\alpha}$ to be nonzero, i.e. $\ket{\Phi}$ might not be cyclic and separating. The definition of the relative Tomita operator, for all $a \in \mathfrak{A}$, becomes :
\begin{equation*}
    S_{\Phi|\Psi} ((a \otimes 1)\ket{\Psi}) = (a^{\dagger} \otimes 1) \ket{\Phi}.
\end{equation*}
We make the same assumption as before for $a$:
\begin{align*}
    a\ket{i} = \ket{\alpha}, &\hspace{15pt} a\ket{j} = 0 \,\,\,\,\,\, \mathrm{for} \,\, \mathrm{all} \,\, j \neq i\\
    a^{\dagger}\ket{\alpha} = \ket{i}, &\hspace{15pt} a^{\dagger}\ket{\beta} = 0 \,\,\,\,\,\, \mathrm{for} \,\, \mathrm{all} \,\, \beta \neq \alpha
\end{align*}
From which follows, in analogy to the previous case, acting on the states $\ket{\Psi}$ with the above operators:
\begin{equation*}
    S_{\Phi|\Psi} \ket{\alpha,i'} = \frac{d_{\alpha}}{\overline{c}_i} \ket{i,\alpha'},
\end{equation*}
which leads to:
\begin{equation*}
    \Delta_{\Phi|\Psi} \ket{\alpha, i'} = \frac{|d_{\alpha}|^2}{|c_i|^2} \ket{\alpha, i'}.
\end{equation*}
We now have all that is needed in order to express the relative modular operator in terms of density matrices. For that, assume the vectors $\ket{\Psi}, \ket{\Phi}$ to be normalized:
\begin{equation*}
    \sum_i |c_i|^2 = \sum_{\alpha} |d_{\alpha}|^2 = 1.
\end{equation*}
Then, we define the density matrices $\rho_{12} = \ketbra{\Psi}{\Psi}$ and $\sigma_{12} = \ketbra{\Phi}{\Phi}$ with the condition:
\begin{equation*}
    \Tr_{12} \rho_{12} = \Tr_{12} \sigma_{12} = 1,
\end{equation*}
where $\Tr_{12}$ is the trace operator over the total Hilbert space $\mathcal{H}$. The notation is chosen in order to write the reduced density matrices the following way:
\begin{align*}
    \rho_1 = \sum_i |c_i|^2 \ketbra{i}{i} &\hspace{15pt} \rho_2 = \sum_i |c_i|^2 \ketbra{i'}{i'}\\
    \sigma_1 = \sum_{\alpha} |d_{\alpha}|^2 \ketbra{\alpha}{\alpha} &\hspace{15pt} \sigma_2 = \sum_{\alpha} |d_{\alpha}|^2 \ketbra{\alpha'}{\alpha'},
\end{align*}
and since, by assumption, $\ket{\Psi}$ is cyclic and separating we must have that all $c_i \neq 0$. Comparing now these expressions with those previously obtained for the Tomita and the relative Tomita operator, we see that:
\begin{equation*}
    \Delta_{\Psi} = \rho_1 \otimes \rho_2^{-1} \hspace{15pt} \Delta_{\Psi|\Phi} = \sigma_1 \otimes \rho_2^{-1}. 
\end{equation*}
Let us now look back at Araki's definition of relative entropy:
\begin{equation*}
    S_{\Phi|\Psi} = - \braket{\Psi}{\log \Delta_{\Phi|\Psi} \Psi},
\end{equation*}
which, in terms of the density matrix $\rho_{12} = \ketbra{\Psi}{\Psi}$, is:
\begin{equation*}
    = -\Tr_{12} \big( \rho_{12} \log \Delta_{\Phi|\Psi} \big).
\end{equation*}
From the previous decomposition of the modular operators, we have:
\begin{align*}
    \log \Delta_{\Phi|\Psi} &= \log (\sigma_1 \otimes \rho_2^{-1})\\
    &= \log (\sigma_1 \otimes \mathbb{1}) - \log(\mathbb{1} \otimes \rho_2),
\end{align*}
where the second equality is a consequence of:
\begin{equation*}
    e^{A \otimes \mathbb{1} + \mathbb{1} \otimes B} = e^A \otimes e^B.
\end{equation*}
By taking the logarithm, one obtains:
\begin{align*}
    \log (e^{\log A} \otimes e^{\log B}) &= \log (e^{\log(A) \otimes \mathbb{1} + \mathbb{1} \otimes \log(B)})\\
    \log(A \otimes B) &= \log(A) \otimes \mathbb{1} + \mathbb{1} \otimes \log(B). 
\end{align*}
For this reason, we have that the relative entropy becomes:
\begin{equation*}
    S_{\Phi|\Psi} = - \Tr_{12} \big(\rho_{12}(\log (\sigma_1 \otimes \mathbb{1}) - \log(\mathbb{1} \otimes \rho_2)) \big).
\end{equation*}
However, if we look at $\Tr_{12} \big( \rho_{12}\log(\sigma_1 \otimes \mathbb{1}) \big)$ and take the trace over $\mathcal{H}_2$ first, we note that:
\begin{equation*}
    \Tr_{12} \big(\rho_{12}\log(\sigma_1 \otimes \mathbb{1}) \big) = \Tr_1 \big(\rho_1 \log \sigma_1 \big).
\end{equation*}
Finally, note that $\rho_1$ is just the conjugate of $\rho_2$ under the exchange $\ket{i} \leftrightarrow \ket{i'}$, where by conjugate we mean that we consider $\ket{i'}$ to be a row vector with the corresponding column vector $\ket{i}$. To see it, choose the orthonormal frames such that all the $c_i$ in the decomposition of the cyclic and separating vector $\ket{\Psi}$, are positive numbers (i.e. there is no relative phase between $\ket{i'}$ and $\ket{i}$). In this way, we have that $J_{\Psi} \ket{i,j'} = \ket{j,i'}$, i.e. there exist an antiunitary operator that acts just as a flipping operator between basis vectors in the two factors of the tensor product Hilbert space. So, it is natural to interpret $\ket{i'}$ as the conjugate of $\ket{i}$. Therefore, we have in particular $\Tr_1 \rho_1 \log \rho_1 = \Tr_2 \rho_2 \log \rho_2$, that gives:
\begin{equation*}
    S_{\Phi|\Psi} = -\Tr_1 \rho_1(\log (\sigma_1) + \log (\rho_1)).
\end{equation*}
The above is precisely the notion of entropy as introduced by von Neumann, proving how Araki's definition is a generalization of it.

\begin{rem}
Before concluding with the section, let us remark that often the definition of entropy measures in QFT is approached differently. The idea is to introduce a $UV$ cutoff $\varepsilon$, in order to regularize the theory. The natural choice is to introduce a lattice on the spacetime, reducing  the number of degrees of freedom of the regularized QFT to a finite amount. In this way, the algebras of the theory become of type $I$ and we shall denote them as $\mathfrak{A}_{\varepsilon}(\mathcal{O})$ in order to emphasize that they depend on the regularization. In this context, an entanglement entropy, with some of the desired properties mentioned in Definition \ref{def: Rentr}, can be defined:
\begin{equation*}
    S_{\varepsilon}(\mathcal{O}) := - \Tr_{\mathfrak{A}_{\varepsilon}(\mathcal{O})}(\rho_{\mathcal{O}}\log \rho_{\mathcal{O}}).
\end{equation*}
However, it will clearly not be Poincaré invariant, as the entire theory ceases to be invariant as a consequence of the introduction of the lattice regularization. In the above definition, $\rho_{\mathcal{O}}$ is the density matrix associated with the restricted state $\omega|_{\mathfrak{A}_{\varepsilon}(\mathcal{O})}$. In this approach, also the von Neumann relative entropy will be well defined. However, even if it seems promising at first sight, one encounters the same problems. This is due to the final dependence on the regularization procedure, that we need to get rid of. In fact, if one takes the continuum limit $\varepsilon \to 0$ in this approach, unsurprisingly, this limit gives divergent results. Therefore, what is commonly done, is the removal of the divergent part as shared among all states (as it is a feature of the spacetime), by subtracting the entropy of the vacuum. For instance, this was the approach taken in \cite{Casini:2008cr} to give a rigorous proof of the Bekenstein bound \cite{Bekenstein}.
\end{rem}

In this section we have proven that Araki's relative entropy is the natural generalization of the notion of relative entropy, that we are familiar with from standard quantum mechanics in the context of QFT. From this analogy, we are allowed to interpret it as a measure of distinguishability between the different states over the abstract local algebra of observables. Moreover, we have argued why relative entropy is the only information-like notion that can be rigorously defined.\\
Since we are interested in applying entropy measures in physically relevant contexts, the next task is to find a way to easily compute it. This will be the purpose of next section where we will see how, for a specific type of excitation of a free scalar QFT, the relative entropy between two states can be computed in terms of only the relative entropy on the one-particle Hilbert space.

\subsection{Coherent states and relative entropy}\label{sec: CohBos}
Remarkably, recent works by Longo \cite{Longo:2019mhx} and Casini et al. \cite{PhysRevD.99.125020}, showed that, for a free scalar QFT, the relative entropy between the vacuum and a coherent excitation of it corresponds to the relative entropy in the one-particle Fock space. In this manner, we will be able to compute, with a rather simple expression, the relative entropy for such excitations in various contexts. This result was already applied in the context of Black-Hole physics for spacetimes admitting wedge-like regions in order to study their thermodynamical properties in more detail (see for example \cite{Kurpicz_2021} and \cite{DAngelo:2021yat}).\\\\
The goal of this section is to present these results for the free scalar field. The main objective will be to proof that the relative entropy of the second quantized theory, can in fact be computed at the one-particle level. The statements presented in this section mainly follow the work in \cite{Ciolli_2019}.\\
We start introducing the notion of entropy for a vector. Take the complex Hilbert space $\mathcal{H}$ (for instance the one yielding the bosonic Fock space for a quasifree state over the $CCR$-algebra, see Theorem \ref{thm: quasi}), and $H$ a closed, real linear subspace of $\mathcal{H}$. Define also the complex Hilbert subspaces $\mathcal{H}_0 = H \cap i H$ and $\mathcal{H}_{\infty} = H' \cap i H'$, where $H' = (iH)^{\bot_{\mathbb{R}}}$ is the orthogonal space of $iH$ with respect to the real part of the inner product defined over $\mathcal{H}$. In particular, we have the following direct sum decompositions:
\begin{equation*}
    \mathcal{H} = \mathcal{H}_0 \oplus \mathcal{H}_s \oplus \mathcal{H}_{\infty} \hspace{15pt} H = H_0 \oplus H_s \oplus H_{\infty},
\end{equation*}
where $H_0 = \{ 0 \}$ as $H$ was real, $H_{\infty} = \mathcal{H}_{\infty}$ and $H_s$ is the remaining part, where the meaning of the subscript $s$ will be clarified in a moment.\\
\begin{defn}
A real linear subspace $H \subset \mathcal{H}$ is:
\begin{itemize}
    \item A standard subspace if $H$ is closed in the topology of $\mathcal{H}$ and:
    \begin{equation*}
        \overline{H + iH} = \mathcal{H} \hspace{15pt} H \cap iH = \{ 0 \}
    \end{equation*}
    \item A standard subspace $H$ is factorial if:
    \begin{equation*}
        H \cap H' = \{ 0 \}, 
    \end{equation*}
    equivalent to saying that the direct sum $H \oplus H'$ is dense in $\mathcal{H}$. 
\end{itemize}
\end{defn}
 From the above decomposition, it follows that each closed, linear, real subspace $H$ is split into the direct sum of two trivial subspaces and a stardard one that we may denote by $H_s$ for this reason. Thus, in what follows, we shall simply deal with standard, closed, real linear subspaces. Moreover, we further restrict ourselves to spaces that are factorial.\\
On standard subspaces, we define the Tomita operator as the anti-linear operator $S: H + iH \to H + iH$, that acts as:
\begin{equation*}
    S(h_1 + ih_2) = h_1 - i h_2, 
\end{equation*}
which is manifestly involutive. For this reason, as it is invertible, we can obtain its unique polar decomposition:
\begin{equation*}
    S = J_H \Delta_H^{1/2}, 
\end{equation*}
where again $J_H$ is an anti-linear, involutive, unitary operator on $\mathcal{H}$, while $\Delta_H$ is positive, non singular, as well as self-adjoint.
\begin{lem}\label{lem: comm}
It holds that:
\begin{equation*}
    S^*_H = S_{H'},
\end{equation*}
which implies $J_H = J_{H'}$ and $\Delta_{H'} = \Delta_H^{-1}$.
\end{lem}
\begin{proof}
Pick $\xi_1, \xi_2 \in H$ and $\xi'_1, \xi'_2 \in H'$ such that we have:
\begin{align*}
    (S_H(\xi_1 + i\xi_2), \xi'_1 + i\xi'_2) &= (\xi_1 - i\xi_2, \xi'_1 + i\xi'_2)\\
    &= (\xi_1,\xi'_1) - (\xi_2, \xi'_2) + i(\xi_1, \xi'_2) + i(\xi_2, \xi'_1)\\
    &= (\xi'_1 - i \xi'_2, \xi_1 + i\xi_2)\\
    &= (S_{H'}(\xi'_1 +i\xi'_2), \xi_1 + i\xi_2),
\end{align*}
proving that $S_{H'} \subset S^*_{H}$.\\
For the reverse inequality, notice that $S^*_H$ is a closed anti-linear involution. Setting $K=\{ \xi \in D(S^*_H): S^*_H \xi = \xi \}$, we see that $K$ is a standard subspace, $H' \subset K$ and $S^*_H = S_K$. With $\xi \in H$ and $\eta \in K$ we have:
\begin{align*}
    (\xi, \eta) &= (\xi, S_K \eta)\\
    &= (\xi, S^*_H \eta)\\
    &= (\eta, S_H \xi)\\
    &= (\eta, \xi).
\end{align*}
This implies $\Im\big( (\xi, \eta)\big) = 0$, so we also have $K \subset H'$ implying $S^*_H = S_{H'}$.\\
For the second part, notice that the result just proven implies that:
\begin{align*}
    S^*_H &= J_{H'} \Delta_{H'}^{1/2}\\
    &= \Delta_H^{1/2} J_H\\
    &= J_H \Delta_H^{-1/2}, 
\end{align*}
which implies $J_H = J_{H'}$ and also $\Delta_{H'}^{1/2} = \Delta_H^{-1/2}$ by uniqueness of polar decomposition.
\end{proof}
With this, we can prove the following result regarding the modular conjugation and flow:
\begin{prop}
For all $t \in \mathbb{R}$:
\begin{equation*}
    J_H H = H' \hspace{15pt} \Delta_H^{it} H = H. 
\end{equation*}
\end{prop}
\begin{proof}
Following the same line of reasoning that we used when discussing the Tomita operator for von Neumann algebras, we have that $\Delta_H^{is}$ commutes with $J_H$ (and obviously also with $\Delta_H$ itself). Then, by picking any $\xi \in H$, we have:
\begin{equation*}
    S_H \Delta^{it}_H \xi = \Delta^{it}_H S_H \xi = \Delta^{it}_H \xi, 
\end{equation*}
which means $\Delta^{it}_H H \subset H$ for any $t \in \mathbb{R}$ and thus $\Delta^{it}_H H = H$.\\
Concerning the first result, pick a $\xi \in H$ and compute:
\begin{align*}
    (J_H \xi, \xi) &= (J_H S_H \xi, \xi)\\
    &= (\Delta_H^{1/2} \xi, \xi) \in \mathbb{R}.
\end{align*}
Thus, for all $\xi, \eta \in H$, we have that:
\begin{equation*}
    (J_H(\xi + \eta), \xi + \eta) = (J_H \xi, \xi) +(J_H \eta, \eta) + (J_H \xi, \eta) + (J_H \eta, \xi).
\end{equation*}
But now, since this expression must be real, we have that $\Im\big( (J_H \xi, \eta) \big) = 0$, namely $J_H H \subset H'$. Moreover, as $J_H = J_{H'}$ we also have that $J_{H'} H' \subset H'' = H$, namely $H' \subset J_H H$. 
\end{proof}
We started off assuming $H$ to be a standard factorial subspace of $\mathcal{H}$. This allows us to take $k \in H \oplus H'$ and define:
\begin{equation*}
    P_H k = h, 
\end{equation*}
where we have a canonical decomposition $k = h + h'$ with $h \in H$ and $h' \in H'$. This defines a real, linear, densely defined operator from $\mathcal{H}$ to $\mathcal{H}$ that we call the \textit{cutting projection} relative to $H$. The cutting projection will be used to define the relative entropy for $H$, but first we need to prove some of its properties and give a useful formula that expresses $P_H$.
\begin{prop}[\textbf{Properties of $P_H$}]
We have:
\begin{itemize}
    \item $P_H$ is a real linear, closed, densely defined operator
    \item $P_H^2 = P_H$
    \item $P_H + P_{H'} = \mathbb{1}|_{H+H'}$
    \item $P^*_H = P_{iH} = -iP_Hi$, where the adjoint is taken with respect to the real part of the inner product over $\mathcal{H}$
    \item $\Delta^{is}_H P_H = P_H \Delta^{is}_H$
\end{itemize}
\end{prop}
\begin{proof}
The second and third properties are trivially fulfilled due to the very definition of $P_H$.\\
We prove the first statement. Let $k_n \in H + H'$ be a sequence such that $k_n \to k$ and $P_H k_n \to h$. Decompose the elements of the sequence: $k_n = h_n + h'_n$ and $h_n + h_n' \to h$, where $h_n \in H$ and $h_n' \in H'$. Then, since $H$ is closed, we have $h_n \to h$ with $h \in H$ and $h'_n \to h'=k-h \in H'$.  So, $k \in \mathrm{Dom}(P_H)$ and we have $P_H k = h$, proving that $P_H$ is closed.\\
We have $P_H^2 = P_H$ and $P^{*2}_H = P_H^*$. Then $P^*_H = P_{iH}$ because:
\begin{equation*}
    \overline{\mathrm{Ran}(P^*_H)} = \ker(P_H)^{\bot_{\mathbb{R}}} = H'^{\bot_{\mathbb{R}}} = iH, 
\end{equation*}
where we used the fact that the range of $P_H^*$ must be the set of vectors $v$ such that if we take any $u \in \ker(P_H)$
\begin{equation*}
    0 = (v, P_H u). 
\end{equation*}
However, taking the adjoint, this implies:
\begin{equation*}
    0 = \Re\big( (P_H^* v, u) \big)
\end{equation*}
and this equality is true for all $u \in \ker(P_H)$. Then, $P_H^* v \in \ker(P_H)^{\bot_{\mathbb{R}}}$ implying the first equality: $\overline{\mathrm{Ran}(P_H^*)} = \ker(P_H)^{\bot_{\mathbb{R}}}$. Moreover, the Kernel is $H'$ and in the last step we made use of the fact that $H' = (iH)^{\bot_{\mathbb{R}}}$. Along the same line of reasoning, we find:
\begin{equation*}
    \ker(P_H^*) = \mathrm{Ran}(P_H)^{\bot_{\mathbb{R}}} = (H)^{\bot_{\mathbb{R}}} = (iH').
\end{equation*}
The last statement follows from the fact that $\Delta^{is}_H H = H$ and $\Delta^{is}_{H'} H' = H'$
\end{proof}
Let us pick an $\epsilon > 0$ and denote by $E_{\epsilon}$ the spectral projection of $\Delta_H$ relative to the subset $(\epsilon, 1 - \epsilon) \cup ((1-\epsilon)^{-1}, \epsilon^{-1})$ of $\mathbb{R}$ and further denote by $\mathcal{H}_{\epsilon} = E_{\epsilon} \mathcal{H}$ the corresponding spectral subspace. Evidently, we have that $E_{\epsilon} H \subset H$ as $E_{\epsilon} H' \subset H'$. Consider also the dense, complex linear subspace of $\mathcal{H}$ given by:
\begin{equation*}
    \mathcal{D}_0 := \bigcup_{\epsilon > 0} \mathcal{H}_{\epsilon}. 
\end{equation*}
Let us consider the functions for $\lambda \in (0, +\infty)$:
\begin{equation*}
    a(\lambda) = \lambda^{-1/2}(\lambda^{-1/2} - \lambda^{1/2})^{-1} \hspace{15pt} b(\lambda) = (\lambda^{-1/2} - \lambda^{1/2})^{-1}
\end{equation*}
and set:
\begin{equation*}
    \mathcal{D} := \mathrm{Dom}(a(\Delta_H)) \cap \mathrm{Dom}(b(\Delta_H)). 
\end{equation*}
Clearly, $\mathcal{D}_0 \subset \mathcal{D}$ and $\mathcal{D}_0$ is a core (a subset of the domain on which the closure of the operator is the operator itself) for both $a(\Delta)$ and $b(\Delta)$. Since the domain of the sum of two operators is the intersection of their domains, $\mathcal{D}$ is the domain of $P_H$. This is expressed via the following:
\begin{thm}
We have $\mathcal{D} \subset H + H'$ and:
\begin{equation}\label{eq: projection}
    P_H|_{\mathcal{D}} = \Delta_H^{-1/2}(\Delta_H^{-1/2} - \Delta_H^{1/2})^{-1} + J_H(\Delta_H^{-1/2} - \Delta_H^{1/2})^{-1}. 
\end{equation}
Moreover, $\mathcal{D}$ is a core for $P_H$, namely:
\begin{equation*}
    P_H = \overline{(a(\Delta_H) + J_H b(\Delta_H))}.
\end{equation*}
Indeed, already $\mathcal{D}_0$ is a core for $P_H$.
\end{thm}
\begin{proof}
First we assume that $0,1 \not \in \mathrm{sp}(\Delta_H)$, thus $\mathrm{sp}(\Delta_H)$ is bounded. Then $H + H'$ and $\mathcal{D}$ are equal to $\mathcal{H}$, as we are avoiding any divergences of $a(\lambda), b(\lambda)$ for any vector on which we act and thus both $a(\Delta_H)$, $b(\Delta_H)$ are bounded operators. Any $k \in \mathcal{H}$ can be written as $k = h + h'$, with $h\in H$ and $h' \in H'$. As both $H$ and $H'$ are real subspaces, we have: $J_H \Delta_H^{1/2} h = h$, $J_{H'} \Delta_{H}^{-1/2} h' = h'$. Then:
\begin{align*}
    \Delta_H^{-1/2}(\Delta_H^{-1/2} &- \Delta_H^{1/2})^{-1}h + J_H (\Delta_H^{-1/2} - \Delta_H^{1/2})^{-1}h\\
    &=  \Delta_H^{-1/2}(\Delta_H^{-1/2} - \Delta_H^{1/2})^{-1}h - (\Delta_H^{-1/2} - \Delta_H^{1/2})^{-1} J_H h\\
    &=  \Delta_H^{-1/2}(\Delta_H^{-1/2} - \Delta_H^{1/2})^{-1}h - (\Delta_H^{-1/2} - \Delta_H^{1/2})^{-1} \Delta_H^{1/2}h\\
    &= h
\end{align*}
and:
\begin{align*}
     \Delta_H^{-1/2}(\Delta_H^{-1/2} &- \Delta_H^{1/2})^{-1}h' + J_H (\Delta_H^{-1/2} - \Delta_H^{1/2})^{-1}h'\\
    &=  \Delta_H^{-1/2}(\Delta_H^{-1/2} - \Delta_H^{1/2})^{-1}h' - (\Delta_H^{-1/2} - \Delta_H^{1/2})^{-1} J_Hh'\\
    &=  \Delta_H^{-1/2}(\Delta_H^{-1/2} - \Delta_H^{1/2})^{-1}h' - (\Delta_H^{-1/2} - \Delta_H^{1/2})^{-1} \Delta_H^{-1/2}h\\
    &= 0. 
\end{align*}
This proves that the sum of these two operators is $P_H$.\\
In the general case, we consider the orthogonal decomposition $\mathcal{H} = \mathcal{H}_{\epsilon} \oplus \mathcal{K}_{\epsilon}$, where $\mathcal{K}_{\epsilon} = \mathcal{H}_{\epsilon}^{\bot}$ is the complementary spectral subspace of $\Delta_H$. We also have the corresponding decomposition $H = H_{\epsilon} \oplus K_{\epsilon}$. As $0,1 \not\in \mathrm{sp}(\Delta_H|_{\mathcal{H}_{\epsilon}})$ (because of our construction of $\mathcal{H}_{\epsilon}$ with $\epsilon > 0$ which does not include $0,1$ in the spectrum), we conclude that $\mathcal{D}_0 \subset H + H'$ and thus Eq. $\eqref{eq: projection}$ holds true with $\mathcal{D}_0$ instead of $\mathcal{D}$ by following the same proof that we used above.\\
Let $k \in \mathcal{D}$ and call $k_{\epsilon} = E_{\epsilon} k \in \mathcal{H}_{\epsilon}$. As $\epsilon \to 0$ we have $k_{\epsilon} \to k$ and:
\begin{align*}
    a(\Delta_H) k_{\epsilon} &\to a(\Delta_H)k\\
    b(\Delta_H) k_{\epsilon} &\to b(\Delta_H)k. 
\end{align*}
If we apply Eq. $\eqref{eq: projection}$ to $k_{\epsilon}$ we have that the limiting element $P_H k_{\epsilon} \to P_H k$ must be well defined as a consequence of $P_H$ being closed. So, $k \in \mathrm{Dom}(P_H) = H + H'$ and:
\begin{equation*}
    P_H k = a(\Delta_H)k + J_Hb(\Delta_H)k,
\end{equation*}
which shows Eq. \eqref{eq: projection} holds on $H + H'$.\\
Let now $k \in H +H'$. As $P_H$ commutes with $E_{\epsilon}$ as well, we have that $k_{\epsilon} \to k$ and $P_H k_{\epsilon} = P_H E_{\epsilon} k = E_{\epsilon} P_H k \to P_H k$. This shows that $\mathcal{D}_0 = \mathcal{K}_{\epsilon}$ is a core for $P_H$ as, by taking the closure, we include the limiting elements $k$ obtaining $P_H$: $\overline{P_H |_{\mathcal{D}_0}} = P_H$.
\end{proof}
The functions $a(\lambda), b(\lambda)$ are continuous functions on $(0, +\infty)$ with a singularity at $\lambda = 1$ as they are defined in terms of $(\lambda^{-1/2} - \lambda^{1/2})^{-1}$, namely $a(\lambda), b(\lambda) = o(\lambda^{-1})$ as $\lambda \to 1$.\\\\
We can now define the notion of entropy for a vector:
\begin{defn}\label{def: EE sub}
Let $H$ be a standard subspace of $\mathcal{H}$. If $h \in H$, the entropy $S_h$ of $h$ w.r.t. $H$ is defined as:
\begin{equation*}
    S_h = - (h, \log \Delta_H h). 
\end{equation*}
\end{defn}
We now want to generalize this to arbitrary vectors of $\mathcal{H}$. In order to do so, let us first represent the modular flow in terms of its spectral decomposition:
\begin{equation*}
    \Delta_H = \int_0^{\infty} \lambda dE(\lambda).
\end{equation*}
Then, we can generalize the above to: 
\begin{defn}
If $k \in \mathcal{H}$, we define the entropy $S_k$ of the vector $k$ with respect to $H$ by:
\begin{equation*}
    S_k = \Im \big( (k, P_H A k)\big)
\end{equation*}
For $A = i \log \Delta_H$. 
\end{defn}
Indeed, by introducing the spectral decomposition of the modular flow and the formula for $P_H$ we see that this is well defined for all $k \in \mathcal{H}$:
\begin{equation*}
    (k, P_H A k) = i \int_0^{\infty} a(\lambda) \log \lambda d(k,E(\lambda)k) -i \int_0^{\infty} b(\lambda) \log \lambda d(k,J_H E(\lambda)k).
\end{equation*}
Notice that $b(\lambda) \log \lambda$ is always a well defined bounded operator on $(0,+\infty)$ as the singularity in $\lambda \to 1$ is removed:
\begin{align*}
    \lim_{\lambda \to 1} \frac{\log \lambda}{(\lambda^{-1/2} - \lambda^{1/2})} &= \lim_{\lambda \to 1} \frac{\log (1+ (\lambda-1))}{((1+(\lambda-1))^{-1/2} - (1+(\lambda -1))^{1/2})}\\
    &\simeq \frac{\lambda-1}{1-\frac{\lambda -1}{2}-1-\frac{\lambda-1}{2}}\\
    &\simeq -1. 
\end{align*}
Secondly, we know that $a(\lambda) \log \lambda$ is bounded in $(1,+\infty)$ as well and positive in $(0,1)$. Therefore, the entropy is finite if $-\int_0^1 a(\lambda) \log \lambda d(k,E(\lambda)k) < \infty$, otherwise $S_k = + \infty$. Indeed:
\begin{prop}
Let $k \in \mathcal{H}$. We have:
\begin{equation}\label{eq: equiv S}
    S_k < +\infty \Leftrightarrow -\int_0^1 \log \lambda\,\, d(k,E(\lambda)k) < +\infty, 
\end{equation}
iff $k \in \mathrm{Dom}(\sqrt{|\log \Delta_H|}E_-)$, with $E_{-}$ the negative spectral projection for $\log \Delta_H$.\\
In particular, all vectors in $\mathrm{Dom}(\log \Delta_H)$ have finite entropy.
\end{prop}
\begin{proof}
Notice that:
\begin{equation*}
    \lim_{\lambda \to 1^-} a(\lambda) \log \lambda \sim - 1 - \frac{\lambda - 1}{2}  
\end{equation*}
with the same limit being finite for $\log \lambda$ as well. Moreover, $a(\lambda) \to 1$ as $\lambda \to 0^+$.
\end{proof}
Finally, we prove the main properties of the entropy for vectors:
\begin{prop}\label{prop: S_h}
Let $H$ be a factorial standard subspace of $\mathcal{H}$ and $k \in \mathcal{H}$. The following holds:
\begin{itemize}
    \item[i)] If $k = h + h'$, with $h \in H$, $h' \in H'$, then $S_k = -(h, \log \Delta_H h)$
    \item[ii)] $S_k = (k, iP_H i \log \Delta_H k) = - (k,P^*_H \log \Delta_H k)$
    \item[iii)] $S_k \geq 0$ and $S_k = 0$ iff $k \in H'$
    \item[iv)] If $\mathcal{H} = \mathcal{H}_1 \oplus \mathcal{H}_2$ with $H = H_1 \oplus H_2$, and $k = k_1 \oplus k_2$, then $S^H_k = S^{H_1}_{k_1} + S^{H_2}_{k_2}$, in particular for $k = k_1$: $S^H_{k_1} = S^{H_1}_{k_1}$
    \item[v)] $S_k = \lim_{\epsilon \to 0^+} S_{k_{\epsilon}}$; where $S_{k_{\epsilon}} < +\infty$ and $S_{k_{\epsilon}}$ is non-decreasing as $\epsilon \to 0^+$
    \item[vi)] If $k_n \to k$ in the graph norm of $\sqrt{|\log \Delta_H|}E_-$, then $S_{k_n} \to S_k$
\end{itemize}
\end{prop}
\begin{proof}
Let us start proving $i)$. If $k \in \mathrm{Dom}(\log \Delta_H)$ or $k \in \mathrm{Dom}(P_H)$, we have $k = h +h'$, with $h \in H$ and $h' \in H'$ and both $h,h' \in \mathrm{Dom}(\log \Delta_H)$. Then:
\begin{align*}
    S_k &= \Im\big( (k,P_H A k) \big)\\
    &= \Im\big( (h+h',P_H A (h+h')) \big)\\
    &= \Im\big( (h+h', A h) \big)\\
    &= \Im \big( (h,A h) \big)\\
    &= - (h, \log \Delta_H h). 
\end{align*}
We then have:
\begin{equation*}
    S_k \geq 0. 
\end{equation*}
Indeed, by assuming $\|h\| = 1$, we find:
\begin{equation*}
    (h, \log \Delta_H \,\, h) = \int \log(\lambda) d(h,E(\lambda) h) \leq - \log \bigg( \int \lambda d(h, E(\lambda) h)  \bigg) = \log \|h\|^2 = 0,
\end{equation*}
where we have used Jensen's inequality.\\
Statement $iv)$, follows from the very definition of entropy.\\
Result $v)$ is proven noticing that $S_{k_{\epsilon}}$ is finite, as the restriction of the cutting projection and of $\log \Delta_H$ on $\mathcal{H}_{\epsilon}$ with them being bounded operators, and thus by the spectral theorem $S_{k_{\epsilon}} \to S_k$ as $\epsilon \to 0^+$. Moreover, we have proven above that $S_{k_{\epsilon}} \geq 0$. Thus, by statement $iv)$, we see that $S_{k_{\epsilon}}$ increases, in converging to $S_k$, as $\mathcal{H} = \mathcal{H}_{\epsilon} \oplus \mathcal{K}_{\epsilon}$ and $k_{\epsilon} \in \mathcal{H}_{\epsilon}$, which becomes larger by reducing $\epsilon$\\
Let us now look at $iii)$. By $v)$, and the positivity of each $S_{k_{\epsilon}}$, we also have $S_k \geq 0$. If, on the other hand, we suppose $S_k = 0$, by the monotone increase proven in $v)$, we must also have $S_{k_{\epsilon}} = 0$. But now, by $i)$ and the strict positivity in Def. \ref{def: EE sub}, we have that $k_{\epsilon} \in H'$ and therefore $k \in H'$.\\
The general case in $i)$ subsequently follows, as we have $S_k = \lim_{\epsilon \to 0} S_{k_{\epsilon}}$, and using that $\mathrm{Dom}(P_H)$ is a core for $P_H$.\\
For $ii)$, start noticing that the first equality simply represents the statement that $(k,iP_H i \log \Delta_H k)$ is real if finite. Now, remembering that $A = i \log \Delta_H$, we have that $(k_{\epsilon}, iP_Hi \log \Delta_H k_{\epsilon})$ is real from the identity we have proven at the beginning:
\begin{equation*}
    \Im\big( (k_{\epsilon},P_H A k_{\epsilon}) \big) = \Im\big( (k_{\epsilon},P_H i \log \Delta_H k_{\epsilon}) \big) =  - (h_{\epsilon}, \log \Delta_H h_{\epsilon}), 
\end{equation*}
and from the simple result:
\begin{equation*}
    \Re\big( (k_{\epsilon},iP_H i \log \Delta_H k_{\epsilon}) \big) = \Im\big( (k_{\epsilon},P_H i \log \Delta_H k_{\epsilon}) \big),
\end{equation*}
in conjunction with statement $v)$, we conclude that also $ii)$ holds. Moreover, we have seen that $P_H^* =- i P_H i$, so the result follows.\\
\end{proof}
We now have everything that is needed to introduce second quantization for bosonic scalar free fields and the notion of coherent excitations.\\
In particular, if $\mathcal{H}$ is a complex Hilbert space, we can construct (namely making it the one-particle Hilbert space) the Fock space on it as:
\begin{equation*}
    \mathfrak{S}(\mathcal{H}) = \bigoplus_{n=0}^{\infty} \mathcal{H}_s^{\otimes^n},
\end{equation*}
which for this reson is also called the \textit{exponential} of $\mathcal{H}$, often denoted as $e^{\mathcal{H}}$. Moreover, $\mathcal{H}_0 = \mathbb{C} \Omega$ is the one-dimensional Hilbert space for the vacuum vector $\Omega$ with $\mathcal{H}_s^{\otimes n}$ being the symmetrized $n$-fold tensor product of $\mathcal{H}$. Now, if we have $h \in \mathcal{H}$, we can define the coherent vector $e^h$ as:
\begin{equation*}
    e^h := \bigoplus_{n=0}^{\infty} \frac{1}{\sqrt{n!}}(h^{\otimes^n})_s, 
\end{equation*}
where the zeroth component of $e^h$ is given by $e^0 = \Omega$. We can prove that:
\begin{align*}
    (e^h, e^k) &= \bigg(\bigoplus_{n=0}^{\infty} \frac{1}{\sqrt{n!}}(h^{\otimes^n})_s, \bigoplus_{n=0}^{\infty} \frac{1}{\sqrt{n!}}(k^{\otimes^n})_s \bigg)\\
    &= \bigoplus_{n=0}^{\infty} \frac{1}{n!}\big((h^{\otimes^n})_s, (k^{\otimes^n})_s\big)\\
    &= \bigoplus_{n=0}^{\infty} \frac{1}{\sqrt{n}!}\big(h^{\otimes^n}, k^{\otimes^n}\big)_s\\
    &= e^{(h,k)}
\end{align*}
and $\{ e^h, h \in \mathcal{H} \}$ is a dense subset of $\Gamma(\mathcal{H})$.\\
For $h \in \mathcal{H}$, we can identify a corresponding \textit{Weyl algebra} (see \ref{sec: Weyl}) generated by $V(h)$:
\begin{equation*}
    V(h) \frac{e^k}{\|e^k\|} = \frac{e^{h+k}}{\|e^{h+k}\|},
\end{equation*}
satisfying the Weyl commutation relations:
\begin{equation*}
    V(h+k) = e^{i \Im\big( (h,k) \big)}V(h) V(k).
\end{equation*}
In particular:
\begin{align*}
    V(h) \ket{\Omega} &= V(h) e^0\\
    &= \frac{e^h}{\| e^h \|},
\end{align*}
and since:
\begin{equation*}
    \|e^h\| = \sqrt{(e^h,e^h)} = \sqrt{e^{(h,h)}} = e^{\frac{1}{2}(h,h)},
\end{equation*}
we have:
\begin{equation*}
    V(h) \ket{\Omega} = e^{-\frac{1}{2}(h,h)} e^h. 
\end{equation*}
Therefore:
\begin{equation*}
    (V(k)\Omega, V(h)\Omega) = e^{-\frac{1}{2}(\|h\|^2+\|k\|^2)}(e^k,e^h) = e^{-\frac{1}{2}(\|h\|^2 + \|k\|^2)}e^{(k,h)}
\end{equation*}
and if we denote by $\omega := (\Omega, \cdot \Omega)$ the vacuum state, we obtain:
\begin{equation*}
    \omega(V(h)) = e^{-\frac{1}{2}\|h\|^2}. 
\end{equation*}
Let $H \subset \mathcal{H}$ be a real linear subspace. Then we can define the von Neumann algebra on $\Gamma(\mathcal{H})$:
\begin{equation*}
    R(H) := \{ V(h): h \in H \}''. 
\end{equation*}
If we further assume that $H$ is a standard subspace, we have:
\begin{prop}
Let $H$ be a standard, real, linear subspace, then:
\begin{itemize}
    \item If $K \subset H$ is dense, then $R(K) = R(H)$
    \item $\Omega$ is cyclic and separating for $R(H)$
    \item $R(H') = R(H)'$
    \item The modular flow and conjugation associated with $(R(H),\Omega)$ are given by:
    \begin{equation*}
        \Delta^{it}_{R(H)} = \mathfrak{S}(\Delta^{it}_H) \hspace{20pt} J_{R(H)} = \mathfrak{S}(J_H)
    \end{equation*}
    Where $\mathfrak{S}(\cdot)$ denotes the second quantization of the operator acting on the one-particle Hilbert space.
\end{itemize}
\end{prop}
\begin{proof}
Let us start from the first. As $K \subset H$ we must have $R(K) \subset R(H)$. Now, as $K$ is dense in $H$, for each $h \in H$ we can find a sequence $(h_n)_{n \in \mathbb{N}} \in K$ such that $h_n \to h$ in $H$. I claim, that this implies $V(h_n) \to V(h)$ in the weak topology. For that purpose, consider the action of $V(h_n) \in R(K)$ on a general element $\Psi \in \mathfrak{S}(\mathcal{H})$ that, w.l.o.g., is assumed to be normalized. We know that the set of normalized independent vectors $e^k/\|e^k\|$ for $k \in \mathcal{H}$ is dense, so we can always find $e^{k_l}/\|e^{k_l}\| \to \Psi$. Then:
\begin{align*}
    \lim_{n \to \infty} V(h_n) \Psi &= \lim_{n \to \infty} V(h_n) \lim_{l \to \infty}\frac{e^{k_l}}{\|e^{k_l}\|}\\
    &= \lim_{n \to \infty} \lim_{l \to \infty}\frac{e^{h_n + k_l}}{\|e^{h_n + k_l}\|}\\
    &= \lim_{l \to \infty} \frac{e^{h+k_l}}{\|e^{h+k_l}\|}\\
    &= V(h) \lim_{l \to \infty}\frac{e^{k_l}}{\|e^{k_l}\|}\\
    &= V(h) \Psi,
\end{align*}
proving the convergence in the weak topology of $V(h_n) \in R(K)$ to $V(h) \in R(H)$. But now, from Theorem \ref{thm: DouComm} it follows that $\{ V(h): h \in H \}'' = \overline{\{ V(h): h \in H \}}$, with closure taken with respect to the weak topology, that is $R(K) = R(H)$.\\
The second statement is simply the Reeh-Schlieder theorem\\
For the third statement, notice that as $\Omega$ is cyclic and separating for $R(H)$ and by Reeh-Schlieder theorem it must be for $R(H)'$ as well: 
\begin{equation*}
    \overline{R(H') \ket{\Omega}} = \mathfrak{S}(\mathcal{H}') = \mathfrak{S}(J_{H}) \mathfrak{S}(\mathcal{H}) = J_{R(H)} \overline{R(H) \ket{\Omega}} = \overline{J_{R(H)} R(H) J_{R(H)} \ket{\Omega}} = \overline{R(H)' \ket{\Omega}},
\end{equation*}
where we have used the fourth statement $J_{R(H)} = \mathfrak{S}(J_H)$ and the one for factorial standard subspaces, $J_H = J_{H'}$, see Lemma \ref{lem: comm}.\\
The last statement follows from the very definition of second quantization of a unitary operator on $\mathcal{H}$. Namely, if $U$ is a unitary operator on $\mathcal{H}$, its second quantization $\mathfrak{S}(U)$ is defined as the direct sum of: 
\begin{equation*}
    \mathfrak{S}(U)|_{\mathcal{H}_s^{\otimes^n}} = U \otimes U \otimes \dots \otimes U. 
\end{equation*}
From this definition, it follows that $\mathfrak{S}(U) e^h = e^{U h}$. Finally, since $J_H$ and $\Delta_H^{it}$ are unitary and since for $h \in \mathcal{H}$, we have:
\begin{equation*}
    e^h = \| e^h\| V(h) \ket{\Omega},
\end{equation*}
from which follows the claim.
\end{proof}

Let us now prove the analogy between relative entropy for coherent excitations of the vacuum, of a free real scalar QFT, and the entropy of a vector of the corresponding one particle Hilbert space. We will first prove the result for $h \in H$ and subsequently generalize it for arbitrary $h \in \mathcal{H}$.

\subsubsection{Case $h\in H$}
We saw that by taking $h \in H$ to be a real, linear, standard subspace of $\mathcal{H}$, the entropy of $h$, relative to $H$ is:
\begin{equation*}
    S_h = - (h, \log \Delta_H h).
\end{equation*}
For what concerns the relative entropy on the Fock space, we start by specifying a coherent state on $R(H)$: $\omega_h = (V(h) \Omega, \cdot V(h) \Omega)$, for some $h \in \mathcal{H}$. We notice that:
\begin{equation*}
    \omega_h = \omega \cdot \mathrm{Ad} V(h)^*|_{R(H)},
\end{equation*}
where, for $A \in R(H)$:
\begin{equation*}
    \mathrm{Ad}V(h)|_{R(H)}[A] := V(h) A V(h)^*. 
\end{equation*}
We first focus on the case in which $h,k \in H$ and we study $S(\omega_h\| \omega_k)$.\\
Let us start proving the following important general results, that will turn out to also be crucial for later purposes:

\begin{prop}\label{prop: ciclico}
Let $\mathcal{A}$ be a von Neumann algebra and denote by $\mathcal{A}'$ its commutant. Let $\Omega$ be a cyclic separating vector for $\mathcal{A}$ and $U \in \mathcal{A}$, $U' \in \mathcal{A}'$ unitary operators. Then, the vector $\Phi = U' U \Omega$ is cyclic and separating and we further have:
\begin{equation*}
    S_{\Omega, \Phi} = U S_{\Omega} U'^*.
\end{equation*}
Additionally, by polar decomposition, we find:
\begin{align*}
    J_{\Omega, \Phi} &= UJ_{\Omega}U'^*\\
    \Delta_{\Omega, \Phi} &= U' \Delta_{\Omega} U'^*.
\end{align*}
\end{prop}
\begin{proof}
We have:
\begin{align*}
    \overline{\mathcal{A} \ket{\Phi}} &= \overline{\mathcal{A}U'U \ket{\Omega}}\\
    &= U' \overline{\mathcal{A}U\ket{\Omega}}\\
    &= U' \overline{\mathcal{A} \ket{\Omega}}\\
    &= \overline{\mathcal{A} \ket{\Omega}} = \mathcal{H},
\end{align*}
where in the third step we have used the fact that $U \in \mathcal{A}$. One can prove the same for $\mathcal{A}'$, obtaining that as $\ket{\Phi}$ is cyclic also for the commutant algebra, it must be separating for the algebra itself. This was proven in Prop. \ref{prop: 112}.\\
For any $A \in \mathcal{A}$ we have:
\begin{align*}
    (U S_{\Omega} U'^*) A \ket{\Phi} &= (U S_{\Omega} U'^*) A U' U \ket{\Omega}\\
    &= U S_{\Omega} U'^*  A U' U \ket{\Omega}\\
    &= U S_{\Omega} A U \ket{\Omega}\\
    &= U U^* A^* \ket{\Omega}\\
    &= A^* \ket{\Omega}\\
    &= S_{\Omega, \Phi} A \ket{\Phi},
\end{align*}
where we have used $U' \in \mathcal{A}'$. From this, it follows by polar decomposition that:
\begin{equation*}
    J_{\Omega, \Phi} \Delta^{1/2}_{\Omega, \Phi} = U J_{\Omega} U'^* U' \Delta^{1/2}_{\Omega} U'^*, 
\end{equation*}
from which we have:
\begin{align*}
     J_{\Omega, \Phi} &= UJ_{\Omega}U'^*\\
    \Delta_{\Omega, \Phi} &= U' \Delta_{\Omega} U'^*. 
\end{align*}
\end{proof}

\begin{rem}
Observe that we can take $U' = \mathbb{1}$ and in this way obtain that any $U \ket{\Omega}$ is cyclic and separating whenever $\ket{\Omega}$ is.
\end{rem}
From this proposition, in the particular case in which we consider $U = V(h) \in R(H)$ and $U' = \mathrm{Ad} J_{\Omega} [V(h)]$, we have that:
\begin{equation*}
    J_{\Omega, \Phi} = V(h) J_{\Omega} J_{\Omega} V^*(h) J_{\Omega} = V(h) V(-h) J_{\Omega} = J_{\Omega}.
\end{equation*}
For what concerns the relative entropy between coherent states, we may cover all possibilities, by computing the relative entropy between a single coherent excitation and the vacuum vector:
\begin{prop}
The relative entropy between two coherent states can be computed to be:
\begin{equation*}
    S(\omega_h\|\omega_k) = S(\omega_{k-h}\| \omega),
\end{equation*}
where $\omega$ denotes the vacuum state.
\end{prop}
\begin{proof}
Let us consider $\ket{\Omega}, V(h)\ket{\Omega}, V(k)\ket{\Omega}$ and denote by $S_{k,h}$ the associated Tomita operator (well defined as $V(h) \ket{\Omega}$ is cyclic and separating for all $h \in H$ as long as $\ket{\Omega}$ is). Now, for $U \in R(H)$ unitary and for any $A \in R(H)$:
\begin{align*}
    (U S_{\Omega, \Phi} U^*) A U \ket{\Phi} &= U S_{\Omega,\Phi} (U^* A U) \ket{\Phi}\\
    &= U (U^* A U)^* \ket{\Omega}\\
    &= U U^* A^* U \ket{\Omega}\\
    &= A^* U \ket{\Omega},
\end{align*}
from which follows:
\begin{equation}\label{eq: unit1}
   U S_{\Omega, \Phi} U^* = S_{U\Omega, U\Phi} \Rightarrow \Delta_{U \Omega, U \Phi} = U\Delta_{\Omega, \Phi} U^*.
\end{equation}
Then, however, if we take as cyclic and separating states $\ket{\Omega}$ and:
\begin{equation*}
   \ket{\Phi} :=  V(k)^* V(h) \ket{\Omega},
\end{equation*}
and we look at the relative Tomita operator $S_{V(k) \Omega, V(k) \Phi} = S_{V(k) \Omega, V(h) \Omega}$, we have:
\begin{equation*}
    S_{V(k) \Omega, V(k) \Phi} = V(k) S_{\Omega, \Phi} V(k)^* \Rightarrow \Delta_{V(k) \Omega, V(k) \Phi} = V(k)\Delta_{\Omega, \Phi} V(k)^*.
\end{equation*}
But since $V(k)^* V(h) = V(h-k)$, we have that the state functional associated to $\ket{\Phi}$ is $\omega_{h-k}$. Therefore, we find:
\begin{align*}
    S(\omega_k\|\omega_h) &= -\braket{V(h) \Omega}{\log \Delta_{V(k) \Omega,V(h) \Omega}V(h) \Omega}\\
    &= -\braket{V(h) \Omega}{V(k) \log \Delta_{\Omega,\Phi}V(k)^* V(h) \Omega}\\
    &= -\braket{V(k)^* V(h) \Omega}{ \log \Delta_{\Omega,\Phi} V(k)^* V(h) \Omega}\\
    &= -\braket{V(h-k) \Omega}{ \log \Delta_{\Omega,V(h-k) \Omega} V(h-k) \Omega} = S(\omega_{h-k}\|\omega).
\end{align*}

\end{proof}
Using the above results, we can rewrite the general relative entropy:
\begin{equation*}
    S(\omega_h\| \omega) = -\braket{\Omega}{\log \Delta_{V(h) \Omega, \Omega}\Omega}
\end{equation*}
in terms of the modular operator $\Delta_{\Omega}$ only, noticing that $V(h) \in R(H)$ is unitary for $h \in H$:
\begin{align*}
    \Delta_{V(h)\Omega, \Omega} &= V(h) \Delta_{V^*(h) V(h) \Omega, V^*(h)\Omega} V^*(h)\\
    &= V(h) \Delta_{\Omega, V^*(h)\Omega} V^*(h)\\
    &= V(h) \Delta_{\Omega} V^*(h), 
\end{align*}
where in the first step we have used Eq. \eqref{eq: unit1} and Prop. \ref{prop: ciclico} in the third. Therefore, it follows that:
\begin{equation}\label{eq: ecb}
    S(\omega_h\| \omega) = -\braket{V^*(h)\Omega}{\log \Delta_{\Omega}V^*(h)\Omega}.
\end{equation}
We are able to prove the main theorem:
\begin{thm}
Let $h \in H$. If $h \in \mathrm{Dom}(\log \Delta_H)$, the relative entropy on $R(H)$ between $\omega$ and $\omega_h$ is given by:
\begin{equation*}
    S(\omega_h\| \omega) = S_h = -(h, \log \Delta_H h).
\end{equation*}
\end{thm}
\begin{proof}
The first thing to notice is that, from Eq. \eqref{eq: ecb}, we have for $t \in \mathbb{R}$:
\begin{align*}
    S(\omega_h\| \omega) &= \braket{V^*(h)\Omega}{K_{\Omega}V^*(h)\Omega}\\
    &= i\frac{d}{dt}\bigg|_{t=0} \braket{V^*(h)\Omega}{e^{iK_{\Omega}t}V^*(h)\Omega}\\
    &= i\frac{d}{dt}\bigg|_{t=0} \braket{V^*(h)\Omega}{\Delta^{it}_{\Omega}V^*(h)\Omega},
\end{align*}
where I have introduced the \textit{modular Hamiltonian}, defined as:
\begin{equation*}
    K_{\Omega} = - \log \Delta_{\Omega},
\end{equation*}
since $\Delta_{\Omega} = \Delta_{R(H)}$ is the modular operator of $R(H)$ with respect to the state $\omega_{\Omega}$ defined on it. By the second quantized equivalent of the modular flow on the one particle Hilbert space, we have: $\Delta_{R(H)}^{it} = \mathfrak{S}(\Delta_H^{it})$. Therefore:
\begin{align*}
    S(\omega_h\| \omega) &= i\frac{d}{dt}\bigg|_{t=0} \braket{V^*(h)\Omega}{\mathfrak{S}(\Delta^{it}_{H})V^*(h)\Omega}\\
    &= ie^{-\|h\|^2}\frac{d}{dt}\bigg|_{t=0}e^{(h,\Delta_H^{it}h)}\\
    &= i\frac{d}{dt}\bigg|_{t=0}(h,\Delta_H^{it}h)\\
    &= -(h, \log \Delta_H h).
\end{align*}
\end{proof}
We note as well that, since $S(\omega_h|| \omega)$ is real, we have:
\begin{equation*}
    S(\omega_h|| \omega) = -(h, \log \Delta_H h) = i(h,i \log \Delta_H h) = \Im(h,i\log \Delta_H h),
\end{equation*}
where $\Im(\cdot,\cdot)$ defines a symplectic form. We have proven that the relative entropy for coherent excitations for $h \in H$ is equivalent to the entropy of the vector in the real and standard subset of the one particle Hilbert space. We now seek to generalize this to the case in which the vector, with respect to which we take the coherent excitations, is a general vector in the one-particle Hilbert space.

\subsubsection{Case $h \in \mathcal{H}$}
Let us now consider the general case in which $k \in \mathcal{H}$ and consider $\omega_k$ and $\omega_k'$ to be states on $R(H)$ and $R(H)'$, respectively.  
\begin{lem}
Let $k \in \mathcal{H}_{\epsilon}$. Then $S(\omega_k \| \omega) = S_k$.
\end{lem}
\begin{proof}
We saw that a general element in $\mathcal{H}_{\epsilon}$ can be written as $k = h + h'$, where $h \in E_{\epsilon}H$, $h' \in E_{\epsilon} H'$ and where $E_{\epsilon}$ is the spectral projection of $\Delta_H$. So, starting from the general definition of relative entropy of $k \in \mathcal{H}_{\epsilon}$ with respect to $H$, one has:
\begin{align*}
    S_k &= \Im\big( (h+h', P_H i \log \Delta_H (h+h')) \big)\\
    &= \Im\big( (h+h', i \log \Delta_H h) \big)\\
    &= \Im\big((h, i \log \Delta_H h)\big) = S_h,
\end{align*}
where we have used in the second step the fact that $P_H$ commutes with $\Delta_H$ and in the third the fact that $\Delta_H$ is an automorphism of $H$.\\
However, on the other hand, we have that $V(k) = V(h) V(h')$ which gives:
\begin{align*}
    S(\omega_k \| \omega) &= -  \braket{V(h) V(h') \Omega}{\log \Delta_{V(k) \Omega, \Omega} V(h) V(h') \Omega} \\
    &= -  \braket{V(h) V(h') \Omega}{ V(h')\log \Delta_{\Omega} V^{*}(h') V(h) V(h') \Omega} \\
    &= -  \braket{V(h) \Omega}{\log \Delta_{\Omega} V(h) \Omega},
\end{align*}
where we used $R(H') = R(H)'$ and again Prop. \ref{prop: ciclico} to write:
\begin{equation*}
    \Delta_{V(h) V(h') \Omega, \Omega} = V(h') \Delta_{\Omega} V(h')^*.
\end{equation*}
Consequently, we have shown that $S_k = S_h$ and $S(\omega_k \| \omega) = S(\omega_h \| \omega)$. Therefore, from the results of the preceding section, we have that:
\begin{equation*}
    S(\omega_k \| \omega) = S_k. 
\end{equation*}
\end{proof}
Before presenting the main result, we need another technical lemma:
\begin{lem}
If $\mathcal{H} = \mathcal{H}_1 \oplus \mathcal{H}_2$ with $H = H_1 \oplus H_2$, and $k = k_1 \oplus k_2$, then:
\begin{equation*}
    S(\omega_k\| \omega) = S(\omega_{k_1}\| \omega_{k_2}),
\end{equation*}
where $\omega_{k_i}$ is the coherent state associated with $k_i$ on $\mathfrak{S}(\mathcal{H}_i)$.
\end{lem}
\begin{proof}
First, let us prove that, by passing to second quantization, $\mathfrak{S}(\mathcal{H}) = \mathfrak{S}(\mathcal{H}_1) \otimes \mathfrak{S}(\mathcal{H}_2)$. Let us look at the finite particle level:
\begin{align*}
    \mathcal{H} \otimes \mathcal{H} &= (\mathcal{H}_1 \oplus \mathcal{H}_2) \otimes (\mathcal{H}_1 \oplus \mathcal{H}_2)\\
    &= (\mathcal{H}_1 \otimes \mathcal{H}_1) \oplus (\mathcal{H}_1 \otimes \mathcal{H}_2) \oplus (\mathcal{H}_2 \otimes \mathcal{H}_1) \oplus (\mathcal{H}_2 \otimes \mathcal{H}_2),
\end{align*}
where we notice that each of the factors of the direct sum is orthogonal to the others. Let us now take the symmetrized tensor product:
\begin{equation*}
    (\mathcal{H} \otimes \mathcal{H})_s = (\mathcal{H}_1 \otimes \mathcal{H}_1) \oplus (\mathcal{H}_1 \otimes \mathcal{H}_2) \oplus (\mathcal{H}_2 \otimes \mathcal{H}_2),
\end{equation*}
from which follows that at every order $n$, we have:
\begin{align*}
    \mathfrak{S}(\mathcal{H}) &= \bigoplus_{n=0}^{\infty} (\mathcal{H}_1 \oplus \mathcal{H}_2)^{\otimes n}_s\\
    &= \bigg( \bigoplus_{n=0}^{\infty} (\mathcal{H}_1)^{\otimes n}_s \bigg) \otimes \bigg( \bigoplus_{n=0}^{\infty} (\mathcal{H}_2)^{\otimes n}_s \bigg) = \mathfrak{S}(\mathcal{H}_1) \otimes \mathfrak{S}(\mathcal{H}_2).
\end{align*}
For what concerns the states, we split them as $\omega = \omega \otimes \omega$ and $\omega_k = \omega_{k_1} \otimes \omega_{k_2}$. Now, using the additivity of relative entropy under tensor products (see Eq. $(5.22)$ in \cite{Petz}), we have the claimed result.
\end{proof}
We can now state and prove the main theorem:
\begin{thm}
Let $k \in \mathcal{H}$, then:
\begin{equation*}
    S(\omega_k\| \omega) = S_k = \Im \big( (k, i P_H \log \Delta_H k) \big)
\end{equation*}
\end{thm}
\begin{proof}
Take $k \in \mathcal{H}$ and decompose $ \mathcal{H} = \mathcal{H}_{\epsilon} \oplus \mathcal{K}_{\epsilon}$. Correspondingly, we have decompositions for the $J_H,\Delta_H, P_H$ operators. Then, the above lemma and the positivity of relative entropy give:
\begin{equation*}
    S(\omega_k\| \omega) = S(\omega_{k_{\epsilon}}\|\omega) + S(\omega_{(1-E_{\epsilon})k}\|\omega) \geq S(\omega_{k_{\epsilon}}\|\omega).
\end{equation*}
Now, in the limit $\epsilon \to 0^+$, we have that $k_{\epsilon} \to k$ and, by lower semicontinuity of relative entropy, we have:
\begin{equation*}
    \liminf_{\epsilon \to 0^+} S(\omega_{k_{\epsilon}}\| \omega) \geq S(\omega_k \| \omega).
\end{equation*}
Thus, taking both inequalities into account, we find:
\begin{equation*}
    \lim_{\epsilon \to 0^+}S(\omega_{k_{\epsilon}}, \omega) = S(\omega_k\| \omega).
\end{equation*}
But since $S(k_{\epsilon}\| k) = S_{k_{\epsilon}}$, we must have:
\begin{equation*}
   S(\omega_{k}\|\omega) = \lim_{\epsilon \to 0^+}S_{k_{\epsilon}} = S_k, 
\end{equation*}
where, the last equality, follows from Prop. \ref{prop: S_h}.
\end{proof}
The convenience of working with coherent excitations of the vacuum is manifest : for such, the relative entropy between two different configurations of the underlying QFT is computed, using first quantization techniques only, namely just in terms of the inner product of vectors in the one-particle Hilbert space.\\
To establish some analogy with the upcoming results in the next chapter, let us reformulate the final result. In the case of a free scalar Quantum Field Theory we have, as abstract algebra $\mathfrak{A}_{(\mathbf{P},\sigma)}$, where $(\mathbf{P},\sigma)$ is the symplectic space of solutions of the classical field equation. Then, by choosing a quasi-free state on $\mathfrak{A}_{(\mathbf{P},\sigma)}$, we quoted in Theorem \ref{thm: quasi} that one can obtain a Fock space representation. In particular, we mentioned that the inner product over the one-particle Hilbert space $\mathcal{H}$ is related to the symplectic structure by:
\begin{equation*}
    \sigma(f,g) = 2\Im(\braket{f}{g}_{\mathcal{H}})
\end{equation*},
for $f,g \in \mathcal{H}$. Therefore, the above result for the relative entropy for a coherent excitation of the vacuum, can be rewritten for $\mathcal{H} \ni f = f_H + f_H'$as:
\begin{align*}
    S(\omega_f\| \omega) &= \Im \big( (f, i P_H \log \Delta_H f) \big)\\
    &= -(f_H, \log \Delta_H f_H)\\
    &= i\frac{d}{dt}\bigg|_{t=0} \braket{\Omega}{V(f_H) V^*(\Delta_H^{it} f_H)}\\
    &= i\frac{d}{dt}\bigg|_{t=0} \braket{\Omega}{e^{i\Im\big( (f_H, \Delta_H^{it}f_H) \big) }V(f_H - \Delta^{it}_H f_H)}\\
    &= - \frac{d}{dt}\bigg|_{t=0}\Im \big( (f_H, \Delta_H^{it}f_H) \big)\\
    &= -\frac{1}{2}\frac{d}{dt}\bigg|_{t=0} \sigma(f_H, \Delta_H^{it}f_H)\\
    &= \frac{1}{2}\frac{d}{dt}\bigg|_{t=0} \sigma(\Delta_H^{it}f_H, f_H). 
\end{align*}
Finally, if we want to compute the relative entropy in a concrete situation, we need to explicitly determine the $t$-dependence of the last expression. In fact, in absence of a time symmetry on the underlying spacetime on which the test functions are defined, this is in general unknown. However, in the specific case in which the fields are assumed to be localized in wedge-like regions, the Bisognano-Wichmann theorem (see Appendix \ref{app: BW}) provides a geometric action of the modular automorphism with the corresponding explicit form of the $t$-dependence.

\newpage
\chapter{Relative entropy for fermionic fields}
In the last chapter, we have presented how the relative entropy for a free scalar QFT, between a quasifree state and a coherent excitation of it is computed, using Tomita-Takesaki modular theory and the Araki formula. As already mentioned, the final explicit formula, allows to calculate relative entropies in contexts like semiclassical gravity. The calculations are simpler, as a quantity in second quantization, Araki's entropy is defined in terms of an operator acting on a vector in the Fock space, is computed just at the one particle level. At the same time, however, the result holds just for the specific type of coherent excitations and its generalization remains an open problem.\\
Another open task, is the search for similar results for a Dirac/Majorana Quantum Field Theory. As far as we know, a completely satisfactory result, in the same spirit as that for the free scalar, seems to be missing in literature.\\
In this chapter, that corresponds to the core of this thesis, we derive an explicit formula for the relative entropy for a Fermionic QFT. The result, allows the computation of the relative entropy between a quasifree state and a specific type of unitary excitation of it, in terms of the inner product of two vectors in the Hilbert space on which the self-dual CAR algebra is defined. Later on, we attempt at giving a first generalization of this formula for different types of excitations. We conclude the chapter presenting concrete examples in finite and infinite dimension, in which the local algebras are Type $I$ factors, where we will show the equivalence between the relative entropy computed using the von Neumann formula and the one obtained starting from out result.\\

\section{Single Unitary Fermionic Excitation of the Vacuum}\label{sec: risultati}
Following the work of Araki \cite{Araki1968}, and what we have discussed in Section \ref{sec: SDCAR}, we start presenting the result in the most general case and just later specify to a fermionic QFT.\\
Therefore, the starting object is a self-dual CAR algebra $\mathfrak{A}_{SDC}(\mathcal{H}, \Gamma)$, that from what we discussed in Section \ref{sec: SDCAR} we know is a $C^*$-algebra, where $\mathcal{H}$ is an Hilbert space, finite or infinite dimensional, on which we have an involution $\Gamma$ with properties reported in Def. \ref{def: self}. Moreover, we assume that we have a (strongly continuous) one parameter group of unitaries $V_t$ describing a dynamics over $\mathcal{H}$ with $t \in \mathbb{R}$. By Stone's theorem, this one parameter family, is related to a self-adjoint, positve operator $\mathbf{h}$:
\begin{equation*}
    V_t = e^{-i t \mathbf{h}}.
\end{equation*}
To be precise, if we want this to define a dynamics (think of the Heisenberg picture), we need to make sure that it lifts to an automorphism over $\mathfrak{A}_{SDC}(\mathcal{H},\Gamma)$. In particular, as $\mathfrak{A}_{SDC}(\mathcal{H},\Gamma)$ is a $C^*$-algebra, we need to demand it to be a $*$-automorphism. Therefore, if we define it as:
\begin{equation*}
    \alpha_t (B(g)) := B(V_{-t} g) \hspace{20pt} \forall g \in \mathcal{H}
\end{equation*}
we see that:
\begin{align*}
    \alpha_t(B^*(g)) &= \alpha_t(B(\Gamma g)) = B(V_{-t} \Gamma g)\\
    \alpha_t(B(g))^* &= B^*(V_{-t}g) = B(\Gamma V_{-t}g).
\end{align*}
and, in order $V_t$ to define a dynamics, we need for each $t \in \mathbb{R}$:
\begin{equation*}
    [V_t, \Gamma] = 0.
\end{equation*}\\
Now that we have fixed the assumptions, we go back to the elements in the algebra and identify a specific class of them, that squares to the identity and is invariant under the $*$-operation. To obtain such a class, we start considering those elements of the Hilbert space $f \in \mathcal{H}$ for which:
\begin{equation*}
    \Gamma f = f
\end{equation*}
Such elements always exist, as we can always construct one starting from a random element in $\mathcal{H}$:
\begin{align*}
    f &:= (1 + \Gamma) h \hspace{20pt} \forall h \in \mathcal{H}\\
    &:= i(1 - \Gamma) h \hspace{20pt} \forall h \in \mathcal{H}
\end{align*}
If we now consider the corresponding elements, in the self-dual CAR algebra, associated to these "involution invariant" vectors of the Hilbert space, we get elements that are invariant under the abstract $*$-operation: $B^*(f) = B(f)$. Then, from the anticommutation relations:
\begin{equation*}
    B(f) B(f) = \frac{(f,f)}{2} \mathbb{1}
\end{equation*}
That, by a proper chooice of the function $h \in \mathcal{H}$ in the definition of $f$, in order to have $(f,f) = 2$, gives the idempotence:
\begin{equation*}
    B(f) B^*(f) = B^*(f) B(f) = B(f) B(f) = \mathbb{1}
\end{equation*}

\begin{rem} \label{rem: 3.1.1.}
If we take two different $f,g \in \mathcal{H}$ such that $\Gamma f = f$ (same for $g$), we have that for the corresponding $B(f),B(g) \in \mathfrak{A}_{SDC}(\mathcal{H}, \Gamma)$:
\begin{equation*}
    [B(f), B(g)]_+ = (f,g)_{\mathcal{H}} \mathbb{1}
\end{equation*}
but at the same time:
\begin{equation*}
    [B(f), B(g)]_+ = [B(g), B(f)]_+ = (g,f)_{\mathcal{H}} \mathbb{1}
\end{equation*}
But this implies:
\begin{equation*}
    (f,g)_{\mathcal{H}} = (g,f)_{\mathcal{H}} = \overline{(f,g)_{\mathcal{H}}} 
\end{equation*}
That gives $(f,g)_{\mathcal{H}} \in \mathbb{R}$.
\end{rem}

These $B(f)$, for $f = \Gamma f$, are the elements of the algebra that we will consider to excite a quasifree state and thus to compute the relative entropy.\\
Now, in order to be able to define the Tomita operator and a notion of relative entropy between states using Araki's formula, we need to represent the abstract algebra as a von Neumann algebra over a Hilbert space. For this purpose, we start by presenting the two following results (See \cite{Araki:1971id} Lemma $3.2$ and $3.3$):
\begin{lem}\label{Lemma 3.2}
For any state $\varphi$ over $\mathfrak{A}_{SDC}(\mathcal{H}, \Gamma)$, there exists a bounded operator $S$ (called basis polarization) on $\mathcal{H}$, satisfying:
\begin{align}
    \varphi(B^*(f) B(g)) &= (f, S g)_{\mathcal{H}} \label{eq: 2point}\\
    \mathbb{1} \geq S^* &= S \geq 0 \label{eq: 2pointt}\\ 
    S + \Gamma S \Gamma &= \mathbb{1}\label{eq: 2point2}
\end{align}
\end{lem}
\begin{proof}
Being $\mathfrak{A}_{SDC}(\mathcal{H}, \Gamma)$ a $C^*$-algebra, we have:
\begin{equation*}
    B^*(f) B(f) \leq B^*(f) B(f) + B(f) B^*(f) = \| f\|^2
\end{equation*}
It follows:
\begin{equation*}
    \| B(f) \| = \| B^*(f) B(f) \|^{1/2} \leq \|f\|
\end{equation*}
Hence, the relation \eqref{eq: 2point} defines an operator $S$ that is bounded and linear.\\
The positivity of $\varphi$ gives:
\begin{equation*}
    (f,Sf)_{\mathcal{H}} = (S^* f, f)_{\mathcal{H}} \geq 0
\end{equation*}
but the inequality implies $(S^* f, f)_{\mathcal{H}} = (f, S^* f)_{\mathcal{H}}$. Follows that $S^* = S \geq 0$.\\
Finally, from the anticommutation relations over the self-dual CAR algebra, we have:
\begin{align*}
    \varphi(B^*(f) B(g)) &= (f,g)_{\mathcal{H}} - \varphi(B(g) B^*(f))\\
    &= (f,g)_{\mathcal{H}} - \varphi(B^*(\Gamma g) B(\Gamma f))\\
    &= (f,g)_{\mathcal{H}} - (\Gamma g, S \Gamma f)_{\mathcal{H}}
\end{align*}
From the definition of the involution $\Gamma$, we have $(h, \Gamma f)_{\mathcal{H}} = (f, \Gamma h)_{\mathcal{H}}$ that gives:
\begin{equation}\label{eq: con1}
    (\Gamma g, S \Gamma f)_{\mathcal{H}} = (S \Gamma g, \Gamma f)_{\mathcal{H}} = (f, \Gamma S \Gamma g)_{\mathcal{H}}
\end{equation}
Hence:
\begin{equation*}
    (f, S g)_{\mathcal{H}} = (f,g)_{\mathcal{H}} - (f, \Gamma S \Gamma g)_{\mathcal{H}}
\end{equation*}
From which it follows Eq. \eqref{eq: 2point2}. Morever, combining the positivity of $S$ with Eq. \eqref{eq: con1}, we get $\Gamma S \Gamma = \mathbb{1} - S \geq 0$.
\end{proof}

\begin{lem}
For any $S$ as above, there exist a unique quasifree state satisfying Eq. \eqref{eq: 2point}.
\end{lem}
\begin{proof}
As a quasifree state is uniquely determined by its two point function and $S$ determines it, follows the uniqueness. The existence follows by Lemma $4.6$ in \cite{Araki:1971id}
\end{proof}
We denote such quasifree states by $\varphi_S$. But, any quasifree state over the self-dual algebra has, by Lemma \ref{Lemma 3.2}, an associated $S$ with the properties Eq. \eqref{eq: 2pointt} and \eqref{eq: 2point2} that determine it uniquely. As a consequence, any quasifree state over the self-dual CAR algebra is of the form $\varphi_S$.\\
In order to formulate Tomita-Takesaki modular theory, we need the state over the self-dual CAR algebra to be faithful. Moreover, for later purposes, we want it also to be quasifree. Therefore, we investigate the existence and the properites of quasifree, faithful states over a self-dual CAR algebra. The faithfulness condition lead us to closely look at Eq. \eqref{eq: 2point}, our state must be:
\begin{equation*}
    \varphi_S(B^*(f) B(f)) = 0 \Longleftrightarrow B(f) = 0
\end{equation*}
However, by the following lemma (Lemma $4.3.$ in \cite{Araki:1971id}) and by Theorem \ref{thm: quasi}, we see that if a quasifree state has an associated basis polarization $S$ that is also a projection, i.e. $S^2 = S$ and thus called basis projection, then it cannot be faithful:
\begin{lem}\label{lem: FockFer}
Let $P$ be a basis projection. If a state $\varphi$ of $\mathfrak{A}_{SDC}(\mathcal{H}, \Gamma)$ satisfies:
\begin{equation*}
    \varphi(B(f) B^*(f)) = 0 \hspace{20pt} \forall f \in P \mathcal{H}
\end{equation*}
Then $\varphi$ is a quasifree state with $S=P$ i.e. $\varphi = \varphi_P$. The representation $\pi_P$ is irreducible.
\end{lem}
Therefore, we need to work with states that do not lead directly to irreducible representations, when we perform the associated \textit{GNS construction}, if we want the \textit{GNS vector} to be separable.\\
Anyway, for each quasifree state, the corresponding GNS construction leads to a Fock space. In fact, for each $\varphi_S$, we have:
\begin{equation*}
    \varphi_S(B^*(f) B(g)) = (f, Sg)_{\mathcal{H}} = (S^{1/2} f, S^{1/2}g)_{\mathcal{H}}
\end{equation*}
We define the one-particle Hilbert space as (see \cite{DAntoni:2001ido} Section II.$2$):
\begin{equation*}
    \mathfrak{h}_{S} = \{ S^{1/2} f | f \in \mathcal{H} \}
\end{equation*}
and correspondingly, by taking the direct sum of the antisymmetrized tensor product of it, the Fermionic Fock space:
\begin{equation*}
    \mathcal{K}_S := \mathbb{C} \oplus \bigg( \bigoplus_{n = 1}^{\infty} (\mathfrak{h}_S)^{\wedge n} \bigg)
\end{equation*}
With $\wedge$ denoting the antisymmetric tensor product. The corresponding vacuum vector $\Omega_S$ is the vector spanning the $\mathbb{C}$ component. The representation is:
\begin{equation*}
    \pi_S(B(f)) = a^*(S^{1/2}f) + a(S^{1/2}\Gamma f)
\end{equation*}
Where $a^*(p)$ and $a(p)$ are creation and annihilation operators over the Fock space $\mathcal{K}_S$. To check that this defines a representation, one can check that this map defines a $*$-isomorphism of $\mathfrak{A}_{SDC}(\mathcal{H}, \Gamma)$ into the algebra of creation and annihilation operators over $\mathfrak{h}_S$. Furthermore, if $S$ happens to be a projection, we can replace, in the latter construction of the Fock space, $S^{1/2}$ with $S$. In particular, given $\mathfrak{A}_{SDC}(\mathcal{H}, \Gamma)$, a representation with respect to a basis projection always exists (Lemma $3.3.$ in \cite{Araki1968}):

\begin{lem}
If $\dim \mathcal{H} = \mathrm{even}$ or $+\infty$, there exists a basis projection $P$, that generates a $*$-representation $\pi_P$ to the Fock space $\mathcal{K}_P$:
\begin{align*}
    \pi_P (B^*(f)) &= a^*(P f) + a(P \Gamma f)\\
    \pi_P (B(f)) &= a(P f) + a^*(P \Gamma f)\\
    \pi_P(\mathbb{1}) &= \mathbb{1}
\end{align*}
\end{lem}
\begin{proof}
From what we said above, we just need to prove existence. Consider a $\Gamma$ invariant basis of $\mathcal{H}$ denoted $\{ f_i \}_{i = 1, \dots}$. Since $\mathcal{H}$ is even dimensional, we can pair $f_{2n}$ with $f_{2n-1}$ for $n = 1,2, \dots$ to define $P$ as the projector onto the subspace spanned by:
\begin{equation*}
    \frac{1}{2^{1/2}}(f_{2n} + if_{2n-1})
\end{equation*}
In this way $\Gamma P\Gamma = \mathbb{1} - P$. The existence of a $\Gamma$ invariant basis for $\mathcal{H}$ follows from the fact that, for a general $h \in \mathcal{H}$ non zero, either $(1 + \Gamma) h$ or $i(1 - \Gamma)h$ are non zero. Therefore, starting from an arbitrary non-zero element of $\mathcal{H}$ we form the corresponding $\Gamma$ invariant vector and then perform a Gram-Schmidt procedure, as the orthogonal of the space spanned by the $\Gamma$ invariant vectors in $\mathcal{H}$, is itself $\Gamma$-invariant.
\end{proof}

We will need the existence of such a basis projection, for any $\mathfrak{A}_{SDC}(\mathcal{H}, \Gamma)$ later on.\\\\
However, for the current purposes, let $\omega$ be just a faithful, quasifree state over the self-dual CAR algebra, with associated basis polarization $S$.
Performing the corresponding \textit{GNS construction}, we get the triple $(\mathcal{K}_{\omega}, \pi_{\omega}, \Omega_{\omega})$ of a Fock representation.\\
Let us call $\mathcal{A} = \pi_{\omega}(\mathfrak{A}_{SDC}(\mathcal{H}, \Gamma))''$ the associated von Neumann algebra. Now that we have a von Neumann algebra, we can define a relative Tomita operator between two cyclic and separating states. For that purpose, we start by easing the notation, denoting from now on $F := \pi_{\omega}(B(f))$, for $f \in \mathcal{H}$ such that $\Gamma f = f$. Then, by the properties of the corresponding $B(f)$ and the $*$-representation $\pi_{\omega}$, the operator $F$ must be unitary. As a consequence, the vector state $F \ket{\Omega_{\omega}}$, with associated vector functional $\omega_F$ over $\mathcal{A}$, must be cyclic and separating as well:

\begin{lem}
Let $\ket{\Omega_{\omega}}$ be a cyclic and separating vector for the algebra $\mathcal{A}$. Then, for $F \in \mathcal{A}$ unitary, the vector $F \ket{\Omega_{\omega}}$ is cyclic and separating.
\end{lem}
\begin{proof}
As $F \in \mathcal{A}$ is unitary we have $\mathcal{A} F = \mathcal{A}$. Then, using the fact that $\Omega_{\omega}$ is cyclic:
\begin{align*}
    \overline{\mathcal{A} F \ket{\Omega_{\omega}}} &= \overline{\mathcal{A}\ket{\Omega_{\omega}}}\\
    &= \mathcal{K}_{\omega}
\end{align*}
Since $F \in \mathcal{A}$.\\
For what concerns the separating property, consider:
\begin{equation*}
    A F \ket{\Omega_{\omega}} = 0 \Longrightarrow F^* A F \ket{\Omega_{\omega}} = 0
\end{equation*}
But, as $\ket{\Omega_{\omega}}$ is separating, we have $F^* A F = 0$. But $F^*A F \in \mathcal{A}$ as well, so by repeating the same argument: $FF^* A F F^* = 0$ that gives $A = 0$.
\end{proof}

Now that we have two cyclic and separating vectors $\ket{\Omega_{\omega}}$ and $F \ket{\Omega_{\omega}}$, for the von Neumann algebra $\mathcal{A}$, we can compute the Araki's relative entropy:
\begin{equation*}
    S(\omega_{F}||\omega) = - \braket{\Omega_{\omega}}{\log \Delta_{F \Omega_{\omega}, \Omega_{\omega}}\Omega_{\omega}}
\end{equation*}
Where $\Delta_{F \Omega_{\omega}, \Omega_{\omega}}$ is the modular operator for the relative Tomita operator $S_{F \Omega_{\omega}, \Omega_{\omega}}$.\\
To simplify this expression, keeping in mind that $F \in \mathcal{A}$ is unitary, we use the result of Proposition \ref{prop: ciclico} and Eq. \eqref{eq: unit1} to express the relative modular operator, in terms of the modular operator of the vector state $\omega$:
\begin{align*}
    \Delta_{F \Omega_{\omega}, \Omega_{\omega}} &= F \Delta_{F^* F \Omega_{\omega}, F^* \Omega_{\omega}} F^*\\
    &= F \Delta_{\Omega_{\omega}, F^* \Omega_{\omega}} F^*\\
    &= F \Delta_{\Omega_{\omega}} F^*
\end{align*}
Therefore, the relative entropy becomes:
\begin{equation*}
    S(\omega_F || \omega) = - \braket{\Omega_{\omega}}{\log \big( F \Delta_{\Omega_{\omega} } F^*\big)\Omega_{\omega}} = - \braket{\Omega_{\omega}}{F \log \big( \Delta_{\Omega_{\omega}} \big) F^* \Omega_{\omega}}
\end{equation*}
Where the second equality is a consequence of the unitarity of $F$.\\

Now, as we later want the modular automorphism to act on the vector $f \in \mathcal{H}$, we present the following lemma (Lemma $4.2.$ together with Theorem $3$ in \cite{Araki:1971id}):
\begin{lem}
Let $\omega$ be a faithful quasifree state with basis polarization $S$, over $\mathfrak{A}_{SDC}(\mathcal{H}, \Gamma)$. If $S$ commutes with the one parameter family of unitaries $V_t$ on $\mathcal{H}$, then we can choose the modular flow such that for all $B(f) \in \mathfrak{A}_{SDC}(\mathcal{H}, \Gamma)$:
\begin{equation*}
    \Delta_{\Omega_{\omega}}^{it} \pi_{\omega}(B(f)) \Delta_{\Omega_{\omega}}^{-it} = \pi_{\omega}(B(V_t f))
\end{equation*}
\end{lem}
\begin{proof}
If $[V_t, S] = 0$, then the state $\omega$ is in particular stationary. Hence, we can define:
\begin{equation*}
    U_t \sum_i c_i \pi_{\omega}(A_i) \Omega_{\omega} = \sum_i c_i \pi_{\omega} (\alpha_t(A_i)) \Omega_{\omega}  
\end{equation*}
and:
\begin{equation*}
    U_t^* \sum_i c_i \pi_{\omega}(A_i) \Omega_{\omega} = \sum_i c_i \pi_{\omega} (\alpha_{-t}(A_i)) \Omega_{\omega}
\end{equation*}
From the cyclicity of the GNS vector, these are isometric maps from a dense subset of $\mathcal{K}_{\omega}$ into $\mathcal{K}_{\omega}$ satisfying:
\begin{align*}
    U_t U^*_t &= U_t^* U_t \subset \mathbb{1}\\
    U_t &\subset (U_t^*)^*
\end{align*}
Where as usual $A \subset B$ if $A = B$ on $\mathcal{D}(A)$, the domain of the operator $A$, and $\mathcal{D}(A) \subset \mathcal{D}(B)$. Follows, that the closure of $U_t$ is unitary and satisfies $U_t \Omega_{\omega} = \Omega_{\omega}$ together with:
\begin{equation*}
    U_t \pi_{\omega}(A) U^*_t = \pi_{\omega}(\alpha_t(A))
\end{equation*}
Now, following the steps in the proof of Theorem $3$ in \cite{Araki:1971id}, one shows the existence of an antiunitary operator $J$ over $\pi_{\omega}(\mathfrak{A}_{SDC}(\mathcal{H}, \Gamma))''$ such that for any $A \in \pi_{\omega}(\mathfrak{A}_{SDC}(\mathcal{H}, \Gamma))''$:
\begin{align*}
    J U_{-i\beta/2} A \Omega_{\omega} &= A^* \Omega_{\omega}\\
    J \Omega_{\omega} &= \Omega_{\omega}\\
    J \pi_{\omega}(\mathfrak{A}_{SDC}(\mathcal{H}, \Gamma))'' J &= (\pi_{\omega}(\mathfrak{A}_{SDC}(\mathcal{H}, \Gamma))'')'\\
    [J, U_t] &= 0
\end{align*}
But then, from the faithfulness assumption, this defines the Tomita operator with respect to $\Omega_{\omega}$ and, by the uniqeuness of the polar decomposition, follows the statement.
\end{proof}
Therefore, from now on, we further assume that the state $\omega$ is such that its associated basis polarization commutes with the one parameter family of unitaries $V_t$ on $\mathcal{H}$ or at least that the modular automorphism induces an action on $\mathcal{H}$ (for example if we are in a framework in which we can apply the Bisognano-Wichmann theorem \ref{thm: BW}).\\
Then, we further rewrite the above expression, introducing the \textit{Modular Hamiltonian}:
\begin{equation*}
    K_{\Omega_{\omega}} = - \log \Delta_{\Omega_{\omega}}
\end{equation*}
Then, implicitly using Stone's theorem, we compute:
\begin{align*}
    K_{\Omega_{\omega}} &= i \frac{d}{dt}\bigg|_{t = 0} e^{-i K_{\Omega_{\omega}} t}\\
    &= i \frac{d}{dt}\bigg|_{t=0} e^{i t \log \Delta_{\Omega_{\omega}}}\\
    &= i \frac{d}{dt}\bigg|_{t=0} \Delta_{\Omega_{\omega}}^{it}
\end{align*}
So we have obtained:
\begin{equation*}
    - \log \Delta_{\Omega_{\omega}} = i \frac{d}{dt}\bigg|_{t=0} \Delta_{\Omega_{\omega}}^{it}
\end{equation*}
that gives for the relative entropy:
\begin{align*}
    S(\omega_F || \omega) &= i \frac{d}{dt}\bigg|_{t = 0} \braket{\Omega_{\omega}}{F \Delta_{\Omega_{\omega}}^{it} F^* \Delta_{\Omega_{\omega}}^{-it} \Delta_{\Omega_{\omega}}^{it} \Omega_{\omega}}\\
    &= i \frac{d}{dt}\bigg|_{t = 0} \braket{\Omega_{\omega}}{F \Delta_{\Omega_{\omega}}^{it} F^* \Delta_{\Omega_{\omega}}^{-it} \Omega_{\omega}}
\end{align*}
Where we have used $\Delta_{\Omega_{\omega}}^{it} \Omega_{\omega} = \Omega_{\omega}$.\\

Finally, we let the modular automorphism act on the elements in $\mathcal{A}$:
\begin{equation*}
    \Delta_{\Omega_{\omega}}^{it} \pi_{\omega}(B(f)) \Delta_{\Omega_{\omega}}^{-it} = \pi_{\omega}(\alpha_{-t}(B(f))) = \pi_{\omega}(B(f_t)) =: F_t
\end{equation*}
and as a consequence, we get for the relative entropy:
\begin{equation*}
    S(\omega_F || \omega) = i \frac{d}{dt}\bigg|_{t = 0} \braket{\Omega_{\omega}}{F F_t \Omega_{\omega}} = i \frac{d}{dt}\bigg|_{t = 0} (f, S f_t)_{\mathcal{H}}
\end{equation*}
Therefore, we have proven the following main result:

\begin{prop}[\textbf{Araki's formula for a Single Unitary Fermionic Excitation of the Vacuum}]\label{prop: 1}
Let $\mathcal{A}$ be the corresponding von Neumann algebra of the abstract self-dual CAR algebra obtained via the quasifree vector $\omega$ in the way described above. Then, denoting by $F = \pi_{\omega}(B(f))$ for $f \in \mathcal{H}$ such that $\Gamma f = f$, the relative entropy between the state $\omega$ and the one obtained under an excitation by $F$, called $\omega_F$, becomes:
    \begin{equation}\label{entropy}
        S(\omega_{F}\,||\,\omega) = i \frac{d}{dt} \bigg|_{t=0} (f, S f_t)_{\mathcal{H}}
    \end{equation}
\end{prop}

\begin{rem}
Our result was derived for general self-dual CAR algebras, but we can of course specialize it to the case of Dirac or Majorana fields since, as we have outlined in Section \ref{sec: quantisation}, the field algebras in that case are self-dual CAR algebras with $\mathcal{H}$ representing the space of solutions of the Dirac equation on the globally hyperbolic spacetime $M$.
\end{rem}
\begin{rem}
The first thing we may notice, is that our result depends on the inner product of the Hilbert space on which our self-dual CAR algebra is constructed. This is in analogy with the coherent excitation in the bosonic case, where the final result was depending just on the symplectic form. In this sense, both relative entropies are computed just using the canonically defined structure $s$ on the underlying space $(\mathbf{K},s)$ (see comment at the end of Section \ref{sec: SDCAR}). This result (together with the comment at the end of Section \ref{sec: CohFer}), suggests that the analogous of the coherent excitation in the bosonic case is, in the fermionic case, the type of excitation that we are considering. A summary of this, is reported in the following table:\\\\
\begin{tabular}{|l|l|l|l|}
\hline
\rule[-4mm]{0mm}{1cm}
Free scalar field (CCR) & $(\mathbf{P, \sigma(\cdot, \cdot)})$ & Coherent excitation: $\omega_{k}$ & $S(\omega_k \| \omega) \propto \sigma(f_t, f)$\\
\hline
\rule[-4mm]{0mm}{1cm}
Free Dirac/Majorana field (CAR) & $(\mathcal{H}, (\cdot, \cdot)_{\mathcal{H}})$ & Unitary field excitation: $\omega_F$ & $S(\omega_F \| \omega) \propto (f, S f_t)_{\mathcal{H}}$\\
\hline
\end{tabular}
\end{rem}

\section{Multiple Unitary Fermion Excitation of the Vacuum}
In this and in the following section, we aim at generalizing the result obtained in Prop. \ref{prop: 1}, to more general types of excitations. The first generalization is to the case of field polynomials, namely we consider a product of such unitary field operators and we want to find the corresponding expression for the relative entropy for this type of excitation.\\
Let us consider a set of vectors $\{f^{(i)}\}_{i=1,\cdots,N} \in \mathcal{H}$ such that $\Gamma f^{(i)}= f^{(i)}$ and, for simplicity, assuming that the bounded operator associated to $\omega$ is $S > 0$:
\begin{equation*}
    (f^{(i)}, S f^{(j)}) = 0 \hspace{20pt} \mathrm{for} j \neq i
\end{equation*}
Taking the same faithful and quasifree state, with the properties that we discussed above, we get, in the GNS representation, the multiply excited state:

\begin{equation*}
    F_1 \cdot \cdots \cdot F_N \ket{\Omega_{\omega}} \doteq \ket{\Psi_{1,\cdots,N}}
\end{equation*}

where once again:

\begin{equation*}
    F_i \doteq \pi_{\omega}(B(f^{(i)}))
\end{equation*}

 Since $\ket{\Psi_{1,\cdots,N}}$ is again cyclic and separating for $\mathcal{A}$, we can use Araki's formula to determine the relative entropy:

 \begin{equation*}
     \begin{aligned}
         S(\omega_{F_1\cdots F_N}|| \omega) &= i \frac{d}{dt}\bigg|_{t=0}\braket{\Omega_{\omega}}{F_1  \dots   F_N \, F'_N  \cdots   F'_1\, \Omega_{\omega}} = i \frac{d}{dt}\bigg|_{t=0} \omega\big(B(f^{(1)}) \cdots B(f^{(N)}) B(f^{(N)}_t) \cdots B(f^{(1)}_t) \big)
     \end{aligned}
 \end{equation*}

 where $F'_i \doteq \pi_{\omega}[B(f^{(i)}_t)]$.\\
 We use now our assumption of the state $\omega$ to be quasi-free, to write the above $2N$-point function as a sum of products of $N$ 2-point functions. In particular, from our choice of the vectors $\{f^{(i)}\}_{i=1,\cdots,N} \in \mathcal{H}$, the only non vanishing contributions come from permutations that couple a non-primed and a primed index. Precisely, introducing the following convention in order to simplify notation:

 \begin{equation*}
     \begin{aligned}
   \omega_{ij} &= \omega\big(B(f^{(i)}) B(f^{(j)})\big)\\
    \omega_{ij'} &= \omega\big(B(f^{(i)}) B(f^{(j)}_t)\big) = \omega\big(B(f^{(i)}) B(f^{'(j)})\big)\\
     \omega_{i'j'} &= \omega\big(B(f^{(i)}_t) B(f^{(j)}_t)\big) = \omega\big(B(f^{'(i)}) B(f^{'(j)}_t)\big)\\
     \end{aligned}
 \end{equation*}

we have that:

\begin{equation*}
    \omega_{1\cdots N\, N'\cdots 1'} = \sum_{\pi \in P_N} (-1)^{\mathrm{sign}(\pi)} \prod_{i=1}^N \omega_{i \pi(i')}
\end{equation*}

which can be computed using the result of the previous section.

\begin{rem}
   Actually we would need to consider the sign of the permutation that realizes:

   \begin{equation*}
       \{1,\cdots,N,N',\cdots,1'\} \mapsto \{1,1', 2, 2', \cdots, N,N'\}
   \end{equation*}

   However, this is attained with a number of steps:

   \begin{equation*}
       \sum_{l=1}^N 2(N-l) = 2 N^2 - N(N+1) = N(N-1)
   \end{equation*}
   which is always even. 
\end{rem}

As a consequence, for the relative entropy, we have:
\begin{equation*}
    S(\omega_{F_1\cdots F_N}|| \omega) = i \frac{d}{dt}\bigg|_{t=0} \sum_{\pi \in P_N} (-1)^{\mathrm{sign}(\pi)} \prod_{i=1}^N \omega_{i \pi(i')}
\end{equation*}
 Where we have kept the above notation. Moreover, from the above orthogonality condition, we compute:
 \begin{align*}
     S(\omega_{F_1\cdots F_N}|| \omega) &= i \sum_{i=1}^N \frac{d}{dt}\bigg|_{t=0} \omega(B(f^{(i)}) B(f_t^{(i)}))\\
     &= \sum_{i = 1}^N S(\omega_{F_i} \| \omega)
 \end{align*}
 In particular, a general polynomial excitation of $\omega$, of this kind, gives a relative entropy that is still computable just in terms of the inner product that we have on $\mathcal{H}$ and in particular just in terms of the result for a single unitary field excitation.

\section{Extension to more general types of excitations}
The result in Prop. \ref{prop: 1} can be generalized to another type of excitation, with respect to a more general type of test functions. In this section we present two of these examples.

\subsection{Standard subspaces for a fermionic QFT}
Let $(\mathcal{H},\Gamma)$ be the Hilbert space of solutions of the Dirac equation giving rise, for simplicity, to a Majorana field algebra. Therefore, the condition on the spinor test function $\Gamma f = f$, makes the element $B(f) \in \mathfrak{A}_{SDC}(\mathcal{H},\Gamma)$ a unitary fermionic field. However, if we choose a particular representation of the Dirac algebra, we obtain:
\begin{equation*}
    \Gamma f = \overline{f}
\end{equation*}
Therefore, the condition for $B(f)$ to be unitary, is that $f \in \mathcal{H}$ should be real. Let us call $H$ the real closed (see Remark \ref{rem: 3.1.1.}) linear subspace of $\mathcal{H}$ containing the real spinor test functions. Then, from the discussion outlined at the beginning of Section \ref{sec: CohBos}, allows us to assume $H$ to be a closed, real standard subspace of $\mathcal{H}$.\\
Then, by considering elements $g \in H + iH$ such that:
\begin{equation*}
    g = f + if \hspace{20pt} \mathrm{for} \,\, f \in H
\end{equation*}
we have in particular:
\begin{equation*}
    B(g) = (1+i)B(f)
\end{equation*}
That we can check to be unitary, after a proper choice of the normalization:
\begin{equation*}
    B(g) B^*(g) = 2 B^2(f) = (f,f)_{\mathcal{H}} \mathbb{1}
\end{equation*}
namely, if we pick $f$ such that $(f,f)_{\mathcal{H}} = 1$. For this reason, we may as well compute the relative entropy between a state $\omega$ (with the properties discussed above) over $\mathfrak{A}_{SDC}(\mathcal{H},\Gamma)$ and another state $\omega_{\Tilde{F}}$ obtained by exciting the corresponding vacuum state using such a $B(g)$. In particular, in the GNS representation of $\omega$, let us call $\Tilde{F} := \pi_{\omega}(B(g))$. Then, repeating the steps performed in the first section:
\begin{align*}
    S(\omega_{\Tilde{F}}\| \omega) &= i \frac{d}{dt}\bigg|_{t=0} \braket{\Omega_{\omega}}{\Tilde{F} \Delta_{\Omega_{\omega}}^{it} \Tilde{F}^* \Omega_{\omega}}\\
    &= i \frac{d}{dt}\bigg|_{t=0} \braket{\Omega_{\omega}}{(1+i)F(1-i)F_t \Omega_{\omega}}\\
    &= 2 i \frac{d}{dt}\bigg|_{t=0} \braket{\Omega_{\omega}}{F F_t \Omega_{\omega}}\\
    &= 2 S(\omega_F \| \omega)
\end{align*}

\subsection{Unitary exponential excitation}\label{sec: CohFer}
Consider in the abstrac self-dual CAR algebra, not necessarily a Dirac/Majorana algebra, for $f \in \mathcal{H}$ such that $\Gamma f = f$, the following element:
\begin{equation*}
    e^{i B(f)} = \sum_{k = 0}^{+\infty} \frac{(i B(f))^k}{k !}
\end{equation*}
The first thing we may notice is that, from the $*$-invariance of the corresponding $B(f)$, we have:
\begin{equation*}
    \big( e^{i B(f)} \big)^* = e^{-i B(f)}
\end{equation*}
As a consequence, using the Baker-Campbell-Hausdorff formula, we can compute:
\begin{align*}
    \big( e^{i B(f)} \big)^* e^{i B(f)} &= \mathbb{1}
\end{align*}
Therefore, the excitation induced by $e^{i B(f)}$ (that from now on we call \textit{unitary exponential excitation}) will be represented by a unitary operator, once we represent the algebra as bounded operators over a Hilbert space. Moreover, using the properties of $B(f)$, we rewrite the unitary exponential excitation as:
\begin{align*}
    e^{i B(f)} &= \mathbb{1} + i B(f) +\frac{(i)^2}{2!}B(f)^2 + \frac{(i)^3}{3!}B(f)^3 + \dots\\
    &= \mathbb{1} + i B(f) +\frac{(i)^2}{2!} \mathbb{1} + \frac{(i)^3}{3!}B(f) + \dots\\
    &= \bigg( 1 - \frac{1}{2!} +  \frac{1}{4!} + \dots\bigg) \mathbb{1} + i \bigg(1 - \frac{1}{3!} + \frac{1}{5!} + \dots\bigg) B(f)\\
    &= \cos(1) \mathbb{1} + i \sin(1) B(f)
\end{align*}
If we call $\omega_C$ the state obtained from $\omega$ by the unitary exponential excitation, we can compute:
\begin{align*}
    S(\omega_C \| \omega) &= i \frac{d}{dt}\bigg|_{t=0} \braket{\Omega_{\omega}}{e^{i B(f)} e^{-iB(f_t)} \Omega_{\omega}}\\
    &= i \frac{d}{dt}\bigg|_{t=0} \braket{\Omega_{\omega}}{\big(\cos(1) \mathbb{1} + i \sin(1) B(f)\big) \big(\cos(1) \mathbb{1} - i \sin(1) B(f_t)\big)\Omega_{\omega}}\\
    &= i \frac{d}{dt}\bigg|_{t=0}\sin^2(1) \braket{\Omega_{\omega}}{B(f) B(f_t) \Omega_{\omega}}\\
    &= i \sin^2(1) \frac{d}{dt}\bigg|_{t=0} (f, S f_t)_{\mathcal{H}}\\
    &= \sin^2(1) S(\omega_F \| \omega)
\end{align*}
Therefore, in the case of a unitary exponential excitation, the relative entropy can again be computed just in terms of the single unitary field excitation.\\
In this sense, despite the analogy of this case with the bosonic coherent excitation, it seems more fundamental to consider the single fermionic unitary excitation as the fermionic analogue of the bosonic case.

\section{Comparison with Von Neumann relative entropy}
To check our result, we shall compute the relative entropy using Araki's formula, for a self dual CAR algebra constructed over a finite dimensional Hilbert space. In this context, we know it should coincide with the relative entropy computed using the usual relative entropy formula due to von Neumann.\\
Let $(\mathcal{H}, \Gamma)$ be such that $\dim \mathcal{H} < +\infty$ and even. Construct over it the self dual CAR algebra $\mathfrak{A}_{SDC}(\mathcal{H},\Gamma)$. Morever, assume that we have a (strongly continuous) one parameter group of unitaries $V_t$, describing a dynamics over $\mathcal{H}$. As explained at the beginning of Section \ref{sec: risultati}, this raises to a $*$-automorphism $\alpha_t$ over $\mathfrak{A}_{SDC}(\mathcal{H},\Gamma)$. The existence of such an automorphism, allows to define notions like ground and KMS states over the abstract self dual CAR algebra.\\

Therefore, let us pick a state $\omega$ over $\mathfrak{A}_{SDC}(\mathcal{H},\Gamma)$ assumed to be a ground state that is also quasifree leading to a Fock representation (see Lemma \ref{lem: FockFer}). Let us call $P$ the associated basis projection over $(\mathcal{H},\Gamma)$, and perform the corresponding \textit{GNS construction} over $\mathfrak{A}_{SDC}(\mathcal{H},\Gamma)$. This leads to a triple of Fock-Hilbert space, Fock-representation map and a corresponding vacuum vector, all denoted as $(\mathcal{K}_{\omega}, \pi_{\omega}, \Omega_{\omega})$. A ground state always exists, as the assumption of $\mathcal{H}$ having a dynamics $V_t$ generated by a self-adjoint operator $\mathbf{h}$ that anticommutes with $\Gamma$, allows us to take $P$ to be the projection onto the positive part of the spectrum of $\mathbf{h}$. Then, the corresponding quasifree state associated to $P$, is the ground state:
\begin{align*}
    -i \partial_t \omega(B^*(g) \alpha_t(B(g)))|_{t=0} &= -i \partial_t|_{t=0} (g, P V_{-t} g)_{\mathcal{H}}\\
    &= -i \partial_t|_{t=0} (g, P e^{i t \mathbf{h}} g)_{\mathcal{H}}\\
    &= -i \partial_t|_{t=0} (g, \sum_{n} e^{it E^+_n} b_n^+ \Psi_n^+)_{\mathcal{H}}\\
    &= -i \partial_t|_{t=0} (g, \sum_{n} e^{it E^+_n} b_n^+ \Psi_n^+)_{\mathcal{H}}\\
    &= \sum_{n} \sum_{m} \overline{b}_m^+ b_n^+ \delta_{nm} E_n^+ \geq 0
\end{align*}
Where we have expanded a general $g \in \mathcal{H}$ on an eigenbasis of $\mathbf{h}$:
\begin{equation*}
    g = \sum_{n} b_n^- \Psi_n^- + b_n^+ \Psi_n^+ + b_n^0 \Psi_n^0
\end{equation*}
where the sum is finite and $b_n^{\pm} \in \mathbb{C}$. Moreover, we have denoted by $E_n^{\pm}$ the positive/negative eigenvalues of $\mathbf{h}$.\\
As a ground state, $\omega$ is also stationary. This means that $\alpha^*_t\omega = \omega$, where we have raised the action of the automorphism $\alpha_t$ to state functionals:
\begin{equation*}
    \alpha^*_t \omega(A) := \omega( \alpha_t A) \hspace{20pt} \forall A \in \mathfrak{A}_{SDC}(\mathcal{H},\Gamma),
\end{equation*}
The stationarity implies, from the uniqueness of the GNS construction up to unitary equivalence, that $\alpha_t$ is implementable in $\mathcal{K}_{\omega}$ by a one-parameter family of unitaries $U_t = e^{-it H}$. In particular, $H$ is the second quantization of the one-particle Hamiltonian $\mathbf{h}$, as $\mathcal{K}_{\omega}$ is the Fock space consructed over the one-particle Hilbert space $P \mathcal{H}$.\\
Moreover, as $\mathcal{H}$ is finite dimensional and we are dealing with fermions: $\dim \mathcal{K}_{\omega} < \infty$. Then, we can consider a general density matrix, associated to a Gibbs state of inverse temperature $\beta > 0$, on $\mathcal{K}_{\omega}$:
\begin{equation*}
    \rho_{\beta} = \frac{e^{-\beta H}}{\Tr(e^{-\beta H} )}
\end{equation*} 
and this is well defined, i.e. is of trace class, as $\dim \mathcal{K}_{\omega} < \infty$.\\

Having a density matrix of a Gibbs state in this finite dimensional setting, allows us to go back from the representation to the abstract algebra:
\begin{equation*}
    \rho_{\beta} \longrightarrow \omega_{\beta}
\end{equation*}
defining $\omega_{\beta}$ as a KMS state over $\mathfrak{A}_{SDC}(\mathcal{H}, \Gamma)$. The argument behind this, is based on the fact that the KMS state $\omega_{\beta}$ is defined over a finitely generated algebra $\mathfrak{A}_{SDC}(\mathcal{H}, \Gamma)$. 

The finite dimension of the algebra, gives that $\omega_{\beta}$ is quasiequivalent to $\omega$. This follows from the fact that the GNS constructions with respect to them, leads to finite dimensional vector spaces that, as such, are all isomorphic to some $\mathbb{C}^n$. Therefore, we can define an isomorphism (denoting the GNS triple associated to $\omega_{\beta}$ as $(\mathcal{K}_{\beta}, \pi_{\beta}, \Omega_{\beta})$):
\begin{equation*}
    B: \mathcal{K}_{\beta} \to \mathcal{K}_{\omega}
\end{equation*}
But then, whenever we take a normal state $\phi$ in the representation $\pi_{\beta}$, denoted as $\phi \in \mathfrak{S}^{(\pi_{\beta})}(\mathfrak{A}_{SDC}(\mathcal{H}, \Gamma))$, we have:
\begin{align*}
    \phi(A) &= \Tr_{\beta}\big( \rho_{\phi} \pi_{\beta}(A) \big)\\
    &= \Tr_{\omega}\big( (B^{-1})^*\rho_{\phi} \pi_{\beta}(A) B^{-1}\big)\\
    &= \Tr_{\omega}\big( (B^{-1})^*\rho_{\phi} B^{-1}\pi_{\omega}(A) B B^{-1}\big)\\
    &= \Tr_{\omega}\big( (B^{-1})^*\rho_{\phi} B^{-1}\pi_{\omega}(A)\big)
\end{align*}
for any $A \in \mathfrak{A}_{SDC}(\mathcal{H}, \Gamma)$. Therefore, $\phi \in \mathfrak{S}^{(\pi_{\omega})}(\mathfrak{A}_{SDC}(\mathcal{H}, \Gamma)) $. In this way, showing with the analogous argument also the opposite inclusion, we have $\mathfrak{S}^{(\pi_{\beta})}(\mathfrak{A}_{SDC}(\mathcal{H}, \Gamma)) = \mathfrak{S}^{(\pi_{\omega})}(\mathfrak{A}_{SDC}(\mathcal{H}, \Gamma))$ proving the quasiequivalence. The immediate consequence, is that $\omega_{\beta}$ has an associated density matrix in $\mathcal{K}_{\omega}$, as it is normal with respect to its representation by considering in $\mathcal{K}_{\beta}$ simply $\ket{\Omega_{\beta}}\bra{\Omega_{\beta}}$.\\

Let us assume, that also the thus constructed state $\omega_{\beta}$ is quasifree and that its associated basis polarization is $S$. Therefore, by considering $\omega_{\beta}$ over $\mathfrak{A}_{SDC}(\mathcal{H}, \Gamma)$ we can perform another \textit{GNS construction}, leading to:  $(\mathcal{K}_{\beta}, \pi_{\beta}, \Omega_{\beta})$.\\\\
After this necessary introduction of the setup, let us go back to our elements of the self dual CAR algebra $B(f)$, for $f \in \mathcal{H}$ such that $\Gamma f = f$, and introduce the following convenient notation:
\begin{align*}
    F_g &:= \pi_{\omega}(B(f))\\
    F_{\beta} &:= \pi_{\beta}(B(f)).
\end{align*}
We aim at computing the relative entropy between the KMS state $\omega_{\beta}$ and the one obtained by acting on it with $B(f)$ that we will denote as:
\begin{equation*}
    \omega_{F_{\beta}}(A) := \omega_{\beta}(B(f) A B(f)) \hspace{20pt} \forall A \in \mathfrak{A}_{SDC}(\mathcal{H}, \Gamma)
\end{equation*}\\\\
We start by computing the von Neumann relative entropy. For that purpose, we need to derive the form of the density matrix in the Hilbert space $\mathcal{K}_{\omega}$ associated to our considered unitary excitation of the KMS state. To derive it, notice that the quasiequivalence discussed above, gives:
\begin{align*}
    \omega_{F_{\beta}}(\cdot) &= \omega_{\beta}(B(f) \cdot B(f))\\
    &= \Tr_{\omega}\big( \rho_{\beta} F_g \pi_{\omega}(\cdot) F_g \big)\\
    &= \Tr_{\omega}\big( F_g \rho_{\beta} F_g \pi_{\omega}(\cdot) \big)\\
    &= \Tr_{\omega}\big( \sigma_{\beta} \pi_{\omega}(\cdot) \big)
\end{align*}
where we have denoted by $\sigma_{\beta}$ the density matrix associated with the new state in $\mathcal{K}_{\omega}$.\\
Therefore, we compute the corresponding von Neumann relative entropy in $\mathcal{K}_{\omega}$:
\begin{align*}
   S_{vN}(\rho_{\beta}, \sigma_{\beta}) := S_{vN} &= \Tr\bigg( \rho_{\beta} (\log \rho_{\beta} - \log \sigma_{\beta}) \bigg)\\
    &= \Tr\bigg( \frac{e^{-\beta H}}{\Tr(e^{-\beta H})} \bigg(\log(e^{-\beta H}) - \log(\Tr(e^{-\beta H})) - \log \big( F_g \rho_{\beta} F_g \big) \bigg) \bigg).
\end{align*}
Now, as the $F_g$ is a unitary operator, we can drag it out of the logarithm:
\begin{align*}
    \log(F_g \rho_{\beta} F_g) &= \log(1+ (F_g \rho_{\beta} F_g -1))\\
    &= \sum_{n=0}^{\infty} \frac{(-1)^{n+1}}{n!}(F_g \rho_{\beta} F_g -\mathbb{1})^{n}\\
    &= \sum_{n=0}^{\infty} \frac{(-1)^{n+1}}{n!} (F_g(\rho_{\beta} -\mathbb{1})F_g)^{n}\\
    &= \sum_{n=0}^{\infty} \frac{(-1)^{n+1}}{n!} F_g(\rho_{\beta} -\mathbb{1})^{n}F_g\\
    &= F_g \bigg(\sum_{n=0}^{\infty} \frac{(-1)^{n+1}}{n!} (\rho_{\beta} -\mathbb{1})^{n}\bigg)F_g\\
    &= F_g \log(\rho_{\beta}) F_g.
\end{align*}
The series expansion is well defined as:
\begin{align*}
    \| F_g \rho_{\beta} F_g - \mathbb{1} \|_{op}^2 &= \| F_g (\rho_{\beta} - \mathbb{1}) F_g \|_{op}^2\\
    &= \sup_{\|v \|_{\mathcal{K_{\omega}}} = 1} \| F_g (\rho_{\beta}  - \mathbb{1})F_g v\|_{\mathcal{K}_{\omega}}^2\\
    &= \sup_{\|v \|_{\mathcal{K_{\omega}}} = 1} \|(\rho_{\beta}  - \mathbb{1})F_g v\|_{\mathcal{K}_{\omega}}^2\\
    &= \sup_{\|v \|_{\mathcal{K_{\omega}}} = 1} \| \rho_{\beta} v - v\|_{\mathcal{K}_{\omega}}^2 = \| \rho_{\beta} - \mathbb{1} \|_{op}^2\\
    &= \sup_{\|v \|_{\mathcal{K_{\omega}}} = 1} \langle (\rho_{\beta} v - v), (\rho_{\beta} v - v)\rangle\\
    &= \sup_{\|v \|_{\mathcal{K_{\omega}}} = 1} (\langle \rho_{\beta} v ,\rho_{\beta} v \rangle + \langle  v , v \rangle - 2\langle v ,\rho_{\beta} v \rangle)\\
    &= \sup_{\|v \|_{\mathcal{K_{\omega}}} = 1} (\langle v ,\rho_{\beta}^2 v \rangle + \langle  v , v \rangle - 2\langle v ,\rho_{\beta} v \rangle)\\
    &\leq \sup_{\|v \|_{\mathcal{K_{\omega}}} = 1} (1 - \langle v ,\rho_{\beta} v \rangle)\\
    &< 1.
\end{align*}
Where we have used that a density matrix is positive definite, self-adjoint and has unit trace so $\rho_{\beta}^2 \leq \rho_{\beta}$. Moreover, at the third step we have used that $F_g$ is unitary, at the fourth that still by unitarity we have $\| F_g v \|_{\mathcal{K}_{\omega}} = 1$ allowing us to include everything in the supremum and at the last step the continuity of the expectation value together with the compacteness of the unit ball in finite dimension.\\
Going back to the computation of the von Neumann relative entropy:
\begin{align*}
    S_{vN} &= \Tr\bigg( \frac{e^{-\beta H}}{\Tr(e^{-\beta H})} \bigg(\log(e^{-\beta H}) - \log(\Tr(e^{-\beta H})) - F_g \log (\rho_{\beta}) F_g  \bigg) \bigg)\\
    &= \Tr\bigg( \frac{e^{-\beta H}}{\Tr(e^{-\beta H})} \bigg(\log(e^{-\beta H}) - \log(\Tr(e^{-\beta H})) - F_g \big(\log (e^{-\beta H}) - \log(\Tr(e^{-\beta H}))\big) F_g  \bigg) \bigg)\\
    &= \Tr\bigg( \frac{e^{-\beta H}}{\Tr(e^{-\beta H})} \bigg(\log(e^{-\beta H}) - F_g \log (e^{-\beta H}) F_g \bigg) \bigg)\\
    &= \Tr\bigg( \frac{e^{-\beta H}}{\Tr(e^{-\beta H})} \big(-\beta H + F_g (\beta H) F_g \big) \bigg)\\
    &= \beta \Tr\bigg( \frac{e^{-\beta H}}{\Tr(e^{-\beta H})} \big(F_g [H, F_g] \big) \bigg).
\end{align*}\\\\
On the other hand, let us now compute the relative entropy using our result starting from Araki's formula:

\begin{align*}
    S_A &:= S(\omega_{F_{\beta}}|| \omega_{\beta})\\
    &= i\frac{d}{dt}\bigg|_{t=0} \braket{\Omega_{\beta}}{F_{\beta} \Delta_{\beta}^{it} F_{\beta}\Omega_{\beta}}.
\end{align*}
In the derivation of our result, we were assuming that the modular flow was the raising of the dynamics on $\mathcal{H}$. Therefore, the abovely discussed dynamics $V_t = e^{-it \mathbf{h}}$ is related to the modular flow in the GNS construction of $\omega_{\beta}$ as follows:
\begin{align*}
    \Delta_{\Omega_{\beta}}^{it} B(f) \Delta_{\Omega_{\beta}}^{-it}  = B(e^{-it \beta \mathbf{h}} f)
\end{align*}

Therefore, the quasiequivalence between $\omega_{\beta}$ and $\omega$, gives:
\begin{align*}
    S_{A} &= i\frac{d}{dt}\bigg|_{t=0} \omega_{\beta}(B(f) B(e^{-it\beta \mathbf{h}}f))\\
    &= i\frac{d}{dt}\bigg|_{t=0} \Tr\big( \rho_{\beta} F_g F_g^t \big)\\
    &= i\frac{d}{dt}\bigg|_{t=0} \Tr\big( \rho_{\beta} F_g e^{-it \beta H} F_g e^{it \beta H} \big)\\
    &= i \beta \Tr\big( \rho_{\beta} F_g (-iH F_g + i F_g H) \big)\\
    &= \beta \Tr\big( \rho_{\beta} F_g [H,F_g] \big)\\
    &= \beta \Tr\bigg( \frac{e^{-\beta H}}{\Tr(e^{-\beta H})} \big(F_g [H,F_g] \big) \bigg)
\end{align*}
where now $H$ is the second quantization of $\mathbf{h}$ with respect to the ground state.\\
The abstract steps performed in the proof for the Araki entropy, are outlined in the following diagram:
\begin{figure}[H]
 		\centering
 		\includegraphics[width=1.20	\columnwidth]{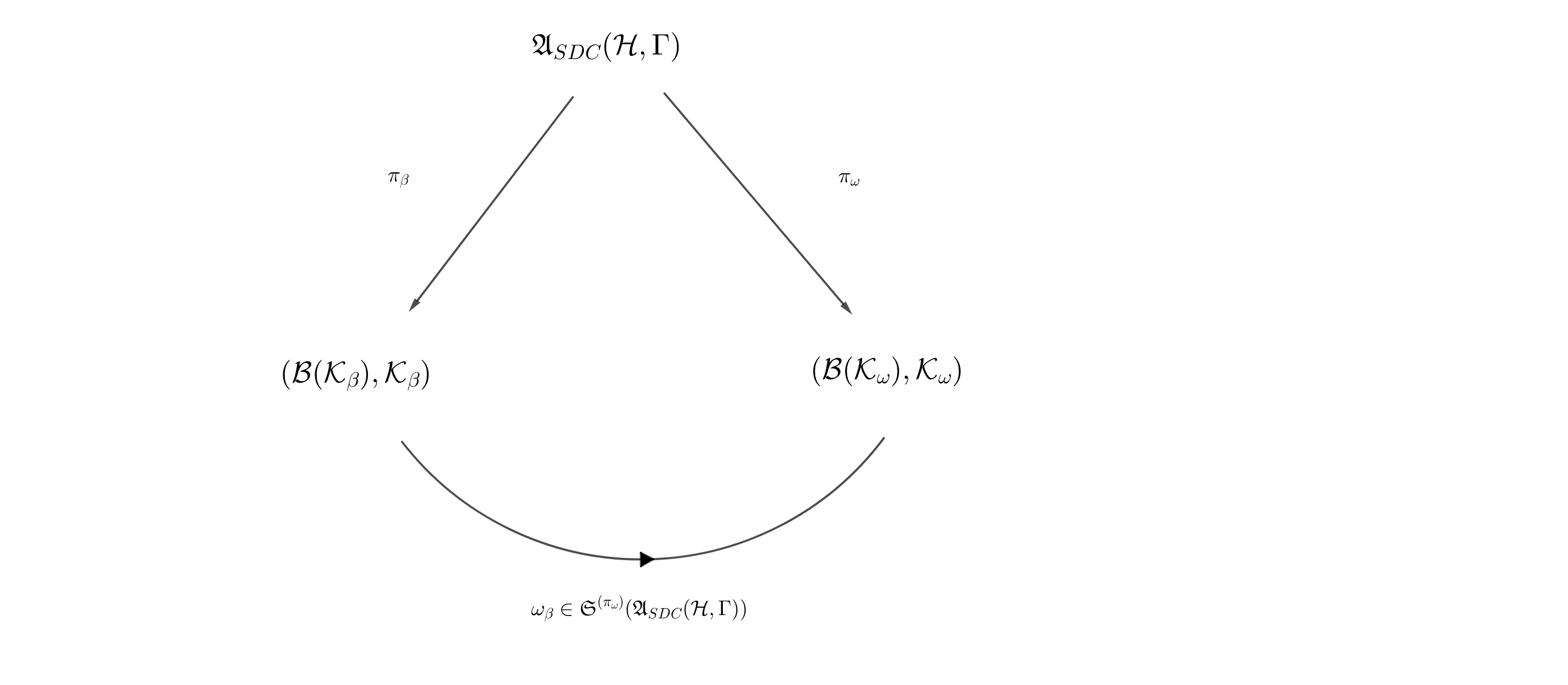}
\end{figure}
In this way, we have proven the equivalence between our result and the von Neumann relative entropy, providing a consistency check for it.\\
The result just presented, gives even a simpler way to compute relative entropies in the finite dimensional context, we just need a specific dynamics over $\mathcal{H}$. Once that is known, the relative entropy becomes just a derivative, with respect to the parameter of the dynamics, of the inner product over $\mathcal{H}$. In this sense, this result shows once more the power of working with abstract algebras, where we can perhaps move from one representation to a more convenient one to compute the same quantity.

\section{Relative entropy for Majorana fields on ultrastatic spacetimes}
In this section we aim at applying our result for the relative entropy between a KMS state and the state obtained from it by the unitary excitation $B(f)$. These, are taken, for simplicity, as states over the abstract algebra of Majorana fields. However, as the explicit form of the modular automorphism on the corresponding von Neumann algebra is generally unknown, we will consider a configuration in which the underlying Hilbert space of solutions of Dirac's equation has a canonically implemented unitary dynamical evolution. This is achieved by assuming the spacetime $(M,g)$, on which the elements of $\mathcal{H}$ are defined, to be ultrastatic.\\\\
An ultrastatic spacetime $(M,g)$ is:
\begin{equation*}
    M = \mathbb{R} \times \mathcal{C} \hspace{20pt} g = dt^2 - h(\mathbf{x})
\end{equation*}
where $h(\mathbf{x})$ is a Riemannian metric over $\mathcal{C}$, that is assumed to be a compact Cauchy surface, and $\mathbf{x}$ denotes the coordinates on it. In order to define the notion of spinors, we further assume $\mathcal{C}$ to be parallelizable. We put this as an assumption as, in general, we do not fix the dimension of $\mathcal{C}$ to $3$ where this will be the case by Stiefel's theorem \cite{Stiefel1935/36}.\\
Under such assumptions, we can write the Dirac's equation $(-i\cancel{\nabla} + m)\psi = 0$ as:
\begin{align*}
    0 &= (i\gamma^0 \nabla_0 + i \gamma^j \nabla_j + m)\psi\\
    i \gamma^0 \partial_t \psi &= -(-i\sum_{j=1}^D\gamma_j \nabla_j + m)\psi\\
    i \partial_t \psi &= -(-i\sum_{j=1}^D \gamma_0 \gamma_j \nabla_j + \gamma_0 m)\psi\\
    &:= \mathbf{h} \psi
\end{align*}
where we have denoted by $D = \dim M$.\\
To find the solutions to this equation, we do the ansatz:
\begin{equation*}
    \psi(t, \mathbf{x}) = e^{-i\epsilon t} \Psi(\mathbf{x}).
\end{equation*}
We remark that such an ansatz is possible, just because of the ultrastatic assumption of the spacetime. Therefore, we can find the stationary solutions by solving the eigenvalue problem:
\begin{equation*}
    \epsilon \Psi(\mathbf{x}) = \mathbf{h} \Psi(\mathbf{x}).
\end{equation*}
Let us postpone the task of finding explicit solutions to this equation for later when we will study a specific model for the spacetime manifold. For the moment, we just call the corresponding Hilbert space of solutions $\mathcal{H}$. In this case, as there're no bounds on the energy, we will have $\dim \mathcal{H} = +\infty$.\\
As mentioned at the beginning, we will study Majorana fields. Therefore, the involution over $\mathcal{H}$ becomes:
\begin{equation*}
    \Gamma v = C \overline{v}
\end{equation*}
where $C$ is the charge conjugation operator that, by definition, satisfies:
\begin{equation*}
    C \gamma_a C^{-1} = -\overline{\gamma_a}
\end{equation*}
for $\gamma_a = \pi(l_a)$ a chosen representation of the Dirac algebra (see Appendix \ref{app}). As a consequence, if we compute:
\begin{align*}
    \Gamma^{-1} \mathbf{h} \Gamma &= -i \Gamma^{-1} \gamma_0 \Gamma \Gamma^{-1} \gamma_i \Gamma \nabla_i - \Gamma^{-1}\gamma_0 \Gamma m\\
    &= -i \gamma_0 \gamma_i \nabla_i + \gamma_0 m = -\mathbf{h}.
\end{align*}
Therefore, the involution $\Gamma$ maps eigenspinors with positive energy, to eigenspinors with negative energy. From this, we can already conclude that the spectrum must be symmetric with respect to $0$ and, from the form of $\mathbf{h}$, has a finite gap $(-m,m)$. Furthermore, from $[\Gamma, \mathbf{h}]_+ = 0$, follows that the associated dynamical evolution $V_t$, commutes with $\Gamma$. Therefore, we will be allowed to raise this map as a $*$-automorphism over the algebra constructed on $\mathcal{H}$. Finally, the spatial compacteness, implies the discretization of the allowed values for the momentum $p$ and consequently also of the energy spectrum.\\
Now, from the definition of $\Gamma$ and the form of $\mathbf{h}$, we define a projection $P$, over the space of solutions, corresponding to the projection into the positive energy eigenspinors of $\mathbf{h}$. Such a projection will, by construction, satisfy:
\begin{equation*}
    \Gamma P \Gamma = \mathbb{1} - P.
\end{equation*}
Therefore, picking the quasifree state $\omega$ over the abstract self dual CAR algebra, that one can construct over $(\mathcal{H}, \Gamma)$, associated with the above projection $P$, by performing a GNS construction, we obtain a Fock representation of the fields on the one-particle Hilbert space $\mathcal{H}_1 = P\mathcal{H}$.\\
We now specify this problem for a specific model of the spacetime, where we will solve the spatial Dirac equation.

\subsection{Majorana fields in 1+1 dimensions}
The goal of this subsection is to explicitly compute the relative entropy for a simple spacetime model. Namely, we will assume the spacetime to be $1+1$ dimensional.\\
First, we need to find stationary solutions to the Dirac equation, that corresponds to solve a $1$ dimensional differential equation with spinorial solutions in $\mathbb{C}^2$, as the fields are taken to be Majorana. Moreover, assuming the spatial sections to be compact intervals $\mathcal{C} \subset \mathbb{R}$, the equation becomes:
\begin{equation*}
    \epsilon \Psi(x) = \bigg(i \gamma_0 \gamma_i \frac{d}{dx} - \gamma_0 m\bigg) \Psi(x).
\end{equation*}
By the compacteness assumption, we can approach the problem as that of a particle in an infinite well potential, imposing vanishing boundary conditions (by conservation of probability).\\
\begin{figure}[H]
 		\centering
 		\includegraphics[width=1.00	\columnwidth]{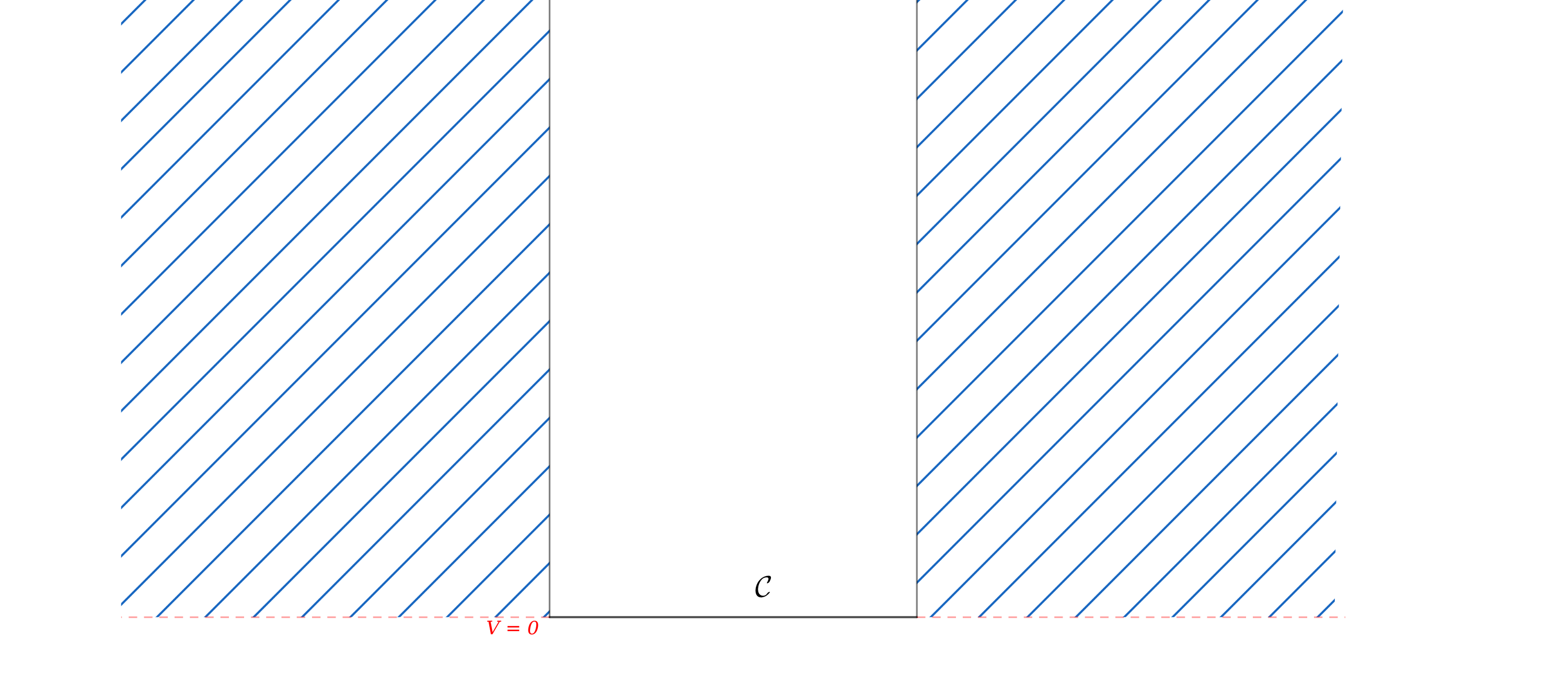}
\end{figure}

To simplify the computations, we choose as representation of the Dirac algebra:
\begin{align*}
    C &= \begin{pmatrix}
    1 & 0\\
    0 & 1
    \end{pmatrix} \hspace{20pt} A = \begin{pmatrix}
    0 & i\\
    -i & 0
    \end{pmatrix}\\
    \gamma_0 &= A \hspace{20pt} \gamma_1 = \begin{pmatrix}
    0 & i\\
    i & 0
    \end{pmatrix}\\
\end{align*}
that one can check to be a consistent choice. The first consequence is that $\Gamma v = \overline{v}$ for any solution $v$ of the Dirac equation. Finally, due to these choices, we make the space of solutions an Hibert space, by completing it with respect to the inner product:
\begin{equation*}
    \braket{f_1}{f_2}_{\mathcal{H}} := \int_{\mathcal{C}} f_1^{\dagger}f_2(x) dx
\end{equation*}
with $dx$ the Lebesgue measure. Namely, our Hilbert space of solutions is $\mathcal{H} = L^2(I, dx; \mathbb{C}^2)$ and we define on it the self dual CAR algebra $\mathfrak{A}_{SDC}(\mathcal{H}, \Gamma)$. Such an Hilbert space is of the kind that was constructed in Section \ref{sec: alternative}.\\

To solve the stationary Dirac equation, we start by finding the eigenvectors of the one particle hamiltonian:
\begin{align*}
    \mathbf{h} &= i \begin{pmatrix}
    -1 & 0\\
    0 & 1
    \end{pmatrix} \frac{d}{dx} - \begin{pmatrix}
    0 & im\\
    -im & 0
    \end{pmatrix}\\
    &= \begin{pmatrix}
    p & im\\
    -im & -p
    \end{pmatrix}
\end{align*}
where we have introduced the momentum operator $p := - i d/dx$. In this simple form, we can diagonalize the one-particle Hamiltonian to get the energy eigenvalues:
\begin{equation*}
    \epsilon_p^{\pm} = \pm \sqrt{p^2 + m^2}
\end{equation*}
for corresponding eigenvectors:
\begin{equation*}
    v^{\pm} = \frac{1}{\sqrt{2\big((\epsilon_p^{\pm})^2 - p \epsilon_p^{\pm}\big)}} \begin{pmatrix}
    -im\\
    p - \epsilon_p^{\pm}
    \end{pmatrix}.
\end{equation*}
One can even get an explicit form for the projection $P$:
\begin{equation*}
    P = \frac{1}{2((\epsilon_p^+)^2 - p \epsilon_p^+)}\begin{pmatrix}
    m^2 & -im(p-\epsilon_p^+)\\
    im(p-\epsilon_p^+) & (p-\epsilon_p^+)^2
    \end{pmatrix}
\end{equation*}
that, from our choice of $\Gamma$, can be checked to be compatible with $\Gamma P \Gamma = \mathbb{1}- P$.\\
For what concerns the $x$ dependence of the solutions, we need to solve the differential equation, in position space, for the infinite well potential. Assume for simplicity, without loss of generality, $\mathcal{C} = [0,a]$ for some $a \in \mathbb{R}$. Then, a general solution is:
\begin{equation*}
    \Psi_p^{\pm}(x) = A v_p^{\pm} e^{ipx} +  B v_p^{\pm} e^{-ipx}.
\end{equation*}
Imposing the vanishing boundary conditions:
\begin{equation*}
    \Psi^{\pm}_n(x) = A v_n^{\pm} \sin\bigg( \frac{\pi n}{a} x \bigg)
\end{equation*}
where we have already substituted:
\begin{equation*}
    p = \frac{\pi n}{a} \hspace{20pt} n \in \mathbb{N}.
\end{equation*}
Therefore, by the compacteness of $\mathcal{C}$, we have proven that the momenta are discretized, by the index $n$, and thus also the energy eigenvalues are:
\begin{equation*}
    \epsilon_{n}^{\pm} = \pm \sqrt{\bigg(\frac{\pi n}{a}\bigg)^2 + m^2}.
\end{equation*}
Moreover, by computing the inner product between two such eigenspinors:
\begin{equation*}
    \int_{\mathcal{C}} (\Psi^{s}_n(x))^{\dagger} \Psi^{s'}_{n'}(x) dx = \delta_{n,n'} \delta_{s,s'} |A|^2\frac{a}{2}
\end{equation*}
for $s,s' = \{+,-\}$. From the above orthogonality, fix the normalization to get:
\begin{equation*}
    \Psi^{\pm}_n(x) = \sqrt{\frac{2}{a}} v_p^{\pm} \sin\bigg( \frac{\pi n}{a} x \bigg).
\end{equation*}
In this way $\{ \Psi_n^+(x), \Psi_n^-(x) \}_{n \in \mathbb{N}}$ provides a basis for $\mathcal{H}$. For this reason, our $f \in \mathcal{H}$ such that $\Gamma f = f$, can be decomposed:
\begin{equation*}
    f(x) = \sum_{n \in \mathbb{N}} (1 + \Gamma)(a_n^+ \Psi_n^+(x) + a_n^- \Psi_n^-(x)).
\end{equation*}
assuming, without loss of generality, that this is nonzero. Moreover, from $[\mathbf{h}, \Gamma]_+ = 0$ we also see that:
\begin{align*}
    \mathbf{h} \Gamma \Psi_n^+ &= - \Gamma \mathbf{h} \Psi_n^+\\
    &= - \epsilon_n^+ \Gamma \Psi_n^+\\
    &= \epsilon_n^- \Gamma \Psi_n^+.
\end{align*}
That gives $\Gamma \Psi_n^+ = \chi_n \Psi_n^-$ for some $\chi_n \in \mathbb{C}$ that, by the orthonormality of the basis $\{ \Psi_n^+, \Psi_n^-\}_{n \in \mathbb{N}}$, is such that: $|\chi_n|^2 = 1$. Finally, from our choice of $\Gamma$ and the abovely computed eigenvectors, we conclude: $\chi_n = -1$. Therefore, we can rewrite:
\begin{equation*}
    f(x) = \sum_{n \in \mathbb{N}} (a_n^+ - \overline{a}_n^-) \Psi_n^+(x) + (a_n^- - \overline{a}_n^+) \Psi_n^-(x)
\end{equation*} 
\\\\
On $(\mathcal{H}, \Gamma)$, we construct the the self-dual CAR algebra $\mathfrak{A}_{SDC}(\mathcal{H}, \Gamma)$ on which we compute now the relative entropy. In particular, that will be done between a general quasifree $KMS$ state $\omega_{\beta}$, for the Majorana QFT over the $1+1$ dimensional ultrastatic spacetime,  with associated basis polarization $S$ and the corresponding state $\omega_{F_{\beta}}$ (where we are adopting the same notation as last section).\\
From our result in Prop. \ref{prop: 1}, we know we need to compute:
\begin{equation*}
    S(\omega_{\beta} \| \omega_{F_{\beta}}) = i\frac{d}{dt}\bigg|_{t=0} (f,Sf_{\beta t})_{\mathcal{H}}
\end{equation*}
for $f = \Gamma f$ and the time evolution is the one at the one-particle level, coming from that on the Fock space of the KMS state (that is why we have also denoted the $\beta$). In order to explicitly compute this expression, we need to know the explicit form of $S$. For that purpose, we first need to show that given a basis polarization, by doubling the Hilbert space, we can construct a basis projection. Therefore, let us define $\hat{\mathcal{H}} = \mathcal{H} \oplus \mathcal{H}$ and on it the involution $\hat{\Gamma} = \Gamma \oplus (- \Gamma)$. With respect to such a doubled Hilbert space, construct a self-dual CAR algebra $\mathfrak{A}_{SDC}(\hat{\mathcal{H}}, \hat{\Gamma})$ and notice that:
\begin{equation*}
    \mathfrak{A}_{SDC}(\mathcal{H}, \Gamma) \subset \mathfrak{A}_{SDC}(\hat{\mathcal{H}}, \hat{\Gamma})
\end{equation*}
As we can see $\mathcal{H}$ as a subset of $\hat{\mathcal{H}}$, identifying it with $\mathcal{H} \oplus 0$.\\
Now, let us consider over $\hat{\mathcal{H}}$ the following operator:
\begin{equation*}
    P_S = \begin{pmatrix}
    S & S^{1/2}(1-S)^{1/2}\\
    S^{1/2}(1-S)^{1/2} & (1 - S)
    \end{pmatrix}
\end{equation*}
Where $S$ is the basis polarization associated to the state $\omega$. But then, by direct computation using that $S^* = S$ and that $\Gamma S \Gamma = \mathbb{1} - S$:
\begin{align*}
    P_S^2 &= \begin{pmatrix}
    S & S^{1/2}(1-S)^{1/2}\\
    S^{1/2}(1-S)^{1/2} & (1 - S)
    \end{pmatrix} = P_S\\
    P_S &= P_S^*\\
    \hat{\Gamma} P_S \hat{\Gamma} &= 1 - P_S
\end{align*}
Where for the last equality, from $\Gamma S \Gamma = \mathbb{1} - S$, we get using the fact that $\Gamma$ is involutive:
\begin{equation*}
    (\Gamma S^{1/2} \Gamma) (\Gamma S^{1/2} \Gamma) = (\mathbb{1} - S)^{1/2}(\mathbb{1} - S)^{1/2}
\end{equation*}
That gives $\Gamma S^{1/2} \Gamma = (\mathbb{1} - S)^{1/2}$.\\
Since $P_S$ is a basis projection over $\hat{K}$, we take the associated quasifree state $\omega_{P_S}$ over $\mathfrak{A}_{SDC}(\hat{\mathcal{H}}, \hat{\Gamma})$ and notice that this is related to the initially chosen state $\omega$ by the following result of Araki (see Lemma $4.6$ in \cite{Araki:1971id}):
\begin{lem}
Let $S$, $P_S$, $\hat{\mathcal{H}}$, $\hat{\Gamma}$ as introduced above. Then, the restriction of the Fock state $\omega_{P_S}$ of $\mathfrak{A}_{SDC}(\hat{\mathcal{H}}, \hat{\Gamma})$ to $\mathfrak{A}_{SDC}(\mathcal{H}, \Gamma)$ is the quasifree state $\omega$.
\end{lem}
\begin{proof}
Since $\omega_{P_S}$ is quasifree, also its restriction will of course remain such. Now, if we take elements in $\mathcal{H}$ as $f = \hat{f} \oplus 0$ and $g = \hat{g} \oplus 0$, we can compute:
\begin{equation*}
    \omega_{P_S}(B^*(f) B(g)) = (f, P_S g) = (\hat{f}, S \hat{g})
\end{equation*}
That follows from the explicit form of $P_S$ given above.
\end{proof}
Let us denote by $(\mathcal{K}_{P_S}, \pi_{P_S}, \Omega_{P_S})$ the GNS triple assocaited to $\omega_{P_S}$. Then, taking advantage of this doubling of the Hilbert space, we can define the von Neumann algebra associated to $\mathcal{A}_{SDC}(\mathcal{H}, \Gamma)$ with respect to the representation $\pi_{P_S}$ for:
\begin{equation*}
    R_S := \pi_{P_S}(\mathfrak{A}_{SDC}(\mathcal{H}, \Gamma))''
\end{equation*}
Then, provided that such an algebra is a factor, by Theorem $3$ in \cite{Araki:1971id} we have that the basis polarization associated to our starting KMS state with inverse temperature $\beta$ is:
\begin{equation*}
    S = (1 + e^{-\beta \mathbf{h}})^{-1}
\end{equation*}
Then, we compute:
\begin{align*}
    S_A &= i\frac{d}{dt}\bigg|_{t=0} (f,Sf_{\beta t})_{\mathcal{H}}\\
    &= i \frac{d}{dt}\bigg|_{t=0} \sum_{n,m \in \mathbb{N}} \bigg[ \overline{(a_n^+ - \overline{a}_n^-)}(a_n^+ - \overline{a}_n^-) \frac{e^{-i\beta t \epsilon_n^+}}{1 + e^{-\beta \epsilon_n^+}}(\psi_m^+, \psi_n^+)_{\mathcal{H}} + \overline{(a_n^+ - \overline{a}_n^-)}(a_n^+ - \overline{a}_n^-) \frac{e^{i\beta t \epsilon_n^+}}{1 + e^{\beta \epsilon_n^+}}(\psi_m^-, \psi_n^-)_{\mathcal{H}}\bigg]\\
    &= \sum_{n \in \mathbb{N}} \beta \epsilon_n^+\overline{(a_n^+ - \overline{a}_n^-)}(a_n^+ - \overline{a}_n^-) \bigg[  \frac{1}{1 + e^{-\beta \epsilon_n^+}} - \frac{1}{1 + e^{\beta \epsilon_n^+}}\bigg]\\
    &= \beta \sum_{n \in \mathbb{N}} \epsilon_n^+\overline{(a_n^+ - \overline{a}_n^-)}(a_n^+ - \overline{a}_n^-) \tanh\bigg( \frac{\beta \epsilon_n^+}{2} \bigg)
\end{align*}
That, defining the energy content in each mode of $f$ as:
\begin{equation*}
    \Tilde{E}_n^+ := \epsilon_n^+\overline{(a_n^+ - \overline{a}_n^-)}(a_n^+ - \overline{a}_n^-)
\end{equation*}
Becomes:
\begin{equation*}
    S_A = \beta \sum_{n \in \mathbb{N}} \Tilde{E}_n^+ \tanh\bigg( \frac{\beta \epsilon_n^+}{2} \bigg)
\end{equation*}
That gives us the relative entropy between the two configurations, of the Majorana field, given by the KMS state $\omega_{\beta}$ and its single unitary excitation $\omega_{F_{\beta}}$ \\\\

The analysis carried out here, is easily generalizable to other ultrastatic spacetime models. For example, we may consider $M = \mathbb{R} \times S$, for $S$ a sphere with radius $r \in [0,a]$. In this case, the problem reduces to solve the spherically symmetric spatial Dirac equation. However, the form of the final result will remain the same: the inverse temperature of the considered KMS state times the energy content of the considered excitation damped by the hyperbolic tangent term.

\section{Conclusion}
We started off introducing the Algebraic approach to Quantum Theory and provided arguments supporting the considerable advantages of working with abstract algebras, in particular in the last chapter the strengths of working with abstract algebras. Subsequently, we introduced fermionic fields, in the algebraic spirit, over a globally hyperbolic spacetime, putting special emphasis on the fact that the corresponding algebra falls under the general category of self-dual CAR algebras.\\
Moreover, thanks to the algebraic perspective, we were able to better understand the issues encountered in defining relative entropy for a theory with an uncountable number of degrees of freedom and that is defined by a type $III$ von Neumann algebra, such as QFT. On the other hand, the algebraic point of view simultaneously allowed us to define a generalization of relative entropy as a measure of distinguishability for states defined over the abstract von Neumann algebra.\\
Finally, in order to deal with such an abstract definition, we presented the case of a coherent excitation of the vacuum for a free scalar Quantum Field Theory, in which the relative entropy is computed at the one-particle level (\cite{Longo:2019mhx}, \cite{PhysRevD.99.125020}). Inspired by this result, we were able to derive an expression for the relative entropy for a unitary excitation of a quasifree faithful state given by:
\begin{equation}\label{entropy}
        S(\omega_{F}\,||\,\omega) = i \frac{d}{dt} \bigg|_{t=0} (f, S f_t)_{\mathcal{H}},
\end{equation}
presented in this work. As became apparent in the last chapter, the computation of relative entropy then reduces to the knowledge of the form of the basis polarization $S$.\\
We argued, presenting a series of examples, how this result can be generalized to different kinds of excitations and how in those cases, the relative entropy is expressible in terms of the one of the unitary excitation. Based on our current understanding, we believe that this result can be generalized further, for instance, in order to obtain a lower bound for the relative entropy. In fact, from the assumption of the $\Gamma$-invariant subspace $H$ to be of the standard type, any $g \in H + iH$ can be written as:
\begin{equation*}
    g = f_1 + i f_2 \Longrightarrow B(g) = B(f_1) - i B(f_2),
\end{equation*}
for $f_1, f_2 \in H$. In this case, the associated $B(g)$ does not square to the identity. This implies that the corresponding vector state representative cannot be cyclic and separating. Nevertheless, as discussed in \cite{Araki1977}, we may define a notion of relative entropy even if only one of the states is faithful. The expression that one should study involves $\Delta_{\pi_{\omega}(B(g))\Omega_{\omega}, \Omega_{\omega}}$, where the state:
\begin{equation*}
    \pi_{\omega}(B(g))\Omega_{\omega} = (\pi_{\omega}(B(f_1)) - i \pi_{\omega}(B(f_2)))\Omega_{\omega},
\end{equation*}
is the sum of two cyclic and separating vectors. We believe it to be worthwhile to investigate this case further in the future as it might lead to some additional and more general result.
\begin{appendices}
\chapter{Appendix}
\section{Elements of Lorentzian geometry}\label{app: Lorentz}
This section aims to recall the most important definitions and results of Lorentzian geometry, which are used in the thesis. The purpose is, by no means, to give a complete review of Lorentzian geometry, for a complete treatement, I instead refer to \cite{Wald:1984rg}, to which our notation will be affine and I further refer to \cite{oneill1983semiriemannian} for a more mathematically rigorous discussion.\\\\
Firstly, remember the convention for the spacetime, to be a Lorentzian Manifold $M$ with metric $g$ of signature $(+---)$. Let me quote a set of standard definitions:
\begin{defn}
For a closed subset $C \subset M$, one defines the causal future/past as:
\begin{align*}
    J^{\pm}(C) := \{ p \in M | &\exists \,\, \mathrm{future}/\mathrm{past} \,\, \mathrm{directed} \,\, \mathrm{causal} \,\, \mathrm{curve} \, \gamma(\tau)\\
    & \mathrm{and} \, \tau_1 \geq \tau_0 \,\, \mathrm{s.t.} \,\, \gamma(\tau_0) \in C, \gamma(\tau_1) = p \}
\end{align*}
\end{defn}
\begin{defn}
Let $S \subset M$ be a closed achronal set. We define the future/past domain of dependence or future/past causal developement of $S$ as:
\begin{equation*}
    D^{\pm}(S) := \{ p \in M| \mathrm{Every} \,\, \mathrm{past/future} \,\, \mathrm{inextendible} \,\,  \mathrm{causal} \,\, \mathrm{curve} \,\, \mathrm{through} \,\, p \,\, \mathrm{intersects} \,\, S\}.
\end{equation*}
The domain of dependence of $S$ is then defined as: $D(S) = D^+(S) \cup D^-(S)$
\end{defn}
\begin{defn}
Given $\mathcal{O} \subset M$, we define its spacelike complement $\mathcal{O}'$ as the set of all points that are spacelike separated with all points of $\mathcal{O}$:
\begin{equation*}
    \mathcal{O}' := \mathrm{int}\{p\in M | p \not \in J^{\pm}(q) \,\, \forall q \in \mathcal{O} \}
\end{equation*}
and its causal completion as $\mathcal{O}''$. It is always true that $(\mathcal{O}'')'= \mathcal{O}'$. We say that $\mathcal{O}$ is causally complete if $\mathcal{O} = \mathcal{O}''$.
Finally, let us denote the set of all causally complete regions on $M$ as $\mathcal{K}$.
\end{defn}
\begin{defn}
A curve $\gamma : I \to M$ has a future (past) endpoint $p \in M$ if for any open neighborhood $U$ containing $p$, there exist some value $\tau_0 \in I$ for the affine parameter such that $\gamma(\tau) \in U$ for $\tau > \tau_0$ (resp. $\tau< \tau_0$)
\end{defn}
\begin{defn}
A curve that does not have future or past endpoints is called \textit{future/past inextendible}.
\end{defn}
\begin{defn}
A Cauchy surface $\Sigma$ is a hypersurface of $(M,g)$, such that any inextendible causal curve intersects $\Sigma$ exactly once\footnote{One can show this to be equivalent to the definition of a Cauchy surface as the hypersurface whose causal developement is the entire spacetime manifold $M$}.
\end{defn}
\begin{defn}
A spacetime $(M,g)$ is said to be globally hyperbolic, if it possesses a Cauchy surface.
\end{defn}
For a globally hyperbolic spacetime, we have the following known result:
\begin{prop}
Let $(M,g)$ be a globally hyperbolic spacetime. For any compact $K,K' \subset M$, $J^{\pm}(K)$ is closed, $J^+(K) \cap J^-(K')$ is compact and, for any Cauchy surface $\Sigma$, the intersection $\Sigma \cap J^{\pm}(K)$ is compact.
\end{prop}
\newpage
\section{The Reeh-Schlieder theorem}
We will present the theorem in the case of a free scalar $QFT$ on Minkowski spacetime but it can be easily generalized also to the case of an arbitrary $QFT$. I refer to \cite{Reeh:1961ujh} for the original proof, while here we present a revisited version of it using the Schwarz reflection principle.
\begin{thm}[\textbf{Reeh-Schlieder}]\label{app: Reeh}
Consider a free scalar Quantum Field Theory on Minkowski spacetime $M_D := (\mathbb{R}^D, \eta)$. Then, given an open subset $\mathcal{V}$ of a Cauchy surface $\Sigma$ with neighbourhood $U_{\mathcal{V}} \subset M_D$ and an arbitrary number of test functions $f_1, \cdots, f_n \in \mathcal{C}^{\infty}_0(U_{\mathcal{V}}, \mathbb{C})$ the set of vectors:
\begin{equation*}
    \phi_{f_1} \cdots \phi_{f_n} \ket{\Omega}
\end{equation*}
Defines a dense subset of the Hilbert space $\mathcal{H}$ of the quantum field theory and $\ket{\Omega}$ is the cyclic and separating vector representing the vacuum of the theory.  The above notation means:
\begin{equation*}
    \phi_f = \int_{M_D} d^Dx \phi(x) f(x)
\end{equation*}
\end{thm}
\begin{figure}[H]
 		\centering
 		\includegraphics[trim = 170 0 180 60,clip,width=0.60	\columnwidth]{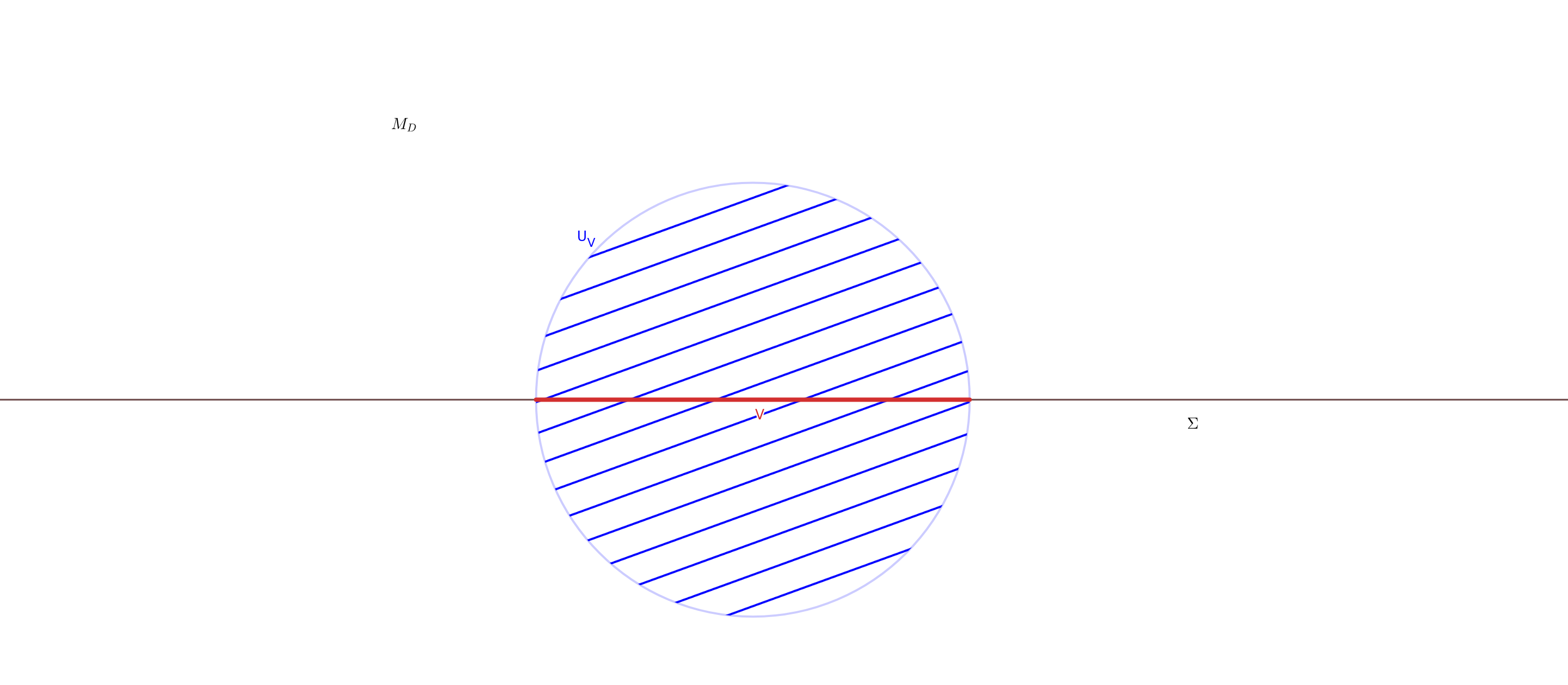}
 \end{figure}
\begin{proof}
The idea of the proof is the following: we take for all $f_1, \cdots, f_n \in \mathcal{C}^{\infty}_0(U_{\mathcal{V}}, \mathbb{C})$ and all $n \in \mathbb{N}$ any vector $\ket{\chi} \in \mathcal{H}$ such that:
\begin{equation*}
    0 = \bra{\chi} \phi_{f_1} \cdots \phi_{f_n} \ket{\Omega} \iff 0 = \bra{\chi} \phi(x_1) \cdots \phi(x_n) \ket{\Omega}
\end{equation*}
i.e. that is in the orthogonal complement in $\mathcal{H}$ of the space generated by the field polynomial for $x_1, \dots, x_n \in U_{\mathcal{V}}$, then we want to show that such $\ket{\chi}$ is trivial, proving that the orthogonal complement contains just the trivial vector.\\
Let us start by considering a future directed timelike vector $\mathbf{t}$ and let $u \in \mathbb{R}$. Consider the following timelike translation of the $n$-th coordinate: $x_n \mapsto x_n + u \mathbf{t}$. Consider then the function:
\begin{equation*}
    g(u) := \bra{\chi} \phi(x_1) \cdots \phi(x_n+ u\textbf{t}) \ket{\Omega}
\end{equation*}
But, from the axioms of an AQFT that we have listed before, we know that the fields must transform covariantly:
\begin{equation*}
    \phi(x_n+ u\textbf{t}) = \exp\big( i H u \big) \phi(x_n) \exp \big( -i H u\big)
\end{equation*}
Where $\exp\big( i H u \big)$ is a time translation operator and $H$ is the self adjoint operator that is positive semidefinite such that $H \ket{\Omega} = 0$. It follows that:
\begin{equation*}
    g(u) = \bra{\chi} \phi(x_1) \cdots \exp\big( i H u \big) \phi(x_n) \ket{\Omega}
\end{equation*}
Since $0 = \bra{\chi} \phi(x_1) \cdots \phi(x_n) \ket{\Omega}$ as long as $x_1, \dots, x_n \in U_{\mathcal{V}}$, we have that $g(u) = 0$ as long as $u$ is small. So $g(u) = 0$ for $u \in I =[-\varepsilon, \varepsilon]$.\\
If we now displace $u$ in the complex plane, as long as the displacement is in the upper half plane, the function $g(u)$ is holomorphic. This can be seen from the $u$-dependence of $g(u)$ that is just in the exponential factor that, by the spectral theorem, can be decomposed as:
\begin{equation*}
    \exp\big( i H u \big) = \int_{\mathbb{R}^+} \exp\big( i \lambda u \big) dE_{\lambda}
\end{equation*}
But the function $\exp\big( i \lambda u \big)$ is holomorphic for $\lambda > 0$ and $\Im(u)> 0$ since:
\begin{equation*}
    \exp\big( i \lambda u \big) = \exp\big( i \lambda \Re(u) \big) \exp\big(- \lambda \Im(u) \big)
\end{equation*}
and if we study the holomorphicity of this function in the upper half plane, we see that there're no pole. So, from Cauchy integral theorem, whenever we integrate on a closed curve in the upper half plane the function $\exp\big( i \lambda u \big)$, we get zero. But then, by Morera's theorem, this implies that $g(u)$ is holomorphic in the upper half plane.
\begin{figure}[H]
 		\centering
 		\includegraphics[trim = 200 0 200 50,clip,width=0.60	\columnwidth]{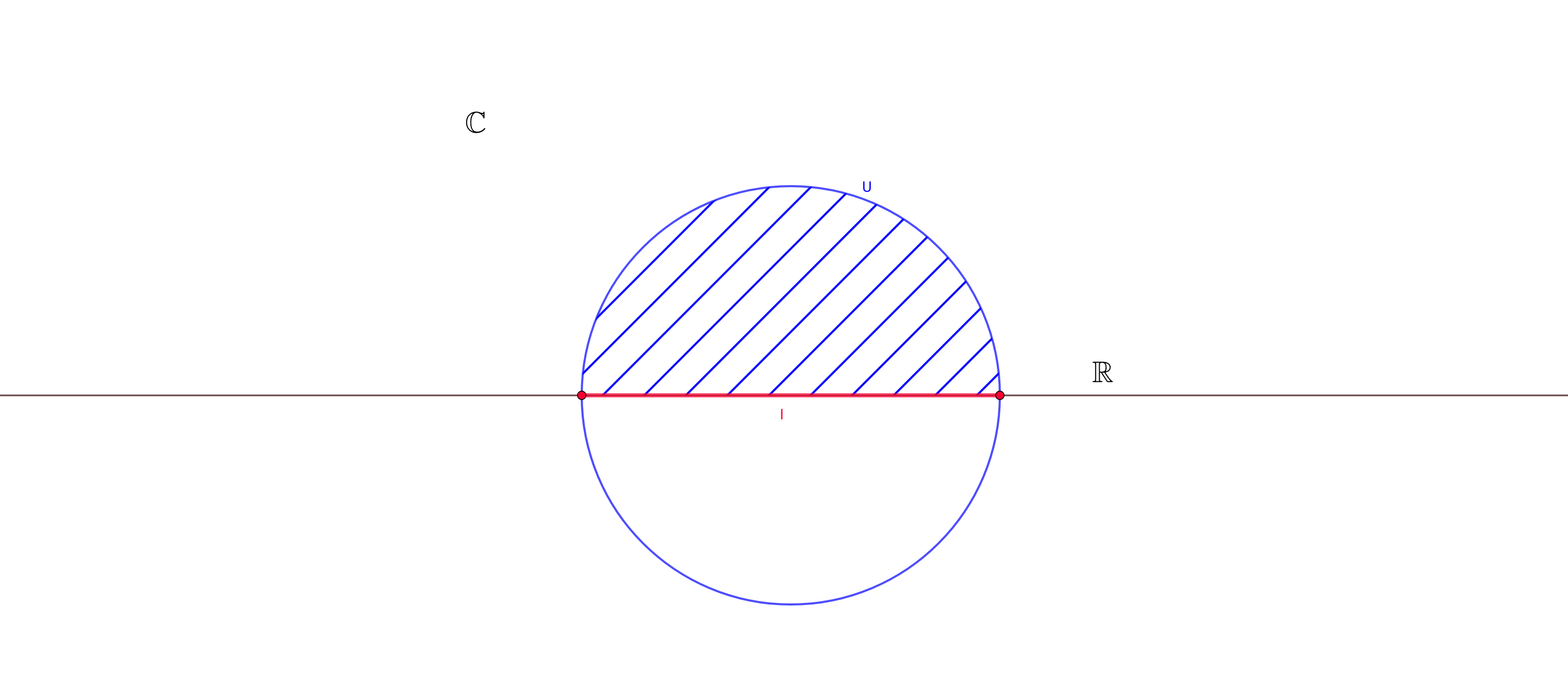}
 		\caption{The choice of a circle is random, we should have chosen any other region symmetric w.r.t. the real line, also a square or a ellipse}
\end{figure}
We can now use a Lemma of complex analysis:
\begin{lem}[\textbf{Schwarz reflection principle}]\label{app: Schw}
Let $U$ be a region symmetric about the real axis. If $f(z)$ is a holomorphic function in the part of $U$ in the upper half plane and is such that:
\begin{equation*}
    \lim_{Im(z) \to 0+} f(z) = 0
\end{equation*}
Then $f(z)$ extends to a holomorphic function on $U$
\end{lem}
\begin{proof}
See Section $11.10$ of \cite{Arfken}.
\end{proof}
This lemma allows us to say that the function $g(u)$ is not just holomorphic in the upper half plane, but has domain of analiticity that extends through the real axis also to a piece of the lower half plane, provided that the symmetric region is taken with respect to $I$ or a subset of it. As a consequence, $g(u)$ for $u \in I$ is a holomorphic function that vanishes identically on the closed interval $I$. A holomorphic function that vanishes identically on a closed interval of the real line must be identically zero for all $u \in \mathbb{R}$, so $g(u) = 0$. This means that, no matter how much we move in the timelike direction $\mathbf{t}$ the coordinate $x_n$, the inner product $g(u)$ is identically zero.\\
This was just a displacement in a chosen timelike direction of $x_n$, if we want to extend this for any $x_n \in M_D$ we can compose timelike translations (either future or past directed) to reach any point in $M_D$
\begin{figure}[H]
 		\centering
 		\includegraphics[trim = 250 0 250 0,clip,width=0.70	\columnwidth]{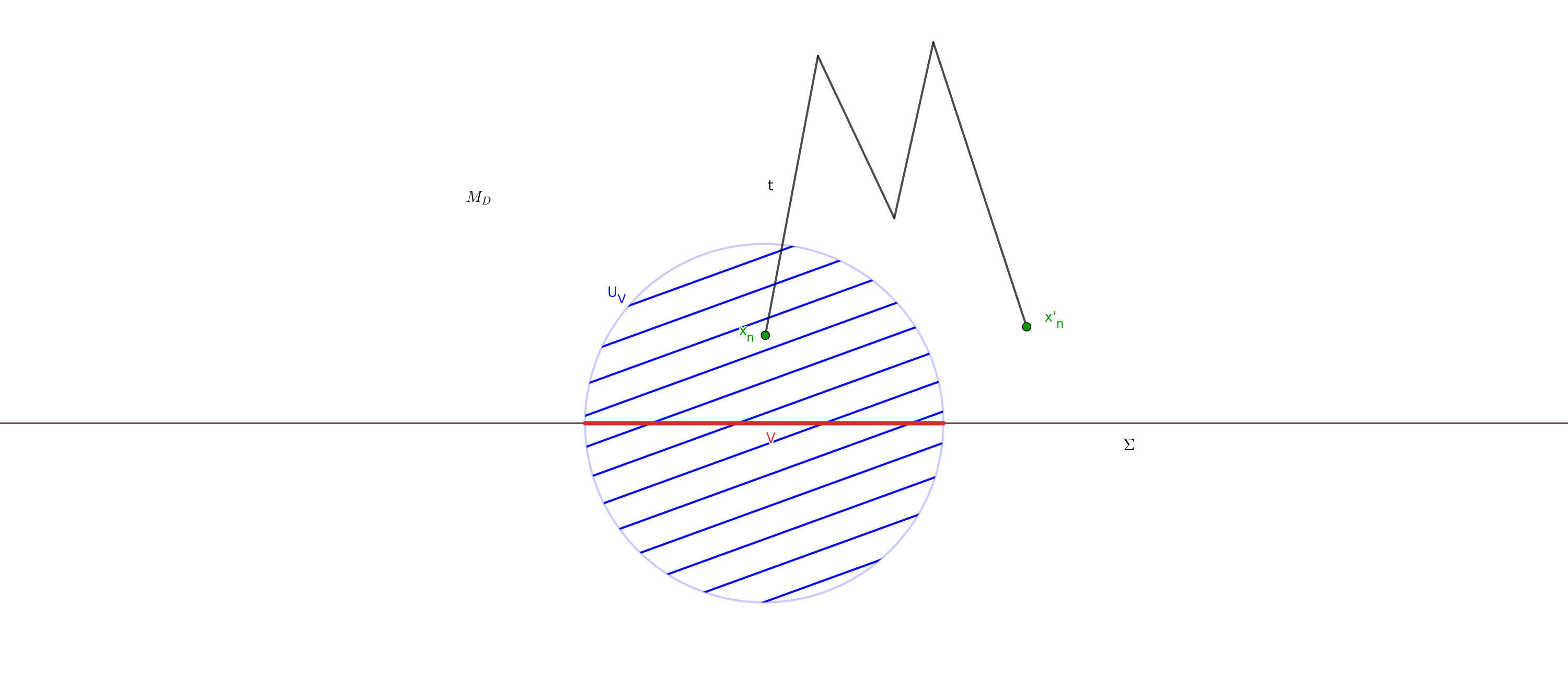}
\end{figure}
Where in the case in which the timelike translation is past directed we will perform the holomorphicity argument starting from the lower half plane. In this way we have shown that, independently from the choice of $x_n \in M_D$: $ 0 = \bra{\chi} \phi(x_1) \cdots \phi(x_n) \ket{\Omega}$.\\
We can now repeat the argument for all the other coordinates in an iterative way. Namely, if we consider now the coordinate $x_{n-1}$ and another timelike vector $\mathbf{t}'$ we perform the following timelike translation:
\begin{equation*}
    x_{n-1} \mapsto x_{n-1} + u\mathbf{t}' \hspace{20pt} x_{n} \mapsto x_{n} + u\mathbf{t}'
\end{equation*}
Being aware that for any $x_n$ we have $ 0 = \bra{\chi} \phi(x_1) \cdots \phi(x_n) \ket{\Omega}$. Then, repeating the same above arguments, we prove that also the choice of $x_{n-1} \in M_D$ does not affect the above quantity. Repeating the argument for all $n \in \mathbb{N}$ we have:
\begin{equation*}
     0 = \bra{\chi} \phi(x_1) \cdots \phi(x_n) \ket{\Omega} \,\,\,\,\,\, \forall x_1, \dots, x_n \in M_D
\end{equation*}
Showing in this way, from axiom \textbf{A4} of an AQFT, that $\phi_{f_1} \cdots \phi_{f_n} \ket{\Omega}$ for $f_1, \cdots, f_n \in \mathcal{C}^{\infty}_0(U_{\mathcal{V}}, \mathbb{C})$ generates a dense subset of $\mathcal{H}$
\end{proof}
Generalizations of the theorem on curved spacetimes were given by Verch in \cite{Verch:1992pn} for the case of a quasifree ground-state of a free scalar massive field theory on an ultrastatic spacetime, by Strohmaier \cite{Strohmaier_2000} for general local Quantum field theories on stationary spacetimes and by \cite{Sanders_2009} in the case of spacetimes diffeomorphic to spacetimes with the Reeh-Schlieder property.
\newpage

\section{The Bisognano-Wichmann theorem}\label{app: BW}
The Bisognano-Wichmann theorem is a statement regarding the algebras of a general QFT localized in wedge-like regions. In this extent is one of the few known cases in which the modular flow of the theory is said to be geometric as, due to this theorem, is expressed in terms of a geometric transformation on the wedge. This result was first proven to hold by the authors for the free scalar case in \cite{Bisognano:1975ih} and generalized later for arbitrary QFT by the same authors in \cite{Bisognano:1976za}. In both cases, the algebras are localized on Minkowski spacetime, a generlization of this result to more general curved backgrounds admitting wedge-like regions was obtained in \cite{Brunetti:2002nt}.\\
In this section I will just present the original version of the theorem for arbitrary Quantum field theories on Mikowski spacetimes and refer to the cited litterature for its generalization to curved backgrounds.\\\\
Let us start recalling the definition of Rindler wedges on Minkowski spacetime $M_D$:
\begin{equation*}
    W_R = \{ x \in M_D : x_1 > |x_0| \}
\end{equation*}
called the right wedge, and:
\begin{equation*}
    W_L = \{ x \in M_D : x_1 < -|x_0| \}
\end{equation*}
called the left wedge. We will denote as usual the local algebras of operators as $\mathfrak{A}(\mathcal{O})$ where $\mathcal{O} \subset M_D$. These are assumed to be $C^*$-algebras, thus, when represented, give rise to von Neumann algebras still denoted as $\mathfrak{A}(\mathcal{O})$ for simplicity (see Theorem \ref{thm: GNS}).\\
Moreover, if the local algebras have a $\mathbb{Z}_2$-grading we can define the twisted algebra:
\begin{equation*}
    \mathfrak{A}(\mathcal{O})^{t'} = \{ ZXZ^{-1} | X \in \mathfrak{A}(\mathcal{O}) \}
\end{equation*}
Where $Z$ is the twisting operator we defined in Definition \ref{def: Z2grad}. In particular, from the discussion we had after the Definition \ref{def: Z2grad}, we see that the only difference occurs if the net of local algebrs $\mathfrak{A}(\mathcal{O})$ is homogeneous Fermi because in the Bose case: $\mathfrak{A}(\mathcal{O})^{t'} = \mathfrak{A}(\mathcal{O})$.\\
Considering the von Neumann algebras $\mathfrak{A}(W_R)$ and $\mathfrak{A}(W_L)^{t'}$, recalling that from Tomita-Takesaki modular theory on von Neumann algebras we have the existence of modular conjugation and flow, the statement of the Bisognano-Wichmann theorem is:
\begin{thm}[\textbf{Bisognano-Wichmann}] \label{thm: BW}
If $J_{W_R}$ and $\Delta_{W_R}$ denote the modular operators for the pair $(\mathfrak{A}(W_R), \Omega)$, then:
\begin{equation*}
    J_{W_R} = Z U(R_{23}(\pi),0) \Theta \hspace{20pt} \Delta_{W_R}^{it} = U(\Lambda_{W_R}(t),0)
\end{equation*}
Where: $U(R_{23}(\pi),0)$ is the unitary representation over the Hilbert space of a rotation of $\pi$ in the plane of the free coordinates, $U(\Lambda_{W_R}(t),0)$ is the unitary representation of a boost of parameter $t$ in the direction $x_1$ and $\Theta$ is a CPT transformation such that:
\begin{align*}
    \Theta^2 &= U(-\mathbb{1}, 0)\\
    \Theta \Omega &= \Omega\\
    \Theta U(g,x) \Theta^{-1} &= U(g, -x)
\end{align*}
for $g \in Spin^0_{1,3}$ and $x \in M_D$.\\
Finally, it holds $J_{W_R} \mathfrak{A}(W_R) J_{W_R} = \mathfrak{A}(W_L)^{t'}$
\end{thm}
\begin{proof}
For a proof of the theorem I refer to Lemma $7$ and Theorem $1$ in the original work \cite{Bisognano:1976za}. Notice that, in the original work, the authors deal with field algebras that are not von Neumann algebras. For this reason, their proof starts by showing the existence of the $J_{W_R}$ operator of the above form. With that, they prove that for any $X$ (possibly unbounded):
\begin{equation*}
    J_{W_R} U(\Lambda_{W_R}(i \pi),0) X \Omega = X^* \Omega
\end{equation*}
That, in the case in which the algebras are turned into von Neumann, in analogy with Tomita-Takesaki modular theory, gives the claimed result.
\end{proof}
The result for the modular conjugation $J_{W_R}$, modulo a twisting factor $Z$, shows that is a $CRT$ symmetry. This is the transformation in coordinate: $t \to -t$, $x_1 \to -x_1$ and $x_i \to x_i$ for all $i > 2$. In this way becomes more intuitive, from a geometrical point of view , that $J_{W_R} \mathfrak{A}(W_R) J_{W_R} = \mathfrak{A}(W_L)^{t'}$ as such a transformation of coordinates maps exactly $W_R \to W_L$. For this reason, and the form of the modular flow, that this theorem is often mentioned as giving a geometric action of the modular theory.
\newpage

\section{Clifford algebra}\label{app}
This section is devoted to a short survey of Clifford algebras that will allow us to study the universal covering group of the Lorentz group, used to define and relate Spin and Frame bundles. Let $\mathbb{R^{r,s}}$ be the real vector space of dimension $n = r+s$ and equip it with a non-degenerate bilinear form $\Omega_{ab}$ that has $r$ many positive and $s$ many negative eigenvalues. A special example of this is Minkowski spacetime $M_D = \mathbb{R}^{1,3}$, where the bilinear form is the usual Minkowski metric tensor that in the orthonormal basis $\{l_a\}_{a=0,\dots,3}$ is: $\eta = \mathrm{diag}(1,-1,-1,-1)$.
\begin{defn}
The Clifford algebra $Cl_{r,s}$ of $\mathbb{R}^{r,s}$ is defined as the associative unital algebra generated by an orthonormal basis $e_a$ of $\mathbb{R}^{r,s}$ subject to the Clifford relations\footnote{One can show that the definition of the Clifford relations is independent on the choice of the basis $e_a$ see \cite{Lawson} Prop. 1.1}:
\begin{equation*}
    e_a e_b + e_b e_a = 2 \Omega_{ab} \mathbb{1}
\end{equation*}
We can identify the subspaces of $Cl_{r,s}$ spanned by monomials of even or odd degree in the basis vectors, and call them $Cl_{r,s}^0$ resp. $Cl_{r,s}^1$.\\
We define the \textbf{Dirac algebra} as the Clifford algebra of Minkowski: $D:=Cl_{1,3}$. Therefore the Clifford condition becomes in this case:
\begin{equation*}
    l_a l_b + l_b l_a = 2 \eta_{ab} \mathbb{1}
\end{equation*}
Where now $\{l_a\}_{a=0,1,2,3}$ is an orthonormal basis of Minkowski spacetime.
\end{defn}
Recognize, that the real vector space itself is a subspace of the algebra, namely $\mathbb{R}^{r,s} \subset Cl_{r,s}$ generated by the monomials of order $1$ in the basis $e_a$.
\begin{rem}
The even subspace $Cl_{r,s}^0$ is a subalgebra as it is closed under product of its elements, while for the odd case the product of two elements will give an even monomial. Despite this, both subspaces are well defined, as the Clifford relations are purely even and thus, independently from the order of the monomials, we have a sum of even monomials.
\end{rem}
\begin{rem}
The dimension of the Clifford algebra is $2^{r+s}$. To see it consider as basis of $\mathbb{R}^{r+s}$ an orthogonal basis $e_a$, then the Clifford relations give:
\begin{align*}
    e_i e_j &= -e_j e_i\\
    e_i^2 &= \Omega_{ii}\mathbb{1}
\end{align*}
but this implies that, in searching for all the independent monomials, we can define a unique independent order and we cannot have repetitions of $e_j$. Then:
\begin{equation*}
    \dim Cl_{r,s} = \sum_{k=0}^{r+s} \begin{pmatrix}
    r+s\\
    k
    \end{pmatrix} = 2^{r+s}
\end{equation*}
\end{rem}
As it is customary in introducing the Dirac fields on Minkowski spacetime, and from the relations with Dirac $\gamma$ matrices that we will introduce later, we shall call the volume element $l_5 := l_0l_1l_2l_3$.\\
The real Dirac algebra $D$ can be represented as complex matrices via a complex representation: $\pi: D \to M(n,\mathbb{C})$ for some $n \in \mathbb{N}$. We look for such a representation because, in treating Dirac fields, we want the Dirac algebra to act on spinor fields that are vectors in a $\mathbb{C}^n$. For this purpose, we quote a theorem due to Pauli (see \cite{Pauli:1936gd}) about the representation theory of the Dirac algebra:

\begin{thm}[\textbf{Fundamental Theorem}]\label{thm: fund}
The Dirac abstract algebra $D$ is simple\footnote{A simple abstract algebra $\mathfrak{A}$ is an algebra for which each homomorphism that has as domain the entire algebra $\mathfrak{A}$ is injective} and has a unique irreducible complex representation up to equivalence. This representation is denoted as:
\begin{align*}
    \pi_0: D &\to M(4,\mathbb{C})\\
    l_a &\mapsto \pi_0(l_a) =: \gamma_a
\end{align*} 
that gives the famous Dirac gamma matrices $\gamma_a$:
\begin{equation*}
    \gamma_0 := \begin{pmatrix}
    0 & \mathbb{1}_{2\times 2}\\
    \mathbb{1}_{2\times 2} & 0
    \end{pmatrix} , \hspace{15pt} \gamma_i := \begin{pmatrix}
    0 & -\sigma_i\\
    \sigma_i & 0
    \end{pmatrix} 
\end{equation*}
Where $\sigma_i$ are the Pauli matrices.\\
The equivalence with another complex irreducible representation $\pi$ of $D$ is impllemented by $\pi(S) = L \pi_0(S) L^{-1}$ for all $S \in D$, where $L \in GL(4,\mathbb{C})$ unique up to a non-zero complex factor. 
Define, for notational convenience, also $\gamma_5 := \pi_0(l_5)$
\end{thm}
This theorem, defines and gives an explicit form to the famous Dirac $\gamma$ matrices, that were used in defining the Dirac derivative in Section \ref{sec: DiMaj}. Furthermore, in that Section, the notions of adjoint and charge conjugation of Dirac spinors were used starting from matrices $A,C \in GL(4,\mathbb{C})$. Now, we define them and list their properties:
\begin{defn}\label{def: hacc}
Let $\pi$ be an irreducible complex representation of the Dirac algebra. We define the matrices $A,C \in GL(4,\mathbb{C})$ via the conditions:
\begin{align*}
    A &= A^* ,\hspace{15pt} \pi(l_a)^* = A \pi(l_a) A^{-1}, \hspace{15pt} A \pi(n) = A n^a \gamma_a > 0\\
    \overline{C} C &= \mathbb{1}, \hspace{15pt} -\overline{\pi(l_a)} = C \pi(l_a) C^{-1}
\end{align*}
for all future pointing timelike vectors $n$ and the Hermitean conjugation and adjoint are the standard ones on $\mathbb{C}^4$.
\end{defn}
\begin{rem}\label{rem: dhacc}
As both sets of gamma matrices $\pi(l_a)^*$ and $-\overline{\pi(l_a)}$ satisfy the Clifford relations, from the above theorem, the matrices $A,C$ are uniquely determined up to a multiplicative constant.
\end{rem}
In fact, one can prove that such matrices always exist given a complex irreducible representation:
\begin{thm}\label{thm: ex}
For any complex irreducible representation $\pi$ of $D$, there are $A,C \in GL(4,\mathbb{C})$ which satisfy properties of Definition \ref{def: hacc} with respect to $\pi$. We also have $A = - C^* \overline{A} C$.\\
Moreover, $A$ is uniquely determined up to a positive factor while $C$ up to a phase factor.\\
Let $A_i,C_i \in GL(4, \mathbb{C})$ for $i = 1,2$ satisfying Definition \ref{def: hacc} with respect to irreducible complex representations $\pi_i$ of $D$. Then there exist $L \in GL(4,\mathbb{C})$, unique up to a sign, such that $L^*A_1 L = A_2$, $\overline{L}^{-1}C_1 L = C_2$ and $\pi_2 = L^{-1} \pi_1 L$ on $D$
\end{thm}
\begin{proof}
We start proving the existence for the representation $\pi_0$. For that, take $A_0 := \gamma_0$ and $C_0 := \gamma_2$. Then, we need to check the properties in Definition \ref{def: hacc} to be fulfilled. First of all, as $\pi_0(l_a) = \gamma_a$, for what concerns $A_0$, we have:
\begin{align*}
    \gamma_0 &= (\gamma_0)^*\\
    \gamma_a^* &= - \gamma_a = \gamma_0 \gamma_a \gamma_0\\
    \gamma_0 n^a \gamma_a &= \begin{pmatrix}
    n^0 \mathbb{1} + n^i \sigma_i & 0\\
    0 & n^0 \mathbb{1} - n^i \sigma_i
    \end{pmatrix} > 0
\end{align*}
Where the last inequality follows from $\det(n^0 \mathbb{1} \pm n^i \sigma_i) = n_0^2 - |\mathbf{n}| = 1$ and also $\Tr(n^0 \mathbb{1} \pm n^i \sigma_i) = 2n^0 > 0$. For what concerns $C_0$:
\begin{equation*}
    \overline{\gamma_2} \gamma_2 = \begin{pmatrix}
    0 & \sigma_2\\
    -\sigma_2 & 0
    \end{pmatrix} \begin{pmatrix}
    0 & -\sigma_2\\
    \sigma_2 & 0
    \end{pmatrix} = \mathbb{1}
\end{equation*}
and also:
\begin{align*}
    \gamma_2 \gamma_a \gamma_2^{-1} &= (-\gamma_a \gamma_2 + 2 \eta_{a2}) \gamma_2^{-1}\\
    &= -\gamma_a + 2 \eta_{a 2} \gamma_2^{-1}
\end{align*}
that, for $a \neq 2$, is precisely $-\overline{\gamma_a}$ from the reality of gamma matrices and as the second term vanishes. For $a = 2$, we have that $-\overline{\gamma_2} = \gamma_2$ and since $\gamma_2^{-1} = - \gamma_2$:
\begin{equation*}
    -\gamma_2 - 2 \gamma_2^{-1} = -\gamma_a + 2 \gamma_2 = -\overline{\gamma_2}
\end{equation*}
In this same representation, we can also compute the relation between the $A$ and $C$ matrices, claimed in the statement of the theorem, to hold:
\begin{equation*}
    -C_0^* \overline{A_0} C_0 = \gamma_2 \gamma_0 \gamma_2 = - \gamma_0 \gamma_2^2 = \gamma_0 = A_0
\end{equation*}
In the case of a general complex representation $\pi$, we use the fundamental theorem to write $\pi(l_a) = K^{-1} \gamma_a K$ for $K \in GL(4, \mathbb{C})$. But then, relating $A = K^* A_0 K$ and $C = \overline{K}^{-1} C_0 K$ and using the proved properties for $A_0, C_0$, we have that the thus defined $A$ and $C$ still satisfy Definition \ref{def: hacc}:
\begin{align*}
    (K^* A_0 K)^* &= K^* A_0 K\\
    A \pi(l_a) A^{-1} &= K^* A_0 K K^{-1} \gamma_a K K^{-1} A_0^{-1} (K^{-1})^*\\
    &= K^* \gamma_a^* (K^{-1})^* = \pi(l_a)^*\\
    A \pi(n) &= K^* A_0 \pi_0(n^a l_a) K > 0\\
    \overline{(\overline{K}^{-1} C_0 K)} \overline{K}^{-1} C_0 K &= K^{-1} \overline{C}_0 C_0 K\\
    &= K^{-1}  K = \mathbb{1}\\
    C \pi(l_a) C^{-1} &= \overline{K}^{-1} C_0 K K^{-1} \gamma_a K K^{-1} C_0^{-1} \overline{K}\\
    &= \overline{K}^{-1} C_0 \gamma_a C_0^{-1} \overline{K}\\
    &= -\overline{K}^{-1} \overline{\gamma_a} \overline{K} = - \overline{\pi(l_a)}
\end{align*}
and they also satisfy the general relation:
\begin{align*}
    -C^* \overline{A} C &= - K^* C_0^* (\overline{K}^{-1})^* \overline{K}^* \overline{A}_0 \overline{K} \overline{K}^{-1} C_0 K\\
    &= - K^* C_0^* \overline{A}_0 C_0 K\\
    &=  K^* A_0 K = A
\end{align*}
For the statement regarding the uniqueness, we start noticing that the matrices $A$ and $C$ are determined uniquely up to non-zero complex factors as noticed in the remark \ref{rem: dhacc} that we may call $a,c$. Because $A = A^*$ we must have $a \in \mathbb{R}$ and because of $\overline{C} C = 1$ we have $|c| = 1$ proving already that $C$ is determined up to a phase factor. Moreover, as we also have $A \pi(n) > 0$ we must have $a >0$.\\
For the last part, fix a $K \in GL(4,\mathbb{C})$ such that $\pi_1 = K\pi_2 K^{-1}$ by the fundamental theorem. Set $A'_2 := K^* A_1 K$ and $C'_2 := \overline{K}^{-1} C_1 K$. Then, with respect to $\pi_2$:
\begin{align*}
    A'_2 \pi_2(l_a) A'_2{}^{-1} &= K^* A_1 \pi_1(l_a) A_1^{-1} (K^*)^{-1}\\
    &= K^* \pi_1(l_a)^* (K^*)^{-1}\\
    &= (K^{-1} \pi_1(l_a) K)^* = \pi_2(l_a)^*\\
    A'_2 \pi_2(n) &= K^* A_1 \pi_1(n) K > 0\\
    C'_2 \pi_2(l_a) C'_2{}^{-1} &= \overline{K}^{-1} C_1 \pi_1(l_a) C_1^{-1} \overline{K}\\
    &= \overline{K}^{-1} (- \overline{\pi_1(l_a)}) \overline{K} = - \overline{\pi_2(l_a)}
\end{align*}
So, as also $A'_2$ and $C'_2$ satisfy \ref{def: hacc} with respect to $\pi_2$ as well as $A_2$ and $C_2$, by the above uniqueness: 
\begin{equation*}
    A'_2 = a A_2 \hspace{20pt} C'_2 = c C_2
\end{equation*}
for $a > 0$ and $|c| = 1$. Then, the desired matrix $L \in GL(4,\mathbb{C})$ must be $L = z K$ for some $z \neq 0$ still by the fundamental theorem. In particular to have the right intertwining relations:
\begin{align*}
    A_2 &= L^* A_1 L = |z|^2 K^* A_1 K = |z|^2 A'_2\\
    C_2 &=\overline{L}^{-1} C_1 L = \overline{z}^{-1}z \overline{K}^{-1} C_1 K = \overline{z}^{-1}z C'_2
\end{align*}
we must have $|z|^2 = a$ and $z = c \overline{z}$ that fixes $z$ up to a sign.
\end{proof}

We are now in the position of introducing the universal covering group of the Lorentz group, using the elegant formalism just introduced:
\begin{defn}
The Pin and Spin groups of $Cl_{r,s}$ are defined as:
\begin{align*}
    Pin_{r,s} &:= \{ S\in Cl_{r,s} | S=u_1 \cdots u_k, \hspace{15pt} k \in \mathbb{N}, \hspace{15pt} u_i \in \mathbb{R}^{r,s}, \hspace{15pt} u_i^2 = \pm \mathbb{1} \}\\
    Spin_{r,s} &:= Pin_{r,s} \cap Cl^0_{r,s}
\end{align*}
\end{defn}
In fact, Pin is a group: it certainly contains the inverse,as $u_i^2 = \pm \mathbb{1}$, is closed under composition, the product operation is by definition associative and the neutral element is given by $S = \mathbb{1}$. This statement, can be seen also from the following equivalent characterization\footnote{The notion of determinant and trace for elements in $D$, are defined via the trace and determinant of the corresponding images under the unique, up to equivalence, complex representation of the Dirac algebra}:
\begin{prop}
\begin{equation*}
    Pin_{1,3} = \{ S\in D| \det X = 1; \,\, \forall v \in M_D \,\,\,\, SvS^{-1} \in M_D\}
\end{equation*}
\end{prop}
In order to prove this proposition, we first need the following lemma:
\begin{lem}
We have $l_5^2 = -\mathbb{1}$ and:
\begin{equation} \label{eq: Cl1}
    l_5 v l_5^{-1} = -v l_5 l_5^{-1} = -v, \hspace{15pt} v \in M_D
\end{equation}
Moreover, if $u \in M_D$ has $u^2 = \| u \|^2 \mathbb{1} \neq 0$, with norm in $M_D$, then $u^{-1} = \frac{1}{\| u \|^2}u$ and $v \mapsto -uvu^{-1}$ defines a reflection of $M_D$ in the hyperplane perpendicular to $u$
\end{lem}
\begin{proof}[Proof of lemma]
From the Clifford relations, we have $l_5 l_a = - l_a l_5$ for each $a$, as one always needs to perform $3$ switches. This implies that Eq. \eqref{eq: Cl1} holds, and by direct computation one sees $l_5^2 = - \mathbb{1}$.\\
Finally, compute:
\begin{equation*}
    -uvu^{-1} = v - (uv + vu)u^{-1} =v - \frac{2(u,v)}{\| u\|^2}u 
\end{equation*}
Where we have used the Clifford relations and denoted with $(\cdot,\cdot)$ the inner product on $M_D$.
\end{proof}
\begin{proof}[Proof of proposition]
Whenever we take an $S \in Pin_{1,3}$ the map $v \mapsto SvS^{-1}$ on $M_D$ is just a product of reflections from the above lemma, so $SvS^{-1} \in M_D$ for all $M_D$. Now, from the identity $u^2 = \| u \|^2 $ for all $u \in M_D$:
\begin{align*}
    \det u^2 &= \det \big( \| u \|^2 \mathbb{1} \big)\\
    (\det u)^2 &= \| u \|^8 \det \big( \mathbb{1} \big)\\
    \det u &= \| u \|^4 
\end{align*}
Then as $\| u_i \| =1$ for the $u_i$ in the definition of the $Pin$ group, we have: $\det S = 1$.\\
For the converse suppose that $S \in D$ is such that $\det S = 1$ and $SvS^{-1} \in M_0$ for all $v \in M_0$. If we look at:
\begin{equation*}
    v u + u v = 2 v^a u_a \mathbb{1}
\end{equation*}
we notice that the right hand side of this remains invariant under the adjoint action of $S$ on the vectors. As a consequence $S$ preserves the Minkowski inner product, hence the adjoint action of $S$ determines a Lorentz transformation $\Lambda$. In general, a Lorentz transformation can always be written as a finite product of reflections in non-null hyperplanes (see Theorem $3.20$ of \cite{Artin}), so we can take them be $v \mapsto -u_ivu_i$ as this is a reflection from previous lemma. Define then:
\begin{equation*}
    T := \left\{ \begin{aligned}
    u_k \cdots u_1 \hspace{15pt} \mathrm{if} \,\, k \,\, \mathrm{is \,\, even}\\
    u_k \cdots u_1 l_5 \hspace{15pt} \mathrm{if} \,\, k \,\, \mathrm{is \,\, odd}\\
    \end{aligned}\right.
\end{equation*}
By definition we have $T \in Pin_{1,3}$ and from the above lemma:
\begin{equation*}
    T^{-1} := \left\{ \begin{aligned}
    u_1^{-1} \cdots u_k^{-1} \hspace{15pt} \mathrm{if} \,\, k \,\, \mathrm{is \,\, even}\\
    -l_5 u_1^{-1} \cdots u_k^{-1} \hspace{15pt} \mathrm{if} \,\, k \,\, \mathrm{is \,\, odd}\\
    \end{aligned}\right.
\end{equation*}
So in both cases, still from the above lemma:
\begin{align*}
    TvT^{-1} &= \left\{ \begin{aligned}
    &u_k \cdots u_1 v u_1^{-1} \cdots u_k\\
    &-u_k \cdots u_1 l_5 v l_5 u_1^{-1} \cdots u_k^{-1} = u_k \cdots u_1 v u_1^{-1} \cdots u_k^{-1}
    \end{aligned}\right.\\
    &= \Lambda(v) = S v S^{-1}
\end{align*}
Where we used the fact that $v \mapsto -u_i v u_i$ is a reflection and $u_i^{-1} = \pm u_i$. In particular, it follows that $T^{-1} S l_a S^{-1} T = l_a$, so, by Theorem \ref{thm: fund} we must have $T^{-1} S = c \mathbb{1}$ and $S = c T$ for some nonzero $c \in \mathbb{C}$. Furthermore, one can prove that $M(4,\mathbb{C}) \simeq \mathbb{C} \otimes_{\mathbb{R}} D$ (see \cite{Lawson} Section $4$ together with Theorem $3.7$). Therefore, as $S,T \in D$ we must have that they differ just by a real constant, i.e. $c \in \mathbb{R}$. Moreover, as $\det S = 1$ and one can compute from its definition that $\det T = 1$, it follows $T = \pm S$. Finally, as also $-T = (l_5)^2 T \in Pin_{1,3}$, we must have  also $S \in Pin_{1,3}$
\end{proof}
From this proposition, that characterizes the Pin group, one can prove that $Pin_{1,3}$ and $Spin_{1,3}$ are indeed Lie groups. This can be seen from the identification we mentioned in the above proof:
\begin{equation} \label{eq: decdir}
    M(4,\mathbb{C}) \simeq \mathbb{C} \otimes_{\mathbb{R}} D
\end{equation}
That implies that $D$ is the subset of square matrices with unit determinant. As a Lie group, let us denote by $Spin^0_{1,3}$ the connected component of $Spin_{1,3}$ containing the identity.\\
We are now ready to relate the Lie group $Pin_{1,3}$ with the Lorentz group, by defining a map:
\begin{align*}
    \Lambda : Pin_{1,3} &\to \mathcal{L}\\
    S &\mapsto \Lambda^a_{\,\,b}(S)
\end{align*}
such that: $Sl_bS^{-1} = l_a \Lambda^a_{\,\,b}(S)$. The matrix $\Lambda^a_{\,\,b}(S)$ exists and is a Lorentz transformation by the previous proposition. At this point, one can show (see Theorem $I.2.10$ in \cite{Lawson}):
\begin{prop}\label{prop: lieal}
The map $\Lambda$, is a surjective double covering homomorphism of Lie groups, which restricts to a double covering homomorphism $Spin_{1,3}^0 \to \mathcal{L}^{\uparrow}_+$. 
We have:
\begin{align*}
    \Lambda^a_{\,\, b}(S) &= \frac{1}{4}\eta^{ac}\Tr(l_c S l_b S^{-1})\\
    \Lambda^a_{\,\, b}(S^{-1}) &= \eta^{ac}\eta_{bd} \Lambda^d_{\,\,c}(S)\\
    (d\Lambda)^{-1}(\lambda^b_{\,\,a}) &= \frac{1}{4} \lambda^b_{\,\, a} \eta^{ac}l_b l_c
\end{align*}
Where $d\Lambda: \mathfrak{lie}(Spin_{1,3}^0) \to \mathfrak{lie}(\mathcal{L}^{\uparrow}_+)$.
\end{prop}
Respectively, one can see that $Pin_{1,3}$, $Spin_{1,3}$ and $Spin_{1,3}^0$ are the universal coverings of $\mathcal{L}$, $\mathcal{L}_+$ and $\mathcal{L}_+^{\uparrow}$.\\
Another result, that was needed in introducing the Dirac bundles is the following:
\begin{lem}\label{lem: tec1}
Let $\pi$ be a complex irreducible representation of $D$ and let $A,C \in GL(4,\mathbb{C})$ be as in Definiiton \ref{def: hacc}. Then for all $S \in Spin_{1,3}^0$ we have:
\begin{equation*}
    \pi(S)^*A\pi(S) = A, \hspace{15pt} \pi(S^{-1})C^{-1}\overline{\pi(S)} = C^{-1}
\end{equation*}
\end{lem}
\begin{proof}
Consider a unit vector $u = u^al_a$, then we have $u^2 = \pm \mathbb{1}$. Hence:
\begin{align*}
    \pi(u)^*A \pi(u) &= u^a u^b \pi(l_a)^* A \pi(l_b) \\
    &= u^a u^b A \pi(l_a l_b)\\
    &= A \pi(u^2) = \pm A
\end{align*}
But then, from the definition of $S \in Pin_{1,3}$ in terms of vectors $u_i$ with these same properties we must have: $\pi(S)^* A \pi(S) = \pm A$. Of course, if we take $S=\mathbb{1}$ we will get a plus sign. But, as $Spin^0_{1,3}$ is the connected component of the Lie group containing the identity, by continuity we must have that also for all $S \in Spin^0_{1,3}$ we must keep having the plus sign.\\
For what concerns the charge conjugation, notice that for $u \in M_D$:
\begin{equation*}
    \pi(u^{-1}) C^{-1} \overline{\pi(u)} = - \pi(u)^{-1} \pi(u) C^{-1} = - C^{-1}
\end{equation*}
Now, as $S \in Spin_{1,3}$ is a product of an even number of such $u$'s we must have: $\pi(S^{-1}) C^{-1} \overline{\pi(S)} = C^{-1}$.
\end{proof}
\end{appendices}

\nocite{*}
\bibliographystyle{phd}
\bibliography{bibliografia}

\end{document}